\documentclass{amsart}

\usepackage[includeheadfoot,
            marginratio={1:1,2:3}, 
            width=412pt, 
            height=688pt,]{geometry}

\usepackage{amsmath}
\usepackage{amsfonts}
\usepackage{amssymb}
\usepackage{graphicx}
\usepackage{ulem}
\usepackage{longtable}
\usepackage{afterpage}
\usepackage{amsthm}
\usepackage{tikz}
\usetikzlibrary{decorations.markings}
\usetikzlibrary{shapes}
\usepackage{mathtools}
\usetikzlibrary{positioning}
\usetikzlibrary{cd}
\usepackage{mathrsfs}
\usepackage{tikz-qtree}
\usepackage{stmaryrd}
\usepackage{nicematrix}
\usepackage{array}
\usepackage{arydshln}
\usepackage{float}
\usepackage{xcolor}
\usepackage{enumitem}
\usepackage{caption}
\usepackage{hyperref}
\usepackage{txfonts}

\makeatletter
\renewcommand\section{\@startsection{section}{1}%
  \z@{.7\linespacing\@plus\linespacing}{.5\linespacing}%
 {\normalfont\bfseries\centering}}
\makeatother

\tikzset{->-/.style={decoration={
  markings,
  mark=at position .5 with {\arrow{>}}},postaction={decorate}}}
  
\newcommand{\eq}[1]{\begin{equation}
                     \begin{aligned} #1 \end{aligned}
                     \end{equation}}

\newcommand{\ov}{\overline}

\newcommand{\R}{\mathbb{R}}
\newcommand{\z}{\mathbb{Z}}

\newcommand{\lbr}{\left\lbrace}
\newcommand{\rbr}{\right\rbrace}

\newcommand{\im}{\mathrm{Im}}

\newcommand{\id}{\mathrm{Id}}

\newcommand{\p}{\partial}

\newcommand{\la}{\left\langle}
\newcommand{\ra}{\right\rangle}

\newcommand{\Hbb}{\mathbb{H}}
\newcommand{\Kbb}{\mathbb{K}}

\newcommand{\Cokerrm}{\mathrm{Coker}}
\newcommand{\Endrm}{\mathrm{End}}
\newcommand{\Graphrm}{\mathrm{Graph}}
\newcommand{\Imrm}{\mathrm{Im}}
\newcommand{\Kerrm}{\mathrm{Ker}}

\newcommand{\coevrm}{\mathrm{coev}}

\newcommand{\evrm}{\mathrm{ev}}

\newcommand{\ordrm}{\mathrm{ord}}
\newcommand{\orrm}{\mathrm{or}}

\newcommand{\trrm}{\mathrm{tr}}

\newcommand{\gcal}{\mathcal{g}}
\newcommand{\Bcal}{\mathcal{B}}

\newcommand{\Gcal}{\mathcal{G}}
\newcommand{\Hcal}{\mathcal{H}}

\newcommand{\Rcal}{\mathcal{R}}
\newcommand{\Scal}{\mathcal{S}}
\newcommand{\Ycal}{\mathcal{Y}}

\newcommand{\Blcal}{\mathcal{B}}
\newcommand{\Gcalcl}{\Gcal_{cl}}
\newcommand{\Gcalop}{\Gcal_{op}}
\newcommand{\Hcalcl}{\Hcal_{cl}}
\newcommand{\Hcalop}{\Hcal_{op}}

\newcommand{\dsf}{\mathsf{d}}

\newcommand{\Asf}{\mathsf{A}}
\newcommand{\Bsf}{\mathsf{B}}
\newcommand{\Csf}{\mathsf{C}}
\newcommand{\Dsf}{\mathsf{D}}

\newcommand{\Isf}{\mathsf{I}}

\newcommand{\Zsf}{\mathsf{Z}}

\newcommand{\Funsf}{\mathsf{Fun}}

\newcommand{\Vectsf}{\mathsf{Vect}}
\newcommand{\WSsf}{\mathsf{WS}}

\newcommand{\corr}{\mathsf{corr}}

\newcommand{\mcl}{m_{cl}}
\newcommand{\Dcl}{\Delta_{cl}}
\newcommand{\etacl}{\eta_{cl}}
\newcommand{\epcl}{\epsilon_{cl}}
\newcommand{\mop}{m_{op}}
\newcommand{\Dop}{\Delta_{op}}
\newcommand{\etaop}{\eta_{op}}
\newcommand{\epop}{\epsilon_{op}}
\newcommand{\iotaclop}{\iota_{cl-op}}

\renewcommand{\hom}{\mathrm{Hom}}

\allowdisplaybreaks[2]
\numberwithin{equation}{section}

\theoremstyle{definition}
\newtheorem{defn}{Definition}[section]
\theoremstyle{plain}
\newtheorem{theo}[defn]{Theorem}
\theoremstyle{plain}
\newtheorem{lem}[defn]{Lemma}
\theoremstyle{remark}

\theoremstyle{remark}

\theoremstyle{plain}
\newtheorem{Cor}[defn]{Corollary}
\theoremstyle{plain}
\newtheorem{prop}[defn]{Proposition}

\theoremstyle{plain}
\newtheorem*{theomain1}{Theorem I}
\theoremstyle{plain}
\newtheorem*{theomain2}{Theorem II}

\theoremstyle{definition}

\begin{document}

\begin{flushright}
  {\tiny
  MPP-2020-172\\
  }
  \end{flushright}
  
\title{Cardy Algebras, Sewing Constraints and String-Nets}

\vspace{2cm}
\date{\today}

\vspace{0.4cm}

\author{Matthias Traube}

\address{Matthias Traube: Max-Planck-Institut f\"ur Physik (Werner-Heisenberg-Institut) \\ 
   F\"ohringer Ring 6,  80805 M\"unchen, Germany } 
\email{mtraube@mpp.mpg.de}

\begin{abstract}
\noindent In \cite{Schweigert:2019zwt} it was shown how string-net spaces for the Cardy bulk algebra in the Drinfeld center $\Zsf(\Csf)$ of a modular tensor category $\Csf$ give rise to a consistent set of correlators. We extend their results to include open-closed world sheets and allow for more general field algebras, which come in the form of $(\Csf|\Zsf(\Csf))$-Cardy algebras. To be more precise, we show that a set of fundamental string-nets with input data from a $(\Csf|\Zsf(\Csf))$-Cardy algebra gives rise to a solution of the sewing constraints formulated in \cite{Kong_2014} and that any set of fundamental string-nets solving the sewing constraints determine a $(\Csf|\Zsf(\Csf))$-Cardy algebra up to isomorphism. Hence we give an alternative proof of the results in \cite{Kong_2014} in terms of string-nets.
\end{abstract}


\maketitle

\tableofcontents
\section{Introduction}
String-net spaces were originally introduced by Levin and Wen in \cite{Levin:2004mi} in order to describe phenomena of topological phases of matter on surfaces. Roughly speaking a string-net is an equivalence class of an embedded graph on a surface $S$ with boundary. Based on earlier works \cite{Koenig_2010}\cite{Kadar:2009fs} a precise mathematical formulation of string-nets was given by Kirillov in \cite{2011arXiv1106.6033K} and it can be seen as a higher genus enhancement of the graphical calculus for tensor categories where the usual relations hold on any embedded disk. Similar discussions of categories on surfaces have appeared in \cite{HARDIMAN_2019}\cite{hardiman2019graphical} for the case of a cylinder. In a series of papers \cite{KirillovBalsam}\cite{balsam}\cite{balsam2010turaevviro}\cite{2011arXiv1106.6033K} it was shown how an appropriately defined vector space of formal linear combinations of string-nets is equivalent to the state space of the Turaev-Viro three dimensional topological field theory (TFT) based on $\Csf$ or equivalently to the state space of the Reshetikhin-Turaev (TFT) for $\Zsf(\Csf)$. 

Two dimensional rational conformal field theory (RCFT) on the other hand has a categorical description in the form of the FRS (Fuchs-Runkel-Schweigert)-formalism developed in \cite{Fuchs:2002cm} \cite{Fjelstad:2005ua} \cite{Fuchs:2004xi} \cite{Fuchs:2004dz} \cite{Fuchs:2003id}, where the state space of the Reshetikhin-Turaev TFT features prominently. The monodromy data of an RCFT is described by a modular tensor category and the bulk field algebra is a Frobenius algebra $\Hcalcl$ in $\Zsf(\Csf)$. The space of closed conformal blocks on a compact surface $S^g_{n|m}$ of genus $g$ with $n$ incoming and $m$ outgoing boundaries is given by the state-space of the Reshetikhin-Turaev TFT on $S^g_{n|m}$, $Z_{RT,\Zsf(\Csf)}(\Hcalcl^\ast,\dots,\Hcalcl^\ast,\Hcalcl,\dots, \Hcalcl)$, with $n$ copies of $\Hcalcl^\ast$ coloring incoming boundary components and $m$ copies of $\Hcalcl$ coloring outgoing boundary components. A consistent set of correlators in the RCFT is an assignment of an element in $Z_{RT,\Zsf(\Csf)}(\Hcalcl^\ast,\dots,\Hcalcl^\ast,\Hcalcl,\dots, \Hcalcl)$ for all surfaces $S^g_{n|m}$ s.th. the vectors are invariant under the action of the mapping class group and behave equivariant under sewings of surfaces. That is, there should exist a map for sewing correlators and sewn correlators should agree with the correlator on the sewn surface. The equivalence between string-net spaces and spaces of conformal blocks was used in \cite{Schweigert:2019zwt} to show that a given set of genus zero string-nets with boundary value given by the Cardy-Frobenius algebra in $\Zsf(\Csf)$ indeed give rise to consistent correlators. 

In this paper we generalize their result in two directions: Firstly we allow for world sheets with open and closed field insertions. This leads to a formulation of sewing constraints in terms of the category $\WSsf$ of open-closed world sheets given in \cite{Fjelstad:2005ua}. String-nets on open-closed world sheets still give a symmetric monoidal functor $\Blcal:\WSsf\rightarrow \Vectsf$. Solving sewing constraints for such a symmetric monoidal functor was reduced in \cite{Kong_2014} to a set of 32 relations, which have to be satisfied. The second generalization concerns the input data for the construction. We allow arbitrary $(\Csf|\Zsf(\Csf))$-Cardy algebras as inputs. Based on an operadic formulation for vertex operator algebras (VOA), Cardy algebras were introduced in \cite{Kong_2008} and formulated entirely in terms of category theory in \cite{Kong_2009}. A $(\Csf|\Zsf(\Csf))$-Cardy algebra encodes the necessary data for an open-closed RCFT in genus one and zero and consists of a triple $(\Hcalcl,\Hcalop,\iota_{clop})$, where $\Hcalop$, $\Hcalcl$ are Frobenius algebras in $\Csf$ and $\Zsf(\Csf)$ respectively, together with an algebra map $\iotaclop:\Hcalcl\rightarrow L(\Hcalop)$. Here $L:\Csf\rightarrow \Zsf(\Csf)$ is the adjoint functor to the forgetful functor. This data has to satisfy three conditions: modularity, the center property and finally the Cardy condition. In \cite{Kong_2014} it was shown, using the Reshetikhin-Turaev TFT, that Cardy algebras give rise to a consistent set of correlators and vice versa. The category $\WSsf$ is generated by a set 
\eq{
\lbr O_m,O_\Delta,O_\eta,O_\epsilon,O_{prop},C_m,C_\Delta,C_\eta,C_\epsilon,C_{prop},I, I	^\dagger \rbr 
}
of fundamental world sheets. With the help of the $(\Csf|\Zsf(\Csf))$-Cardy algebra we define correlators 
\eq{
\corr=\lbr \corr_{prop}^{op},\corr_m^{op}\corr_\Delta^{op},\corr_\eta^{op},\corr^{op}_\epsilon,\corr^{cl}_{prop},\corr_m^{cl},\corr^{cl}_\Delta,\corr_\eta^{cl},\corr_\epsilon^{cl},\corr_I,\corr_{I^\dagger}\rbr 
}
in terms of string-nets and the first main result, whose precise formulation in the main text is Theorem \ref{maintheo1}, is
\begin{theomain1} The set of correlators $\corr$ gives a solution to the sewing constraints for the conformal block functor $\Blcal$ with boundary coloring determined by the $(\Csf|\Zsf(\Csf))$-Cardy algebra.
\end{theomain1}

Assuming a solution to the sewing constraints for $\Blcal$ exists, we also show that the converse is true, which is our second main result (for the precise statement see Theorem \ref{maintheo3} in the text).

\begin{theomain2} An assignment of fundamental string-net correlators based on boundary colorings $(\widehat{\Hcalcl},\widehat{\Hcalop})$ in $\Csf$, and $\Zsf(\Csf)$, solving the sewing contraints, determines a $(\Csf|\Zsf(\Csf))$-Cardy algebra $(\Hcalcl,\Hcalop,\iotaclop)$, which is unique up to isomorphism.
\end{theomain2}

The proofs of the theorems very much use the fact that the graphical calculus in $\Csf$ carries over to string-nets on surfaces. Hence the graphical representation of consistency relations for Cardy algebras appear directly as string-nets on surfaces and can be manipulated accordingly. This renders the proof very tractable. As an example, figure \ref{randomcorrel} shows a correlator world sheet. Purple curves denote open insertions, orange ones correspond to closed insertion. Vertices stand for structure morphisms of the Frobenius algberas in The Cardy algebras and squares denote open-closed embedding maps. In addition there are circle decorations where world sheets are glued. All of these ingredients will be explained in the core of the text.

\begin{figure}
\centering
\includegraphics[scale=0.15]{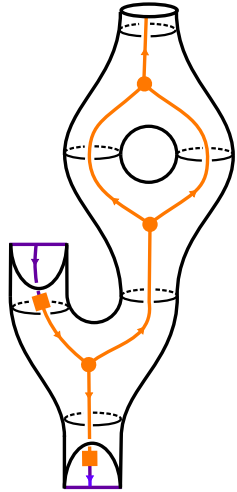}
\caption{Example of an open-closed correlator on a genus 1 surface with two open insertions and one closed insertion. The colored graph runs on the front of the surface.}
\label{randomcorrel}
\end{figure}

The paper is organized as follows. In section \ref{sec1} we give an account of the necessary categorical preliminaries, including modular tensor categories, the Drinfeld center and Frobenius algebras in tensor categories. In section \ref{sec2} we motivate Cardy algebras and recall their definition from \cite{Kong_2009}. In section \ref{sec3} string-net spaces are discussed based on \cite{2011arXiv1106.6033K}. In all of the paper we abusively use the term string-net to refer to an whole equivalence class of string-nets. The definition of the category of world sheets $\WSsf$ from \cite{Fjelstad:2005ua} and the formulation of sewing constraints given in \cite{Kong_2014} is recalled in section \ref{sec4}. Section \ref{sec5} is the main part of the paper and contains the precise formulation of the above theorems, as well as their proofs. In the appendix we display the generating world sheets.

\section{Categorical Preliminaries}\label{sec1}

As stated in the introduction two dimensional (rational) conformal field theory can be conveniently described in categorical terms. This section serves as a reminder of the necessary terms, together with the relevant graphical calculus. A classical source for the material presented here is e.g. \cite{etingof2016tensor}. In all about to come, $\mathbb{K}$ is an algebraically closed field of characteristic $0$.

\subsection{Basic Categorical Identities}

Let $\Csf$ be an abelian $\Kbb$-linear category, i.e. for any morphism $\phi$, $\Kerrm(\phi)$, $\Cokerrm(\phi)$ and $\Imrm(\phi)$ exist and moreover for $A$, $B\, \in \Csf$, $\hom_\Csf(A,B)$ is a $\Kbb$-vector space. A monoidal structure  on $\Csf$ is a bilinear bifunctor $\otimes:\Csf\times \Csf\rightarrow \Csf$ with associativity and unit natural isomorphisms which are assumed to be identities in this paper. Thus we always consider strictly monoidal categories. The unit object for $\otimes$ is denoted by $\mathbf{1}$. A \textit{braiding} on $(\Csf,\otimes)$ is a natural isomorphism $\beta_{A,B}:A\otimes B\rightarrow B\otimes A$. In the graphical calculus about to be introduced, all diagrams run from bottom to top. Graphically $\beta_{A,B}$ is depicted by

\begin{figure}[H]
\centering
\begin{minipage}{0.45\linewidth}
\centering
\includegraphics[scale=0.1]{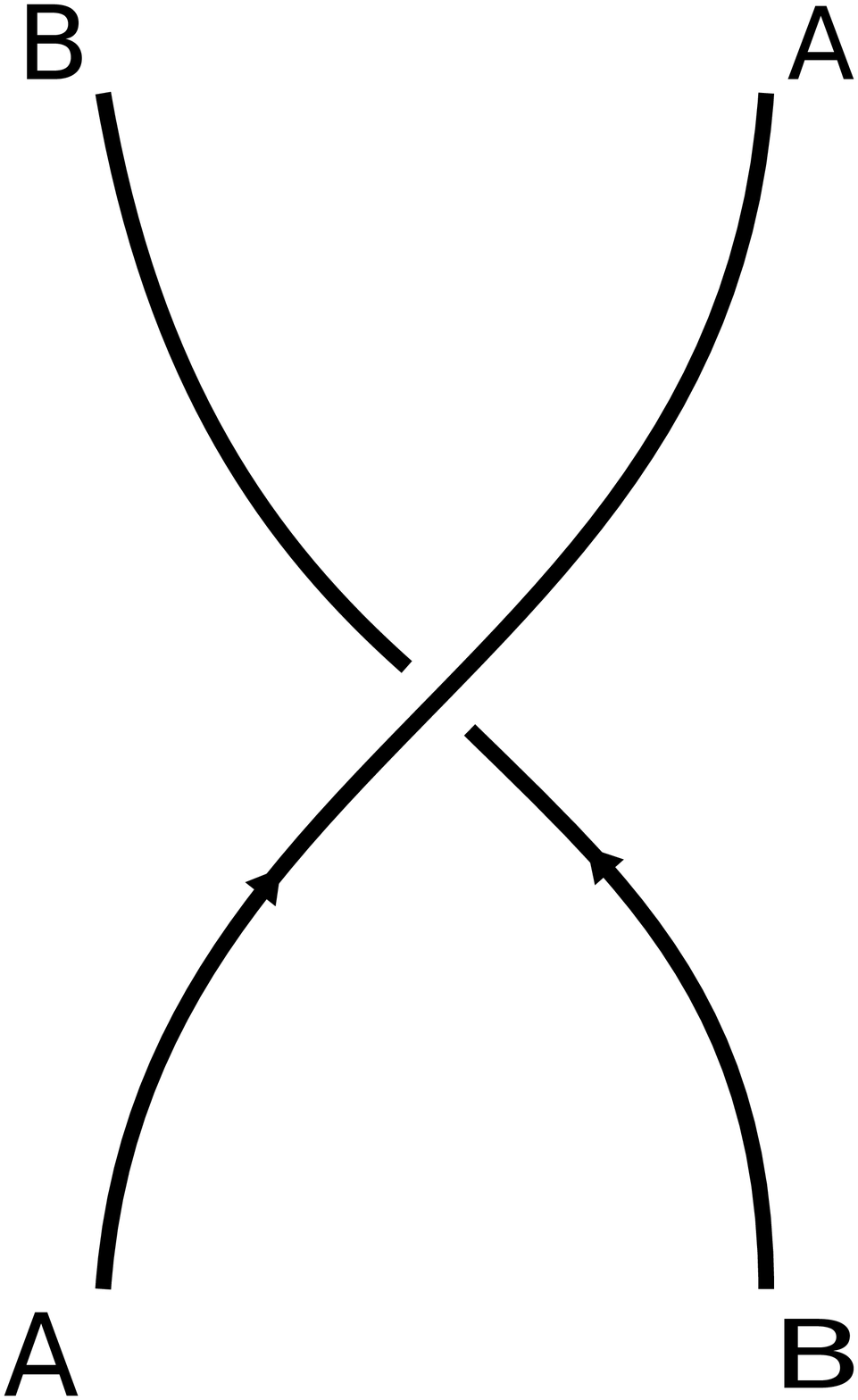}
\caption{$\beta_{A,B}$}
\end{minipage}
\begin{minipage}{0.45\linewidth}
\centering
\includegraphics[scale=0.1]{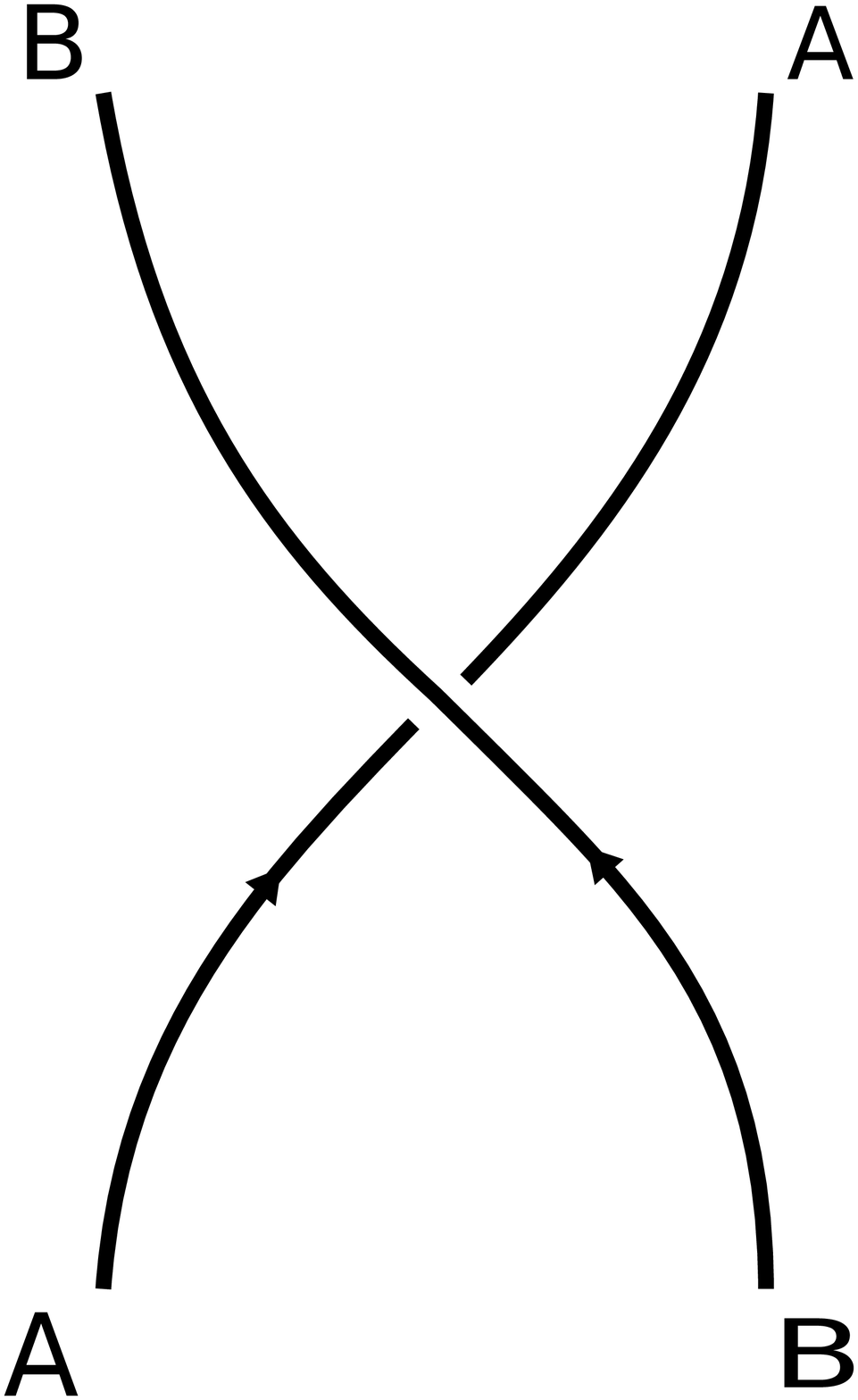}
\caption{$\beta^{-1}_{B,A}$.}
\end{minipage}
\end{figure}

Besides a braiding, $\Csf$ is required to have dualities. For $A\in \Csf$ a \textit{right (left) dual} is a triple $(A^\ast, \evrm_A,\coevrm_A)$ $(^\ast A,\widetilde{\evrm_A},\widetilde{\coevrm_A})$ of an object $A^\ast$ $(^\ast A)$ with morphisms
\begin{figure}[H]
\centering
\begin{tabular}{c @{\hspace{1cm}} c @{\hspace{1cm}}c @{\hspace{1cm}} c @{\hspace{1cm}}} 
\includegraphics[scale=0.1]{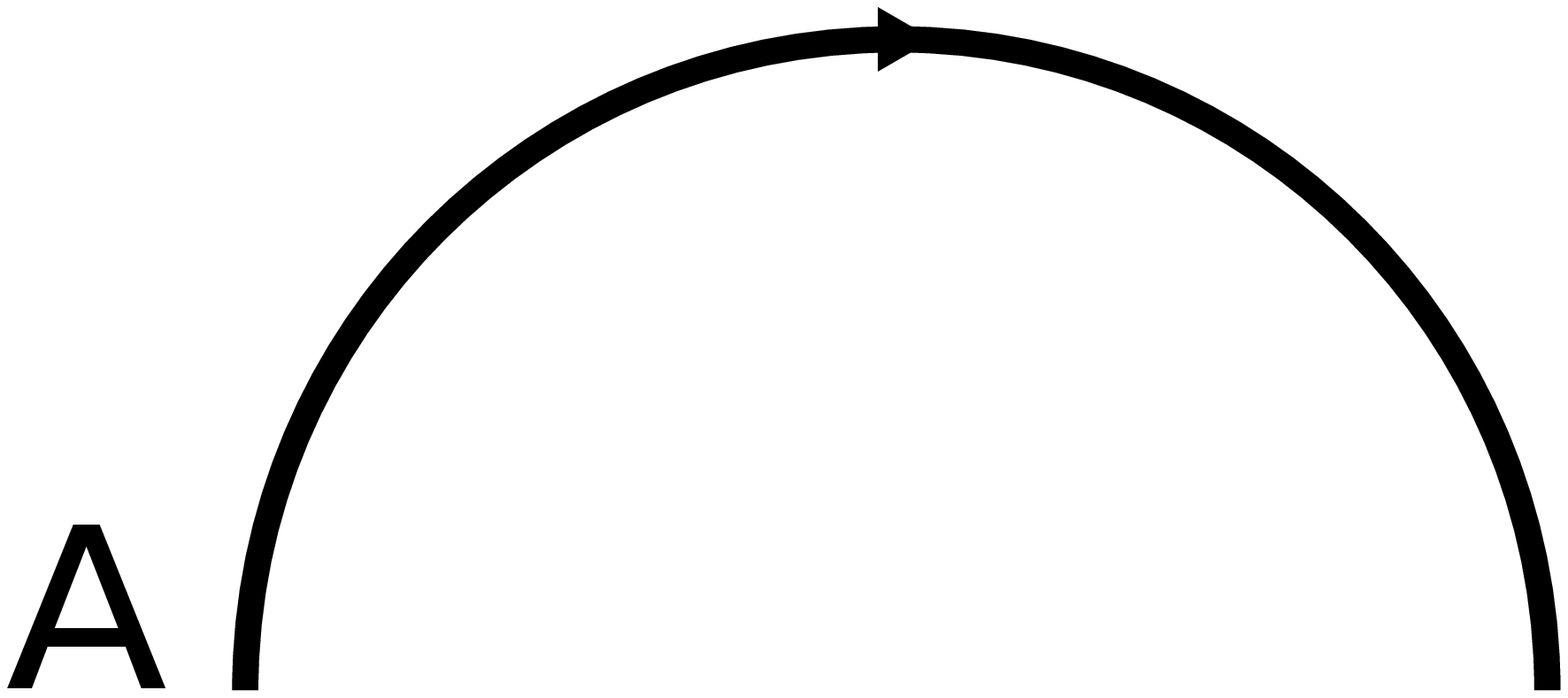} & \includegraphics[scale=0.1]{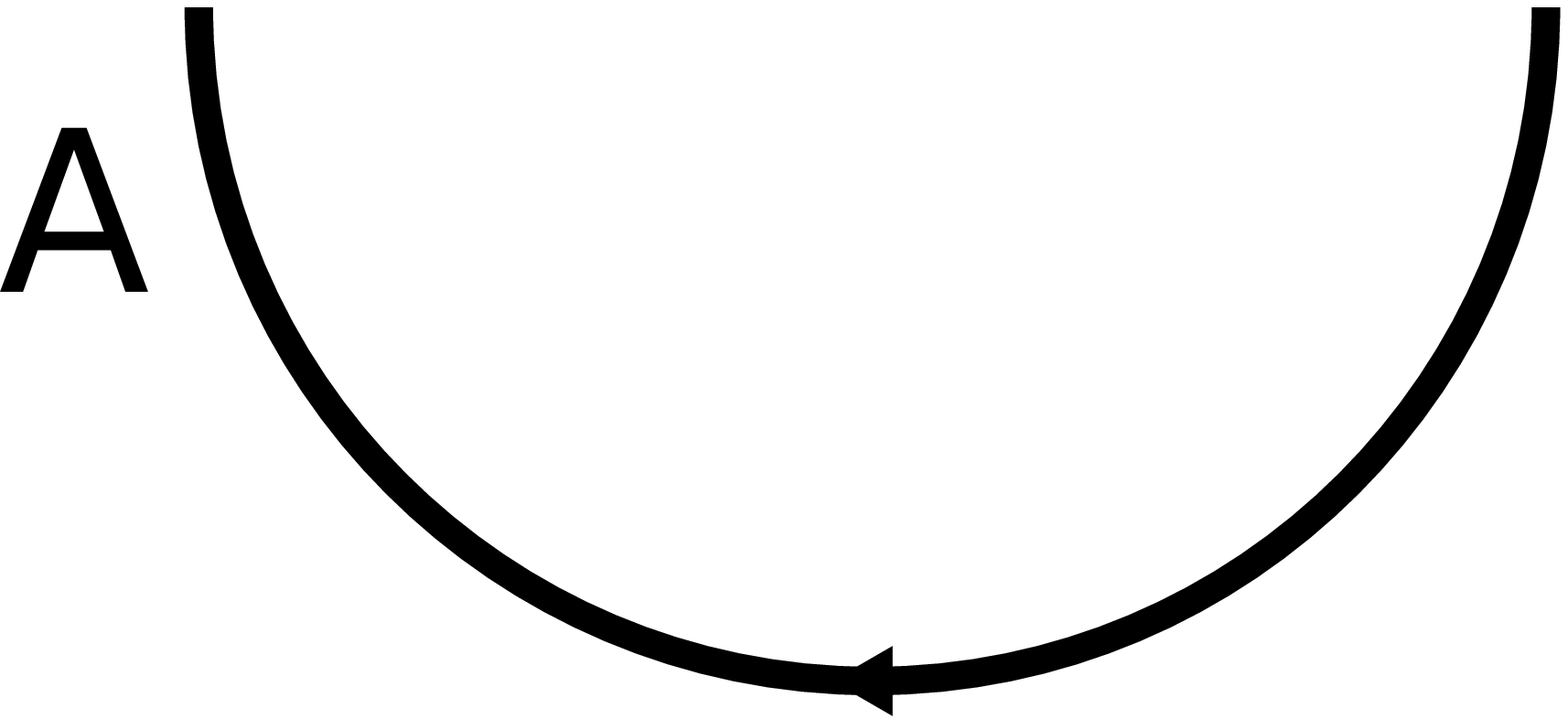} & \includegraphics[scale=0.1]{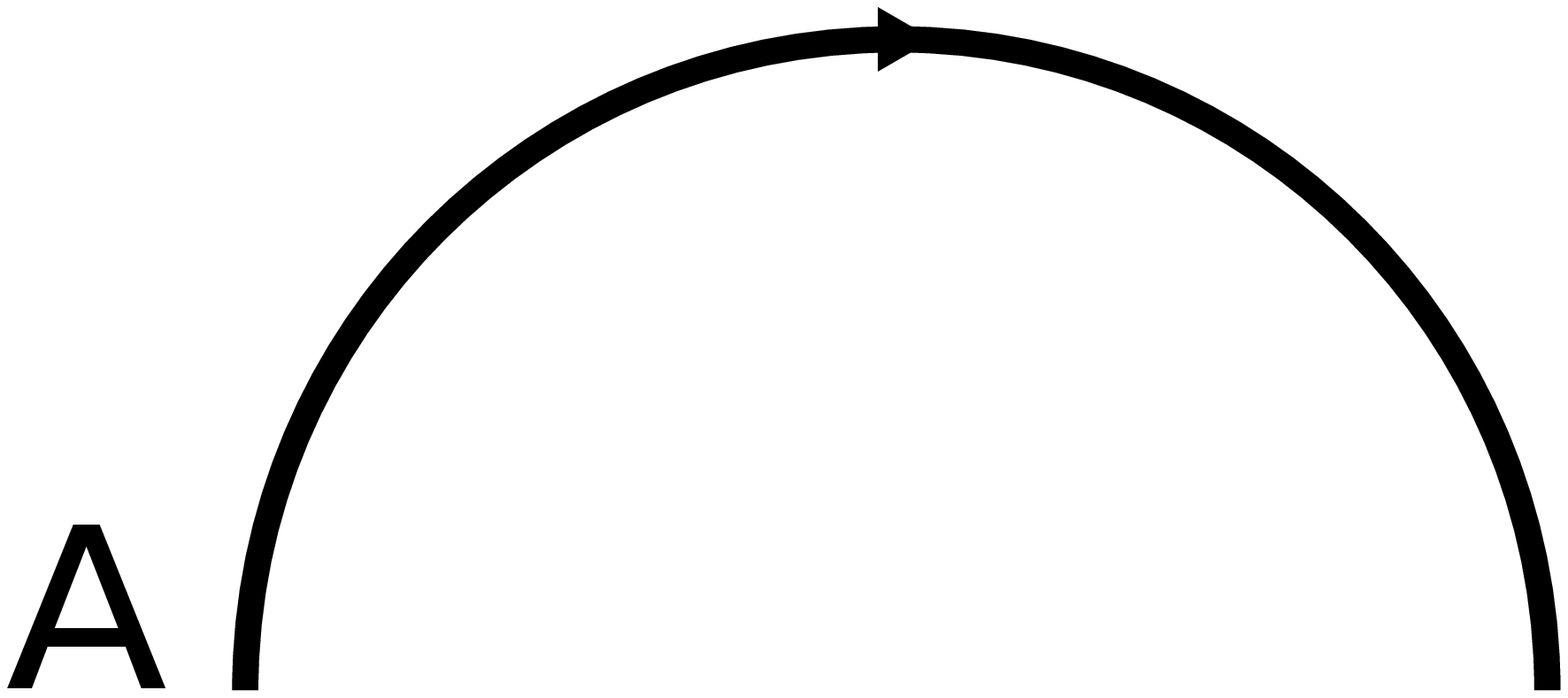} & \includegraphics[scale=0.1]{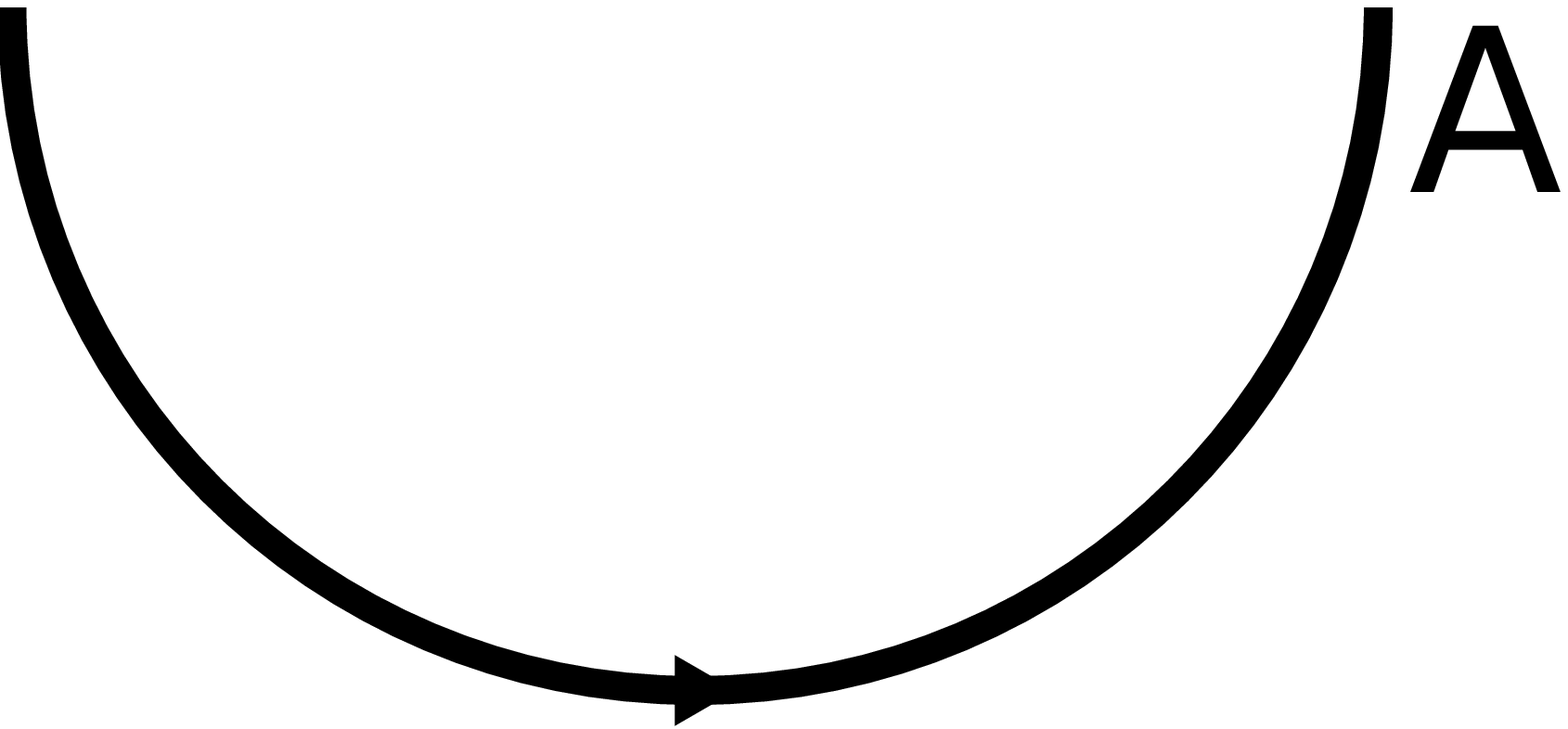}\\
$\evrm_A$ & $\coevrm_A$ & $\widetilde{\evrm_A}$ & $\widetilde{\coevrm_A}$
\end{tabular}.
\end{figure}

The morphisms satisfy straightening properties
\begin{figure}[H]
\centering
\begin{tabular}{c @{\hspace{2cm}} c}
\includegraphics[scale=0.12]{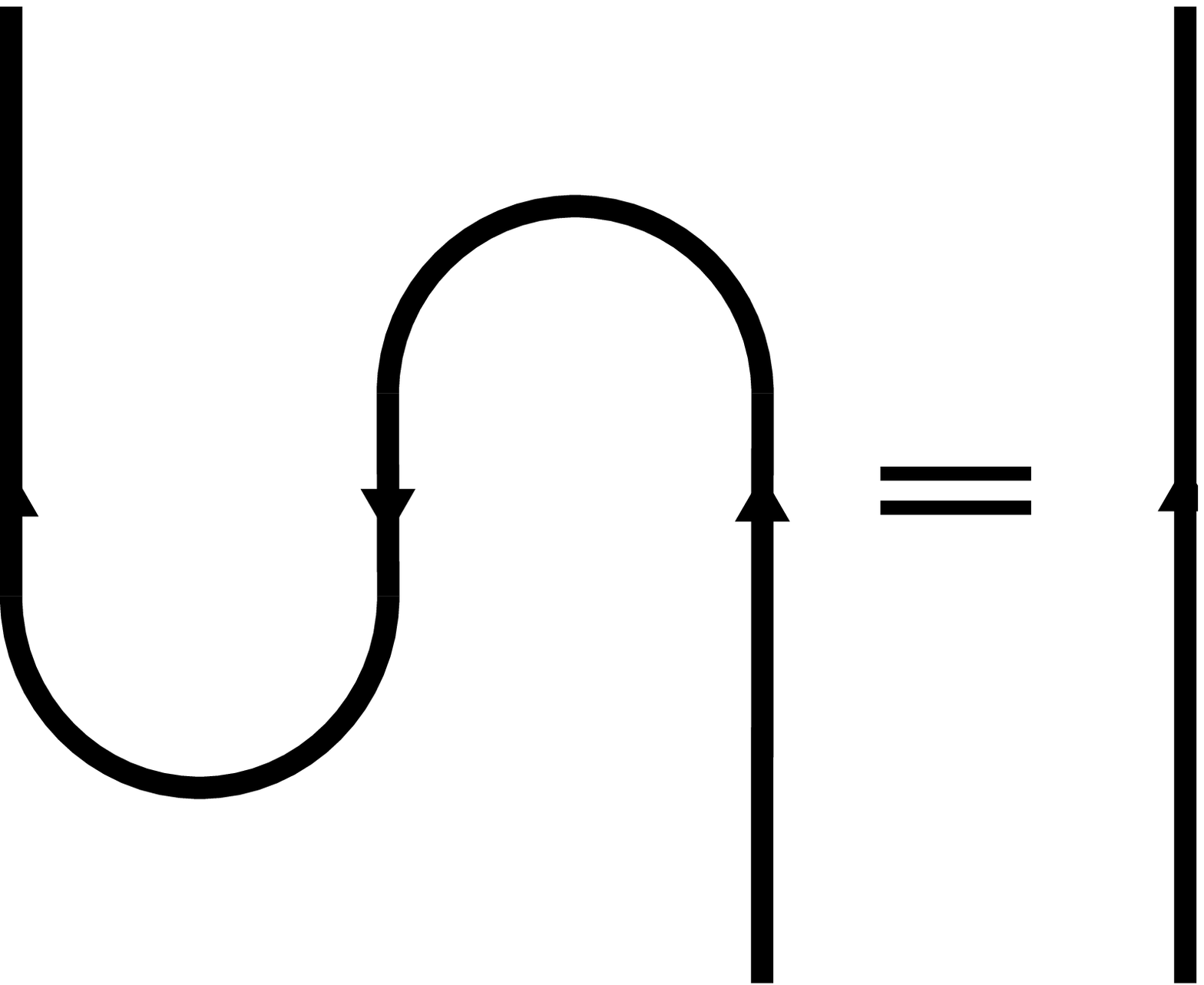} & \includegraphics[scale=0.12]{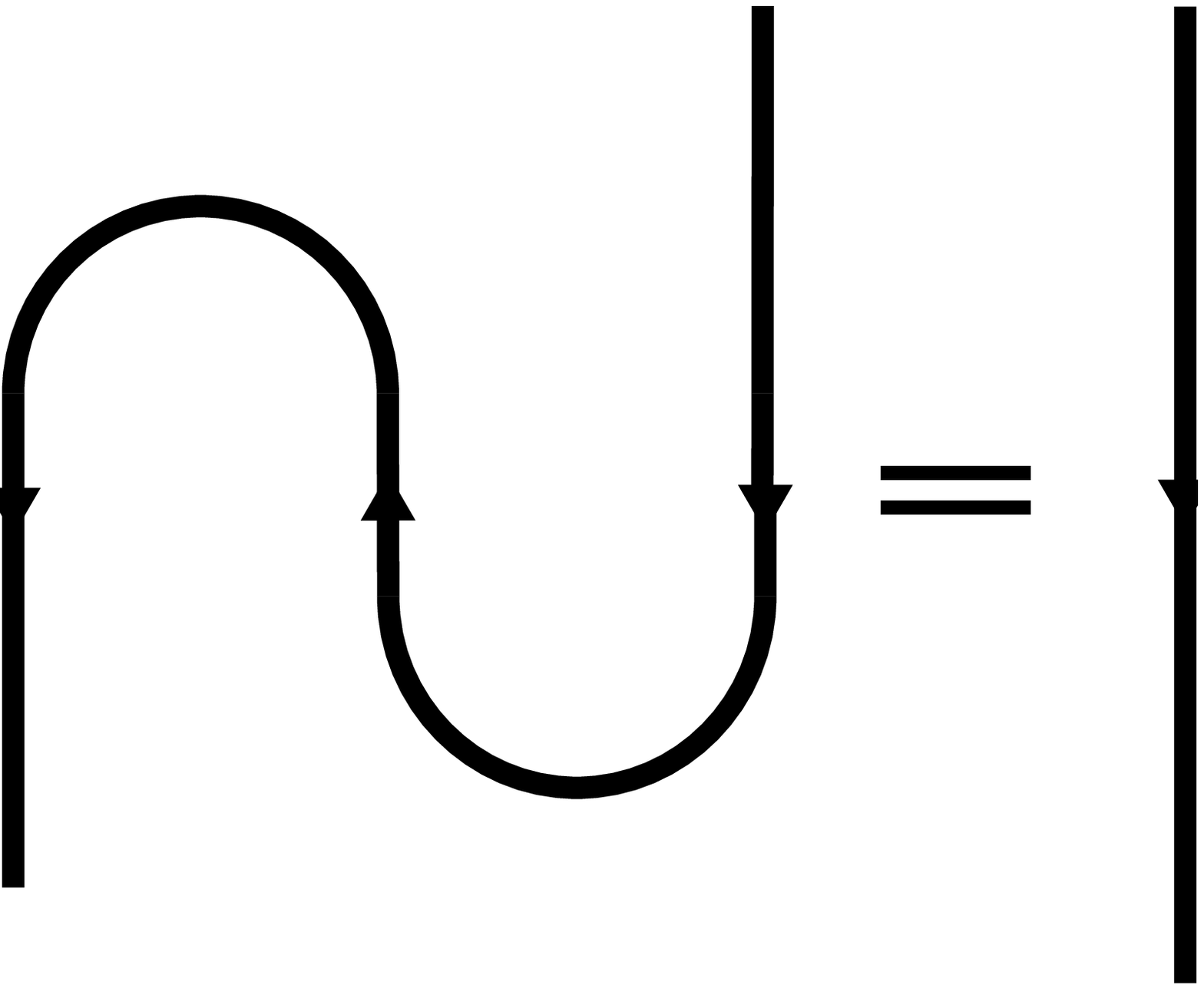}
\end{tabular}
\end{figure}
with similar pictures for left duality. A monoidal category for which every object has right and left duals is called \textit{rigid}. In a rigid category there is the obvious isomorphism $^\ast\left(A^\ast\right)\simeq A \simeq \left(^\ast A\right)^\ast$. A \textit{pivotal structure} in $(\Csf,\otimes)$ is a natural isomorphism $\pi:\id_\Csf\rightarrow \left(\bullet\right)^{\ast\ast}$. In fact any pivotal category is equivalent to a strict pivotal category, i.e. $\pi=\id_{\id_\Csf}$. It easily follows that for a strict pivotal category left and right dualities coincide. In a rigid, strictly pivotal category there are left and right traces for endomorphisms

\begin{figure}[H]
\centering
\begin{tabular}{c @{\hspace{2cm}} c}
\includegraphics[scale=0.1]{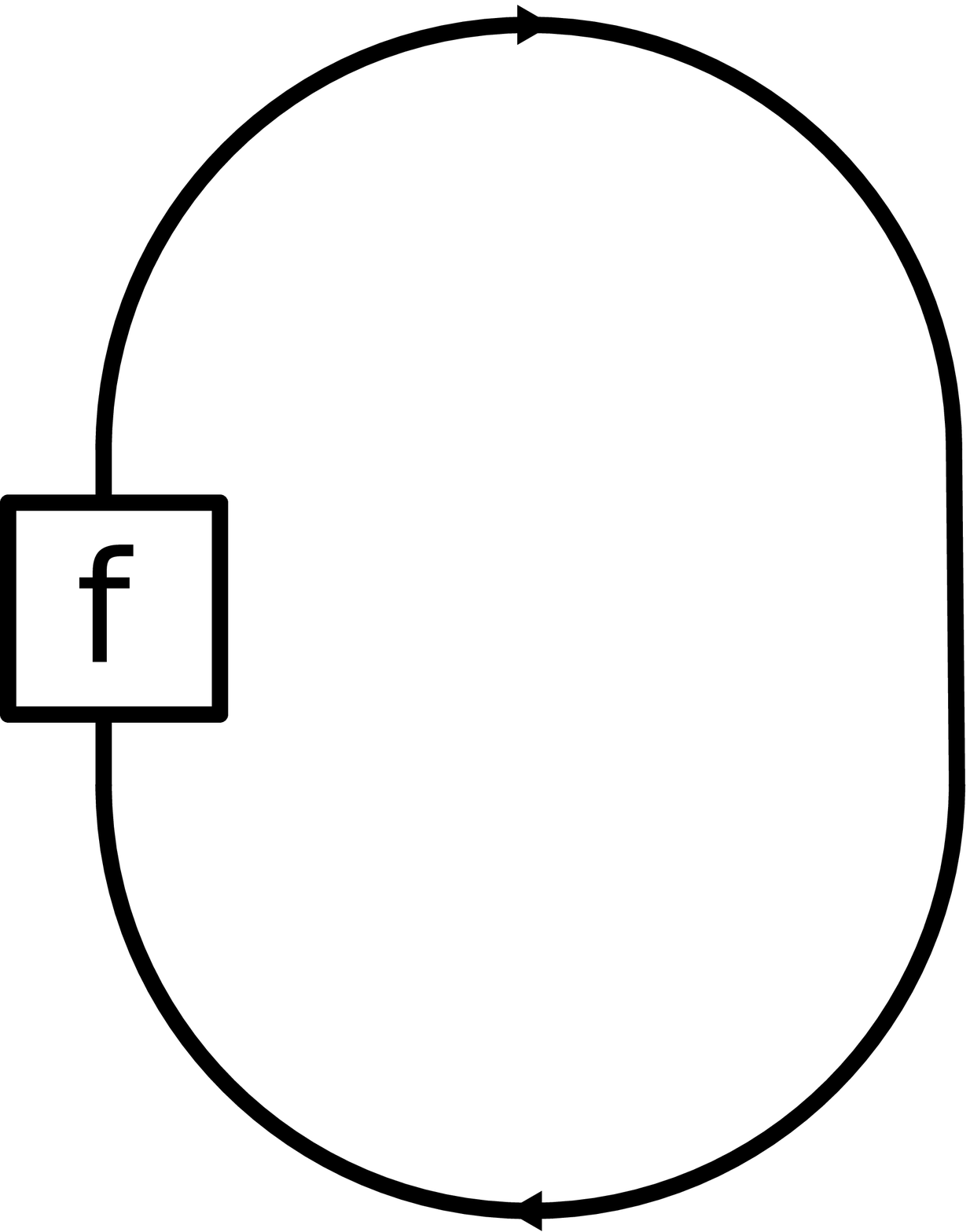} & \includegraphics[scale=0.1]{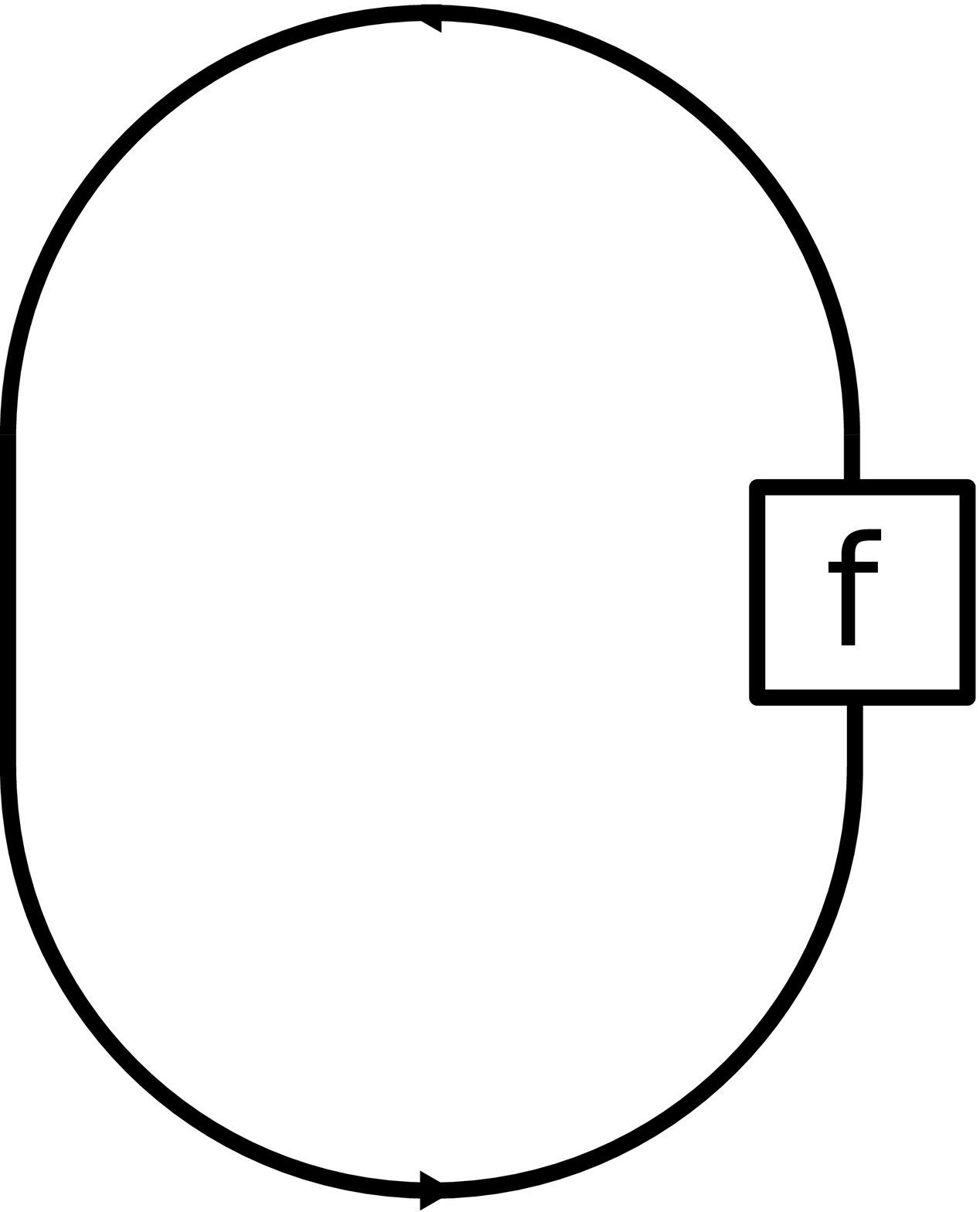}\\
$\trrm_r(f)$ & $\trrm_l(f)$.
\end{tabular}
\end{figure}
If $\trrm(f)=\trrm_r(f)=\trrm_l(f)$ holds in $\Csf$, $\Csf$ is called \textit{spherical}. We introduce the notation $d_A=\trrm(\id_A)$. A \textit{twist} on $(\Csf,\otimes,\beta)$ is a natural isomorphism $\theta:\id_\Csf\Rightarrow \id_\Csf$ satisfying $\theta_{A\otimes B}=\left(\theta_A\otimes \theta_B\right)\circ \beta_{B,A}\circ \beta_{A,B}$. The twist isomorphism is depicted as 
\begin{figure}[H]
\includegraphics[scale=0.12]{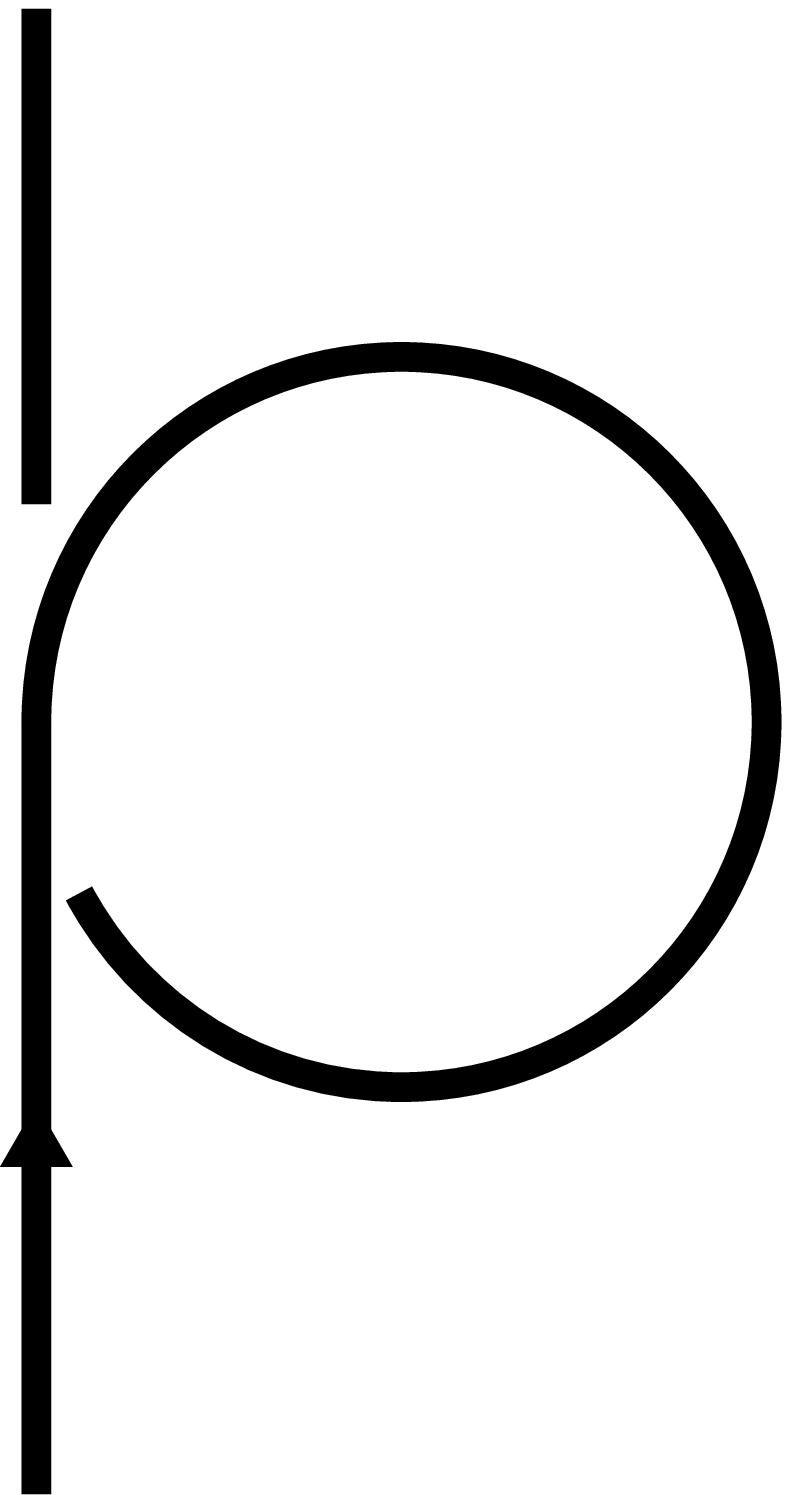}.
\end{figure}
A rigid, pivotal, braided tensor category is \textit{ribbon} if it has twist. Having defined all the necessary structure on $\Csf$ we also want to control its size. Thus we require morphism vector spaces to be finite dimensional. Recall that an object $A\in \Csf$ is \textit{simple} if $\hom_\Csf(A,A)\simeq \Kbb\, \id_\Csf$. If the set of isomorphism classes of simple objects is finite and every object is isomorphic to a direct sum of finitely many simple objects, the category is called \textit{semisimple}. A ribbon category $\Csf$ is \textit{fusion} if it is semisimple and $\mathbf{1}$ is simple. Let $\lbr U_i \rbr_{i\in \Isf(\Csf)}$, $U_i\in \Csf$, be the simple objects of $\Csf$. In string-diagrams, a label by a simple object $U_i$ will be abbreviated by just labeling the edge with $i$. Being semisimple has far reaching consequences, e.g. for any object there are maps $b_i^\alpha:A\overset{\simeq}{\rightarrow} \bigoplus_{i\in \Isf(\Csf)}U_i^{\oplus n_i}\rightarrow U_i$, with $\alpha\in \lbr 1,\dots, n_i\rbr $, where the second map is the projection to the $\alpha$'s $U_i$ summand and a dual map $b^j_\beta:U_i\rightarrow \bigoplus_{i\in \Isf(\Csf)} U_i^{\oplus n_i}\overset{\simeq}{\rightarrow }A$. These maps are dual in the following sense
\eq{
\sum_{i\in \Isf(\Csf)}\sum_{\alpha=1}^{n_i} b^i_\alpha\circ b_i^\alpha=\id_A,\qquad b_i^\alpha\circ b^j_\beta=\delta_{ij}\delta_{\alpha\beta} \id_{U_i}\quad .
}
Graphically we represent the duality as
\begin{figure}[H]
\centering
\includegraphics[scale=0.18]{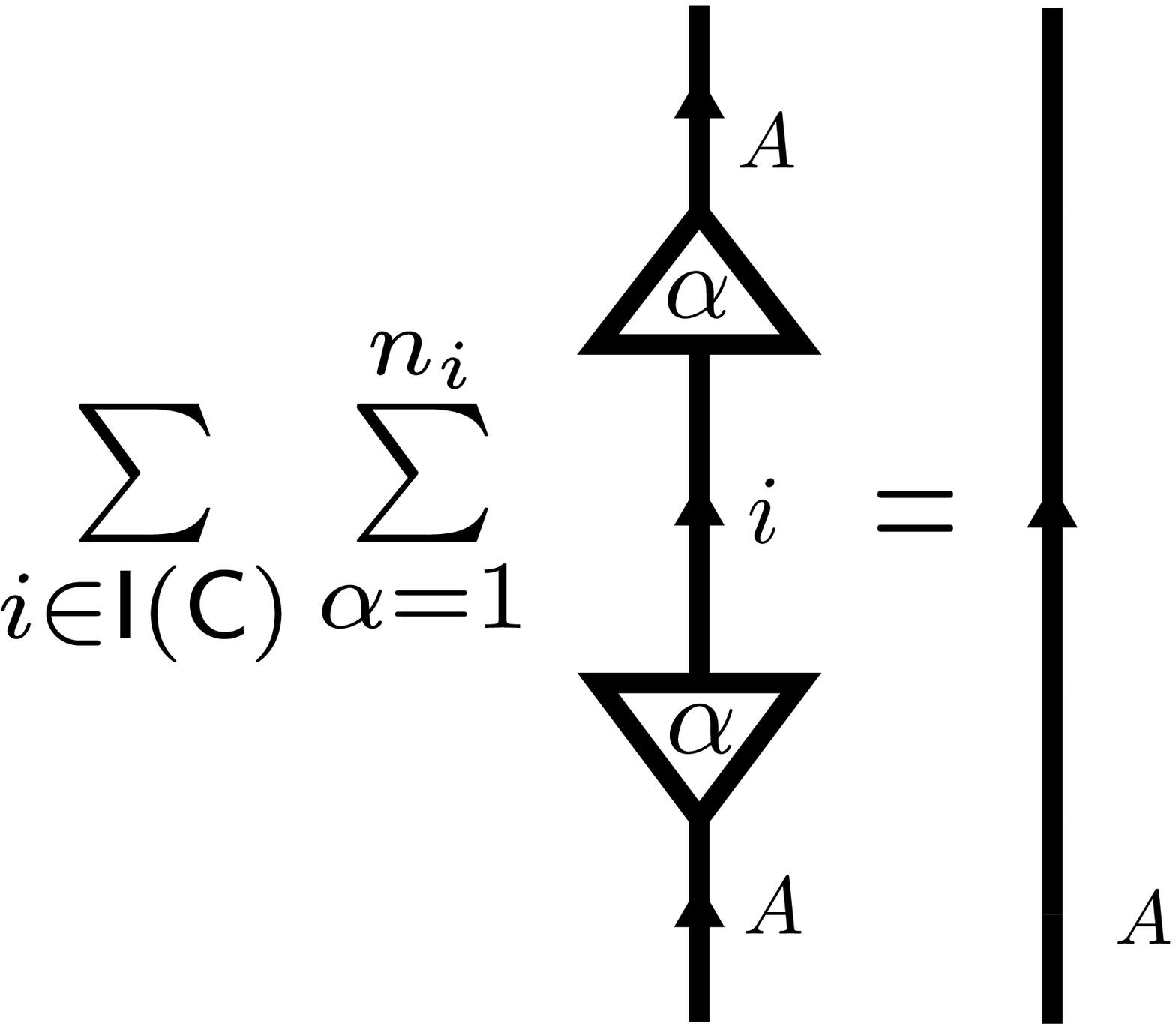}
\end{figure}
and 
\begin{figure}[H]
\centering
\includegraphics[scale=0.18]{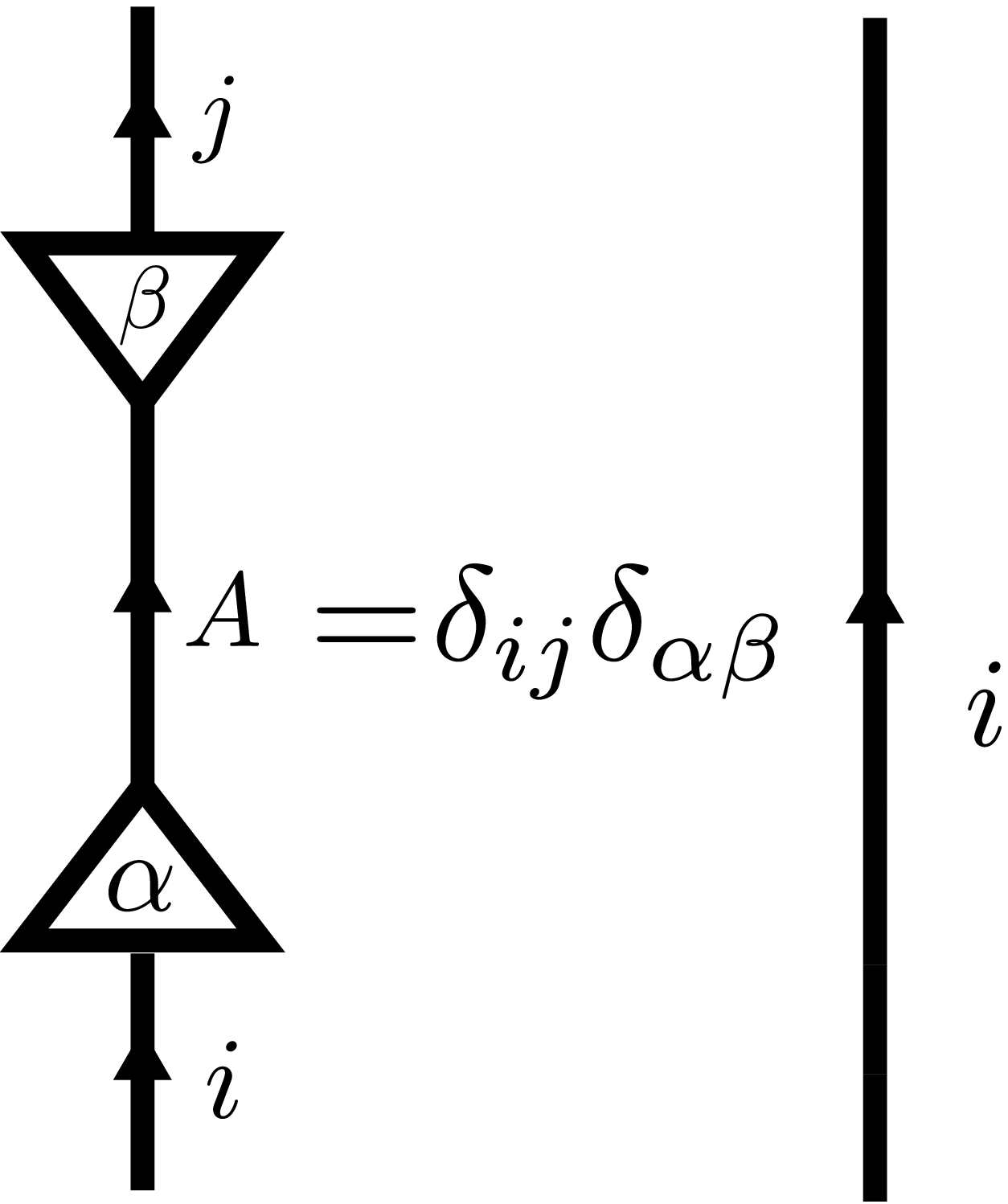}.
\end{figure}
In addition, for a semisimple category, $d_i=\trrm\left(U_i\right)\neq 0$, and the \textit{global dimension} is defined as $\Dsf^2=\sum_{i\in \Isf(\Csf)} d_i$. Besides that, we introduce a basis $\lbr \theta_{(ij);k}^\alpha\rbr_{\alpha=1,\dots, N^k_{ij}}$ for $\hom(U_i\otimes U_j,U_k)$, where $N_{ij}^k=\dim \hom(U_i\otimes U_j,U_k)$ are the fusion coefficients. There is a dual basis $\lbr \theta^{k;(ij)}_{\beta}\rbr _{\beta=1,\dots , N_{ij}^k}$, s.th. in graphical notation we have 

\begin{figure}[H]
\centering
\includegraphics[scale=0.1]{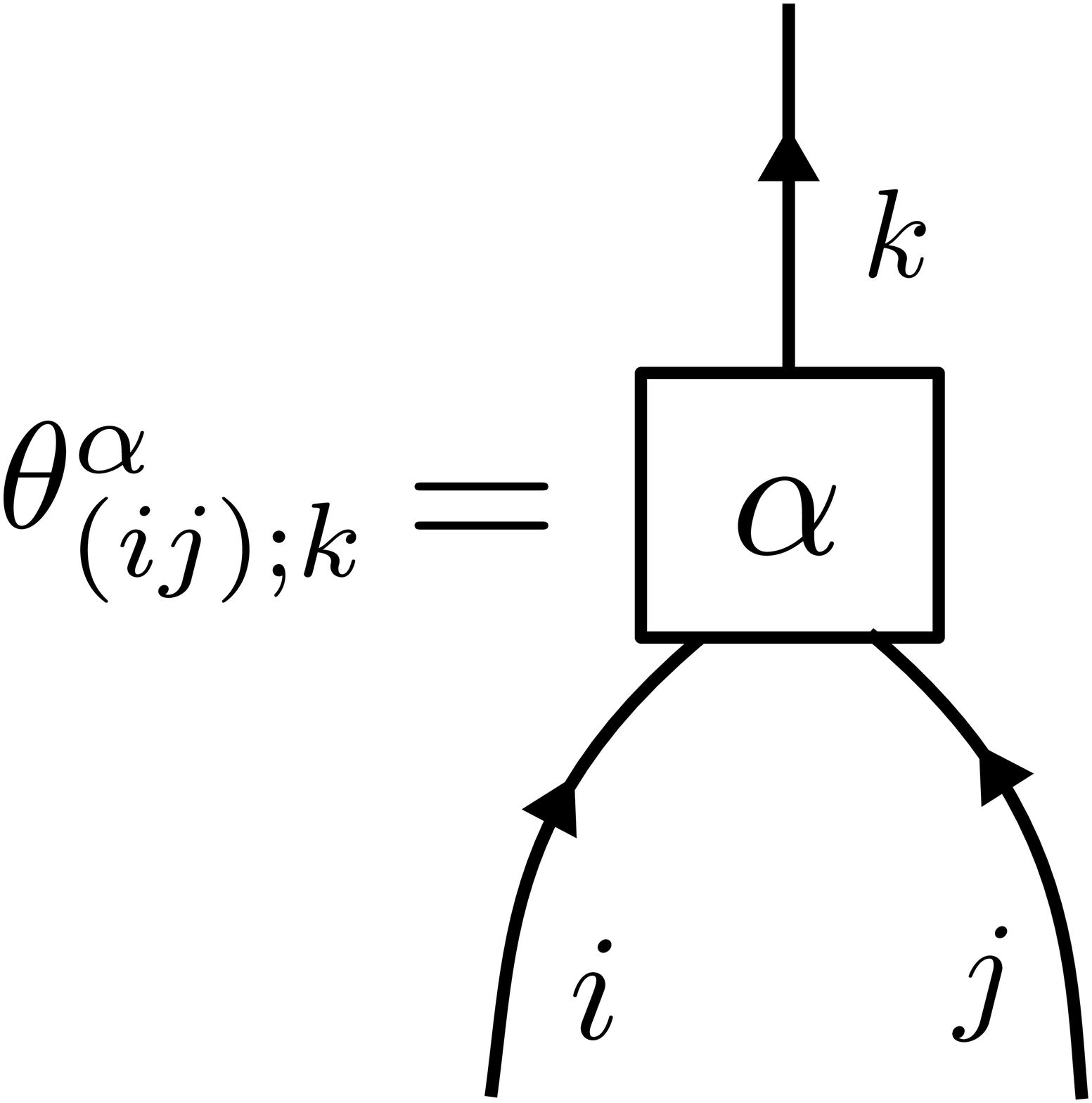}\hspace*{3cm}
\includegraphics[scale=0.1]{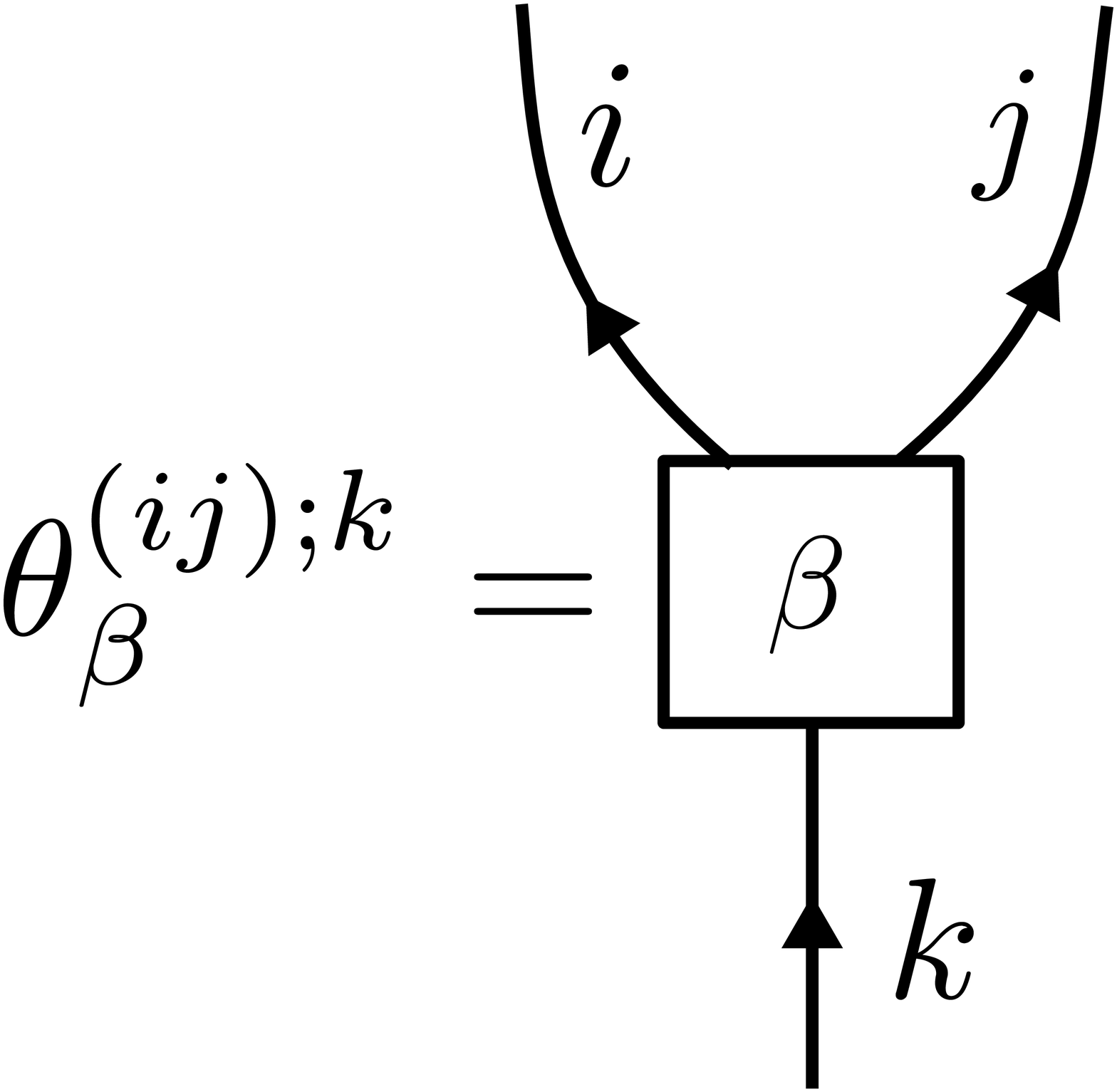}
\end{figure}
and

\begin{figure}[H]
\centering
\includegraphics[scale=0.1]{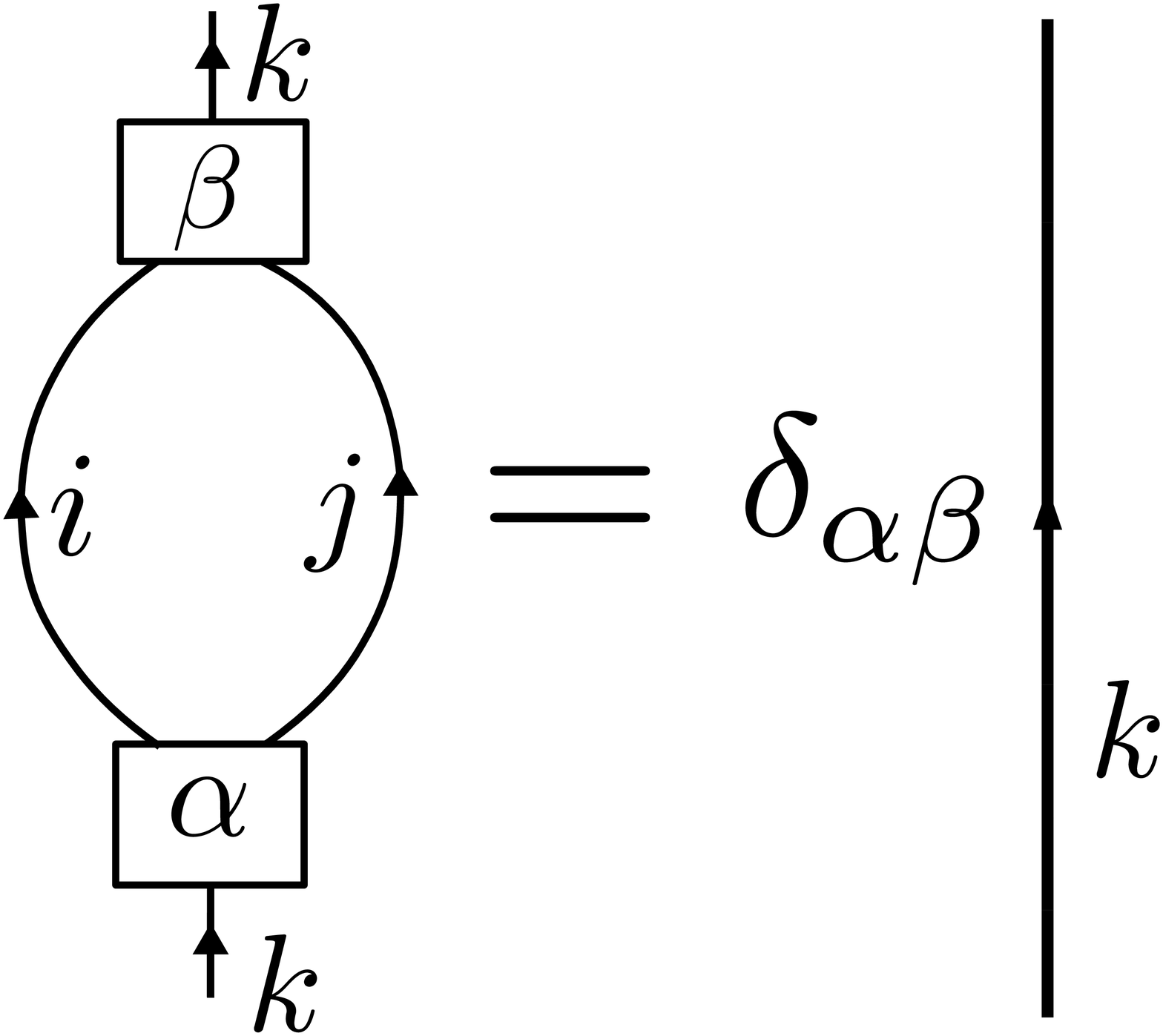}.
\end{figure}

One more ingredient for a modular tensor is needed, namely the $S$-isomorphism
\begin{figure}[H]
\centering
\includegraphics[scale=0.15]{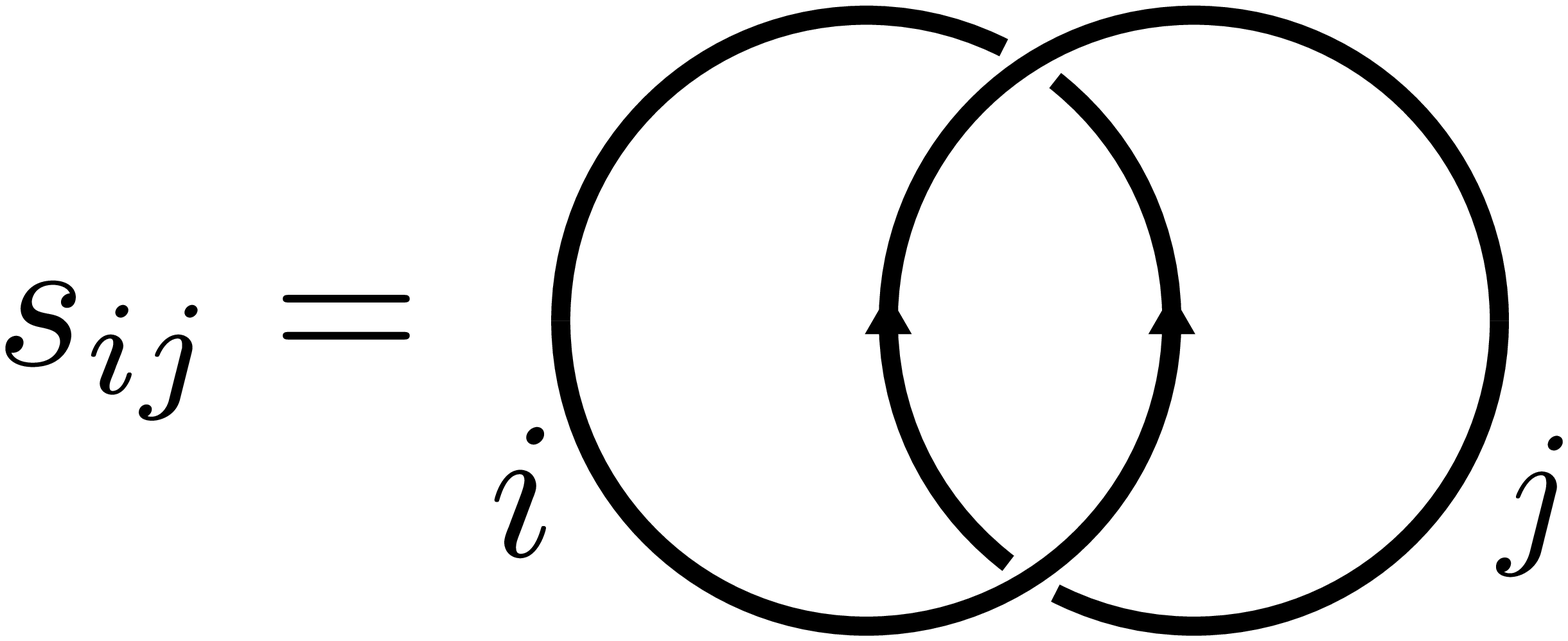}.
\end{figure}

\begin{defn} A \textit{modular tensor category} is a spherical, ribbon fusion category, s.th. $(s_{ij})$ is an invertible $|I|\times |I|$-matrix.
\end{defn}

We state one more relation, which will be used in the proof of theorem \ref{maintheo1}. A proof can be found in e.g. \cite[Corollary~3.1.11]{BakalovKirillov}.
\begin{lem}\label{circleline}
For $\Csf$ a modular tensor category and $l\in \Isf(\Csf)$ it holds
\begin{figure}[H]
\centering
\includegraphics[scale=0.15]{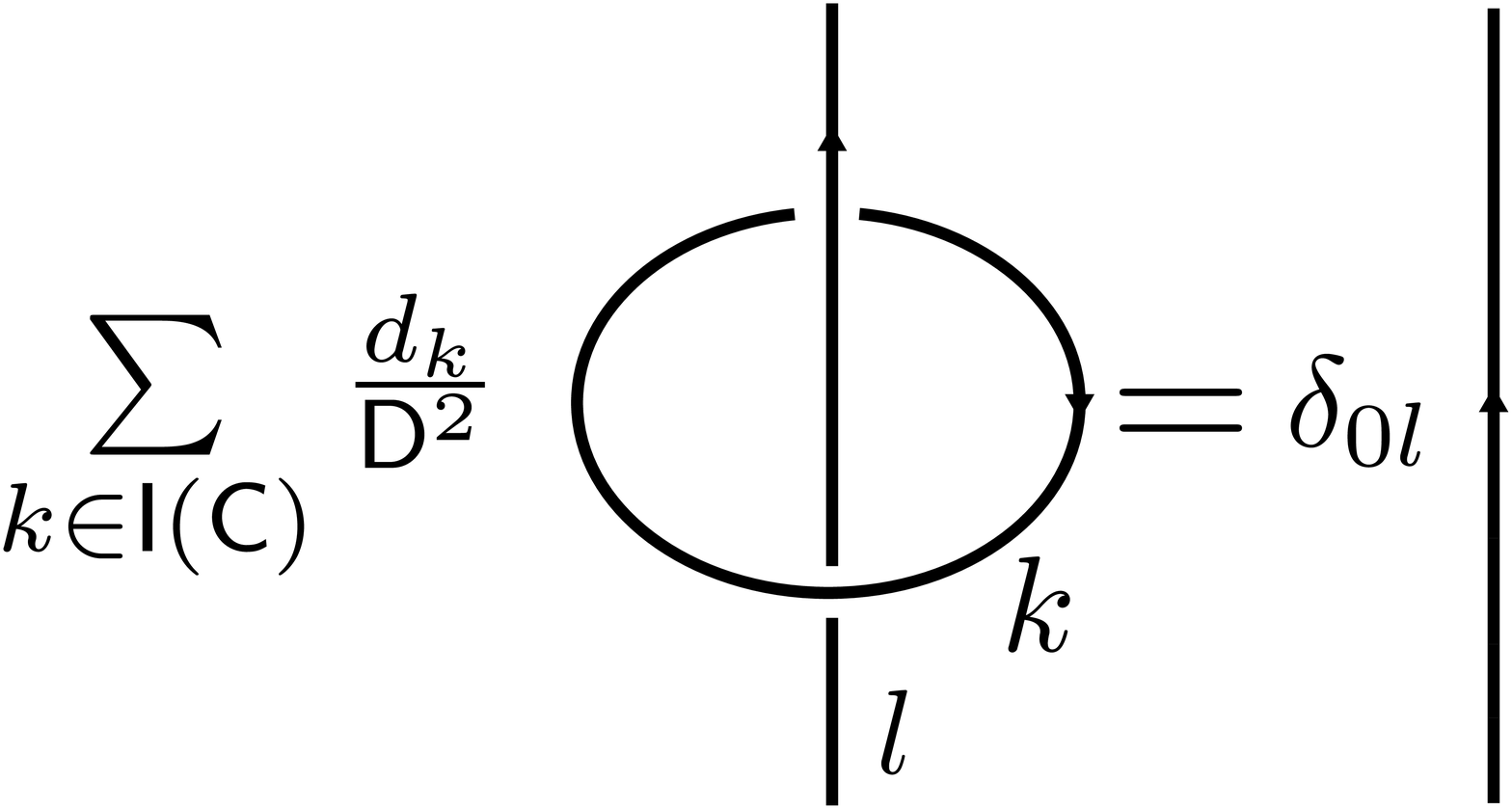}.
\end{figure}
\end{lem}

In $\Csf$ there exists a pairing
\eq{
\hom_\Csf(A,B)\otimes_\Kbb \hom_\Csf(B,A)&\rightarrow \Kbb\\
f\otimes g \mapsto (f,g)\equiv \trrm( g\circ f)
}
which, for $\Csf$ semisimple, is non-degenerate. We are mainly interested in morphism spaces $\hom_\Csf(\mathbf{1},\bullet)$ for which we introduce some notation.
\begin{defn} Let $A_1,\dots , A_n\in \Csf$, then we define
\eq{
\la A_1,\dots, A_n \ra\equiv \hom_\Csf(\mathbf{1},A_1\otimes \dots \otimes A_n)\; .
}
\end{defn}
\begin{lem} There is a functorial isomorphism of vector spaces $\la A_1,\dots , A_n \ra\simeq \la A_n,A_1,\dots, A_{n+1}\ra  $.
\end{lem}
\begin{proof}
A functorial isomorphism is given by
\begin{figure}[H]
\centering
\includegraphics[scale=0.12]{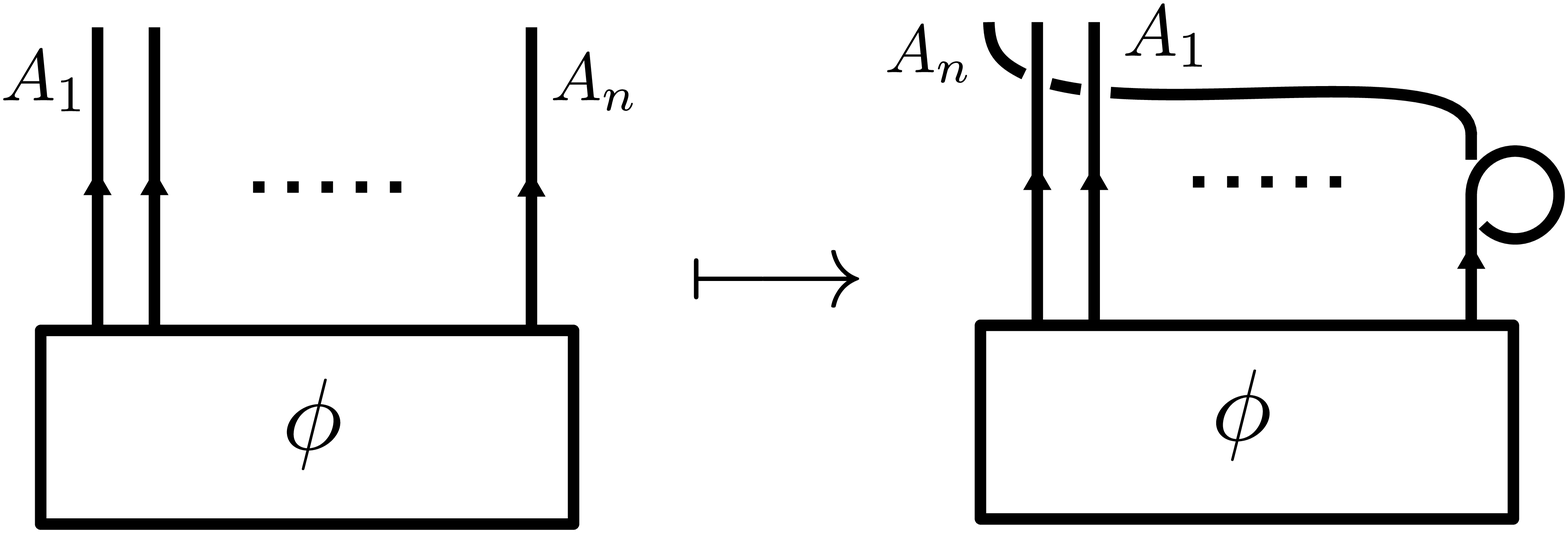}
\end{figure}

which clearly has an inverse given by composing inverses of braiding and twist.
\end{proof}

As only the cyclic order of $\la A_1,\dots, A_n \ra $ matters, instead of rectangular boxes for morphisms $\phi\in \la A_1,\dots, A_n \ra$, we introduce coupons in the graphical calculus 
\begin{figure}[H]
\centering
\includegraphics[scale=0.12]{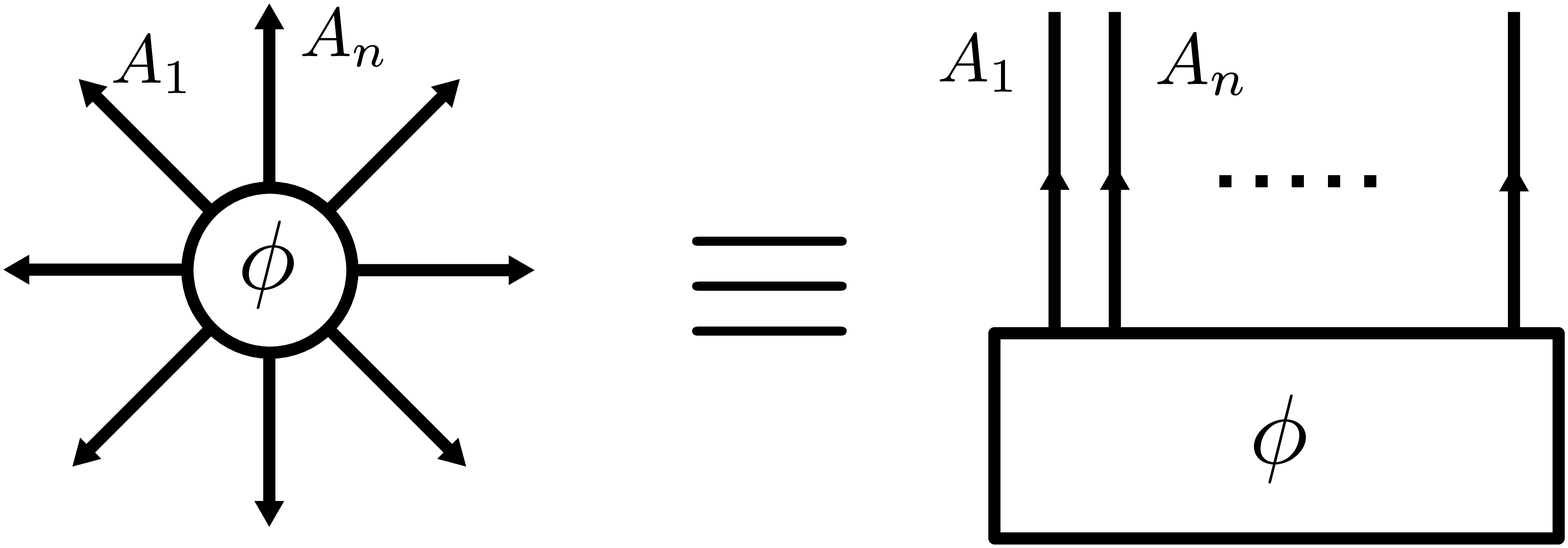}.
\end{figure}
For an arrow oriented towards the coupon with label $A$, the respective element gets replaced by $A^\ast$. Coupons can be composed with the help of the evaluation morphisms

\begin{center}
\begin{tikzpicture}
\node (couponlhs) at (0,0) {\includegraphics[scale=0.12]{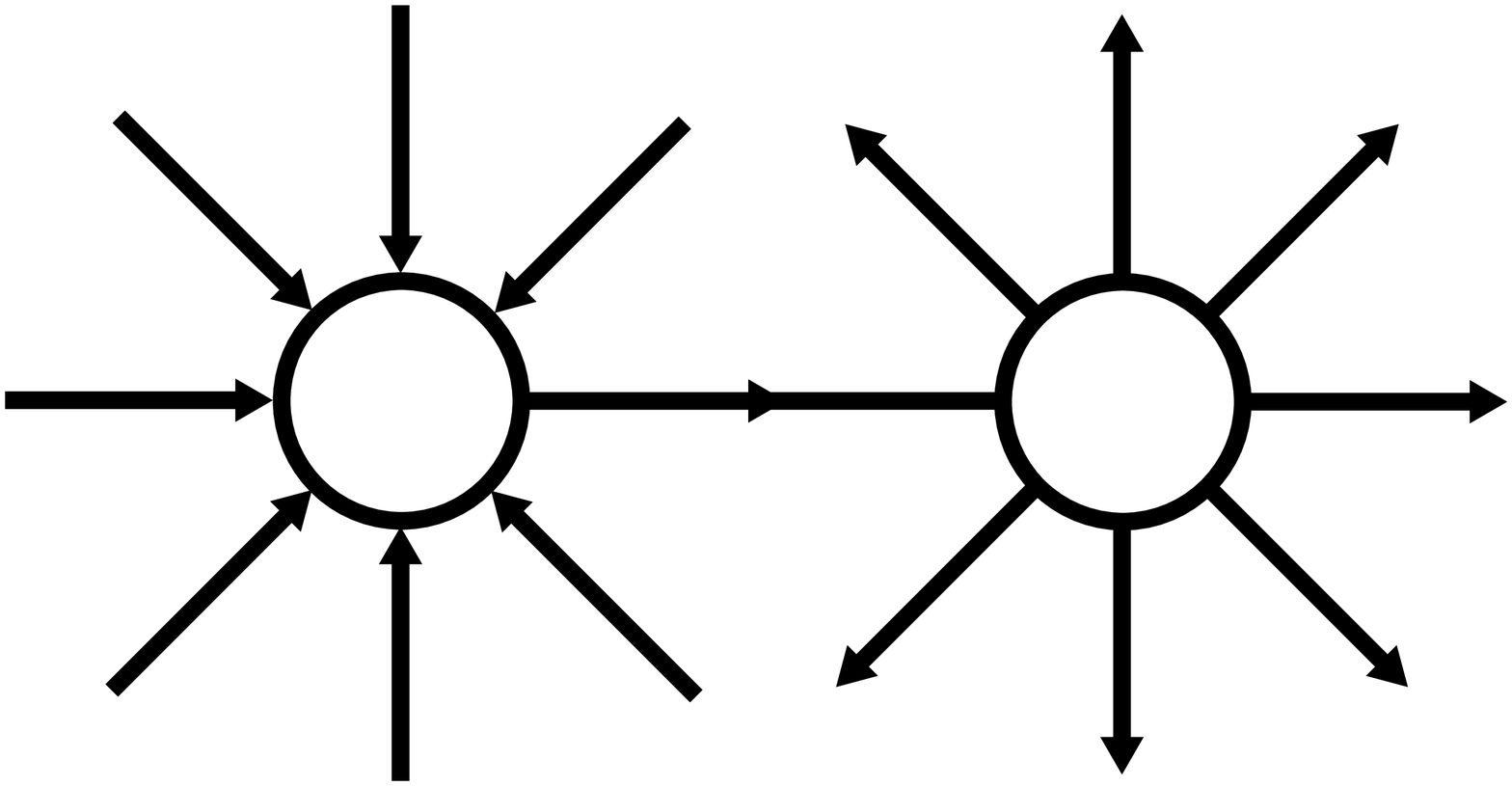}};
\node (couponrhs) at (5,0) {\includegraphics[scale=0.12]{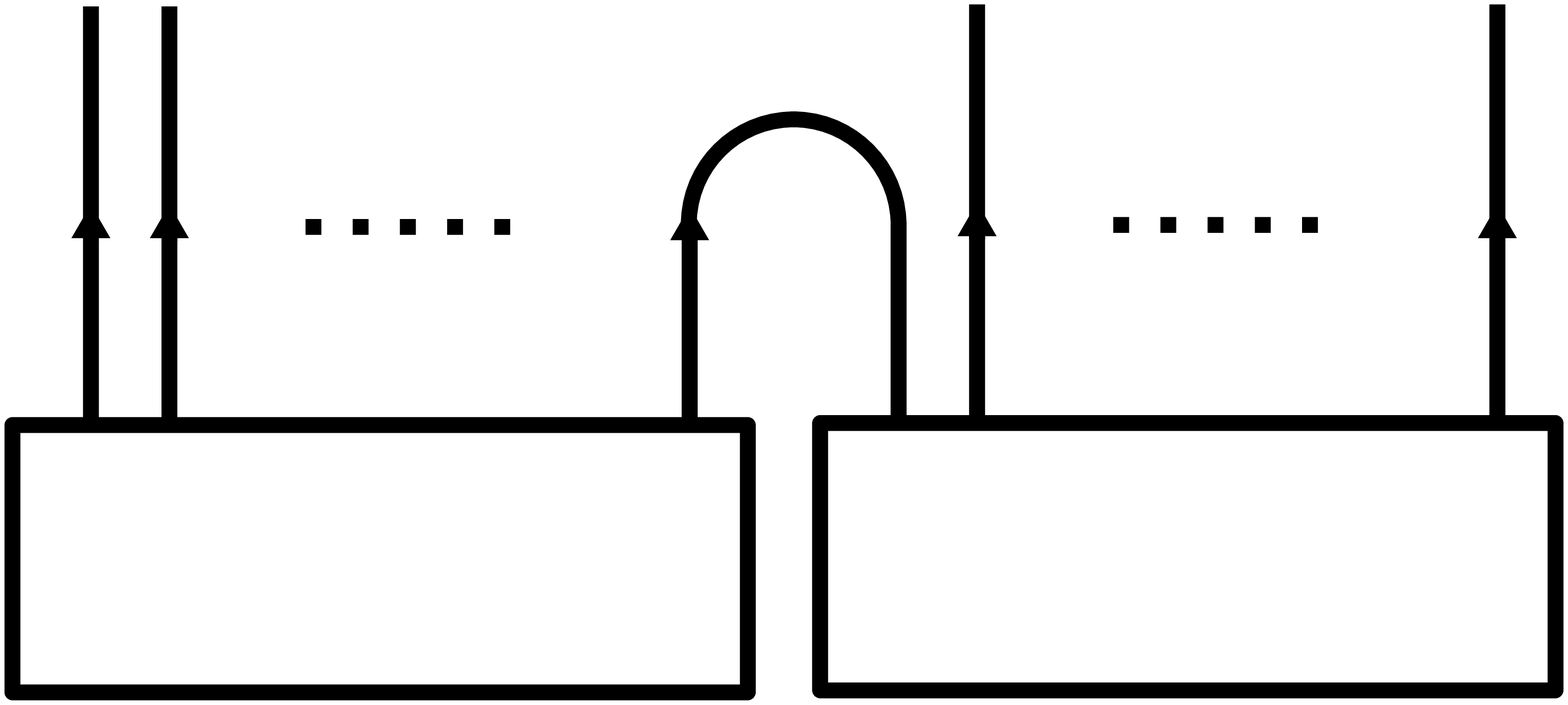}};
\node (equiv) at (2.5,0) {$\equiv$};
\node (phi) at (-0.75,0) {\scriptsize $\phi$};
\node (psi) at (0.75,0) {\scriptsize $\psi$};
\node (phi2) at (4.1,-0.5) {\scriptsize $\phi$};
\node (psi2) at (5.9,-0.5) {\scriptsize $\psi$};
\end{tikzpicture}.
\end{center}

The following lemma can be found in \cite{KirillovBalsam}
\begin{lem}
For any $A\in \Csf$ there are isomorphisms
\begin{enumerate}[label=\alph*)]
\item 

\begin{minipage}{\linewidth}
\centering
\begin{tikzpicture}
\node (lhs) at (0,0) {\includegraphics[scale=0.12]{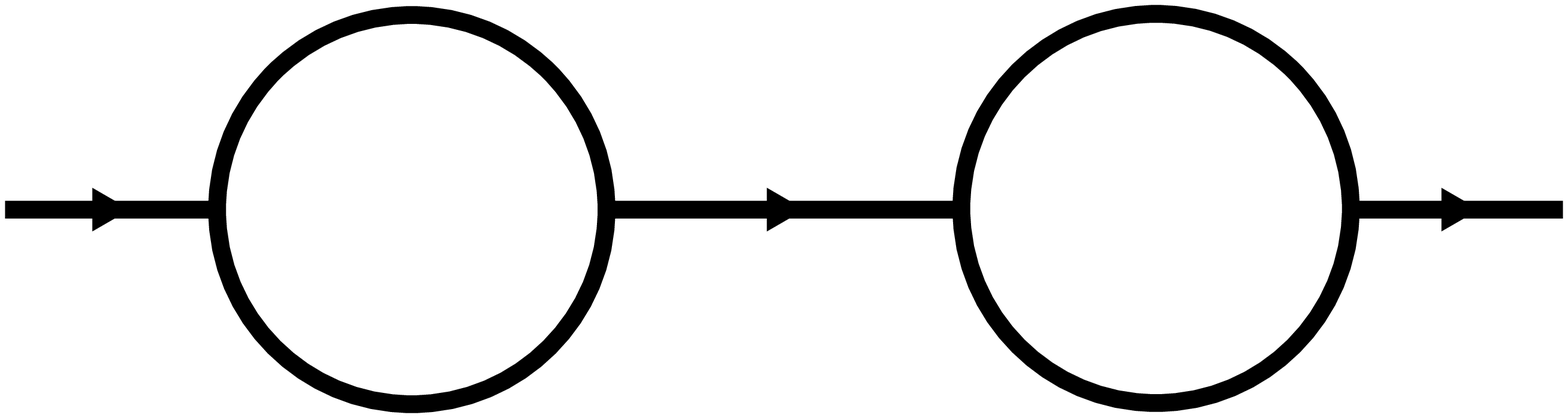}};
\node (rhs) at (5.3,0) {\includegraphics[scale=0.12]{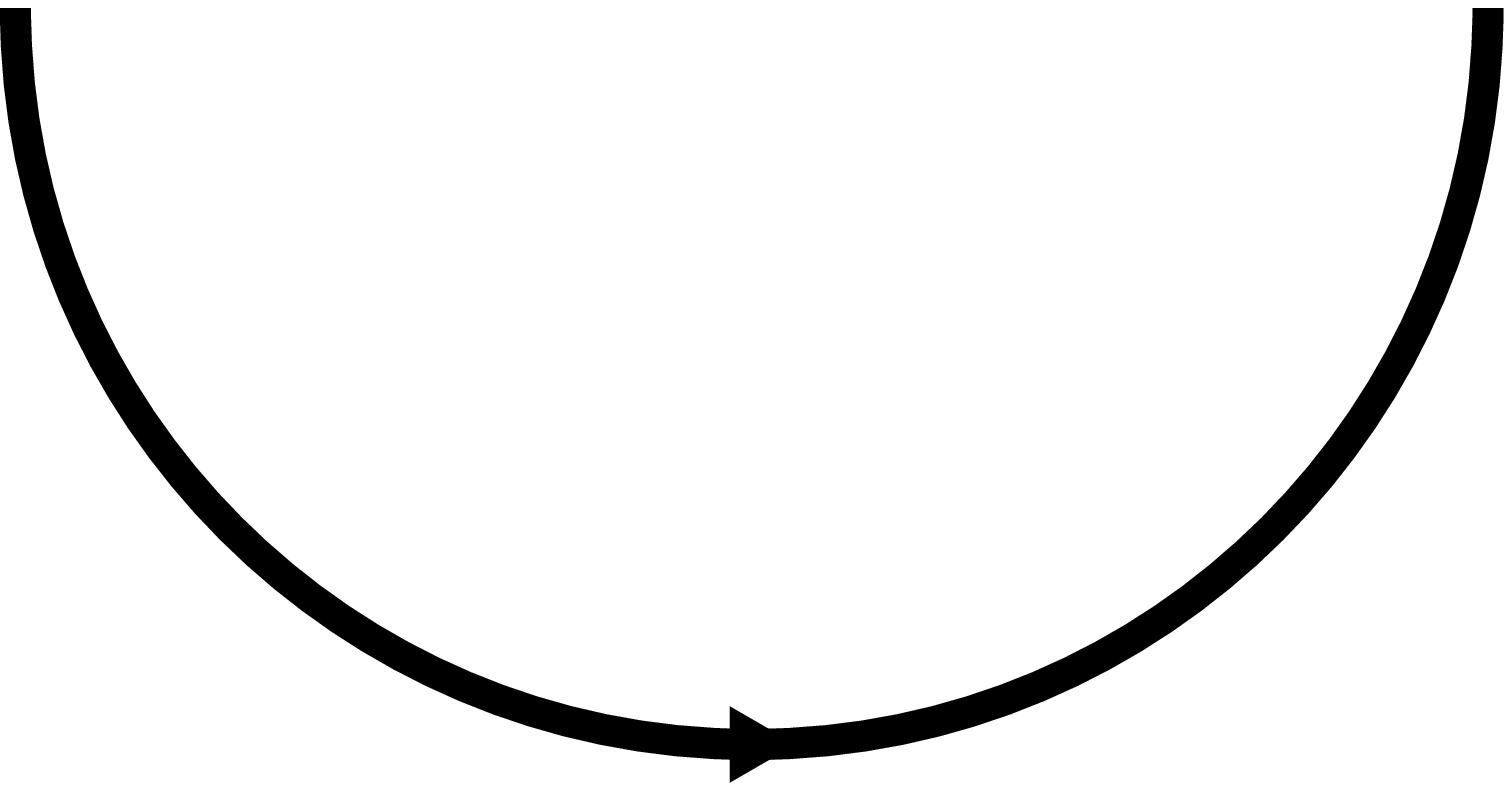}};
\node (equiv) at (2.5,0) {$=$};
\node (i) at (-1.5,0.2) {\scriptsize $i$};
\node (A) at (0,0.2) {\scriptsize $A$}; 
\node (i1) at (1.5,0.2) {\scriptsize $i$};
\node (phi) at (-0.75, 0) {\scriptsize $\phi$};
\node (psi) at (0.75,0) {\scriptsize $\psi$};
\node (pair) at (3.2,0) {$\frac{(\psi,\phi)}{d_i}$};
\node (i2) at (4.1,0.2) {\scriptsize $i$};
\end{tikzpicture}.
\end{minipage}
\item For $\lbr b_\alpha^i\rbr $ a basis in $\la i ,A^\ast_n,\dots A^\ast_1\ra $ with dual basis $\lbr b^\alpha_i\rbr $ in $\la i,A_1,\dots A_n\ra $ in the sense that $(b_\alpha^i,b^\beta_j)=\delta_{ij}\delta_{\alpha\beta}$ it holds

\begin{center}
\begin{tikzpicture}
\node (lhs) at (0,0) {\includegraphics[scale=0.12]{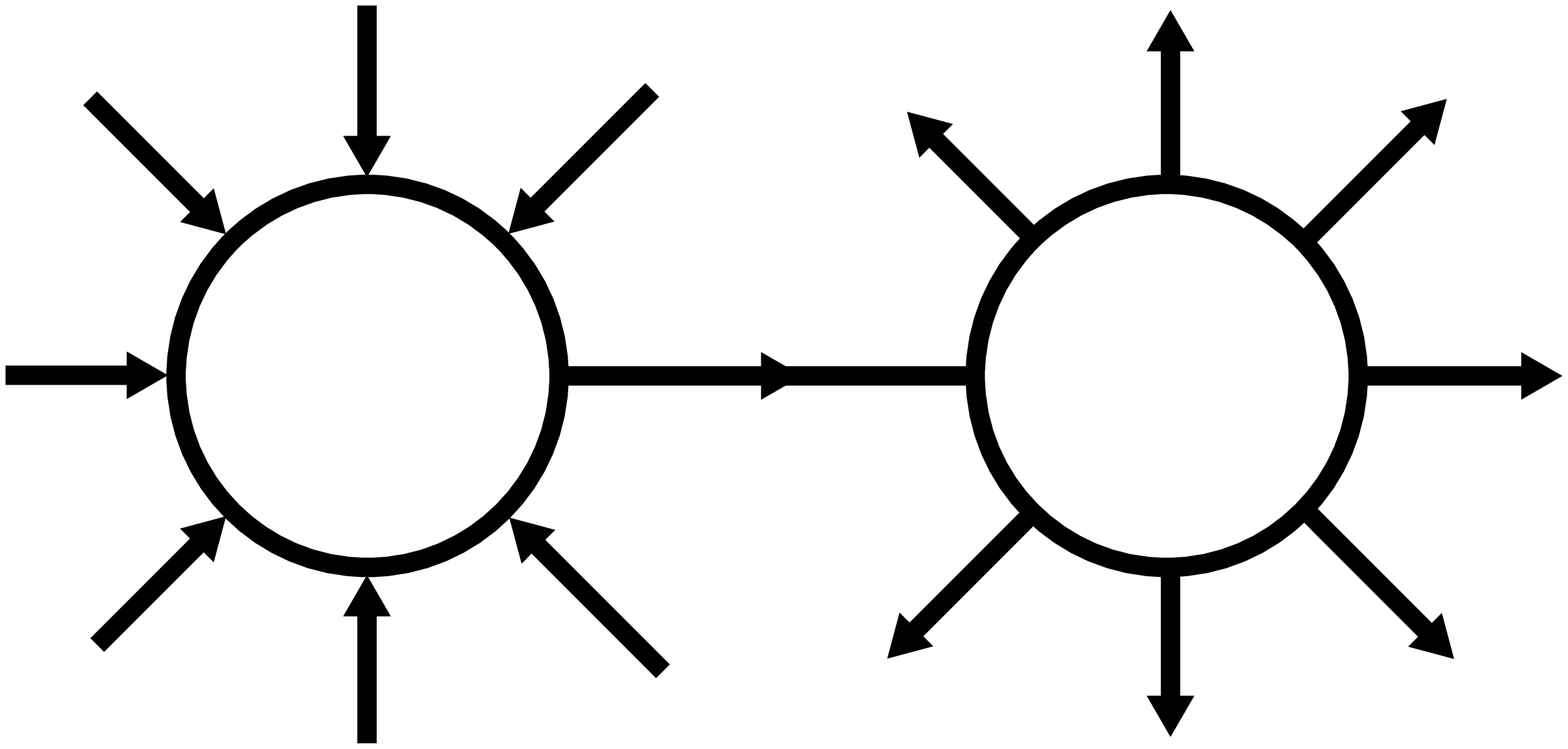}};
\node (rhs) at (5,0) {\includegraphics[scale=0.12]{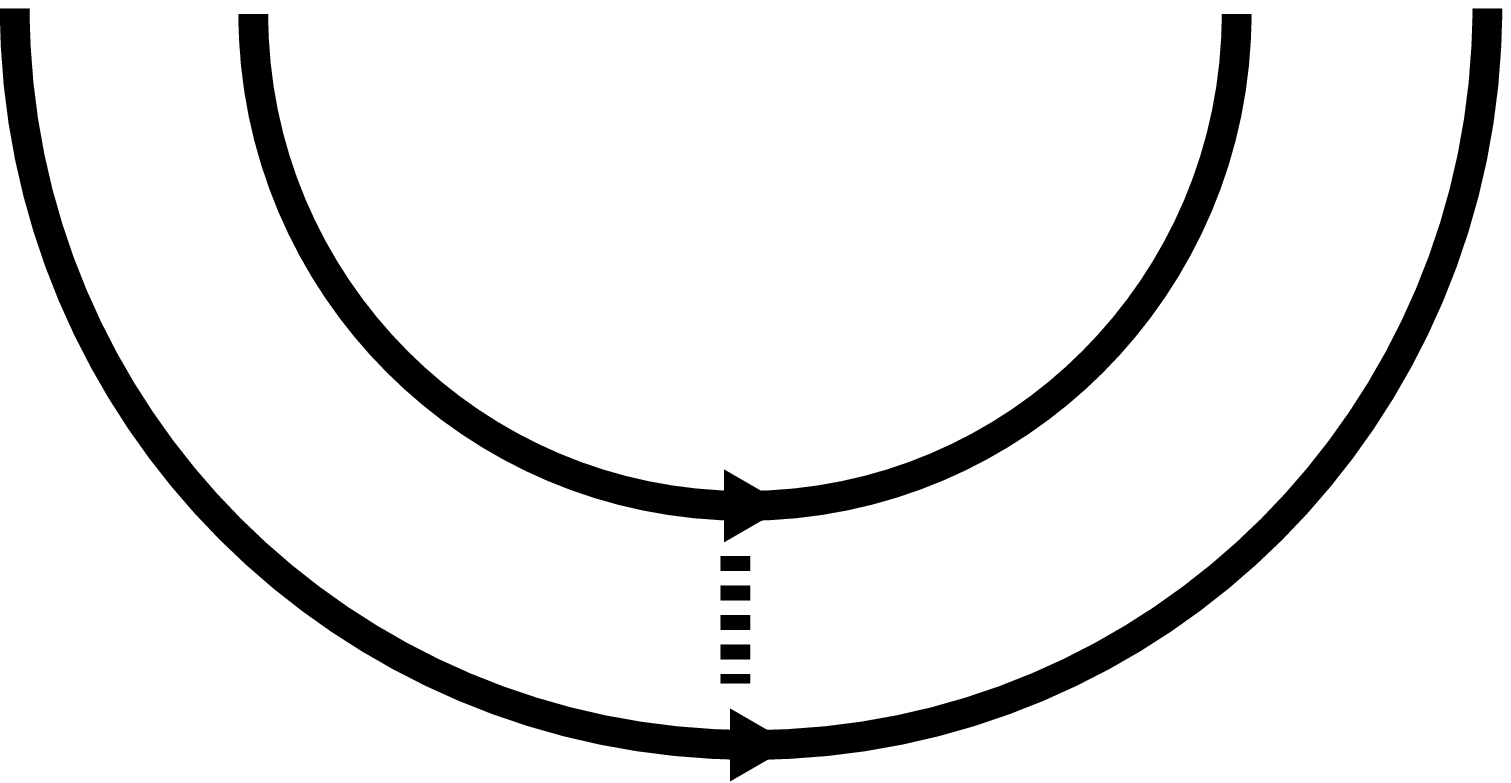}};
\node (i) at (0,0.2) {\scriptsize $i$};
\node (sum) at (-2.5,0) {$\underset{i\in \Isf(\Csf)}{\sum} d_i$};
\node (An) at (-0.2,0.7) {\scriptsize $A_n$};
\node (A1) at (-0.2,-0.7) {\scriptsize $A_1$};
\node (An2) at (0.3,0.65) {\scriptsize $A_n$};
\node (A12) at (0.4,-0.65) {\scriptsize $A_1$};
\node (=) at (2.5,0) {$=$};
\node (b) at (-0.75,0) {$b$};
\node (b2) at (0.75,0) {$b$};
\node (A111) at (3.8,0.4) {\scriptsize $A_1$};
\node (An22) at (4.7,0.4) {\scriptsize $A_n$};
\end{tikzpicture}
\end{center}
where the $b$-$b$ cupons stand for a summation 
\begin{figure}[H]
\includegraphics[scale=0.1]{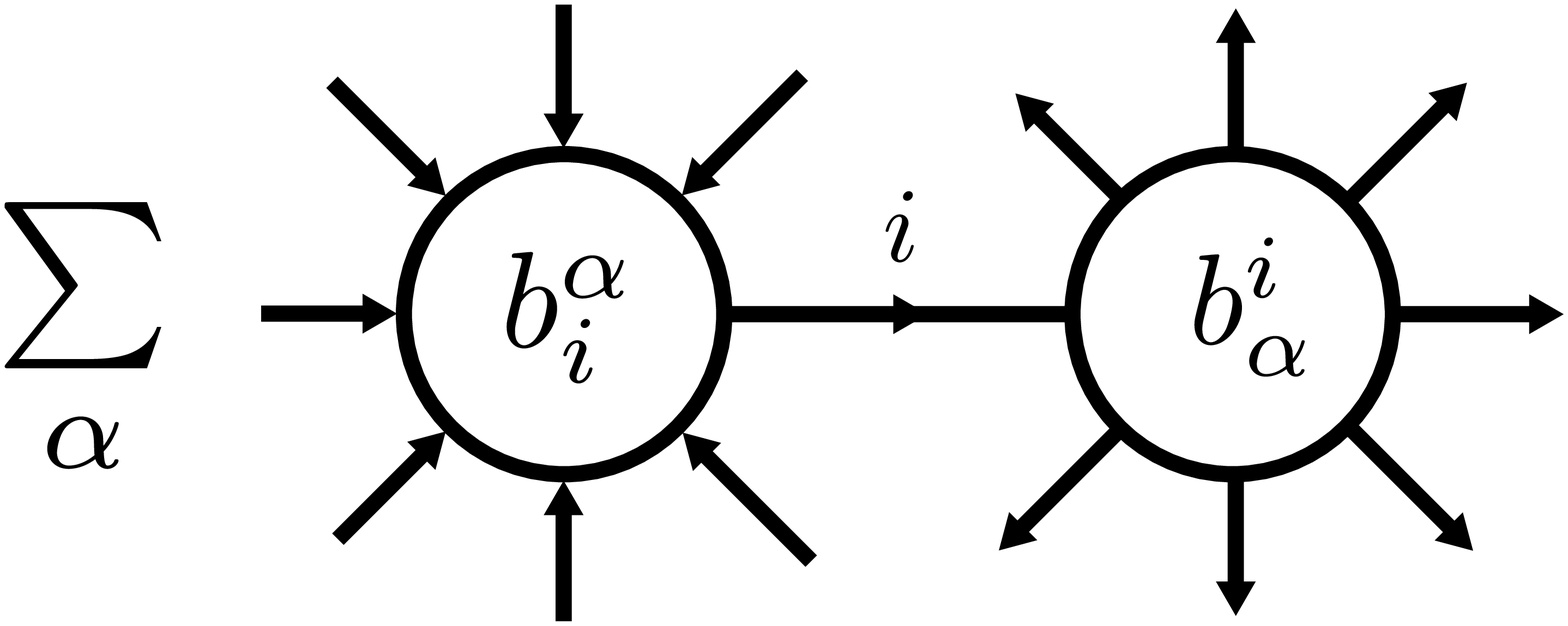}.
\end{figure}

\end{enumerate}
\end{lem}

The second relation is called \textit{completeness property} and will be used several times.

\subsection{Frobenius Algebras in Tensor Categories}

Frobenius algebras are usually defined as associative algebras on finite dimensional vector spaces having a non-degenerate bilinear form, compatible with the algebra multiplication. In this form Frobenius algebras correspond to two dimensional TFTs, see e.g. \cite{Lauda_2008}. The notion of a Frobenius algebra has an enhancement to categories and the above notion corresponds to a Frobenius algebra in $\Vectsf$, the category of finite dimensional vector spaces. 
\begin{defn} Let $(\Asf,\otimes , \mathbf{1})$ be a tensor category. A \textit{Frobenius algebra in $\Asf$ with underlying object $A$} consists of morphisms
\begin{center}
\begin{tikzpicture}
\node at (0,0) {\includegraphics[scale=0.13]{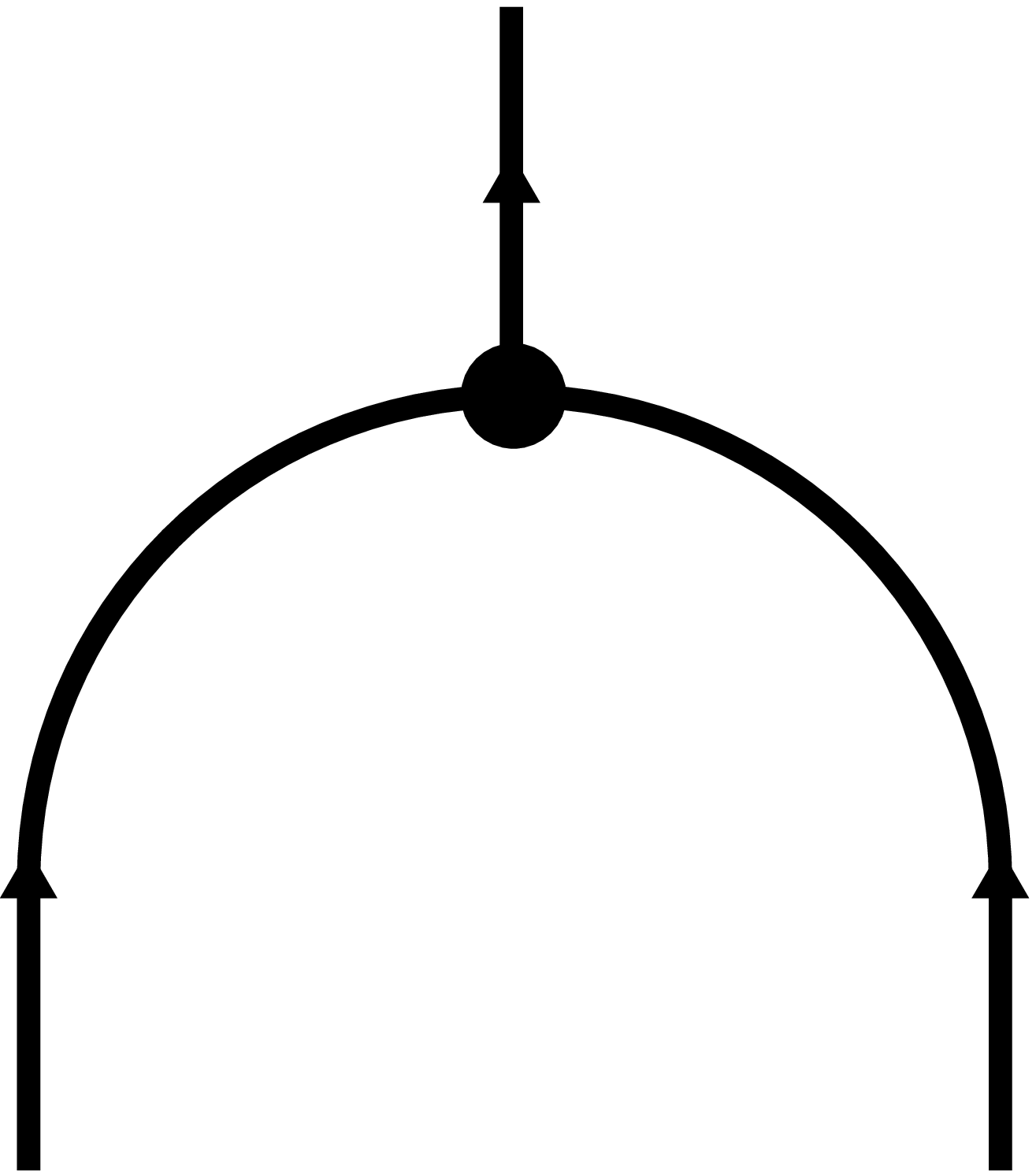}};
\node at (3,0) {\includegraphics[scale=0.13]{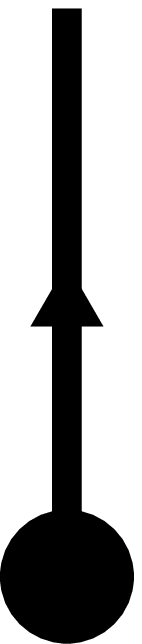}};
\node at (6,0) {\includegraphics[scale=0.13,angle=180]{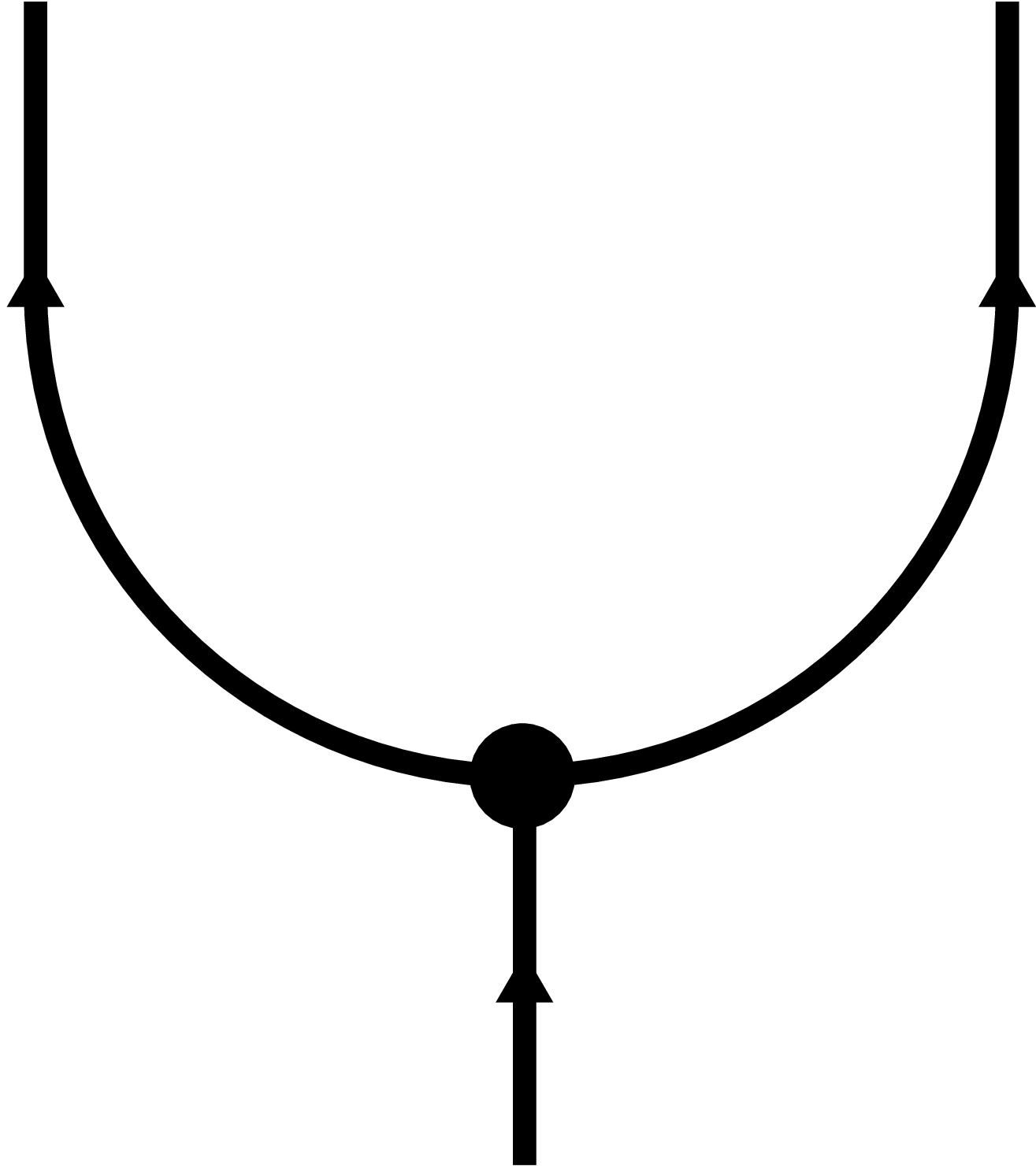}};
\node at (9,0) {\includegraphics[scale=0.13]{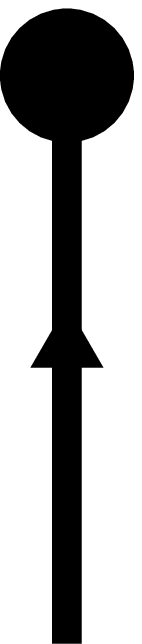}};
\node at (0,-1.5) {$m:A\otimes A\rightarrow A$};
\node at (3,-1.5) {$\eta:\mathbf{1}\rightarrow A$};
\node at (6,-1.5) {$\Delta:A\rightarrow A\otimes A $};
\node at (9,-1.5) {$\epsilon:A\rightarrow \mathbf{1}$};
\end{tikzpicture}.
\end{center}
where all strands are colored with $A$. These have to satisfy:
\begin{enumerate}[label=\Roman*)]
\item $(m,\eta)$ define an associative algebra on $A$:
\begin{center}
\begin{tikzpicture}
\node at (0,0) {\includegraphics[scale=0.12]{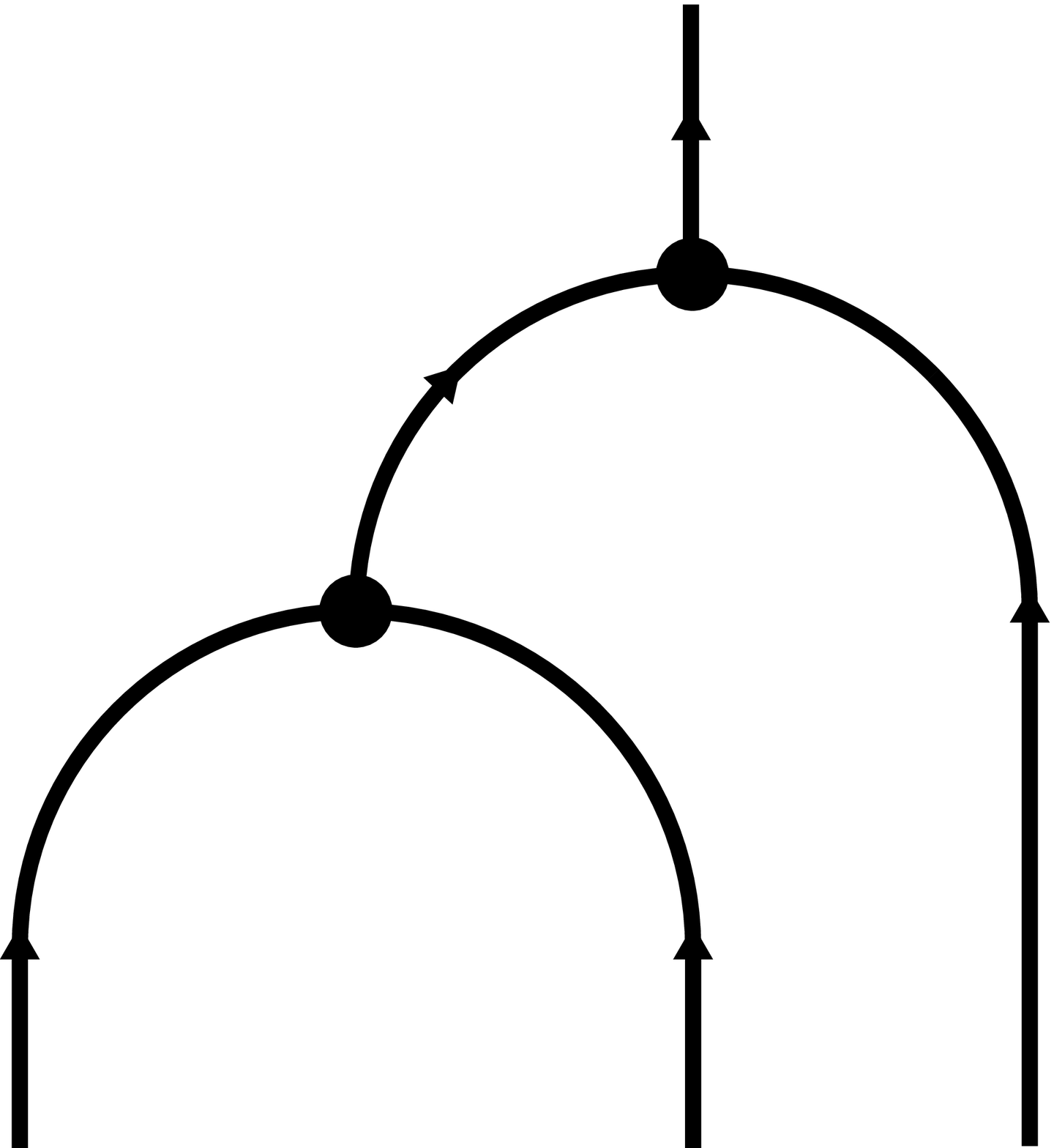}};
\node at (3,0) {\includegraphics[scale=0.12]{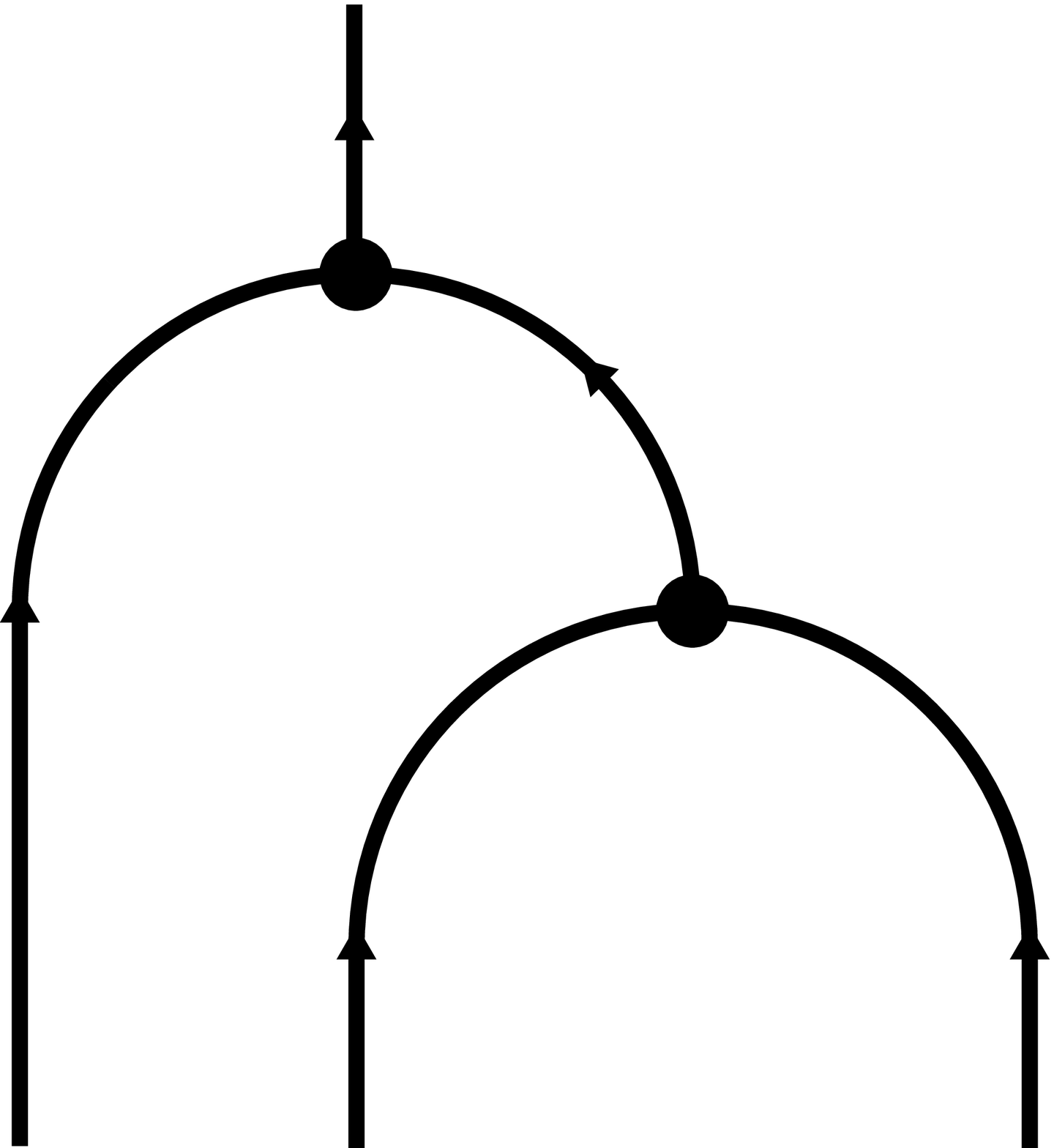}};
\node at (7,0) {\includegraphics[scale=0.12]{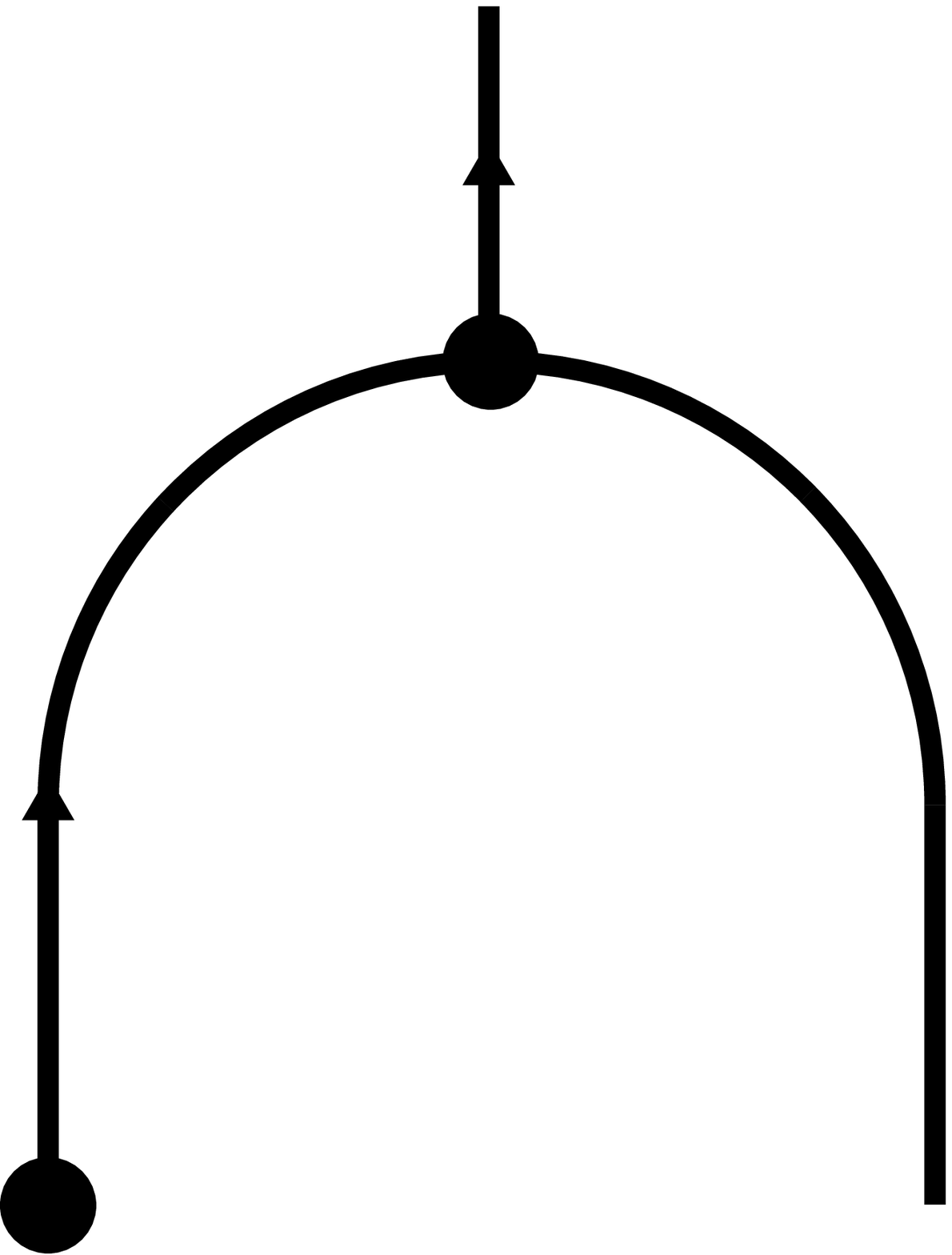}};
\node at (9,0) {\includegraphics[scale=0.12]{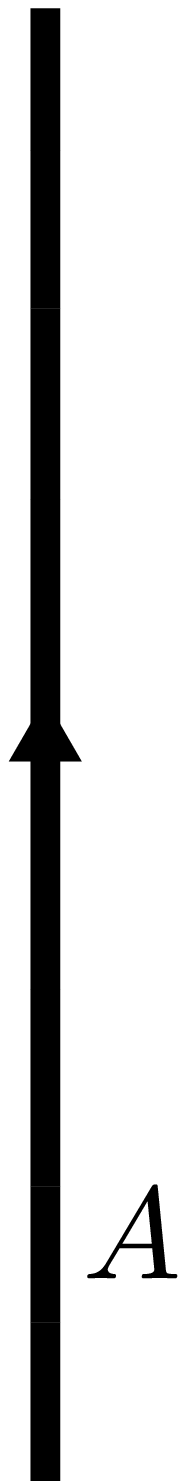}};
\node at (11,0) {\includegraphics[scale=0.12]{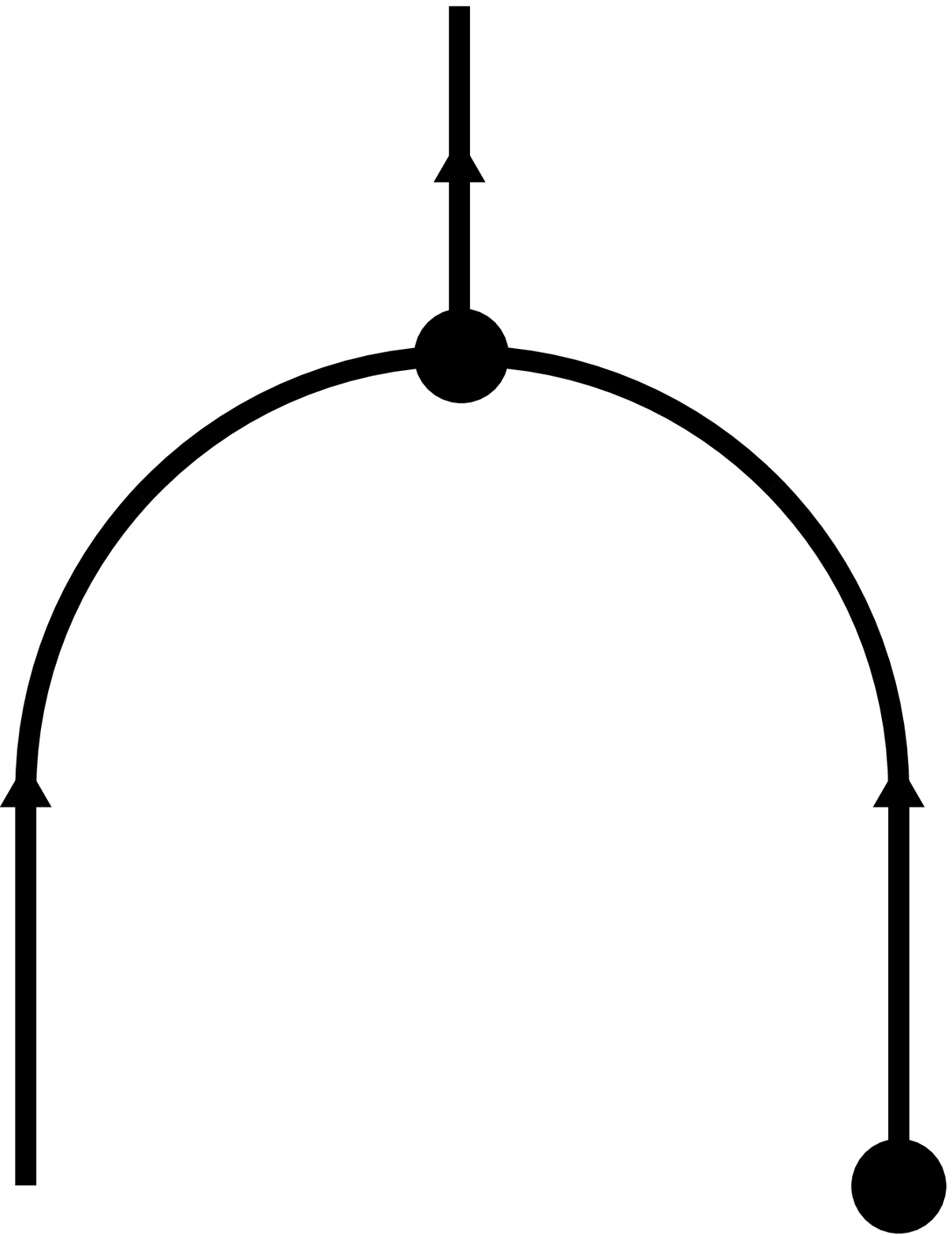}};
\node at (1.5,0) {$=$};
\node at (8.3,0) {$=$};
\node at (9.7,0) {$=$};
\end{tikzpicture}.
\end{center}
\item $(\Delta,\epsilon)$ define a coassociative coalgebra on $A$:
\begin{center}
\begin{tikzpicture}
\node at (0,0) {\includegraphics[scale=0.12,angle=180]{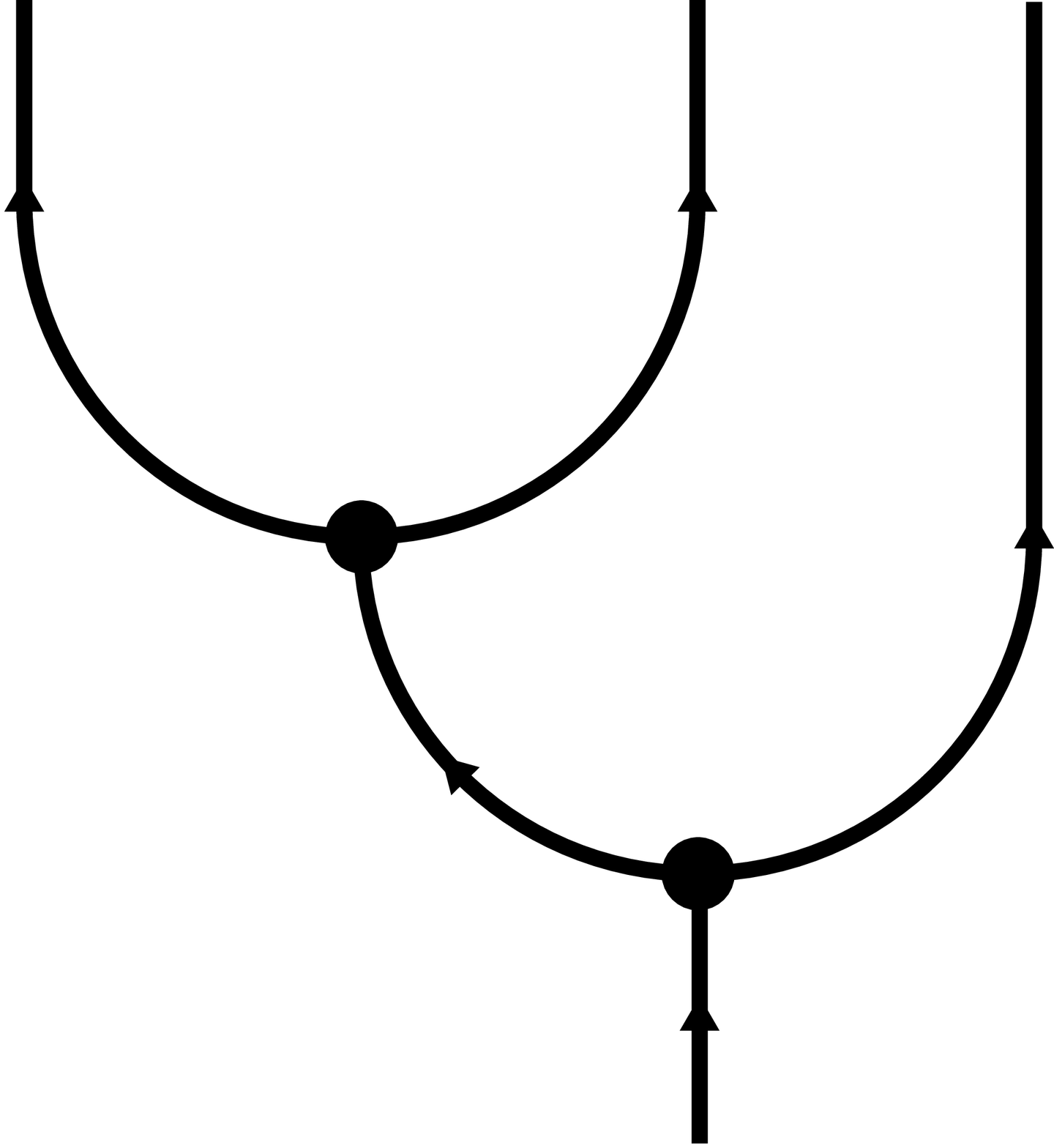}};
\node at (3,0) {\includegraphics[scale=0.12,angle=180]{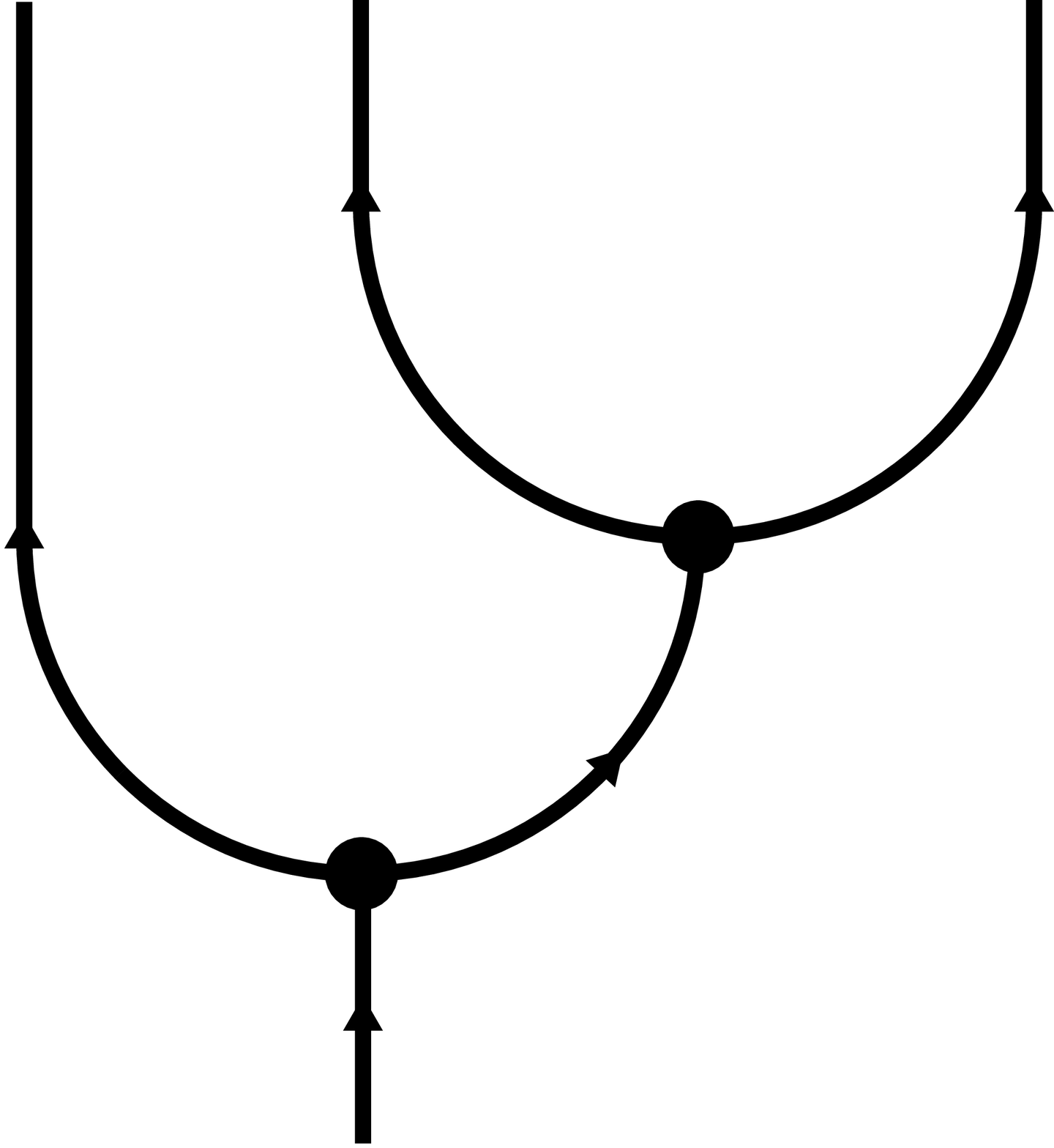}};
\node at (7,0) {\includegraphics[scale=0.12,angle=180]{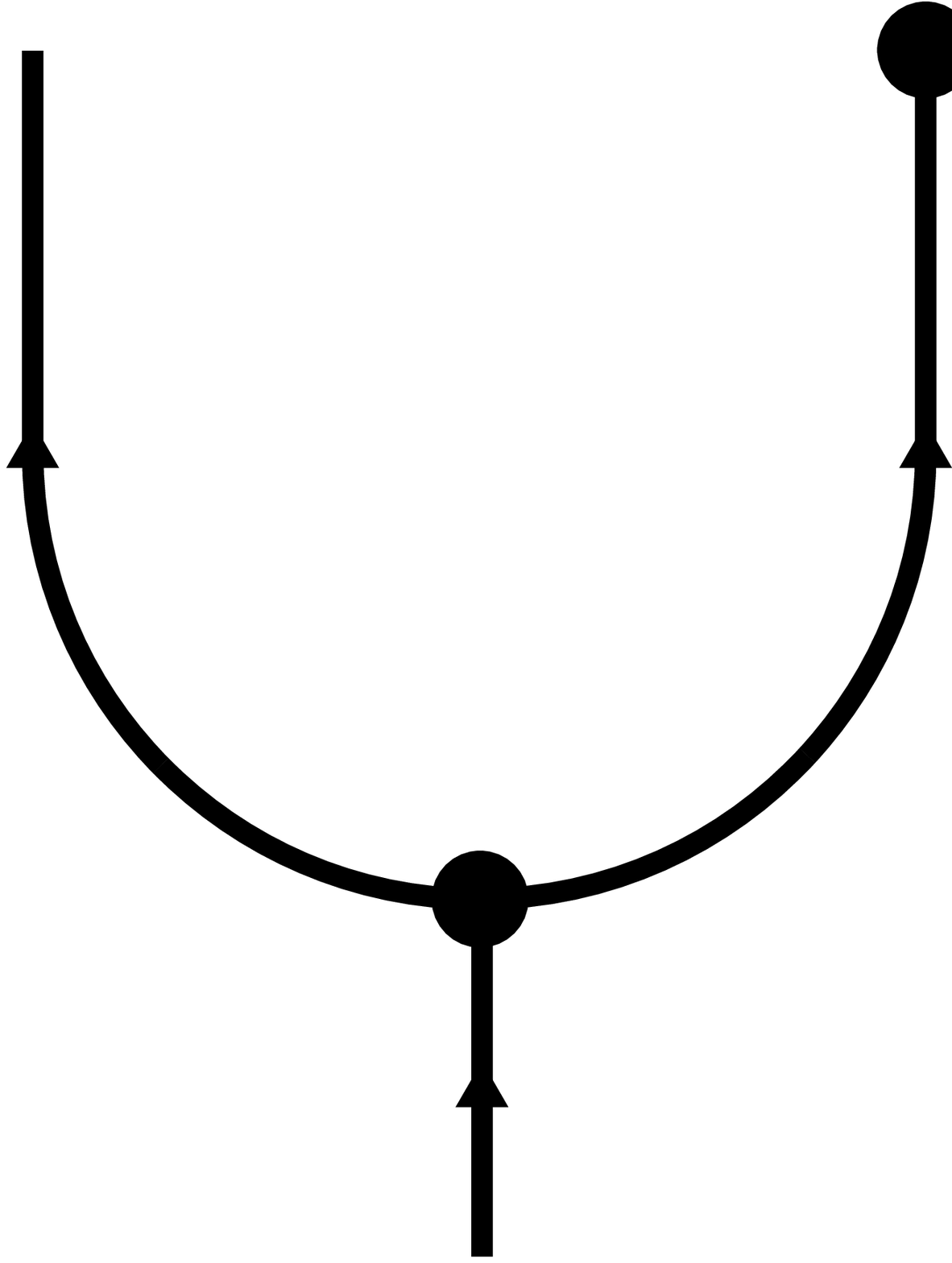}};
\node at (9,0) {\includegraphics[scale=0.12]{figure14.eps}};
\node at (11,0) {\includegraphics[scale=0.12,angle=180]{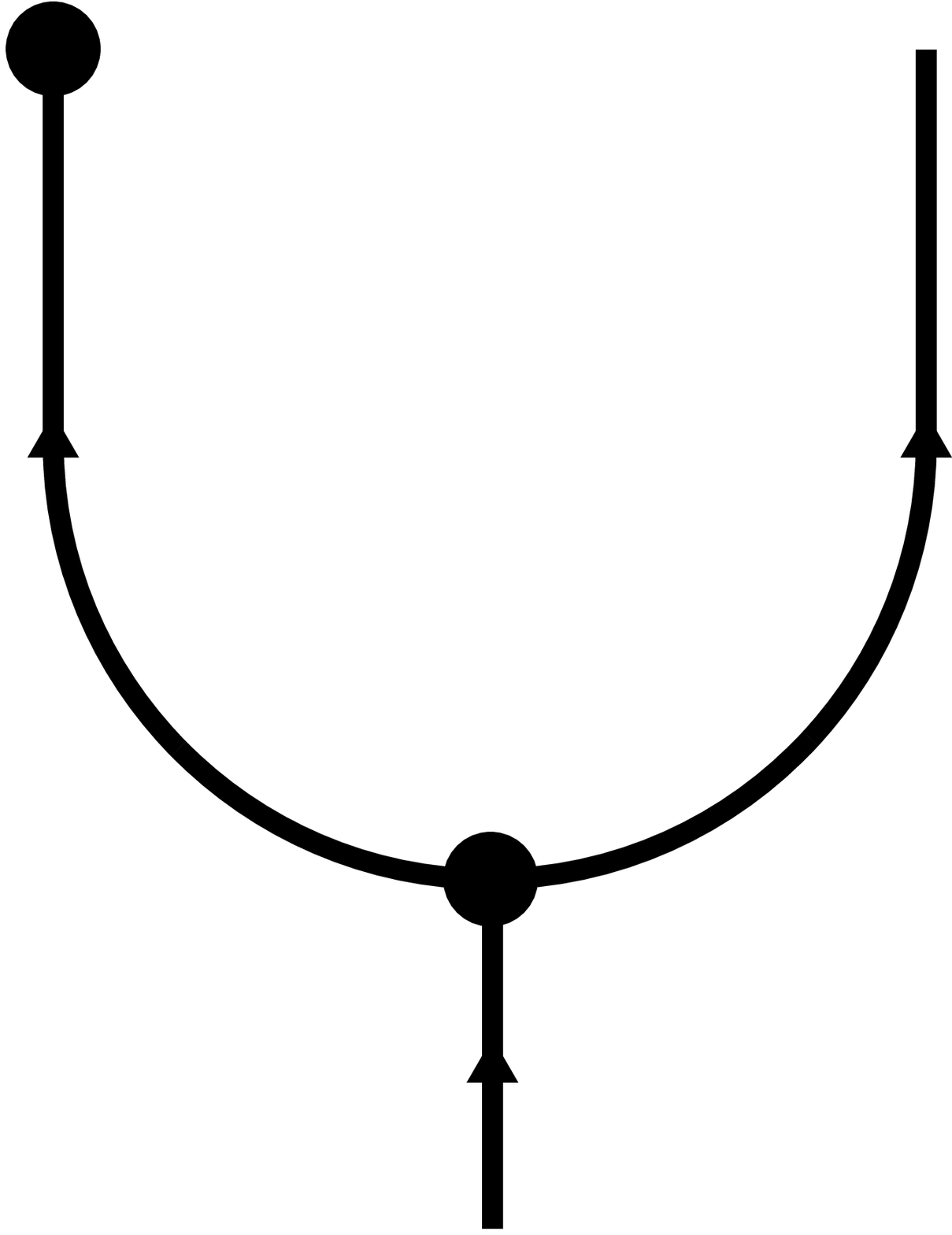}};
\node at (1.5,0) {$=$};
\node at (8.3,0) {$=$};
\node at (9.7,0) {$=$};
\end{tikzpicture}.
\end{center}
\item The \textit{Frobenius properties} hold
\begin{center}
\begin{tikzpicture}
\node at (0,0) {\includegraphics[scale=0.11]{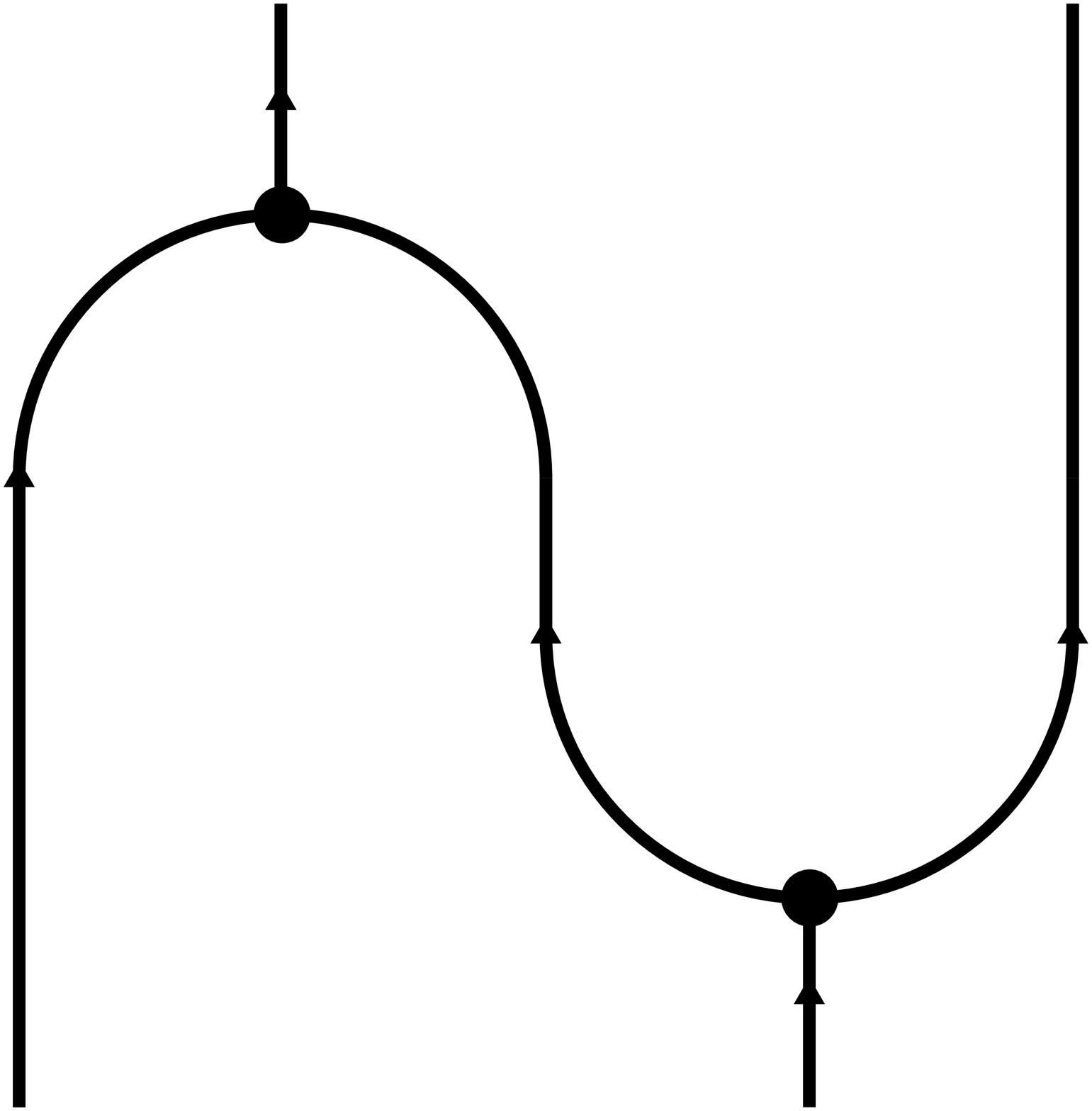}};
\node at (4,0) {\includegraphics[scale=0.11]{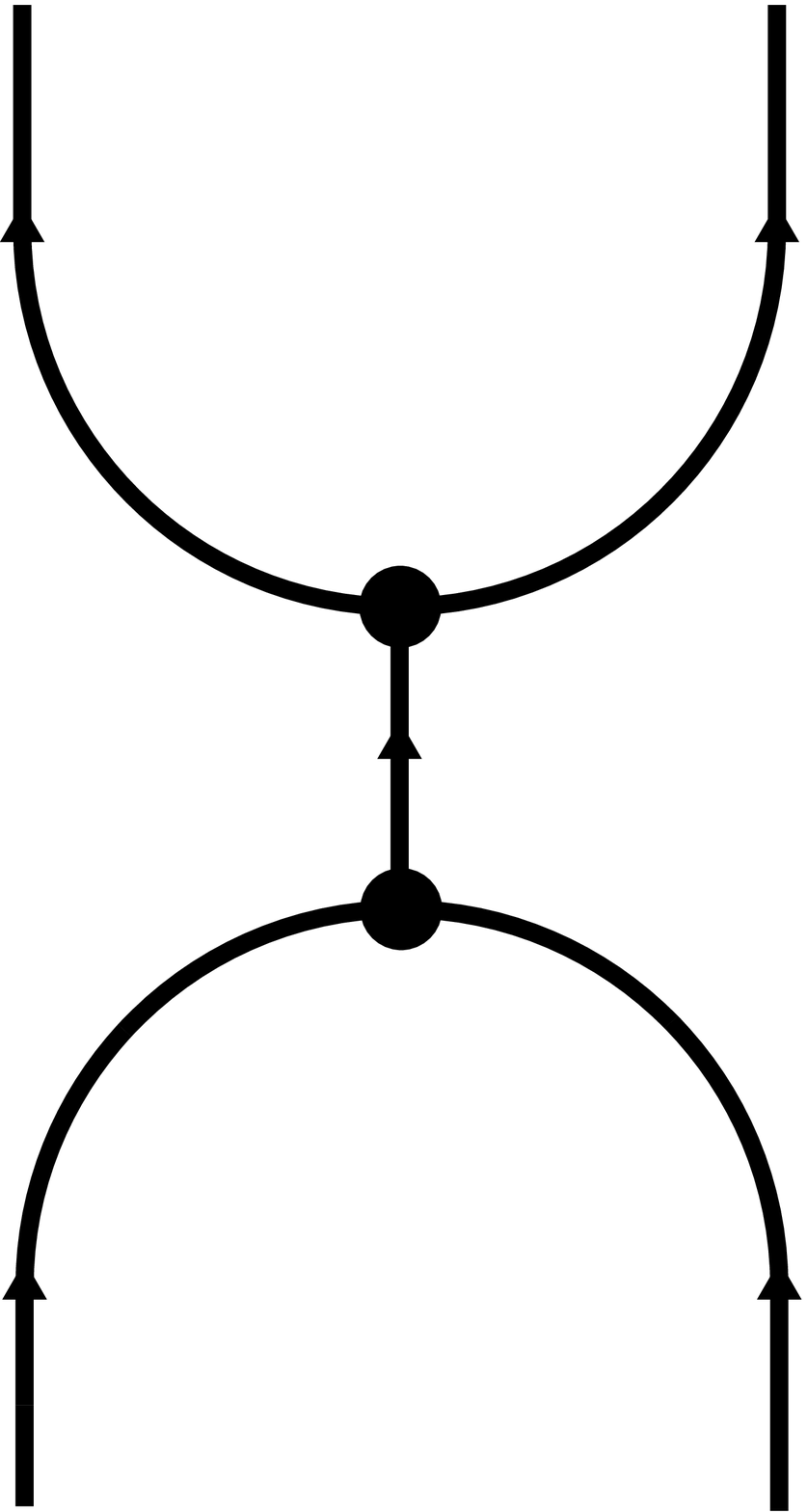}};
\node at (8,0) {\includegraphics[scale=0.11]{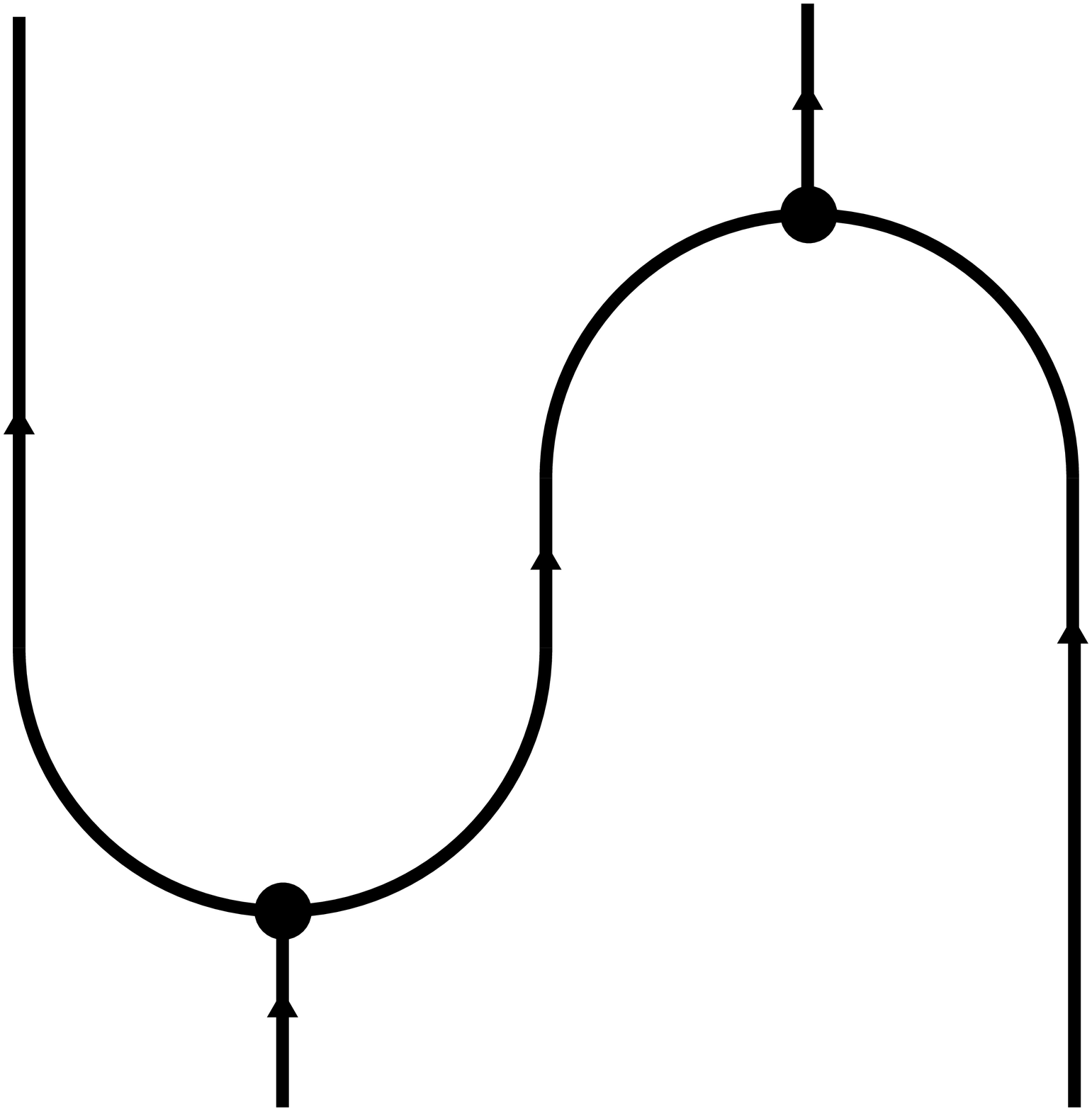}};
\node at (2,0) {$=$};
\node at (6,0) {$=$};
\end{tikzpicture}.
\end{center}
\end{enumerate}
If $\Asf$ is in addition pivotal, we can ask for $(A,m,\eta,\Delta,\epsilon)$ to be \textit{symmetric}, i.e. there is an equality of morphisms
\begin{center}
\begin{tikzpicture}
\node at (0,0) {\includegraphics[scale=0.12]{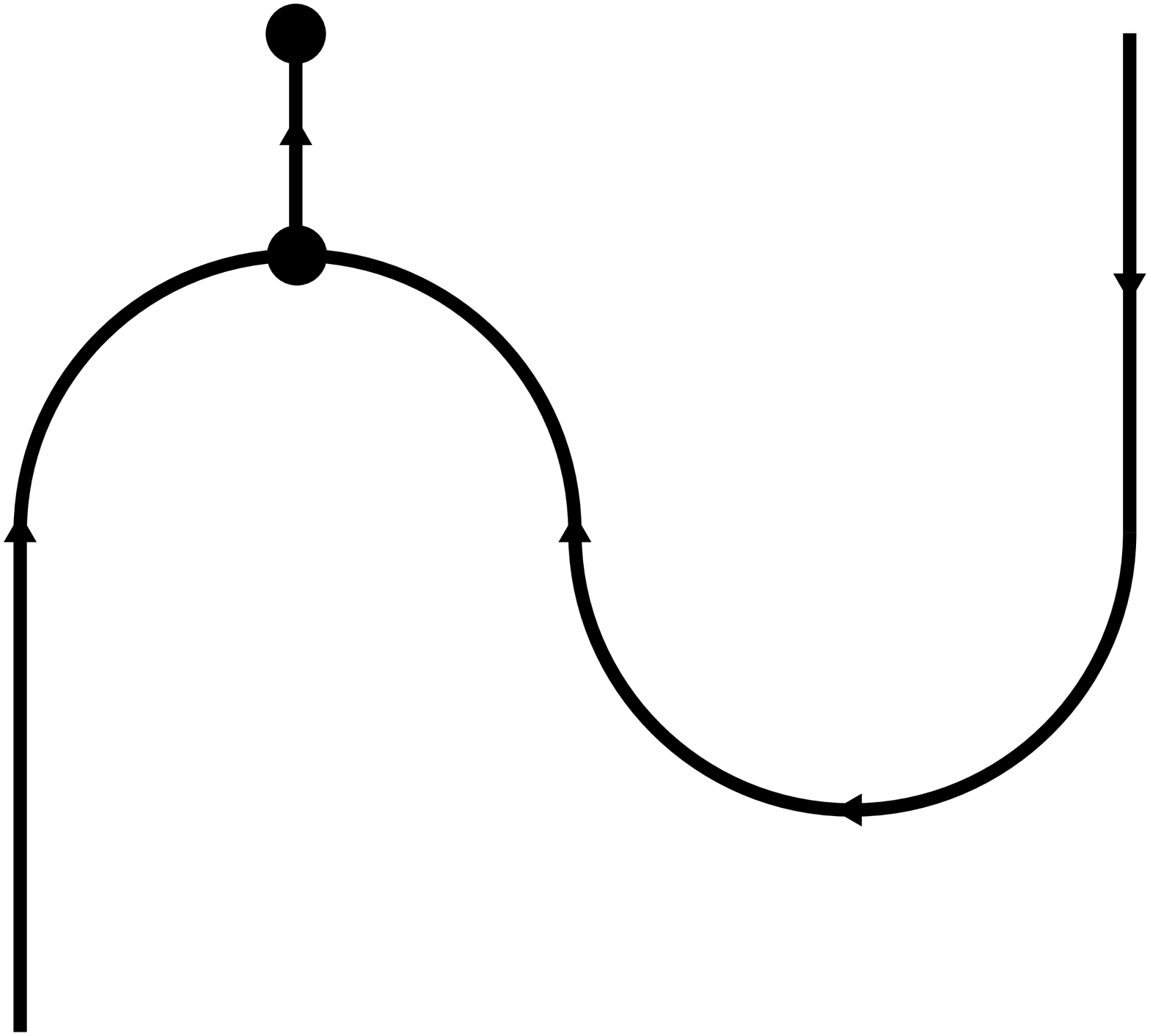}};
\node at (5,0) {\includegraphics[scale=0.12]{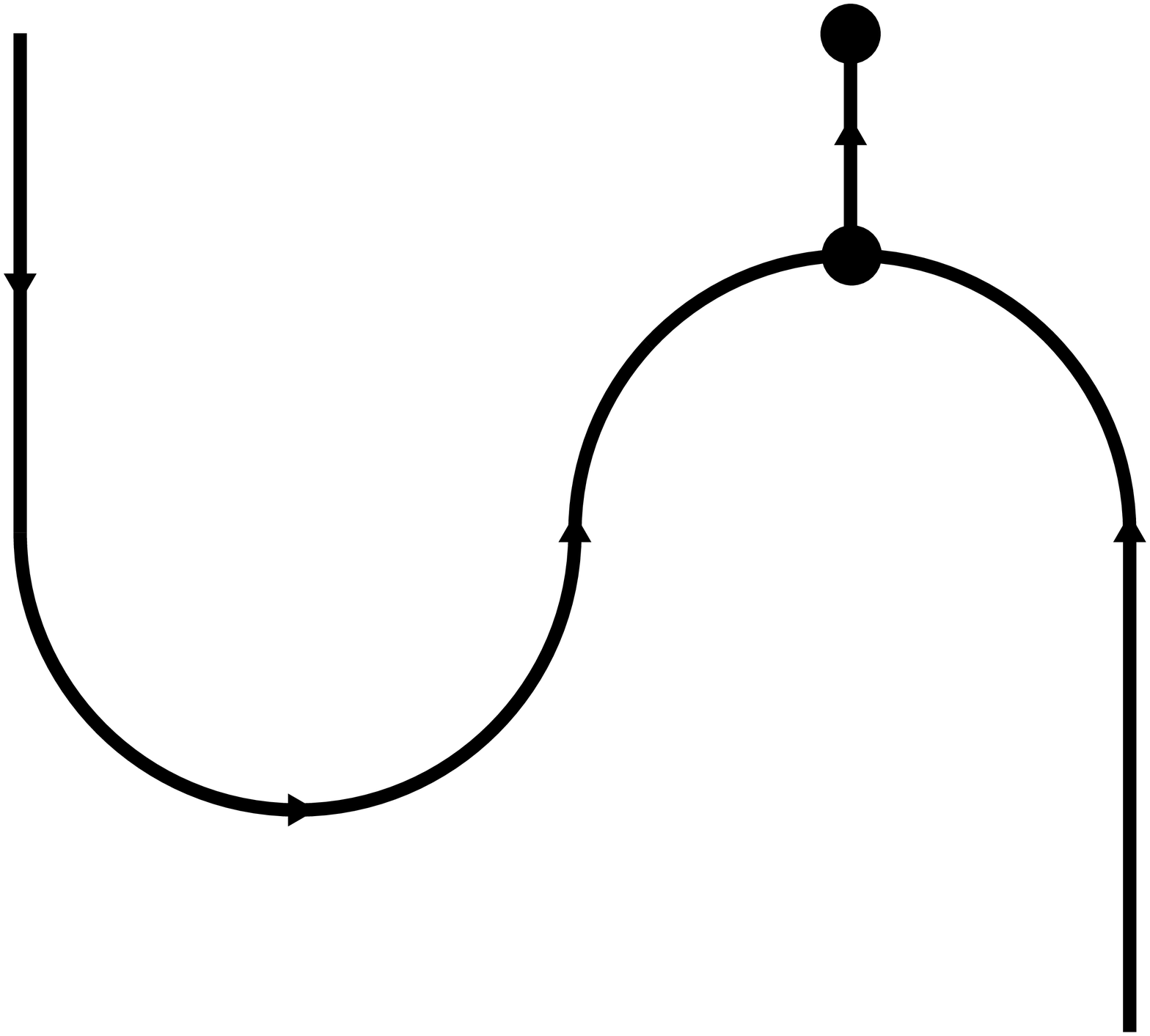}};
\node at (2.5,0) {$=$};
\end{tikzpicture}.
\end{center}
For $\Asf$ braided, we can require $(A,m,\eta,\Delta,\epsilon)$ to be (co-)commutative, i.e. 
\begin{center}
\begin{tikzpicture}
\node at (0,0) {\includegraphics[scale=0.15]{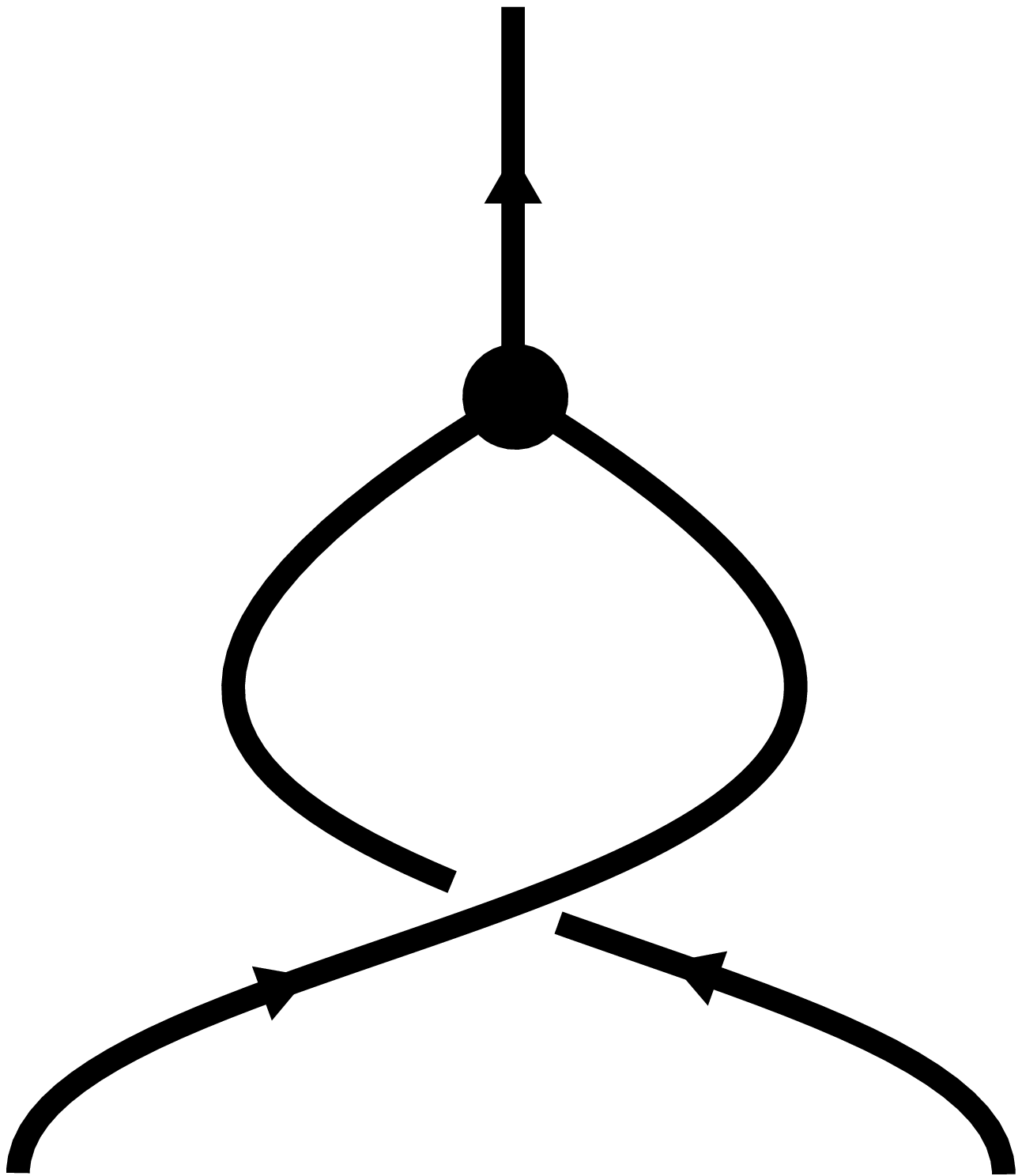}};
\node at (4,0) {\includegraphics[scale=0.15]{figure30.eps}};
\node at (9,0) {\includegraphics[scale=0.15,angle=180]{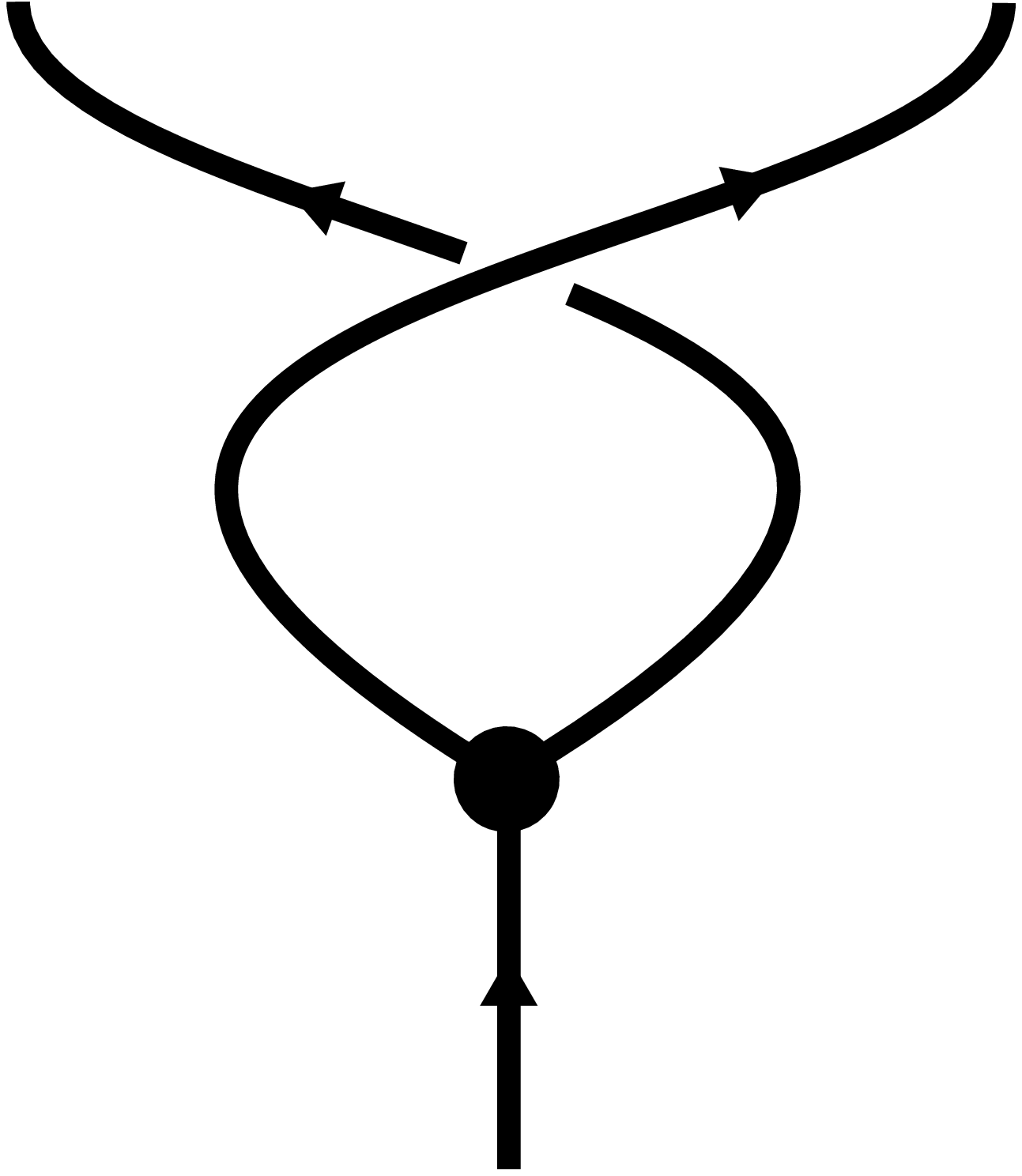}};
\node at (12,0) {\includegraphics[scale=0.15,angle=180]{figure32.eps}};
\node at (2,0) {$=$};
\node at (10.5,0) {$=$};
\end{tikzpicture}.
\end{center}
\end{defn}
For a morphism $f\in \hom_\Csf(A,B)$, if $A$ and $B$ are Frobenius algebras, one can define a \textit{Frobenius adjoint} $f^\dagger\in \hom_\Csf(B,A)$.
\begin{defn}\cite[Definition~2.17]{Kong_2009} Let $(A,m_A,\eta_A,\Delta_A,\epsilon_A)$ and $(B,m_B,\eta_B,\Delta_B,\epsilon_B)$ be Frobenius algebras in a monoidal category $\Csf$ and $f\in \hom_\Csf(A,B)$. The morphism $f^\dagger\in \hom_\Csf(B,A)$ is defined as
\begin{center}
\begin{tikzpicture}
\node at (0,0) {\includegraphics[scale=0.2]{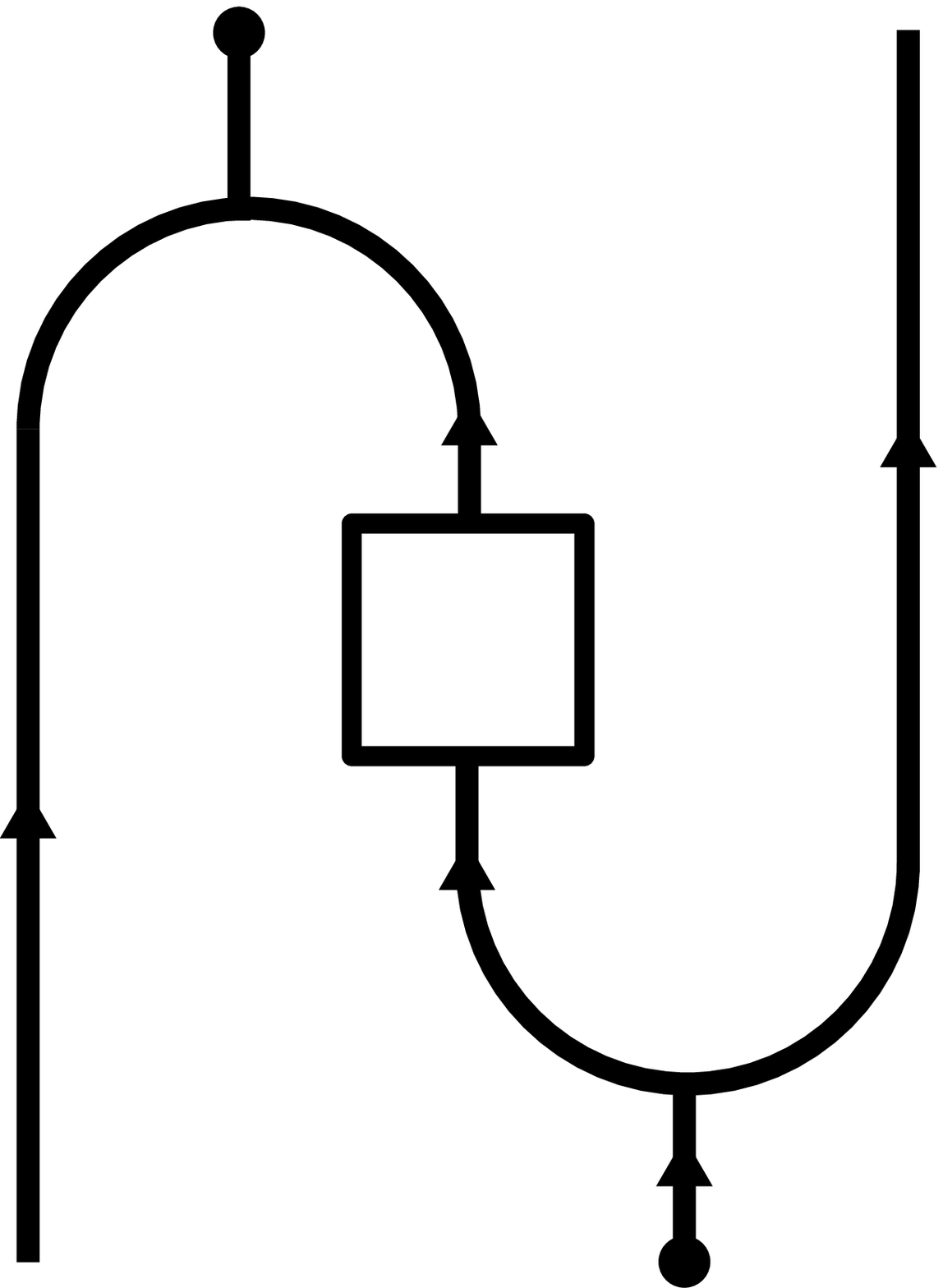}};
\node at (-2,0) {$f^\dagger$};
\node at (-1.5,0) {$=$};
\node at (0,0) {$f$};
\node at (-1.4,-1) {\scriptsize $B$};
\node at (1.3,-1) {\scriptsize $A$};
\end{tikzpicture}.
\end{center}

\end{defn}
We will use Frobenius algebras in modular tensor categories. If the category in question is the representation category of a rational vertex operator algebra, this corresponds to RCFTs very much like commutative Frobenius algebras in $\Vectsf$ correspond to oriented TFTs.
\subsection{Drinfeld-Center}

For $\Bsf$ a strictly monoidal category, its Drinfeld center $\Zsf(\Bsf)$ has objects $(B,\beta_{B,\bullet})$, where $\beta_{B,\bullet}:B\otimes \bullet\Rightarrow \bullet\otimes B$ is a natural isomorphism called \textit{half braiding} s.th. 
\eq{
\beta_{A,B\otimes C}=(\id_B\otimes \beta_{A,C})\circ (\beta_{A,B}\otimes \id_C)\; . 
}
Morphisms in $\Zsf(\Bsf)$ are morphisms $f\in \hom_\Bsf(A,B)$ s.th. 
\eq{
\beta_{B,C}\circ (f\otimes \id_C)=(\id_C\otimes f)\circ \beta_{A,C}\quad .
}
The Drinfeld center becomes a monoidal category with tensor product
\eq{
(A,\beta_{A,\bullet})\otimes (B,\beta_{B,\bullet})=(A\otimes B,\beta_{A\otimes B,\bullet}),\qquad
\beta_{A\otimes B,C}=(\beta_{A,C}\otimes \id_B)\circ (\id_A\otimes \beta_{B,C})\, . 
}
Note that $\Bsf$ does not need to be braided. But $\Zsf(\Bsf)$ is naturally braided with braiding given by
\eq{
\beta^{\Zsf(\Bsf)}_{(A,\beta_{A,\bullet}),(B,\beta_{B,\bullet})}=\beta_{A,B}\quad .
}
There exists an obvious forgetful functor $F:\Zsf(\Bsf)\rightarrow \Bsf$ forgetting the half braiding. For $\Bsf$ a modular tensor category it is shown in \cite{Kong_2009} that the adjoint of the forgetful functor reads
\eq{
L:\Bsf &\rightarrow \Zsf(\Bsf)\\
B &\mapsto \left( L(B)=\bigoplus_{i\in \Isf(\Bsf)} B\otimes U_i^\ast \otimes U_i, \beta^{ou}_{L(B),\bullet}\right) 
}
with over-under half braiding 
\eq{
\beta^{ou}_{L(B),A}=\bigotimes_{i\in \Isf(\Csf)} (\beta_{B,A}\otimes \id_{U_i^\ast\otimes U_i})\circ (\id_B\otimes \beta_{U_i^\ast,A}\otimes \id_{U_i})\circ (\id_{B\otimes U_i^\ast}\otimes \beta^{-1}_{U_i,A})
}
which has graphical representation

\begin{center}
\begin{tikzpicture}
\node at (0,0) {\includegraphics[scale=0.15]{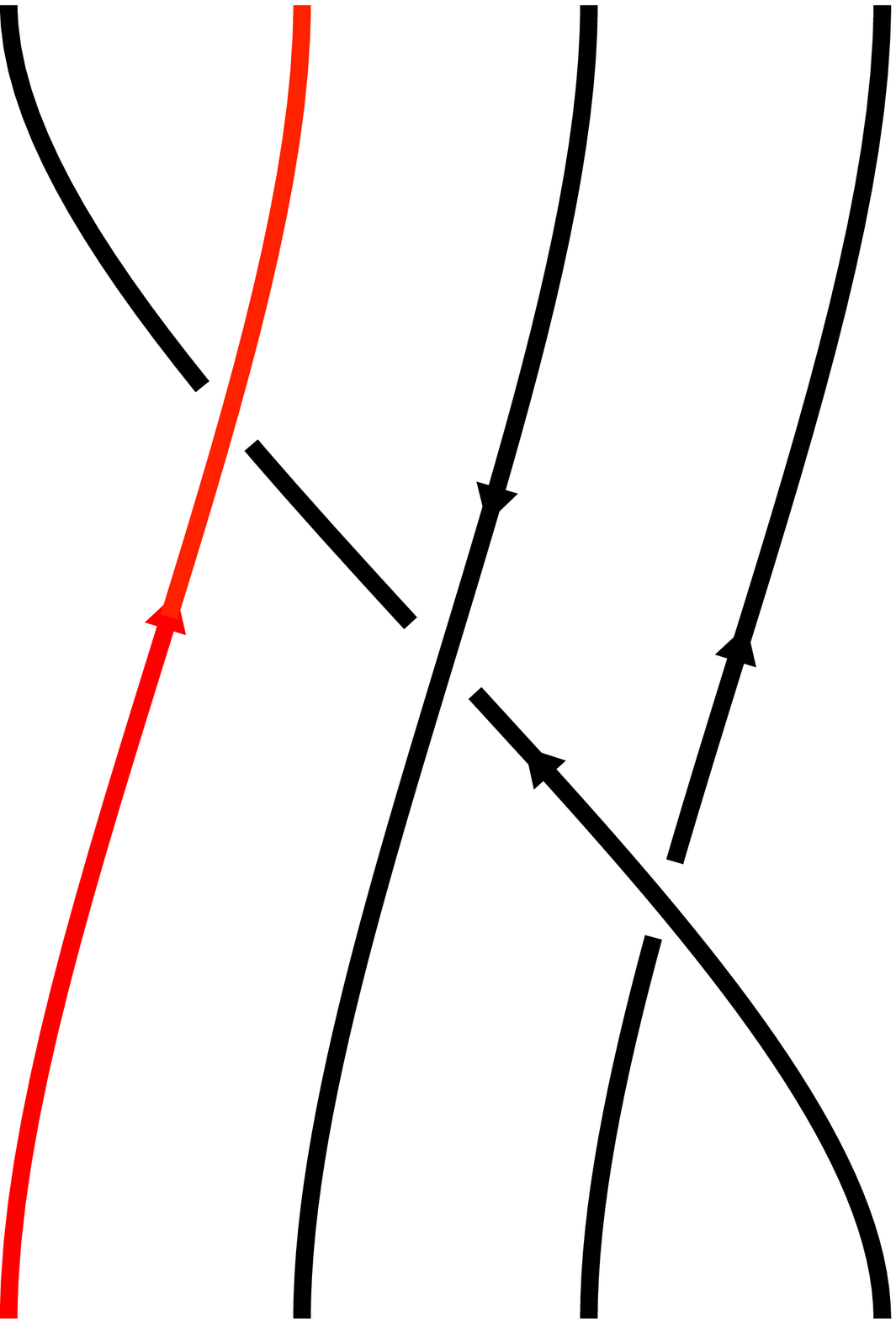}};
\node at (-3.5,0) {$\beta^{ou}_{L(B),A}$};
\node at (-2.3,0) {$=$};
\node at (-1.7,0) {$\begin{aligned}\bigoplus_{i\in \Isf(\Csf)}\end{aligned}$};
\node at (-0.9,-1.5) {$\color{red} B$};
\node at (-0.2,-1.5) {$i$};
\node at (0.6,-1.5) {$i$};
\node at (1.4,-1.5) {$A$};
\end{tikzpicture}.
\end{center}
On morphisms the functor is defined as $L(f)=\bigoplus_{i\in \Isf(\Csf)} f\otimes \id_{{U_i^\ast}\otimes U_i}$. Note that this is a faithful functor \cite[Lemma~2.22]{Kong_2009}, but not a tensor functor, since it doesn't map the identity on $\Bsf$ to the identity of $\Zsf(\Bsf)$. Nevertheless, it transports Frobenius algebras from one category to the other preserving symmetry. 
\begin{prop}\cite[Proposition~2.25]{Kong_2009}
For $A$ a Frobenius algebra in $\Bsf$, the object $L(A)$ has the structure of a Frobenius algebra in $\Zsf(\Bsf)$. In addition, $A$ is symmetric, if and only if $L(A)$ is symmetric.
\end{prop}

For finite categories $\Asf$, $\Bsf$ there exist a tensor product of categories. 
\begin{defn} The \textit{Deligne tensor product} $\Asf\boxtimes \Bsf$ has objects finite sums 
\eq{
\bigoplus A_i\boxtimes B_i,\quad A_i\in \Asf,\, B_i\in \Bsf
}
and morphism spaces 
\eq{
\hom_{\Asf\boxtimes \Bsf}(A_1\boxtimes B_1,A_2\boxtimes B_2)=\hom_\Asf(A_1,A_2)\otimes_\Kbb \hom_\Bsf(B_1,B_2)\quad.
}
\end{defn}
If $\Asf,\Bsf$ are fusion, their tensor product is also fusion with representatives for simple objects given by $U_i\boxtimes V_j$ where $U_i$ ($V_j$) are representatives of simple objects in $\Asf$ ($\Bsf$). For $\Asf,\Bsf$ braided tensor categories, $\Asf\boxtimes \Bsf$ is also a braided tensor category with tensor product 
\eq{
(A_1\boxtimes B_1)\otimes (A_2\boxtimes B_2)\equiv (A_1\otimes A_2)\boxtimes (B_1\otimes B_2)
}
and braiding
\eq{
\beta^{\Asf\boxtimes \Bsf}_{(A_1\boxtimes B_1),(A_2\boxtimes B_2)}=\beta_{A_1,A_2}\boxtimes \beta_{B_1,B_2}\; .
}
For a finitie ribbon category $\Csf$, let $\ov{\Csf}$ be the category with the same objects and morphisms, but with inverse braiding and twist.
In its most general form (even dropping semisimplicity) the following theorem is proven in \cite{SHIMIZU2019106778}. For semisimple categories, a proof is given in \cite{Muger}.
\begin{theo}\cite[Theorem~3.3]{SHIMIZU2019106778}\cite{Muger}
Let $\Csf$ be a finite ribbon category. If $\Csf$ is modular, there is a braided equivalence
\eq{
\Csf\boxtimes \ov{\Csf}&\rightarrow \Zsf(\Csf)\\
(A\boxtimes B)& \mapsto (A\otimes B,\beta^{ou}_{A\otimes B,\bullet})\quad .
}
Conversely, if $\Csf\boxtimes \ov{\Csf}\simeq \Zsf(\Csf)$ are braided equivalent, $\Csf$ is modular.
\end{theo}

This implies in particular that the finite set of simple objects in $\Zsf(\Csf)$ is given by $(U_i\otimes U_j,\beta^{ou}_{ij,\bullet})$, for $i,j\in \Isf(\Csf)$.

\section{Cardy Algebras}\label{sec2}

\subsection{Motivation from Physics}

In a series of paper \cite{Huang:2005gz} \cite{Kong:2006wb} \cite{Huang:2006ar} \cite{Huang9964} \cite{huang1997intertwining} \cite{Huang:2002mx}\cite{huang2005differential}\cite{huang2004open} Huang and Kong gave a rigorous formulation of genus 0,1 two dimensional open-close conformal field theory in the language of partial operads. A textbook account of the results appeared in \cite{Huangbook}. The major outcome can be described in purely categorical terms and is given by the notion of a Cardy algebra. The abstract formulation and its relation to sewing constraints was developed in \cite{Kong_2009,Kong_2014}. Though we don't need it in the core of the paper, we still spend the next paragraph giving some intuition from CFT for the abstract formulation about to come.

One way of formalizing two dimensional CFT is given by vertex operator algebras (VOA) which describe the chiral and antichiral symmetry algebras in full CFT. Roughly speaking, a VOA encodes the operator-state correspondence and operator product expansions (OPEs). It has an underlying graded state space $V$ and a vertex operator map $Y:V\rightarrow \Endrm (V)[\![z^{-1},z]\!]$, where $z$ is a (formal) coordinate on the complex plane. There is a well studied notion of representations of VOAs, including fully reducible and irreducible representations. Under the assumption that $V$ is a \textit{rational} VOA,  \footnote{For the precise definition see e.g. \cite[Theorem~3.9]{doi:10.1142/S0219199705001799}.} its representation category $\Rcal_V$ is a modular tensor category. Assuming that a CFT at hand has chiral symmetry algebra $V^L$ and antichiral symmetry algebra $V^R$, both of which are rational, the closed state space decomposes into a sum $\Hcal_{cl}=\bigoplus_{ij} N_{ij} H^L_i\times  H^R_j$, where $H^L_i$, $H^R_j$ are the simple representations in $\Rcal_{V^L}$ and $\Rcal_{V^R}$, respectively. Hence it is naturally an object in $\Rcal_{V^L}\boxtimes \Rcal_{V^R}$. From a physics perspective the crucial object to compute are correlation functions, which in our situation split into products of chrial and antichiral correlation functions. It is well know that chiral correlation functions have an expansion in terms of so called \textit{conformal blocks} which contain all the information about conformal weights and insertion points of the chiral insertions. For a vector space $V$ and a complex number $z$ we denote 
\eq{
V\lbr z\rbr =\lbr \sum_{n\in \mathbb{Q}} v_n z^n\, \middle| \, v_n\in V\rbr 
}
for the space of fractional power series with coefficients in $V$. Three point conformal blocks $\Bcal(v_1,v_2,v_3)$ on the sphere for field insertions $v_i\in H_i^L$ at $(z_3,z_2,z_1)=(\infty,z,0)$ can be described in terms of \textit{intertwining operators} 
\eq{
\Ycal_{12}^3(\bullet, z):H_1\rightarrow \hom(H_2,H_3)\lbr z\rbr
}
satisfying
\eq{
\Bcal(v_1,v_2,v_3)=\la v_3,\Ycal(v_1,z)v_2\ra_{H_3}\, ,
}
where $\la \, \bullet\, ,\, \bullet\, \ra_{H_3}$ is a well defined invariant inner product on $H_3$\footnote{We again refer to e.g. \cite[section~3]{Huang:2005gz} for the precise details of invariant bilinear forms on representations of VOAs.}. The map $\Ycal_{12}^3(\bullet, z)$ is said to be of type $\binom{H_3}{H_1 H_2}$ and the dimension of the vector space of intertwiners of type $\binom{H_k}{H_i\, H_j}$ are precisely the fusion rules $N_{ij}^k$. Intertwining operators have an algebra structure, the so called \textit{intertwining operator algebra (IOA)}. The upshot is that the OPE algebra in the CFT can be conveniently casted in the form of intertwining operators and the state space $\Hcal_{cl}$ inherits an algebra structure in $\Rcal_{V^L}\boxtimes \Rcal_{V^R}$ from the tensor product of chiral and antichiral IOAs. Under the assumption that $\Hcalcl$ carries a non degenerate invariant bilinear form (which we assume in the presentation of intertwiners above already) $\Hcal_{cl}$ becomes a \textit{commutative Frobenius algebra with trivial twist in $\Rcal_{V^L}\boxtimes \Rcal_{V^R}$}. The genus one enhancement is possible if it is in fact a modular invariant Frobenius algebra, a notion we discuss shortly in categorical terms. So far the discussion was solely for closed states.

Including boundaries one also has to consider open states, which due to their localization on some boundary, have only half of the symmetry of closed states. Hence a boundary CFT with symmetry algebra $V$ has an open state space $\Hcal_{op}=\bigoplus_i N_i\, H_i$. Following the same lines as in the closed case, the open state space becomes a symmetric Frobenius algebra in $\Rcal_V$. Note that it will in general not be commutative owing to the fact that field insertions on an interval can't be interchanged along the interval. 

Lastly boundary and bulk fields should interact, i.e. there are bulk-boundary OPEs for bulk fields approaching the boundary. This should correspond to a map $\iota_{cl-op}:\Hcal_{cl}\rightarrow \Hcal_{op}$ satisfying certain compatibility relations. For this to work, we have to assume that left and right symmetry algebra of the closed theory agree. First of all, $\iotaclop$ should be an algebra map, since taking first bulk OPEs and then approaching the boundary and taking bulk-boundary OPEs better give the same as first approaching the boundary and taking bulk-boundary OPEs followed by taking boundary OPEs. Next it should be compatible with boundary OPEs and lastly it should commute with boundary OPEs as bulk fields can be transported along the bulk as shown in figure \ref{closedthroughopen}.

\begin{figure}[H]
\centering
\begin{tikzpicture}
\node at (0,0) {\includegraphics[scale=0.2]{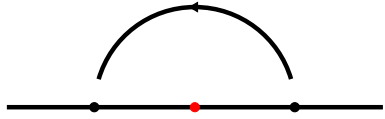}};
\end{tikzpicture}
\caption{The closed field insertion is moved along the half circle through the bulk past the open insertion (red dot).}
\label{closedthroughopen}
\end{figure}

Dually we could consider a map $\iota^\ast_{cl-op}:\Hcal_{op}\rightarrow \Hcal_{cl}$ mapping open field insertions to boundary states in the closed theory. Though only discussed heuristically we note that this can be made entirely concrete. It is shown e.g. in \cite[Proposition~2.8]{Kong_2008} that modulo technicalities $\iota^\ast_{cl-op}$ constructs Ishibashi states. It is well known that Ishibashi states in general don't correspond to true boundary states. Only linear combinations of Ishibashi states satisfying the \textit{Cardy condition} are valid boundary states due to open-close duality. 

\subsection{Categorical Definition of Cardy Algebras}

As stated in section \ref{sec1}, for a modular tensor category $\Csf$, there is a braided equivalence $\Csf\boxtimes \ov{\Csf}\simeq \Zsf(\Csf)$. As an open-closed CFT necessarily has coincident left and right symmetry algebras, the appropriate representation category for the closed theory to take place, is $\Rcal_{V\otimes V}\simeq \Rcal_V\boxtimes \ov{\Rcal_V}$. To match the description for string-net spaces we formulate Cardy algebras in terms of $\Zsf(\Rcal_V)$ instead of $\Rcal_V\boxtimes \ov{\Rcal_V}$.

\begin{defn}\cite[Definition~3.7]{Kong_2009}
Let $\Csf$ be a modular tensor category. A \textit{$(\Csf|\Zsf(\Csf))$-Cardy algebra} $(\Hcalcl,\Hcalop,\iotaclop)$ is the data of 
\begin{enumerate}[label=\Alph*)]
\item a commutative symmetric Frobenius algebra $(\Hcalcl,\mcl,\etacl,\Dcl,\epcl)$ in $\Zsf(\Csf)$.
\item a symmetric Frobenius algebra $(\Hcalop,\mop,\etaop,\Dop,\epop)$ in $\Csf$.
\item a morphism $\iotaclop\in \hom_{\Zsf(\Csf)}(\Hcalcl,L(\Hcalop))$.
\end{enumerate}
This has to satisfy the following conditions
\begin{enumerate}[label=\Roman*)]
\item $\Hcalcl$ has to be \textit{modular}, i.e. there is the equality

\begin{center}
\begin{tikzpicture}
\node at (0,0) {\includegraphics[scale=0.2]{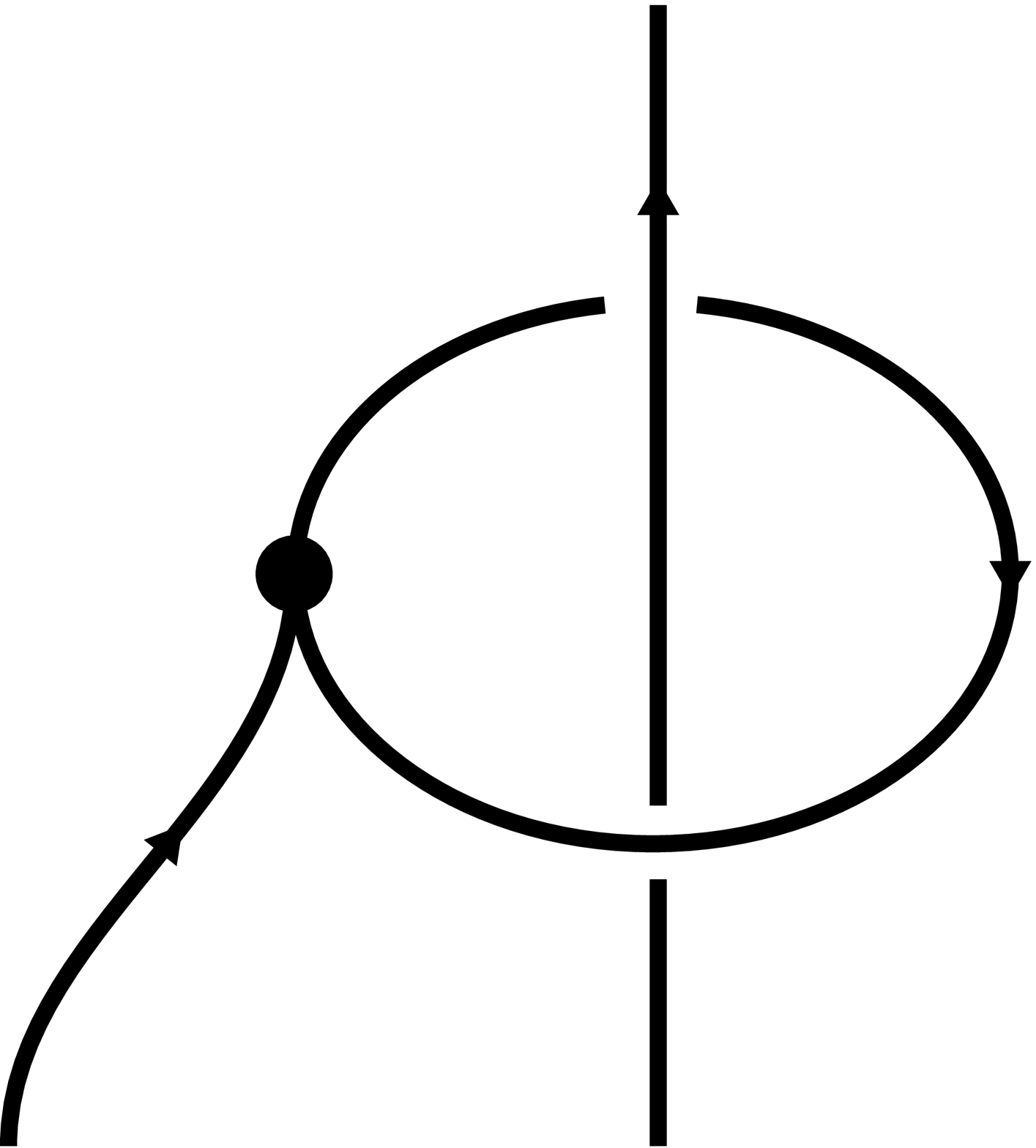}};
\node at (5,0) {\includegraphics[scale=0.2]{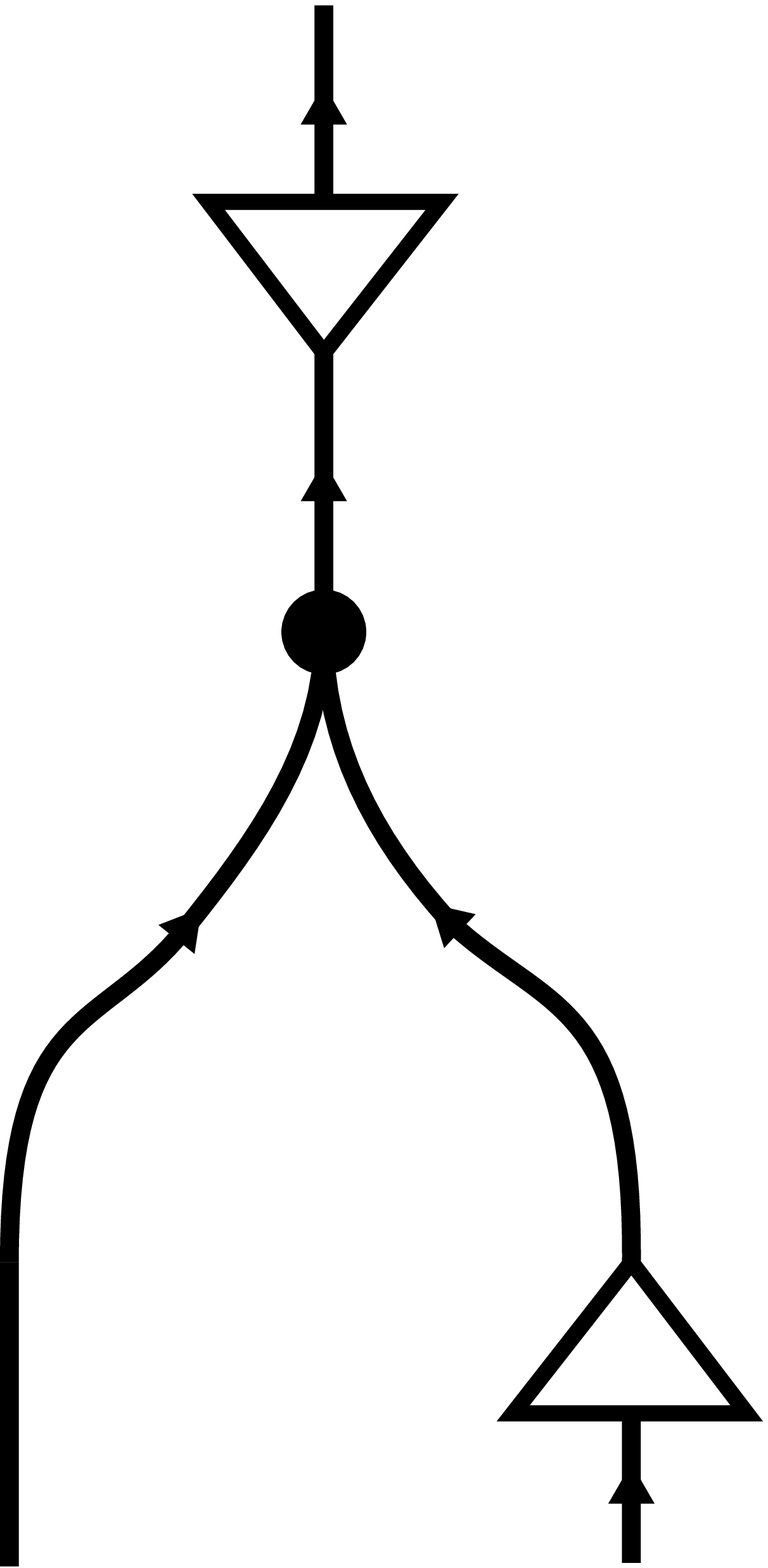}};
\node at (2.3,0) {$=$};
\node at (-2,-1.5) {$\Hcalcl$};
\node at (1,-1.5) { $i\otimes j$};
\node at (-2,0) {$\begin{aligned}\frac{d_id_j}{\Dsf^2}\end{aligned}$};
\node at (4.2,-1.8) {$\Hcalcl$};
\node at (6.3,-2.5) {$i\otimes j $};
\node at (5.8,-1.85) {\scriptsize $\alpha$};
\node at (5.2,0.9) {$\Hcalcl$};
\node at (5.4,2.2) {$i\otimes j$};
\node at (4.82,1.72) {\scriptsize $\alpha$};
\node at (3,0) {$\begin{aligned}\sum_{\alpha} \end{aligned}$};
\node at (6.2,-1.1) {$\Hcalcl$};
\end{tikzpicture}.
\end{center}
\item $\iotaclop$ is an algebra homomorphism.
\item The \textit{center condition} holds:
\begin{center}
\begin{tikzpicture}
\node at (0,0) {\includegraphics[scale=0.2]{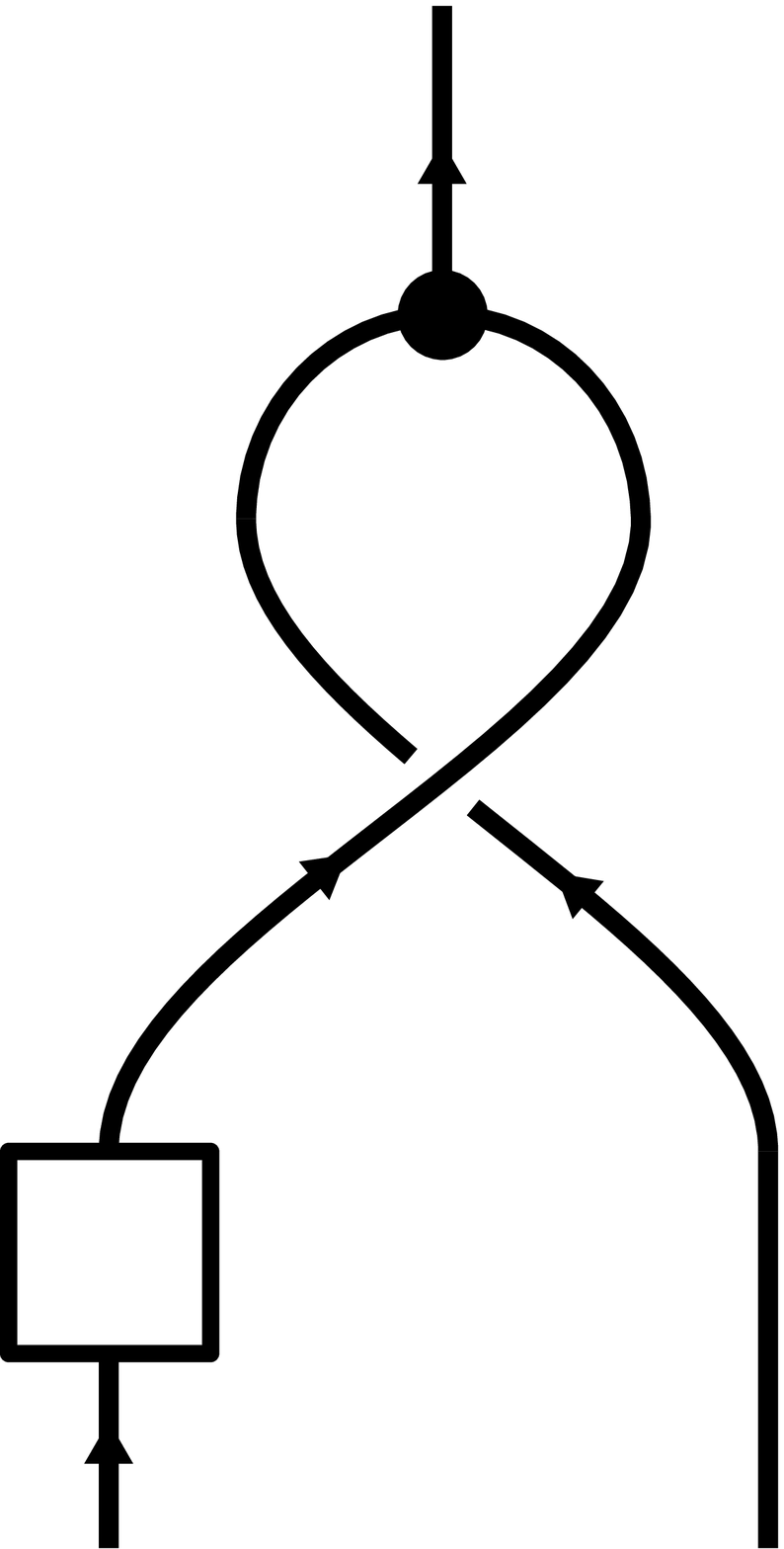}};
\node at (4,0) {\includegraphics[scale=0.2]{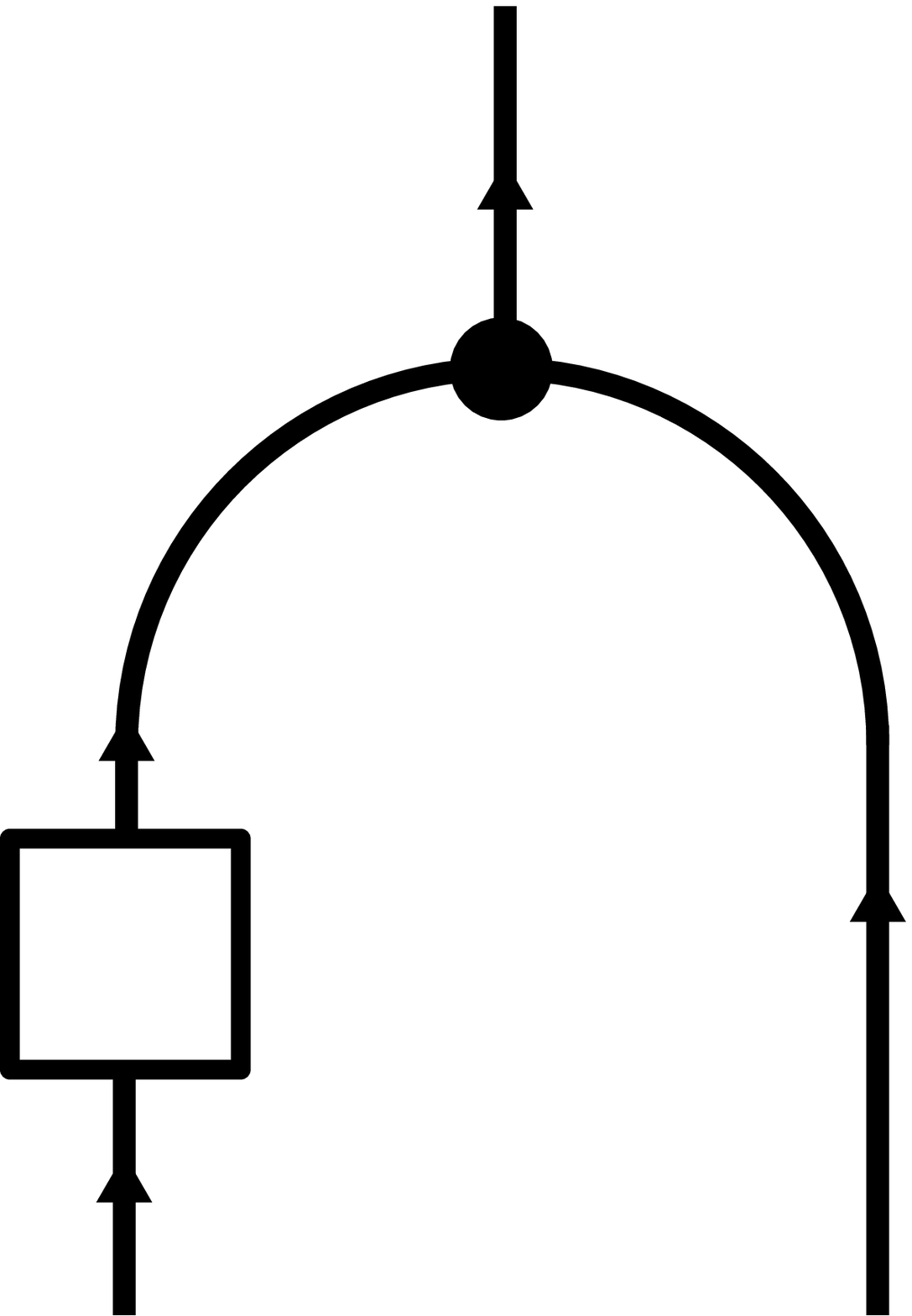}};
\node at (2,0) {$=$};
\node at (-1.4,-2.1) {$\Hcalcl$};
\node at (2,-2) {$L(\Hcalop)$};
\node at (1,2) {$L(\Hcalop)$};
\node at (-0.8,-1.45) {$\iota$};
\node at (3.1,-0.8) {$\iota$};
\node at (3.5,-1.6) {$\Hcalcl$};
\node at (6,-1.6) {$L(\Hcalop)$};
\node at (5,1.5) {$L(\Hcalop)$};
\end{tikzpicture}.
\end{center}
\item The \textit{Cardy condition} holds:
\begin{center}
\begin{tikzpicture}
\node at (0,0) {\includegraphics[scale=0.2]{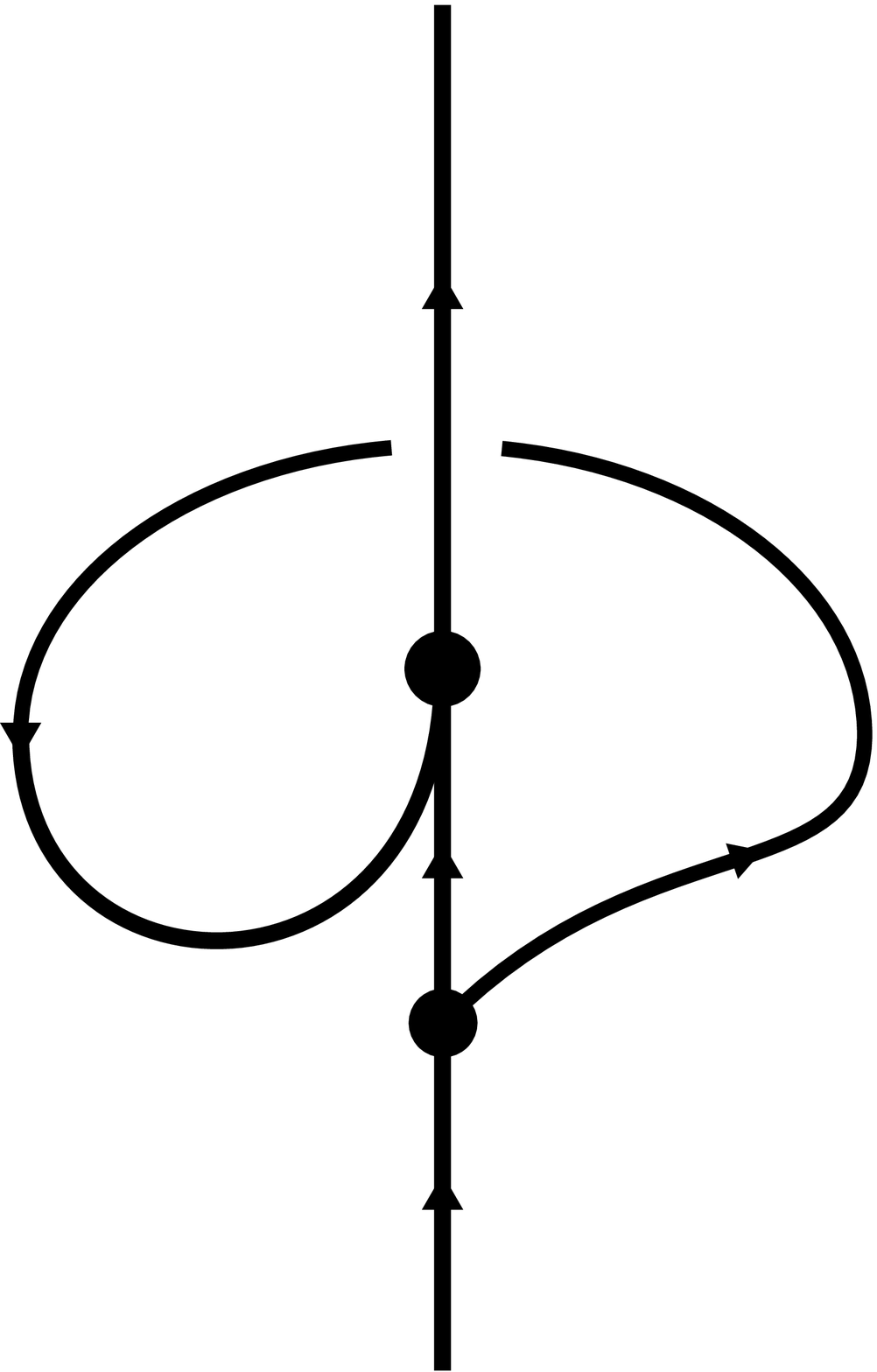}};
\node at (4,0) {\includegraphics[scale=0.2]{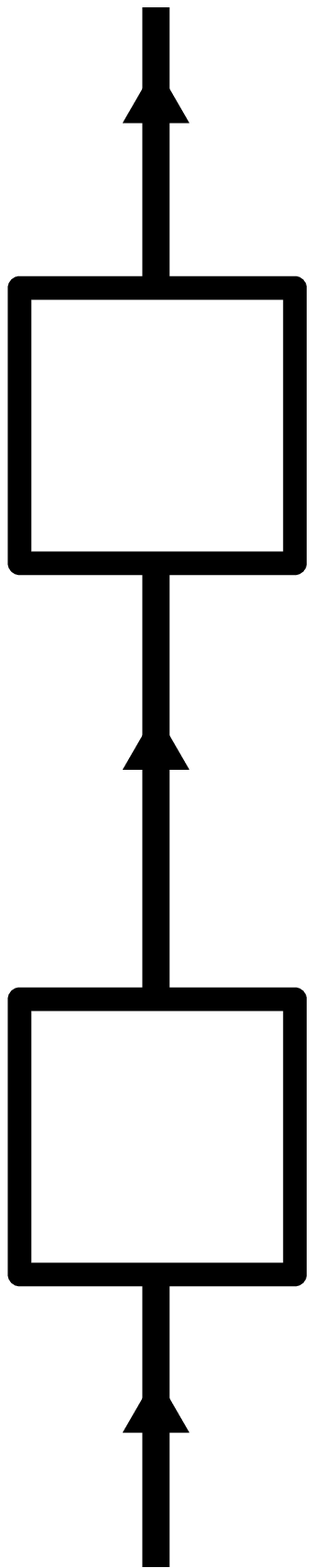}};
\node at (2,0) {$=$};
\node at (1,-2) {$L(\Hcalop)$};
\node at (5,-1.5) {$L(\Hcalop)$};
\node at (4,-0.75) {$\iota^\dagger$};
\node at (4,0.78) {$\iota$};
\node at (5,1.5) {$L(\Hcalop)$};
\node at (4.5,-0.1) {$\Hcalcl$};
\end{tikzpicture}.
\end{center}
\end{enumerate}
\end{defn}

\begin{defn} A morphism between $(\Csf|\Zsf(\Csf))$-Cardy algebras $(\Hcalcl,\Hcalop,\iotaclop)$ and $(\Gcalcl,\Gcalop,\iotaclop^\prime)$ is a pair of maps $f_{cl}\in \hom_{\Zsf(\Csf)}(\Hcalcl,\Gcalcl)$, $f_{op}\in \hom_\Csf(\Hcalop,\Gcalop)$ s.th.:
\begin{enumerate}[label=\Roman*)]
\item Both, $f_{cl}$ and $f_{op}$, are homomorphisms of Frobenius algebras.
\item The following diagram commutes
\begin{center}
\begin{tikzcd}
\Hcalcl\ar[rr, "f_{cl}"]\ar[dd, "\iotaclop"] & & \Gcalcl\ar[dd, "\iotaclop^\prime"]\\
& & \\
L(\Hcalop) \ar[rr, "L(f_{op})"'] & & L(\Gcalop)\; .
\end{tikzcd}
\end{center}

\end{enumerate}
\end{defn}
Using the map $(\bullet)^\dagger$, it is not hard to show that any morphism of Frobenius algebras has an inverse (see \cite[Lemma~2.18]{Kong_2009}). Thus, any morphism of Cardy algebras is in fact an isomorphism.

\section{String-Net Spaces}\label{sec3}

String-net spaces can be seen as a higher genus enhancement of graphical calculus for spherical categories which reduces to the usual graphical calculus on every embedded disk. For a surface $S$ with boundary $\p S$, let $\Gamma\subset S$ be an embedded, finite, oriented graph, which we always consider up to isotopy. Intersection points of $\Gamma$ with $\partial S$ are required to be vertices of valence one of $\Gamma$. For an oriented edge $\mathbf{e}$ the same edge with reversed orientation is denoted by $\ov{\mathbf{e}}$. 
\begin{defn} Let $\Csf$ be a spherical fusion category. A $\Csf$-coloring of $\Gamma$ consists of two parts: First, an assignment of an object $A(\mathbf{e})\in \Csf$ to any oriented edge $\mathbf{e}$ s.th. $A(\ov{\mathbf{e}})=A(\mathbf{e})^\ast$. Second, to a vertex $v$ with incident edges $\mathbf{e}_1,\dots, \mathbf{e}_n$, taken in counterclockwise order, the $\Csf$-coloring assigns an element
\eq{
\phi_v\in \la A(\mathbf{o}_1),\dots , A(\mathbf{o}_n)\ra , 
}
where $\mathbf{o}_i$ is the edge $\mathbf{e}_i$ oriented away from $v$. Hence, if an edge is oriented towards a vertex and colored with $A$, the morphism  corresponding to the vertex is an element in $\la \cdots A^\ast\cdots \ra $. As an example, consider a vertex $v$ with 6 incident edges $\lbr \mathbf{e}_1,\cdots, \mathbf{e}_6\rbr$, two of which are incoming, the rest outgoing. A $\Csf$-coloring is then given by

\begin{center}
\begin{tabular}{m{4cm} m{7cm}}
\includegraphics[scale=0.15]{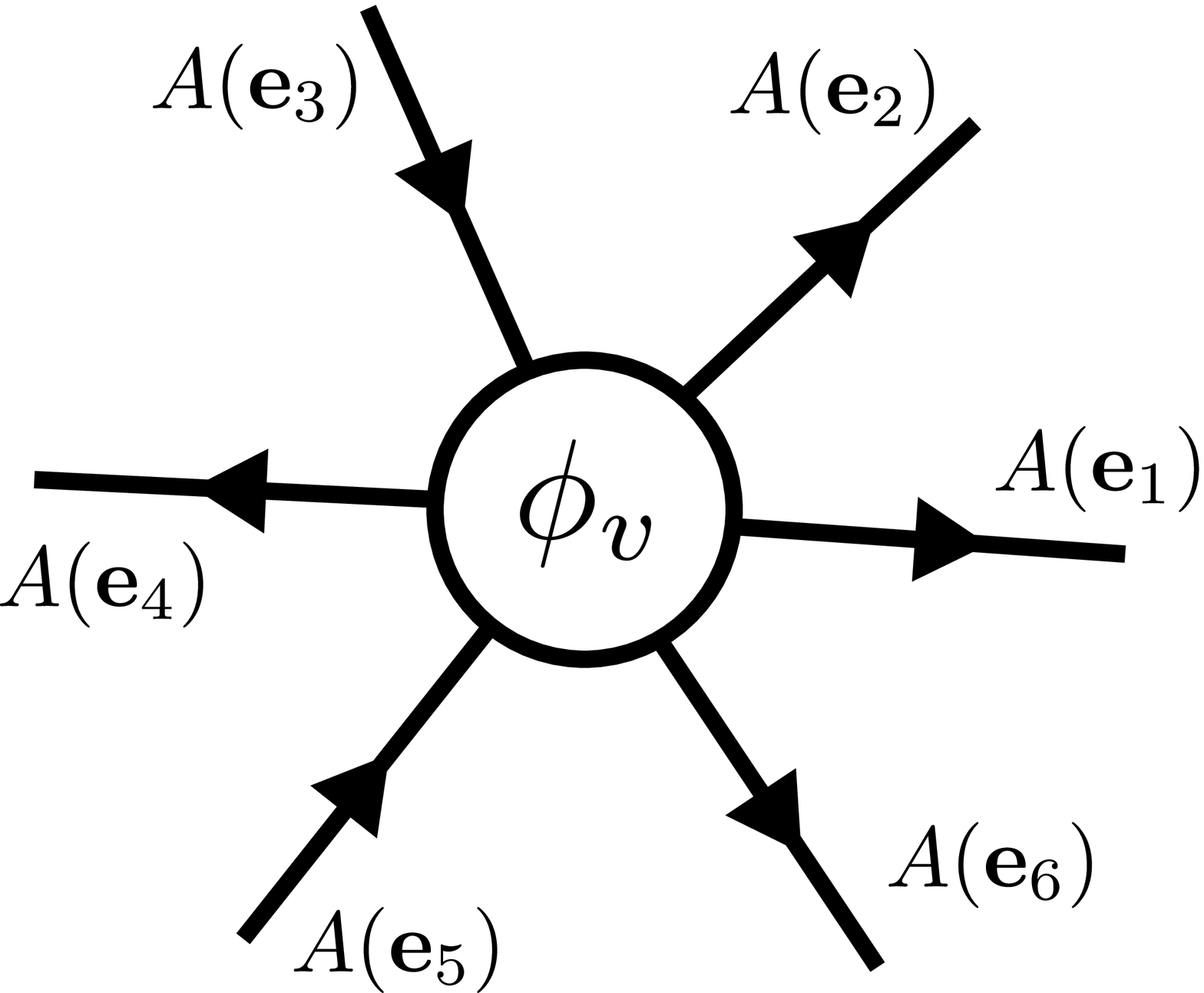} & $\phi_v\in \la A(\mathbf{e}_1),A(\mathbf{e}_2),A(\mathbf{e}_3)^\ast,A(\mathbf{e}_4),A(\mathbf{e}_5)^\ast,A(\mathbf{e}_6)\ra $.
\end{tabular}
\end{center}

An \textit{isomorphism of $\Csf$-colorings} is a collection of isomorphisms $f_\mathbf{e}:A(\mathbf{e})\rightarrow B(\mathbf{e})$ respecting orientations and mapping $\phi_v=f\circ \phi^\prime_v$. Here, $\phi_v$, $\phi^\prime_v$ are the maps assigned to the vertex $v$ by the two $\Csf$-colorings. The \textit{boundary value} for a $\Csf$-colored graph is a tuple $\left(\lbr p_1,\cdots , p_n\rbr ,\, \lbr A_1,\cdots, A_n\rbr \right)$, with $\lbr p_1,\cdots, p_n\rbr\simeq \lbr \Gamma\cap \p S\rbr $, and $A_i$ is the $\Csf$-color of the edge incident to the boundary vertex $v_i$, which corresponds to the intersection point $p_i$.
\end{defn}

Colored vertices have the graphical representation as coupons introduced in section \ref{sec1}. Let $\Gamma$ be an embedded graph as above and $D\subset S$ be an embedded closed disk whose boundary $\partial D$ is transversal to $\Gamma$. Let $\lbr A(\mathbf{e}_1),\dots , A(\mathbf{e}_n)\rbr $ be the colors of edges of $\Gamma$ intersecting $\p D$, taken in counterclockwise order. Then there is a unique surjective evaluation map \cite[Theorem~2.3]{2011arXiv1106.6033K}
\eq{
\la \bullet \ra_D:\Gamma\cap D\rightarrow \la A(\mathbf{o}_1),\dots, A(\mathbf{o}_n)\ra , 
}
satisfying a list of natural properties explicitly given in \cite[Theorem~2.3]{2011arXiv1106.6033K}. Here, similar to the definition of a $\Csf$-coloring, we use the notation $\mathbf{o}_i$ for the edge $\mathbf{e}_i$ with orientation towards the boundary $\p D$. The properties stated in \cite[Theorem~2.3]{2011arXiv1106.6033K} include local relations $\Gamma_1\cap D= \Gamma_2\cap D$, which have to be understood in the sense that $\la \Gamma_1\cap D \ra_D=\la \Gamma_2\cap D\ra_D$. One of these relations is e.g. 

\begin{center}
\begin{tikzpicture}
\node at (0,0) {\includegraphics[scale=0.2]{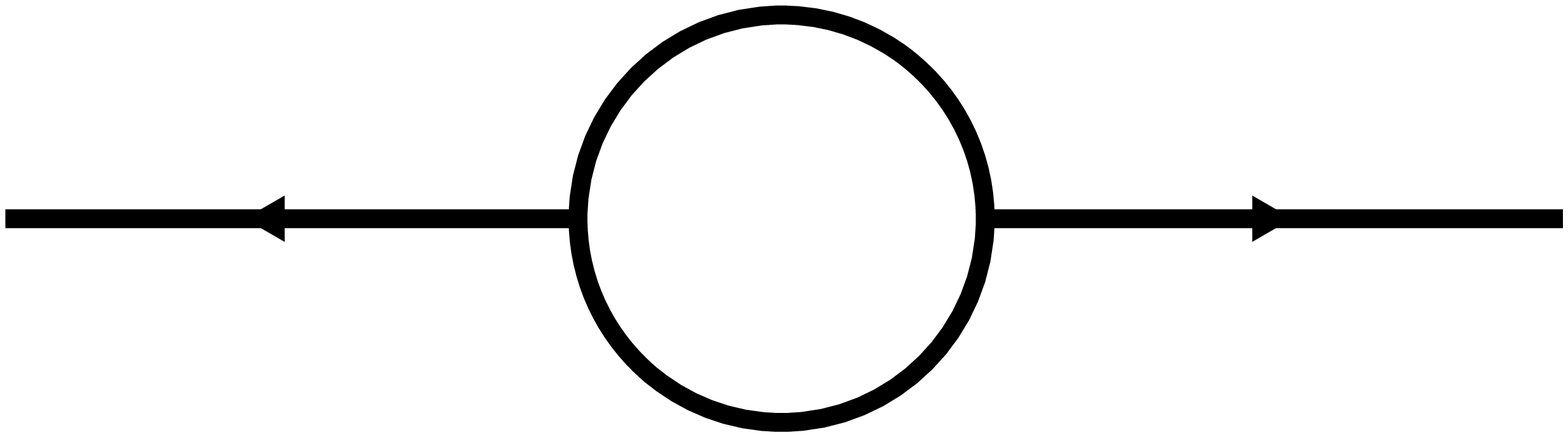}};
\node at (6,0) {\includegraphics[scale=0.2]{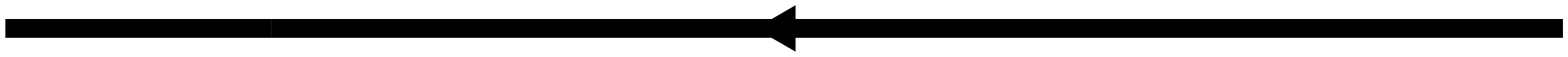}};
\node at (3,0) {$=$};
\node at (0,0) {$\coevrm_A$};
\node at (-2,-0.3) {$A$};
\node at (2,-0.3) {$A$};
\node at (3.8,-0.3) {$A$};
\end{tikzpicture}.
\end{center}
Another property is e.g. the fact, that vertices lying in a disk can always be merged into a single coupon.

\begin{defn} To any surface $S$ there is an associated vector space
\eq{
V\Graphrm(S, \mathbf{A})&=\text{formal finite $\Kbb$-linear combinations of colored graphs with }\\
&\phantom{=================}\text{boundary value $\mathbf{A}$} 
}
\end{defn}

\begin{defn} Let $\Gamma_i$ be $\Csf$-colored graphs and $x_i\in \mathbb{K}$, let $\mathbf{\Gamma}=\sum x_i \Gamma_i \in V\Graphrm(S,\mathbf{A})$. The vector $\mathbf{\Gamma}$ is called a \textit{null graph}, if there exists an embedded disk $D\subset S$ intersecting $\mathbf{\Gamma}$ transversally, s.th. $\Gamma_i|_{S\backslash D}=\Gamma_j|_{S\backslash D}$ and
\eq{
\la \mathbf{\Gamma}\ra_D=\sum x_i \la \Gamma_i \ra_D=0\quad .
}
The vector space of all null graphs with fixed boundary value will be denoted $N\Graphrm(S,\mathbf{A})$. 
\end{defn}

\begin{defn} The \textit{string-net space} on a surface with boundary value $\mathbf{A}$ is defined to be 
\eq{
H(S,\mathbf{A})=\frac{V\Graphrm(S,\mathbf{A})}{N\Graphrm(S,\mathbf{A})}\quad .
}
\end{defn}

So far boundary conditions are just sets of points on the boundary labeled by objects of $\Csf$. Of course boundary conditions should be subject to some natural relation as described in \cite[section~6]{2011arXiv1106.6033K}, which turn them into a category of boundary conditions. Let $N$ be an oriented one dimensional manifold and $\lbr p_1,\dots, p_n\rbr \subset N$ a finite subset of points. Let $\Bsf(N)$ be the category with objects $\left( \lbr p_1,\dots, p_n\rbr , \lbr B_1,\dots , B_n\rbr \right)\equiv \mathbf{B}$, where $B_i\in \Csf$ and morphism spaces are given by
\eq{
\hom_{\Bsf(N)}(\mathbf{B},\mathbf{B}^\prime )\equiv H(N\times I; \mathbf{B}^\ast, \mathbf{B}^\prime),
}
where we denote $\mathbf{B}^\ast=\left(\lbr p_1,\dots, p_n\rbr , \lbr B_1^\ast,\dots , B_n^\ast\rbr \right)$. Composition of morphisms is given by stacking cylinders on top of each other and concatenating string-nets across the internal copy of $N$, followed by a rescaling of the cylinder to unit length. The category of boundary values is defined to be the Karoubi envelope of $\Bsf(N)$, which by abuse of notation will be denoted by the same symbol. This category has all the nice properties to be expected, e.g. $\Bsf(N)\simeq \Bsf(N^\prime)$ for $N\simeq N^\prime $ and $\Bsf(N\sqcup N^\prime)\simeq \Bsf(N)\boxtimes \Bsf(N^\prime)$.

\begin{lem}\cite[Theorem~6.4]{2011arXiv1106.6033K} There are equivalences of categories $\Bsf(S^1)\simeq \Zsf(\Csf)$ and $\Bsf(\R)\simeq \Csf$.
\end{lem}

Using this, one can give an enhancement of string-net spaces taking excited states\footnote{For the terminology we refer to \cite{2011arXiv1106.6033K}.} on the boundary of a surface into account. Let $S$ be a surface with boundary $\p S$ and $\mathbf{B}\in \Bsf(\p S)$. The extended string-net space is defined as the quotient 
\eq{
\hat{H}^s(S,\mathbf{B})=\frac{V\Graphrm (S,\mathbf{B})}{N\Graphrm (S,\mathbf{B})}
}
where 
\eq{
V\Graphrm (S,\mathbf{B})&=\text{formal vector space of finite $\Kbb$-linear combinations of pairs}\, (f,\Gamma) \\
&\phantom{==} \text{with}\, \Gamma \, \text{a graph on } S \, \text{with boundary value}\, \mathbf{A} \\
&\phantom{==}\text{and}\,  f\in \hom_{\Bsf(\p S)}(\mathbf{A},\mathbf{B})
}
and 
\eq{
N\Graphrm(S,\mathbf{B})&=\text{subspace of null graphs under local relations as before plus}\\
 &\phantom{==} \text{relation} \, (f\gamma,\Gamma)=(f,\gamma\Gamma)\, 
 \text{where} \, \gamma\in \hom_{\Bsf(\p S)}(\mathbf{B},\mathbf{B}^\prime), \\
 &\phantom{==} f\in \hom_{\Bsf(\p S)}(\mathbf{B}^\prime, \mathbf{A})
}
The following is a result of a series of papers \cite{KirillovBalsam}\cite{balsam2010turaevviro}\cite{balsam2010turaevviro} and \cite[Theorem~7.3]{2011arXiv1106.6033K}.
\begin{theo}\label{Stringnetvectorspace}
Let $S$ be a compact oriented surface of genus $g$ with boundary parameterized circles and objects $\mathbf{A}=\lbr A_1,\dots, A_{|\pi_0(\p S)|}=A_n\rbr $ objects in $\Zsf(\Csf)$. Then there are isomorphisms
\eq{
\hat{H}^s(S,\mathbf{A})\simeq Z_{TV,\Csf}(S,\mathbf{A})\simeq Z_{RT,\Zsf(\Csf)}(S,\mathbf{A})= \hom_{\Zsf(\Csf)}(\mathbf{1},A_1\otimes \dots \otimes A_n \otimes \left(L\right)^g)
}
where in the last vector space $L=\bigotimes_{i\in \Isf(\Zsf(\Csf))} U_i\otimes U_i^\ast $.
\end{theo}

The following lemma is straightforward. 
\begin{lem} Let $D$ be a closed disk and $A\in \Csf$. Then $\hat{H}^s(D,L(A))=\hom_\Csf(\mathbf{1},A)$.
\end{lem}

Let $P\in \hom_{\Bsf(S^1)}(A,A)$ be the string-net

\begin{figure}[H]
\centering
\begin{tikzpicture}
\node at (0,0) {\includegraphics[scale=0.2]{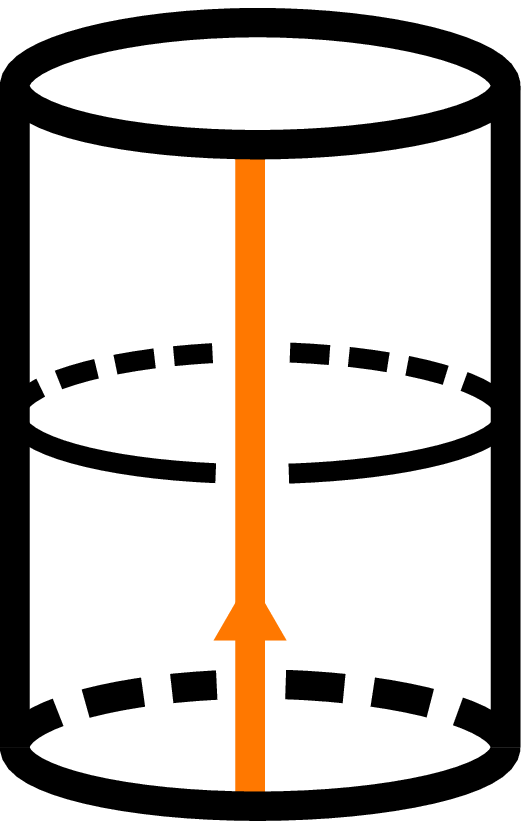}};
\node at (4,0) {\includegraphics[scale=0.2]{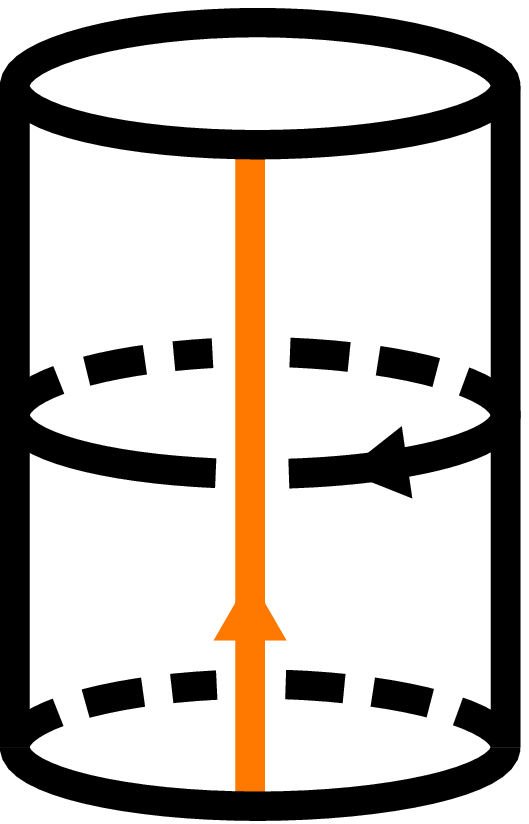}};
\node at (1.3,0) {$=$};
\node at (2.3,0) {$\begin{aligned} \sum_{k\in \Isf(\Csf)}\frac{d_k}{\Dsf^2}\end{aligned}$};
\node at (4.3,-0.3) {\scriptsize $k$};
\end{tikzpicture}
\caption{Drinfeld-center projector}
\label{projector}
\end{figure}
which we call the \textit{projector} and the circle winding around the circumference we call the \textit{projector circle}. In \cite{2011arXiv1106.6033K} it is shown, that the extended string-net space is the image of the projector $P$ in the string-net space. Graphically this means, that elements of the extended string-net spaces are represented by string-nets with additional projectors $P$ placed at each boundary component, as shown in figure \ref{stringnet}.

\begin{figure}[H]
\centering
\begin{tikzpicture}
\node at (0,0) {\includegraphics[scale=0.2]{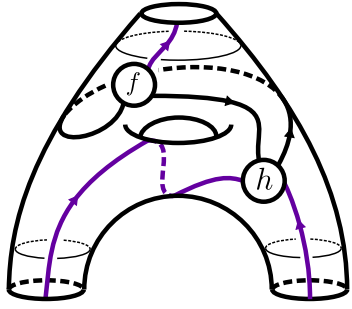}};
\end{tikzpicture}
\caption{Example of a projected string-net on a genus 1 surface.}
\label{stringnet}
\end{figure}

\section{World Sheets, Sewing Constraints and the Block Functor}\label{sec4}

\subsection{The Category of Open-Closed World Sheets} 

In order to define a consistent system of correlators we have to give an appropriate category of open-closed world sheets for which our construction computes correlators. Since we are considering open-closed world sheets, this needs a fair bit of data. Luckily an appropriate category is defined in \cite{Fjelstad:2006aw,Kong_2014} and the first part of this section recalls this definition as well as the notion of sewing constraints given in \cite[section~3.2]{Kong_2014}. Most of the problems concerning open-closed world sheets is caused by properly disentangling open and closed boundaries, which leads to the orientation double. 

\begin{defn} An \textit{open-closed world sheet} is the data 
\eq{
\hat{S}=(\tilde{S},\iota_S,B^i_S,B^o_S,\orrm_S,\ordrm)
}
where 
\begin{enumerate}[label=\Alph*)]
\item $\tilde{S}$ is an oriented topological surface with boundary $\p \tilde{S}$.
\item $\iota_S$ is an orientation reversing involution whose fixed point set is a submanifold. The quotient $S=\tilde{S}/\lbr x\sim  \iota_S(x)\rbr $ is a manifold and $\pi_S:\tilde{S}\rightarrow S $ is a $\z_2$-bundle. Thus, $\tilde{S}$ is the orientation double of $S$.
\item $B^i_S$, $B^o_S$ is a disjoint partition of $\pi_0(\partial \tilde{S})$ in incoming and outgoing boundary components which is fixed by $\iota_S$. Fixed points of the induced map $\iota_S:\pi_0(\partial\tilde{S})\rightarrow \pi_0(\partial \tilde{S})$ are called open boundaries. The set of open boundaries is denoted $B_{op}$ and its complement in $\pi_0(\partial \tilde{S})$ is $B_{cl}$.
\item $\orrm_S:S\rightarrow \tilde{S}$ is a global section of $\pi_S$.
\item $\delta_S:\p \tilde{S}\rightarrow S^1$ is a boundary parameterization being a homeomorphism on every connected component s.th. $\delta_S\circ \iota_S(y)=\ov{\delta_S(y)}$, where $\ov{\bullet}$ denotes complex conjugation and $y\in \partial\tilde{S}$. For a fixed point $b\in \pi_0(\tilde{S})$ of $\iota_S$, it has to hold $\delta|_b^{-1}(S^1\cap \Hbb)=\im(\orrm)|_b$.
\item $\ordrm:\pi_0(\p \tilde{S})\rightarrow \lbr 1, \dots, |\pi_0(\tilde{S})|\rbr $ is an ordering function of boundary components for which we first demand that $\ordrm(o)<\ordrm(c)$ for $o\in B_{op}$ and $c\in B_{cl}$. Second, for a connected set $P\subset \p S$ having non-trivial intersection with a physical boundary (see the next paragraph) and $\tilde{P}\subset \p\tilde{S}_{op}$ with $\pi_S(\tilde{P})\subset P$ there has to exist an $n\in \lbr 1,\dots , |\pi_0(\tilde{S})|\rbr$ s.th. $\ordrm(\tilde{P})=\lbr n,n+1,\dots, n+|\tilde{P}|-1\rbr $ where the ordering of components is cyclically along the orientation of $P$.
\end{enumerate}
\end{defn}

\begin{defn} A \textit{sewing} of a world sheet $\hat{S}$ has data 
\begin{enumerate}[label=\Alph*)]
\item a subset $S_B\subset B_S^i\times B^o_S$ s.th. if $(i,o)\in S_B$ there are no elements $(i,o^\prime)$ or $(i^\prime,o)$ in $S_B$.
\item for $(i,o)\in S_B$ it follows that $(\iota_S(i),\iota_S(o))\in S_B$.
\item either $(i,o)\in B^i_{op}\times B^0_{op}$ or $(i,o)\in B_{cl}^{i;l,r}\times B_{cl}^{o;l,r}$.
\end{enumerate}
The \textit{sewn world sheet} $\widehat{\Scal(S)}$ has 
\begin{enumerate}[label=\Roman*)]
\item $\widetilde{\Scal(S)}=\tilde{S}/\sim$, where $\delta|_i^{-1}(z)\sim \delta^{-1}|_o\left(\ov{-z}\right)$. Let $p_\Scal:\tilde{S}\rightarrow \widetilde{\Scal(S)}$ be the projection.
\item condition B) ensures that there is a well defined involution $\iota_{\Scal(S)}$, defined via $\iota_{\Scal(S)}\circ p_\Scal=p_\Scal\circ \iota_S$.
\item $B^i_{\Scal(S)}=\lbr i\in B_S^i| (i,\bullet)\notin S_B\rbr $ and $B^o_{\Scal(S)}=\lbr o\in B_S^o| (\bullet,o)\notin S_B\rbr $
\item $\orrm_{\Scal(S)}$ is the unique section whose image in $\widetilde{\Scal(S)}$ is the image of $\pi_{\Scal(S)}\circ \orrm_S$.

\end{enumerate}

\end{defn}

There is an additional requirement on the ordering function. For the details we refer to \cite{Kong_2014}. Note that the glueing defined above gives in addition a glueing projection $\Scal_S:S\rightarrow \Scal(S)$ given by $\Scal_S=\pi_{\Scal(S)}\circ  \pi_\Scal\circ \orrm_S$. We will be mainly concerned with $S$ instead of its orientation double $\tilde{S}$. Its boundary components decompose into three different types: open, closed and physical.
\begin{figure}[h]
\centering
\includegraphics[scale=0.15]{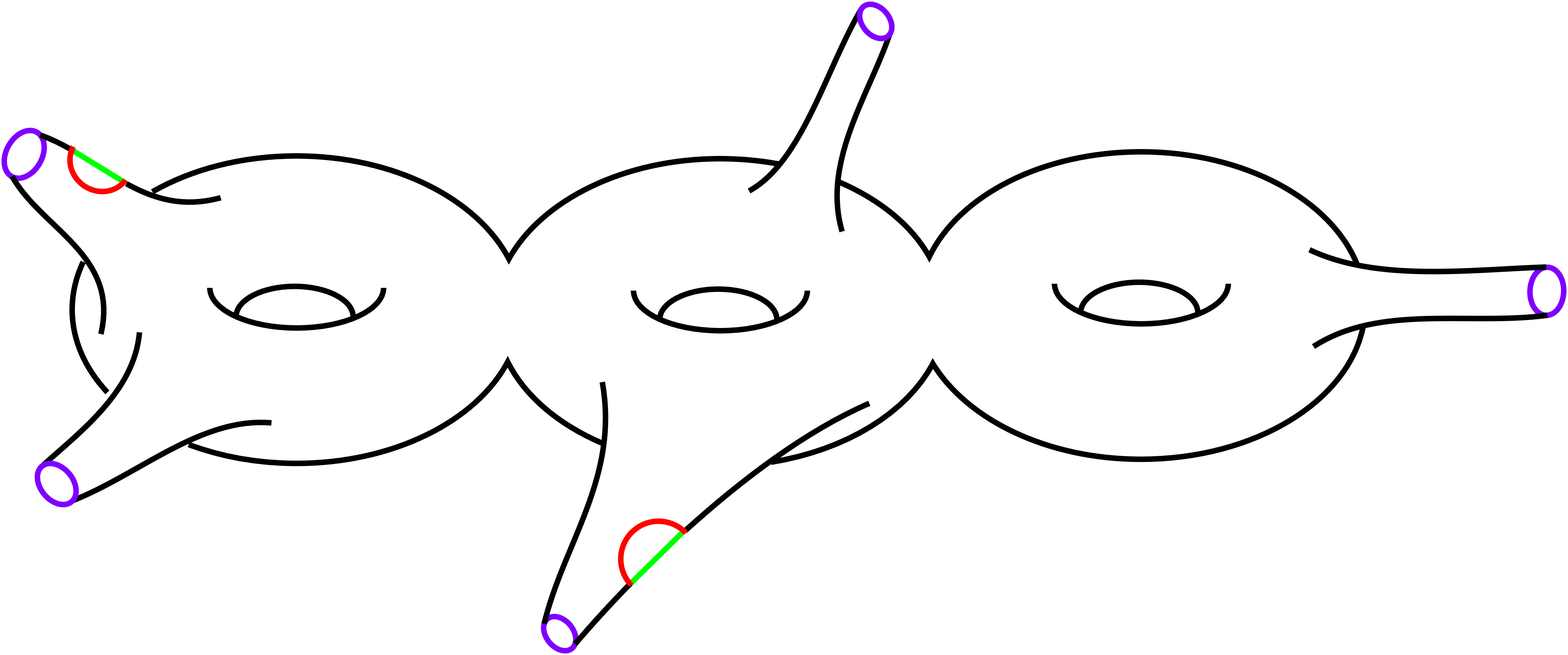}
\caption{The quotient surface of a genus 3 open-closed world sheet with closed boundaries shown in purple. Open boundaries are colored green and physical boundaries are red.}
\end{figure}

\begin{enumerate}[label=\roman*)]
\item Closed boundaries: A point $p\in \p S$ is on a closed boundary if its preimages under $\pi_S$ lie on different connected components of $\p \tilde{S}$. This implies that connected components of closed state boundaries are homeomorphic to $S^1$. Their preimages are pairs $(b_{cl},\iota_S(b_{cl}))$ of connected components of boundaries in $\tilde{S}$ and $\orrm_S$ identifies them with one of the two boundaries.
\item Open boundaries: A point $p\in \p S$ is on an open boundary if its preimages under $\pi_S$ are on the same connected component in $\p \tilde{S}$. Hence its preimage is on a component $b_{op}\in B_{op}$. Since $\iota_S$ was orientation reversing it acts on $b_{op}$ as a reflection. A reflection on $S^1$ has two fixed points and connected components of open boundaries map to one of the open intervals stretching between the fixed points.

\item Physical boundaries: $p\in \p S$ is on a physical boundary if its preimage is on the fixed point set of $\iota_S$. In particular, preimages of physical boundaries aren't boundary components of $\tilde{S}$ except for the fixed points of $\iota_S$ on open components of $\pi_0(\tilde{S})$. Rather they correspond to curves on $\tilde{S}$ s.th. cutting $\tilde{S}$ along the curves results in two copies of $S$ mapped to each other by the involution.
\end{enumerate}

\begin{defn} A \textit{homeomorphism of world sheets} is a homeomorphism $F:\tilde{S}\rightarrow \tilde{T}$ s.th. 
\eq{
F\circ \iota_S=\iota_T\circ F,\quad \delta_T\circ F=\delta_S,\quad F(B^{i;o}_S)=B^{i,o}_T,\quad F\left(\im(\orrm_S)\right)=\im(\orrm_T)
}
\end{defn}
The last point implies in particular, that $f$ steps down to a homeomorphism $f:S\rightarrow T$ preserving all types of boundaries.

\begin{defn} The \textit{category of world sheets $\WSsf$} has objects world sheets and morphisms $\hom_\WSsf(\hat{S},\hat{T})$ are given by pairs $(\Scal,F)$ where $\Scal$ is a sewing of $\hat{S}$ and $F:\widetilde{\Scal(S)}\rightarrow \tilde{T}$ is a homeomorphism of world sheets. For the definition of the composition we refer to \cite{Kong_2014}.
\end{defn}

$\WSsf$ is a symmetric monoidal category with the usual tensor product given by disjoint union. In addition two morphisms $(\Scal_1,F_1)$, $(\Scal_2,F_2)$ in $\WSsf$ are \textit{homotopic} if $\Scal_1=\Scal_2$ and $F_1$, $F_2$ are isotopic maps.

\begin{defn} Let $\Funsf_\otimes (\WSsf,\Vectsf)$ be the category of symmetric monoidal functors assigning the same map to homotopic sewings. Morphisms are monoidal natural transformations.
\end{defn}

\subsection{Generating Set and Sewing Constraint Relations}\label{relations}

The category $\WSsf$ has a set of generating world sheets $\lbr S_i\, |\, i\in G\rbr $ which we give in appendix A. It is generating in the sense that for any other world sheet $S$, there exists a list of generating world sheets $S_1,\dots , S_n$ and a morphism $(\Scal,F):S_1\otimes \dots \otimes S_n\rightarrow S$. None of this data needs to be unique. The generating set allows to reduce the discussion of functors and natural transformations almost completely to the generating set and a set of relations among them. To be precise, consider triples of generating data $(S,\lbr S_i,\rbr ,(\Scal,F))$ and functors $\Phi,\Psi\in \Funsf_\otimes(\WSsf,\Vectsf)$. In addition assume that $\Psi((\Scal,F))$ is an invertible linear map. Consider a collection of linear maps 
\eq{
\Gcal_i:\Psi(S_i)\rightarrow \Phi(S_i),\quad i\in G
}
defined for the generating set. To any world sheet $S$ one can associate the map
\eq{
\Gcal(S)\equiv \Phi((\Scal,F)\circ (\gcal_{i_1}\otimes \dots \otimes \gcal_{i_r})\circ \Psi((\Scal,F))^{-1}
}
where $(\Scal,F):S_{i_1}\otimes \dots \otimes S_{i_r}\rightarrow S$ is the morphism from the generating property. Next, there are 32 important different glueings of world sheets. We present them in terms of 32 relations $\lbr R_i\rbr $, where the lhs $R_{i,l}$, and rhs $R_{i,r}$, of the $i$-th relation are different glueings of the same underlying manifold, or together indicate an action of the mapping class group (R24, R25 and R32). In the following we present the relations in terms of the corresponding quotient surfaces of world sheets (cf. \cite{Kong_2014}). Red curves indicate how the world sheet displayed is glued from easier parts. The blue flag on glueing curves indicates the direction of glueing. For the part containing the flag, an incoming boundary is glued. In the figures blue boundaries denote in-boundaries and green ones correspond to out-boundaries.
\begin{enumerate}[label=\Roman*)]
\item \textbf{Open Relations:}

\begin{center}
\begin{tikzpicture}
\node (R1l) at (0,0) {\includegraphics[scale=0.12]{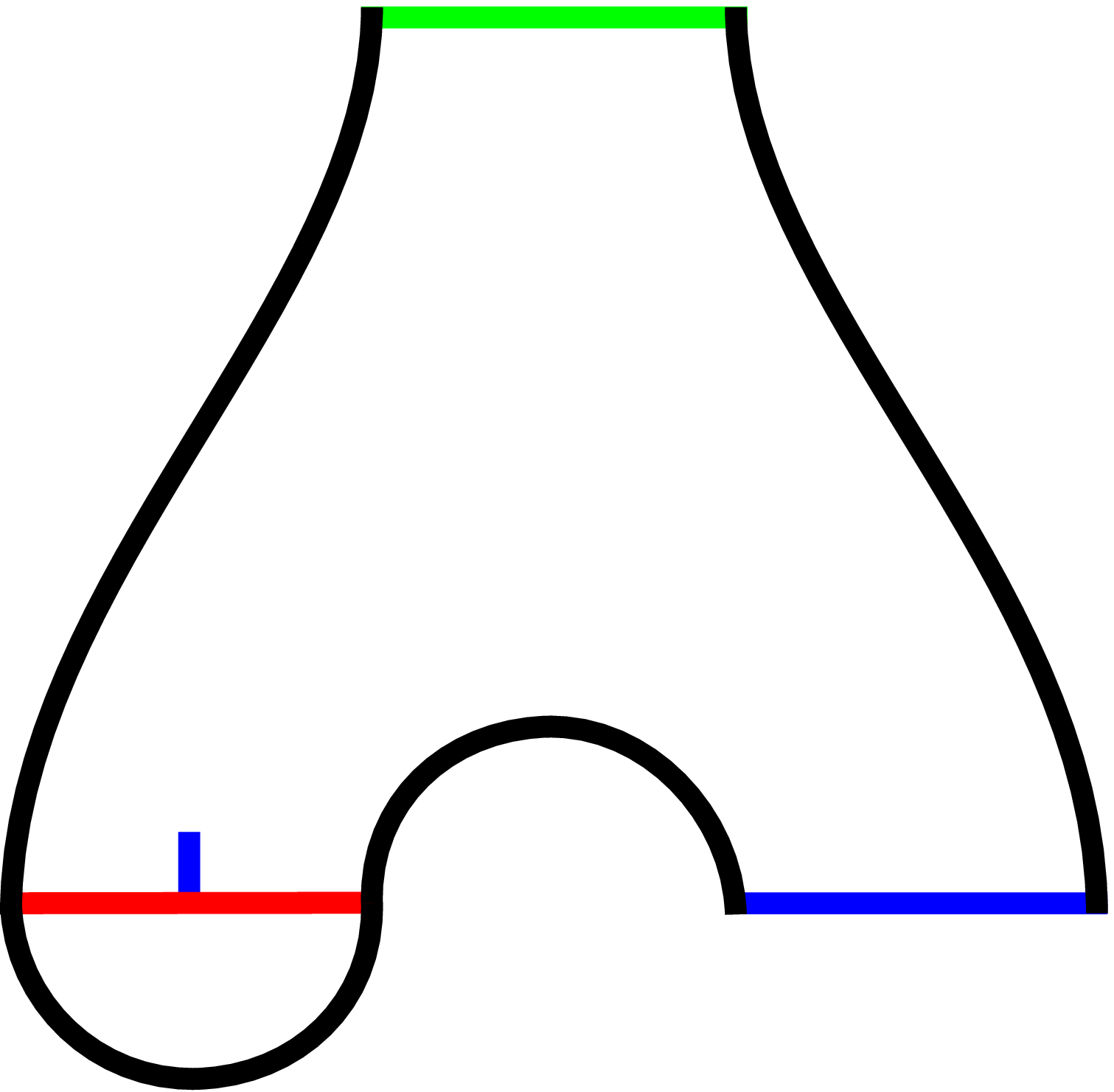}};
\node (R1r) at (3,0) {\includegraphics[scale=0.12]{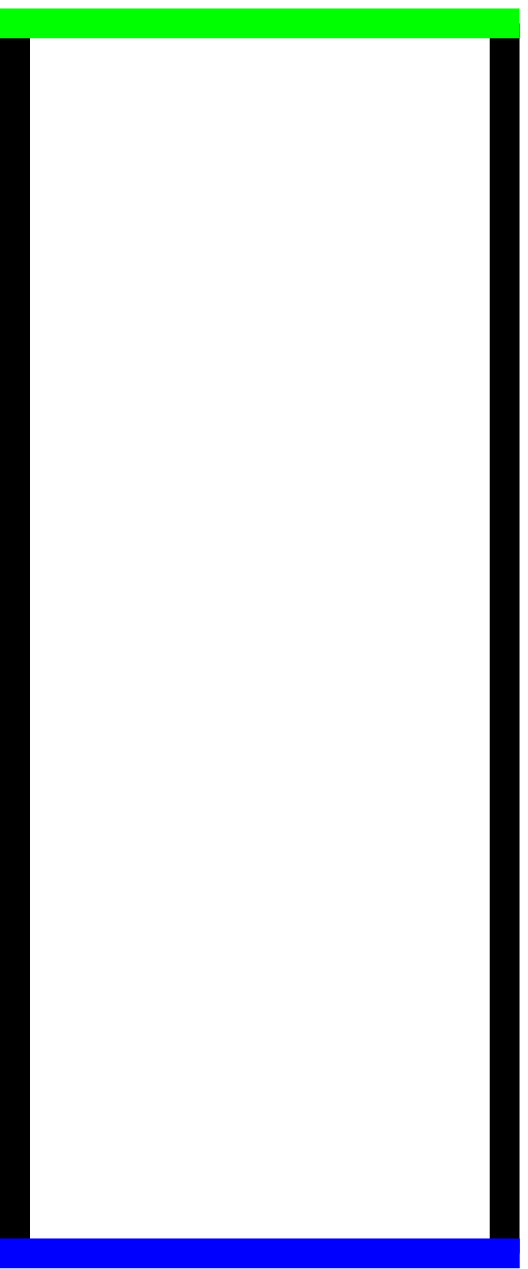}};
\node (R2l) at (7,0) {\includegraphics[scale=0.12]{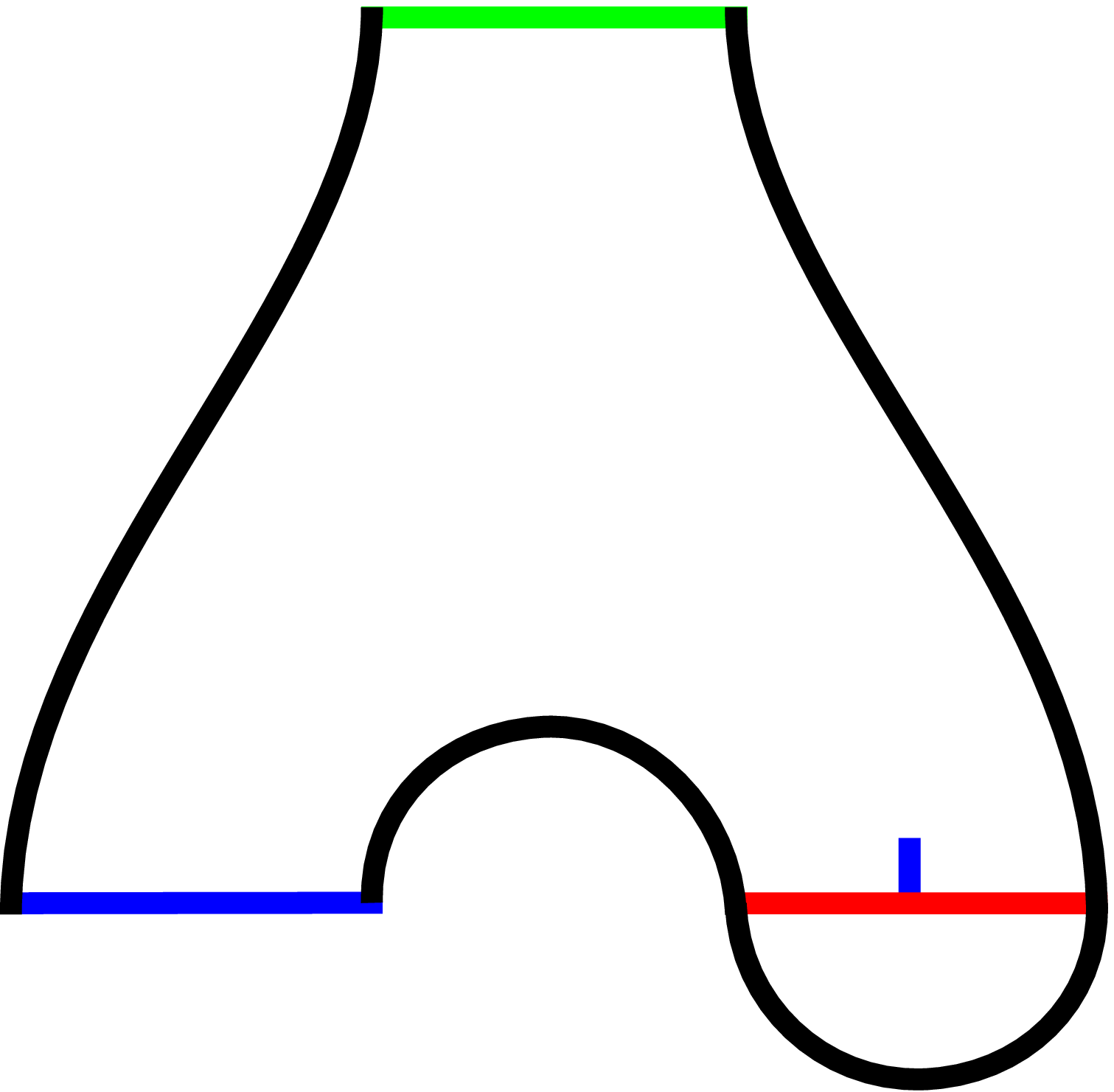}};
\node (R2r) at (10,0) {\includegraphics[scale=0.12]{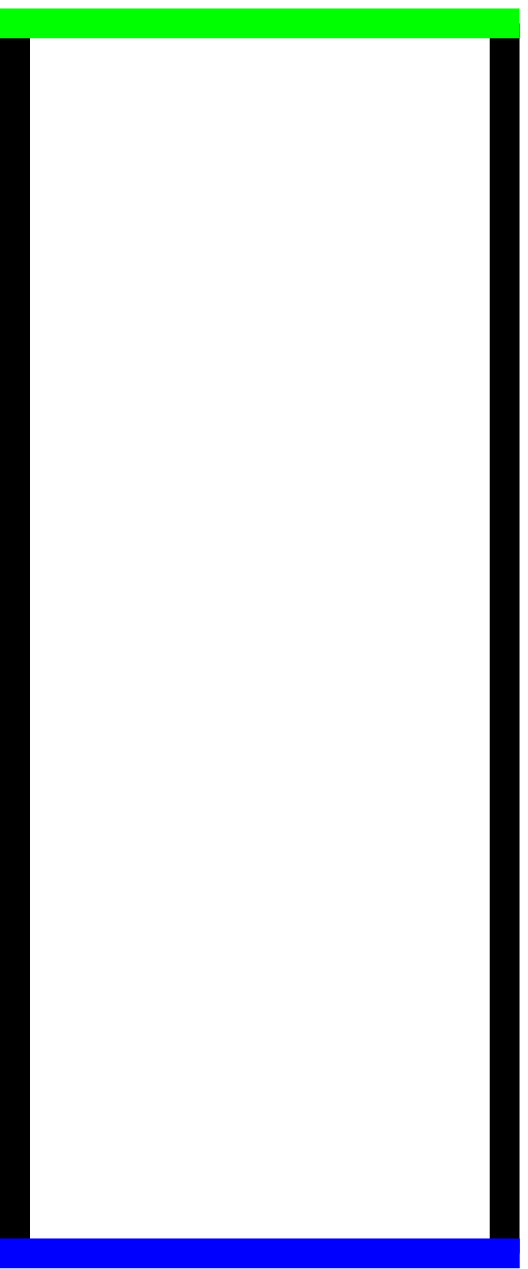}};
\node (=) at (1.7,0) {$\longleftrightarrow$};
\node (=) at (8.7,0) {$\longleftrightarrow$};
\node at (-1.2,1) {R1)};
\node at (5.8,1) {R2)};
\end{tikzpicture}
\end{center}

\begin{center}
\begin{tikzpicture}
\node (R1l) at (0,0) {\includegraphics[scale=0.12,angle=180]{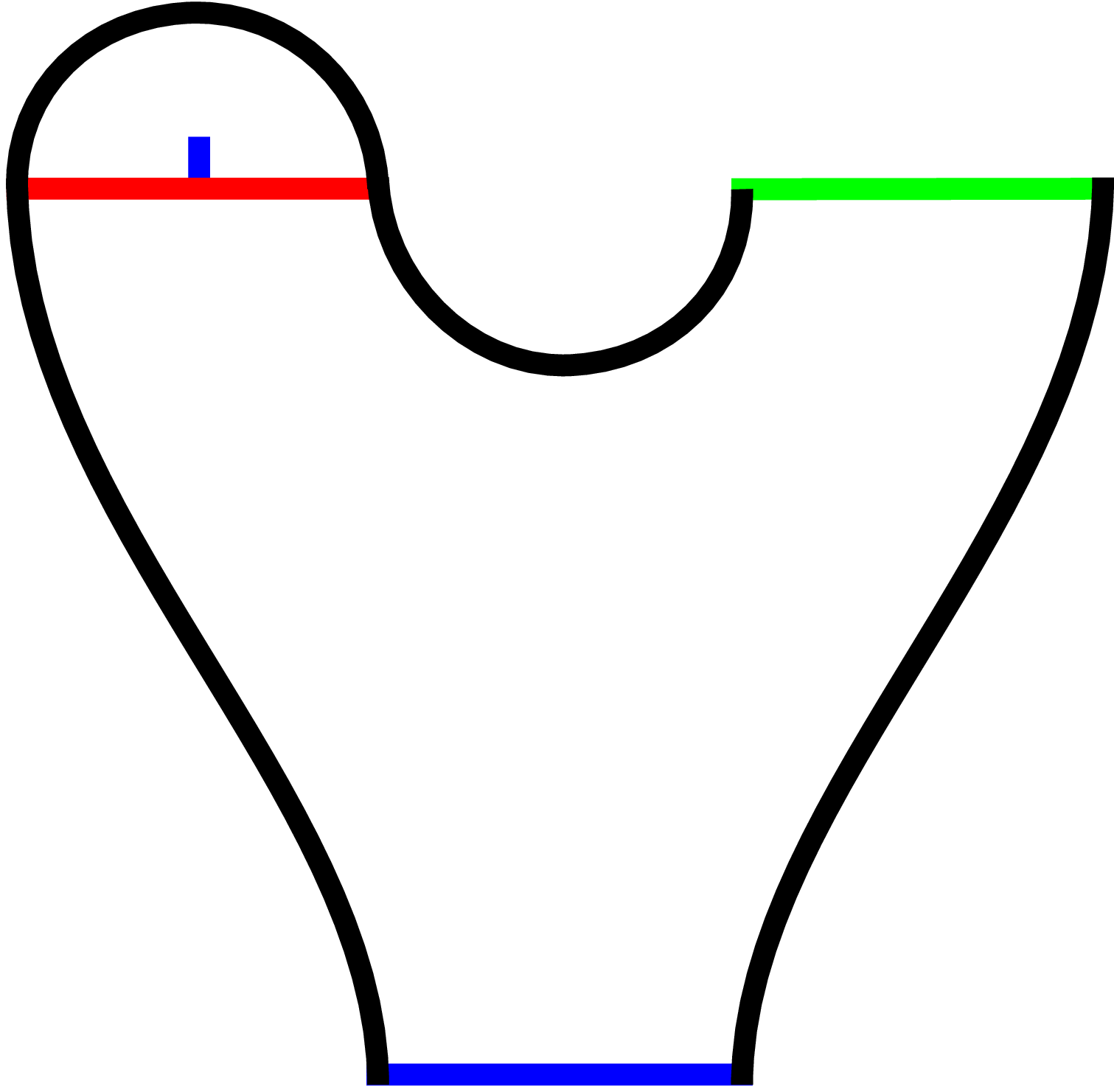}};
\node (R1r) at (3,0) {\includegraphics[scale=0.12]{figure150.eps}};
\node at (7,0) {\includegraphics[scale=0.12,angle=180]{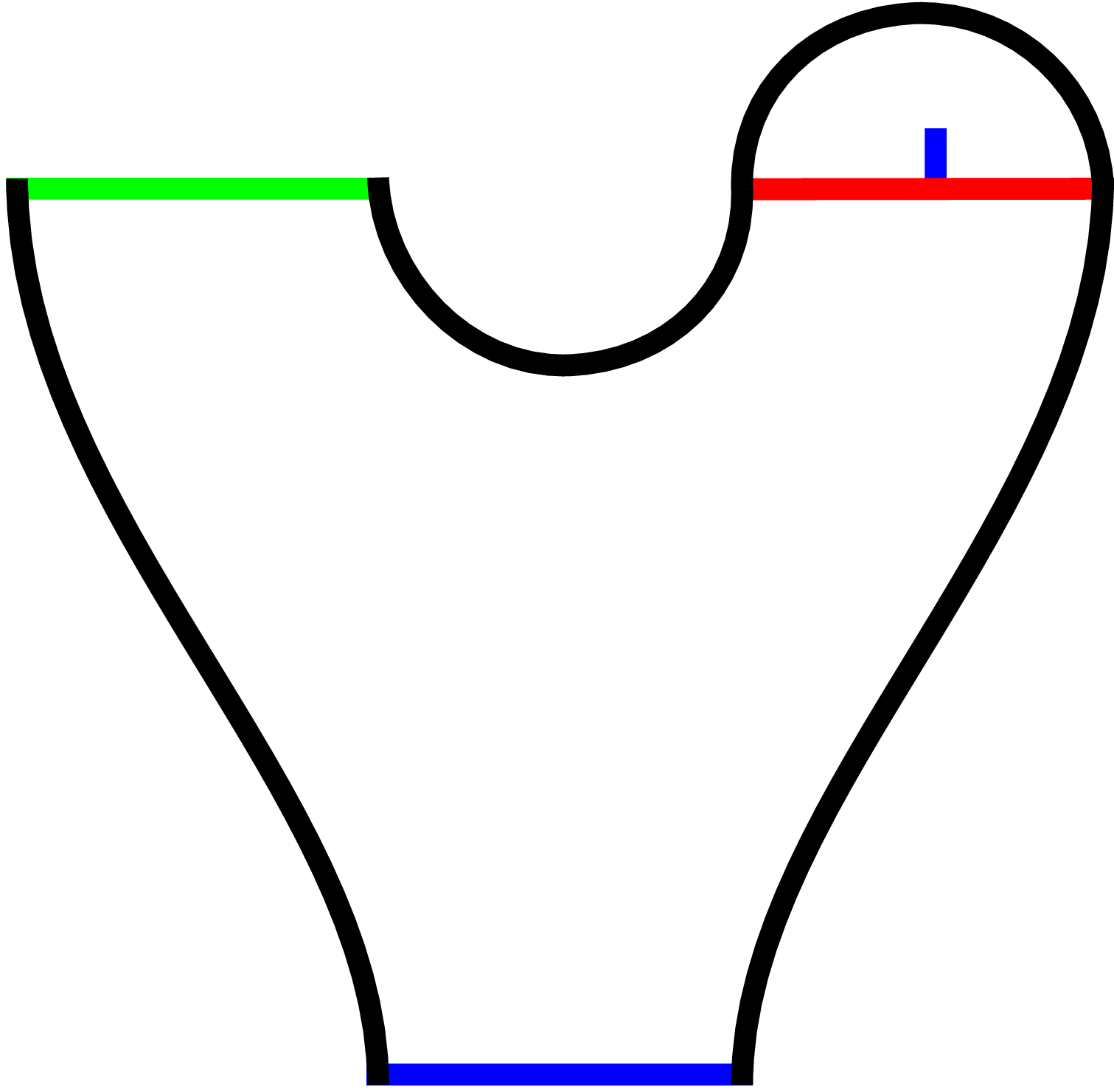}};
\node (R2r) at (10,0) {\includegraphics[scale=0.12]{figure150.eps}};
\node (=) at (1.7,0) {$\longleftrightarrow$};
\node (=) at (8.7,0) {$\longleftrightarrow$};
\node at (-1.2,1) {R3)};
\node at (5.8,1) {R4)};
\end{tikzpicture}
\end{center}

\begin{center}
\begin{tikzpicture}
\node (R1l) at (0,0) {\includegraphics[scale=0.12]{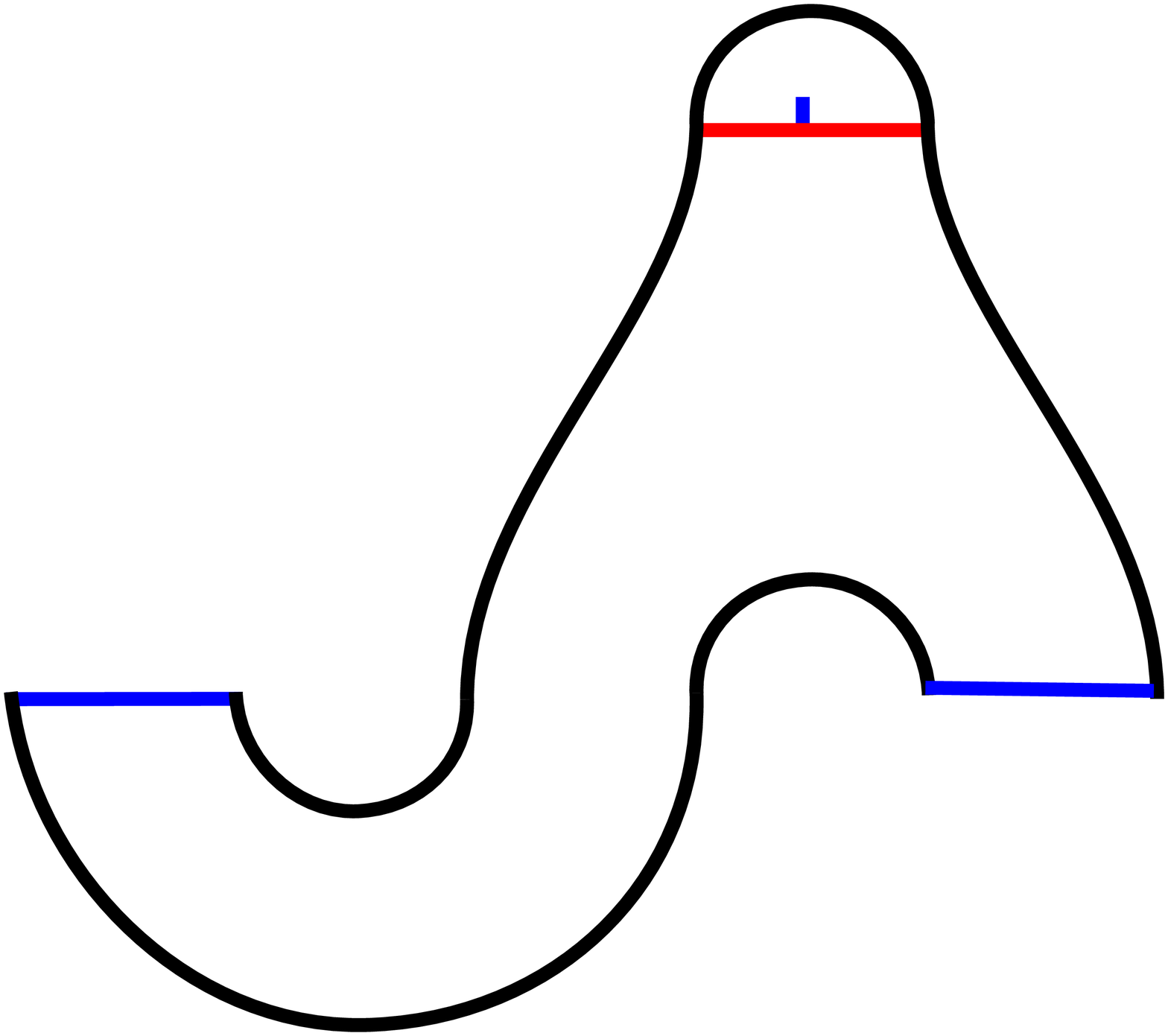}};
\node (R1r) at (5,0) {\includegraphics[scale=0.12]{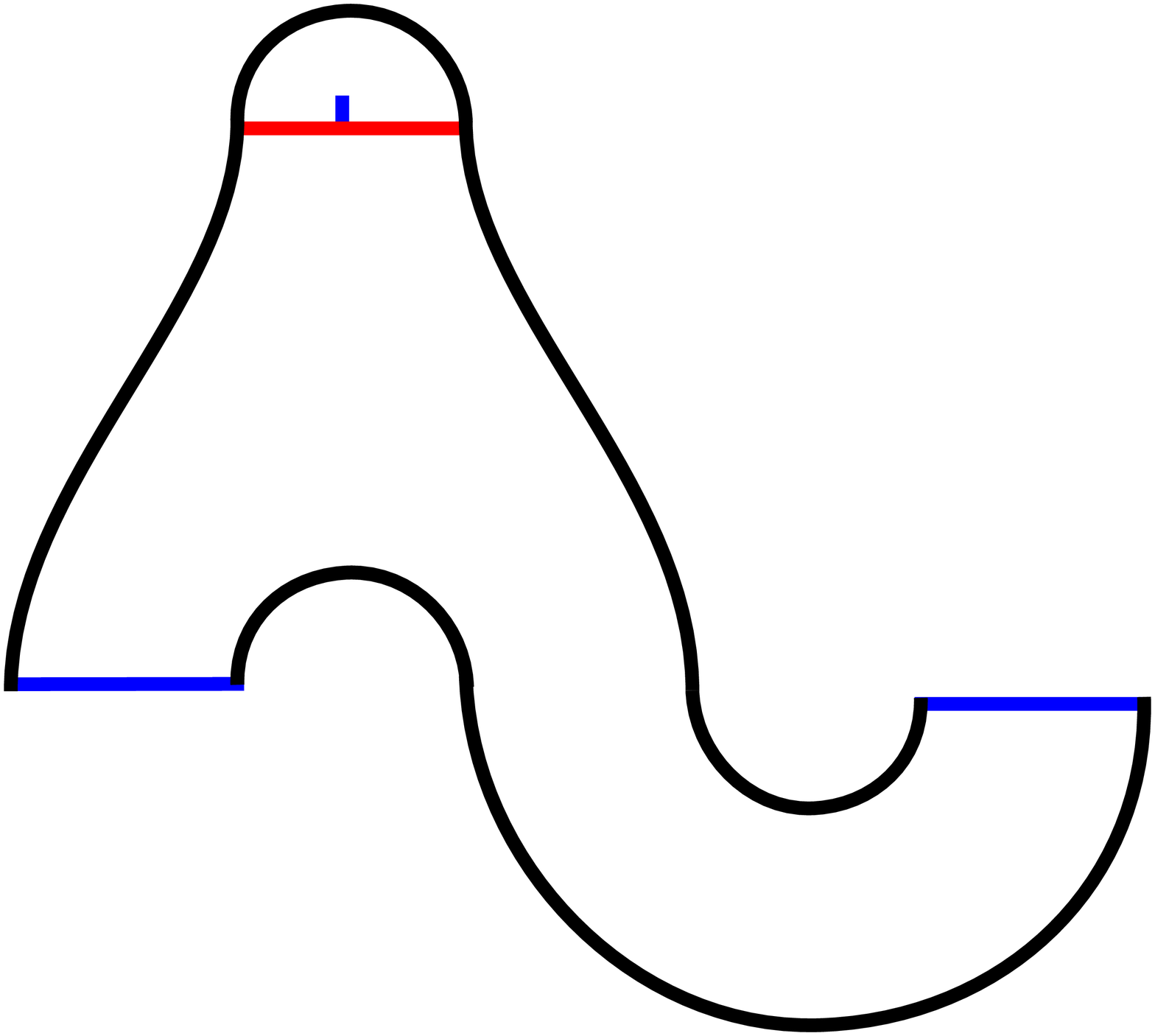}};
\node (=) at (2.5,0) {$\longleftrightarrow$};
\node at (-1.5,1) {R5)};
\end{tikzpicture}
\end{center}

\begin{center}
\begin{tikzpicture}
\node (R1l) at (0,0) {\includegraphics[scale=0.12]{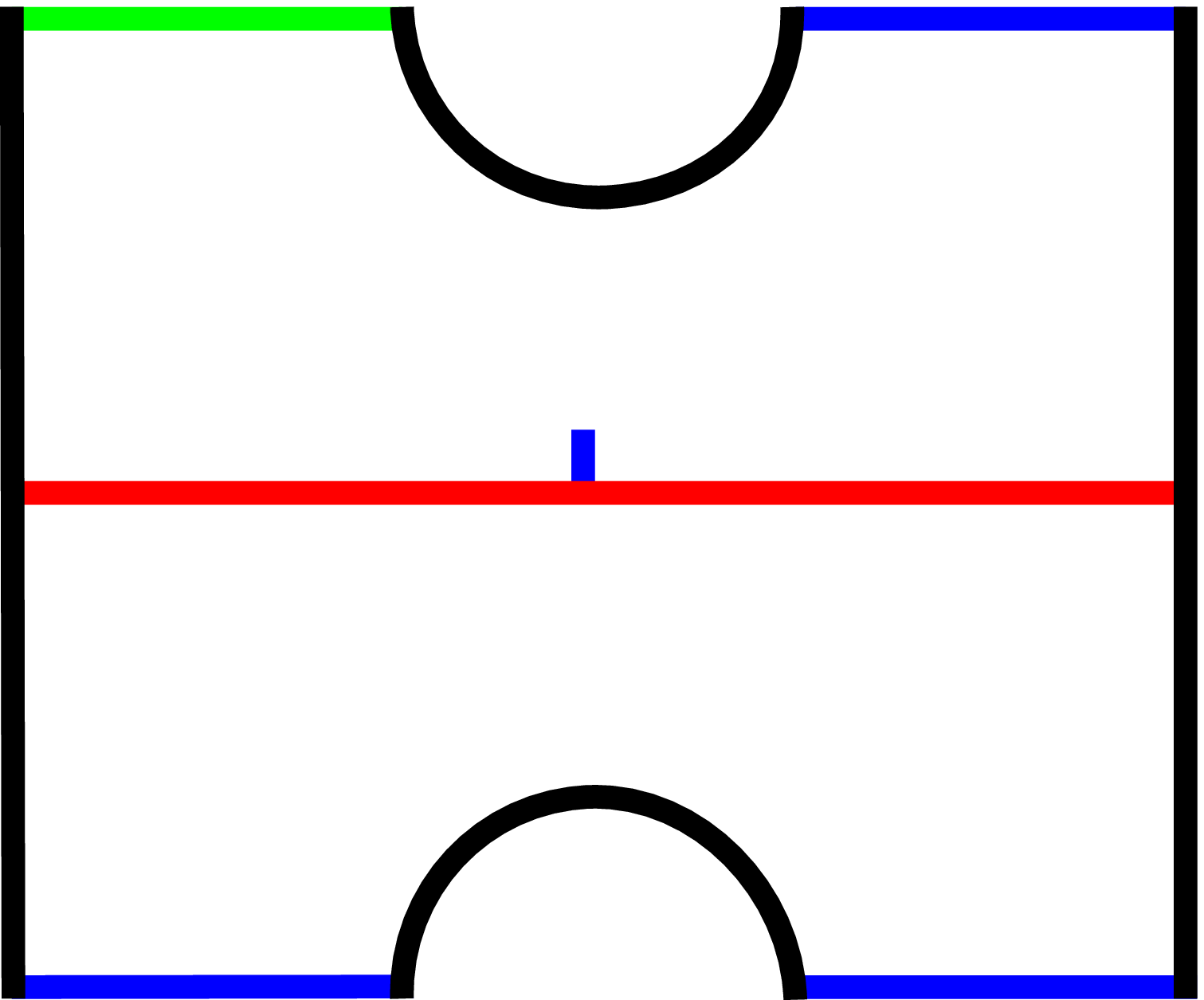}};
\node (R1r) at (3,0) {\includegraphics[scale=0.12]{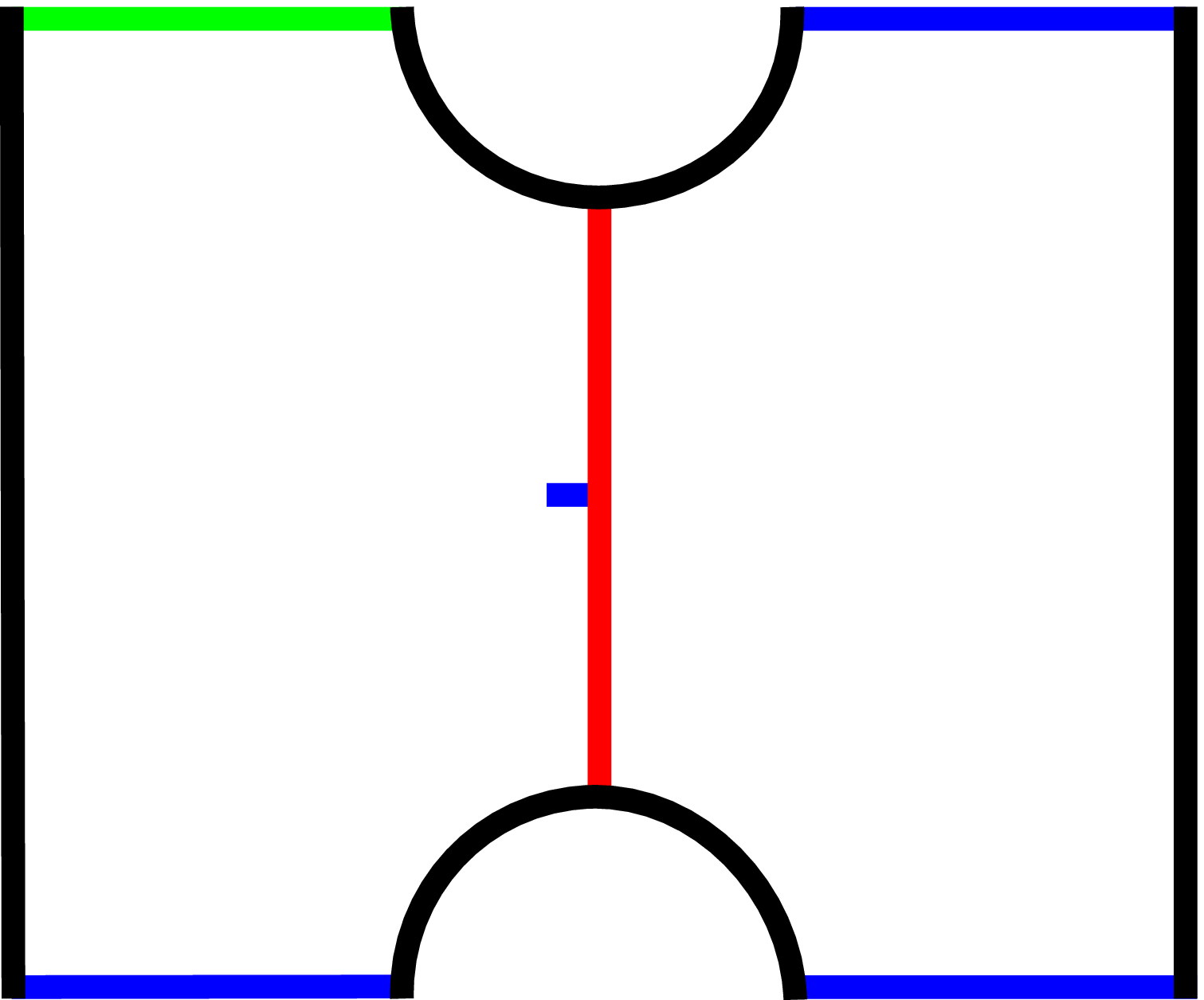}};
\node (R2l) at (7,0) {\includegraphics[scale=0.12]{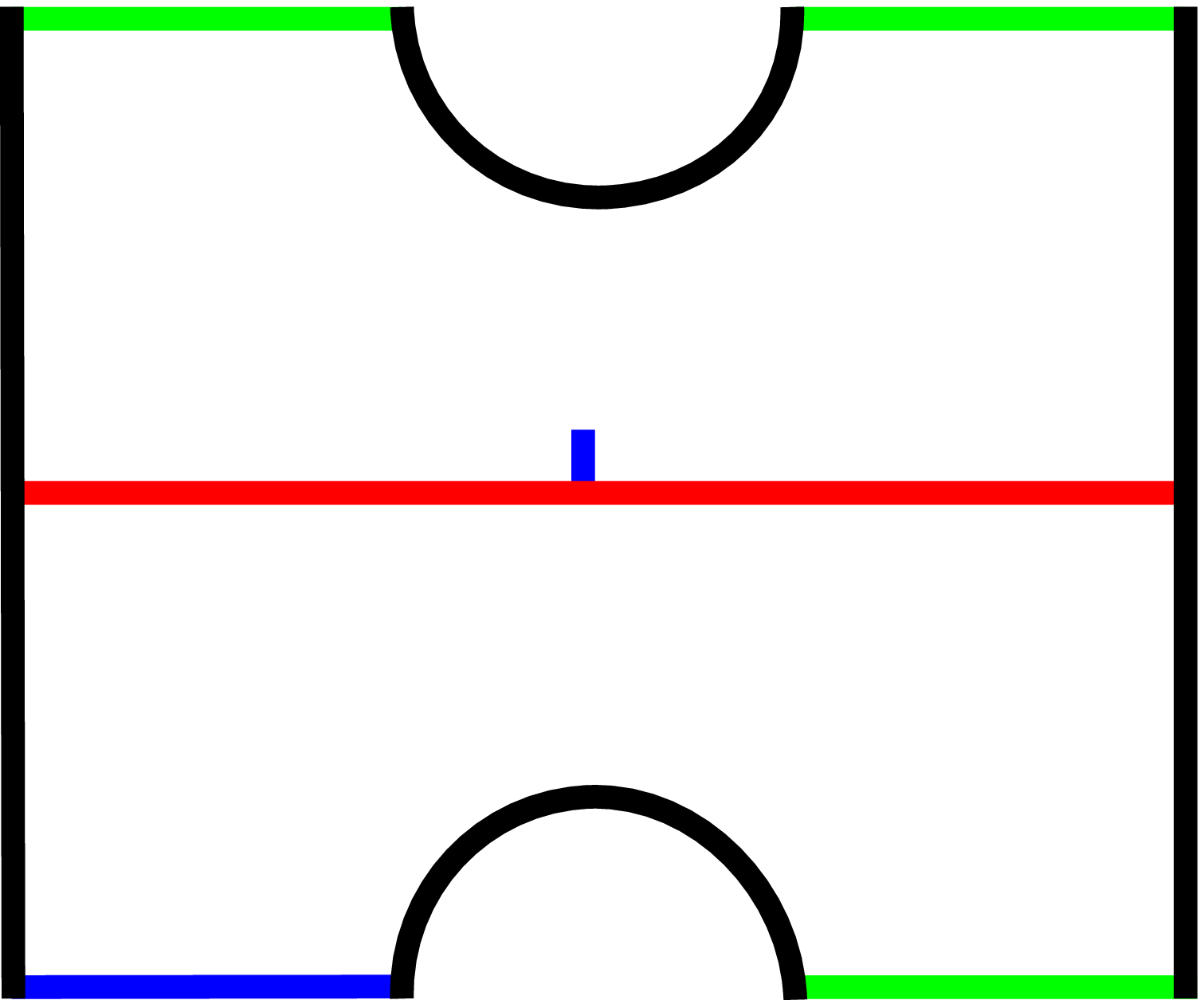}};
\node (R2r) at (10,0) {\includegraphics[scale=0.12]{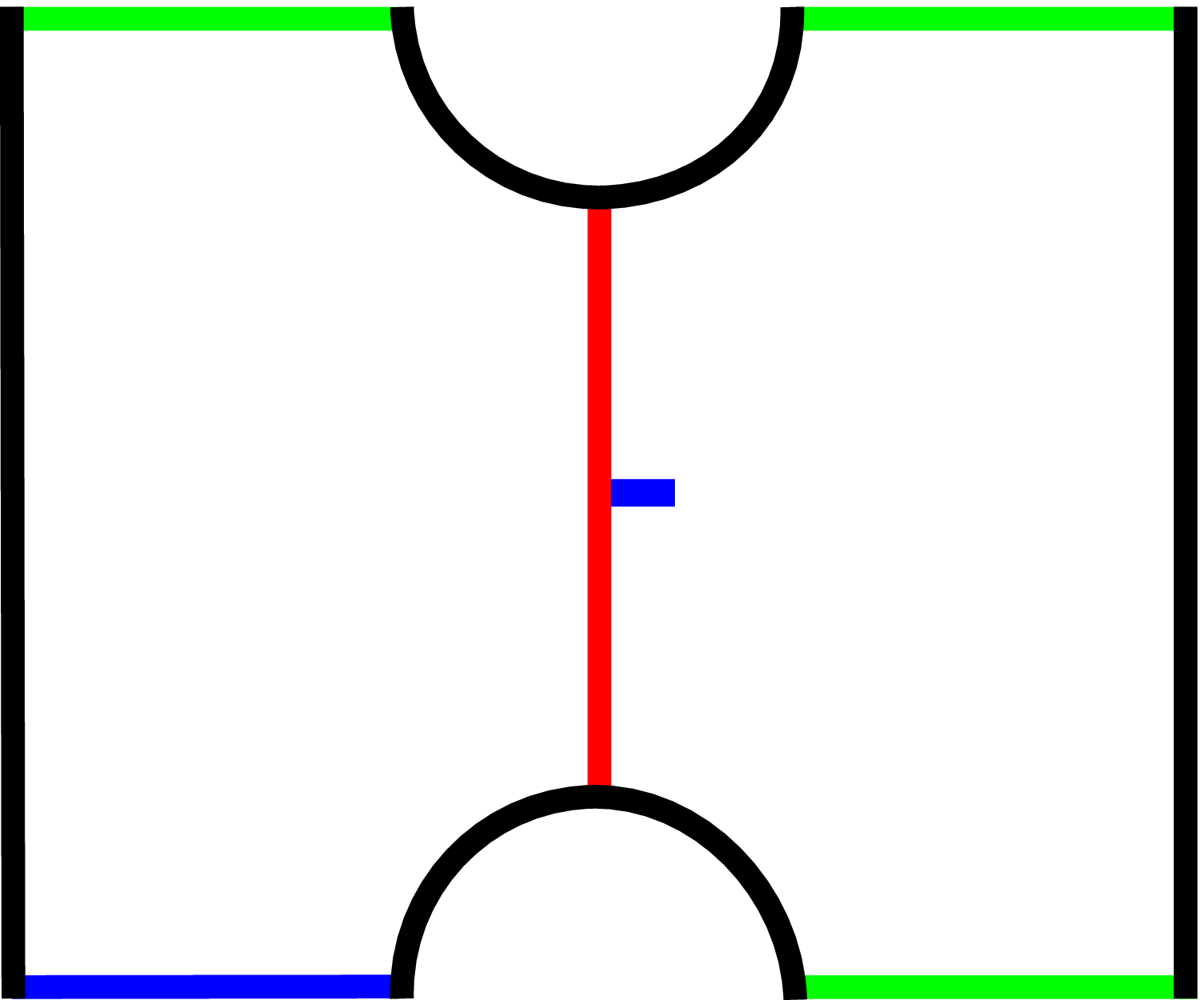}};
\node (=) at (1.5,0) {$\longleftrightarrow$};
\node (=) at (8.5,0) {$\longleftrightarrow$};
\node at (-1.5,1) {R6)};
\node at (5.5,1) {R7)};
\end{tikzpicture}
\end{center}

\begin{center}
\begin{tikzpicture}
\node (R1l) at (0,0) {\includegraphics[scale=0.12]{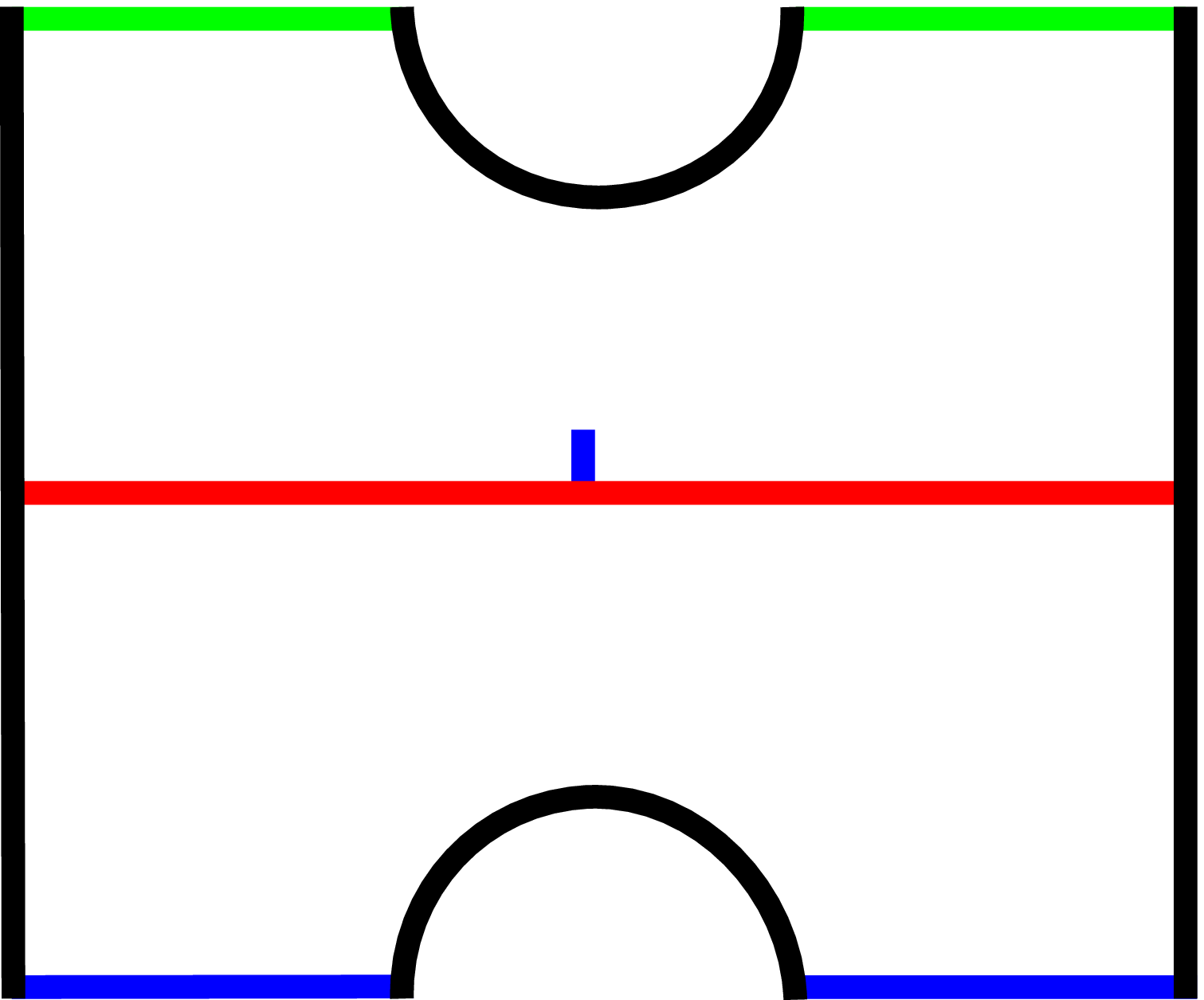}};
\node (R1r) at (3,0) {\includegraphics[scale=0.12]{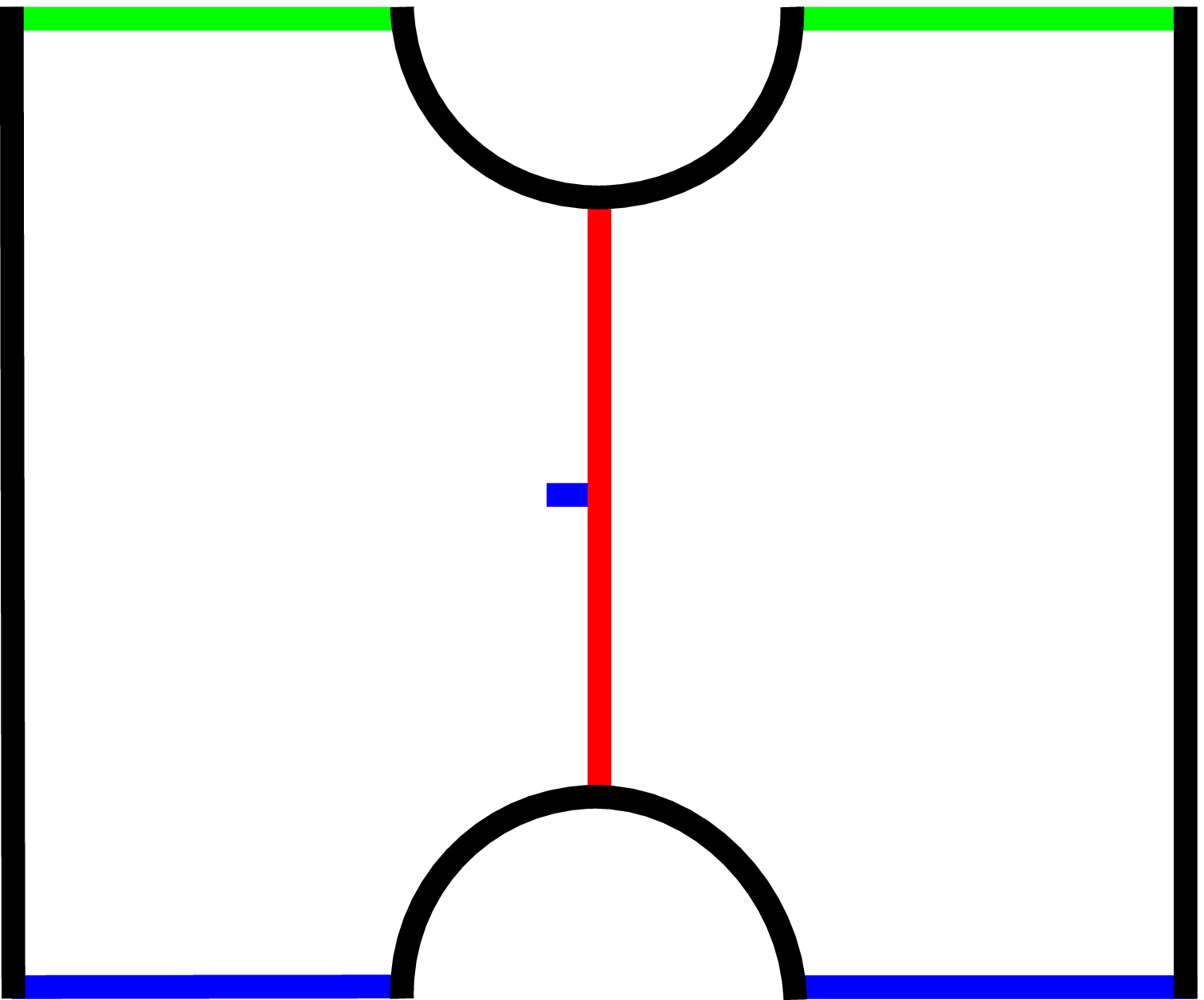}};
\node (R2l) at (7,0) {\includegraphics[scale=0.12]{figure159.eps}};
\node (R2r) at (10,0) {\includegraphics[scale=0.12]{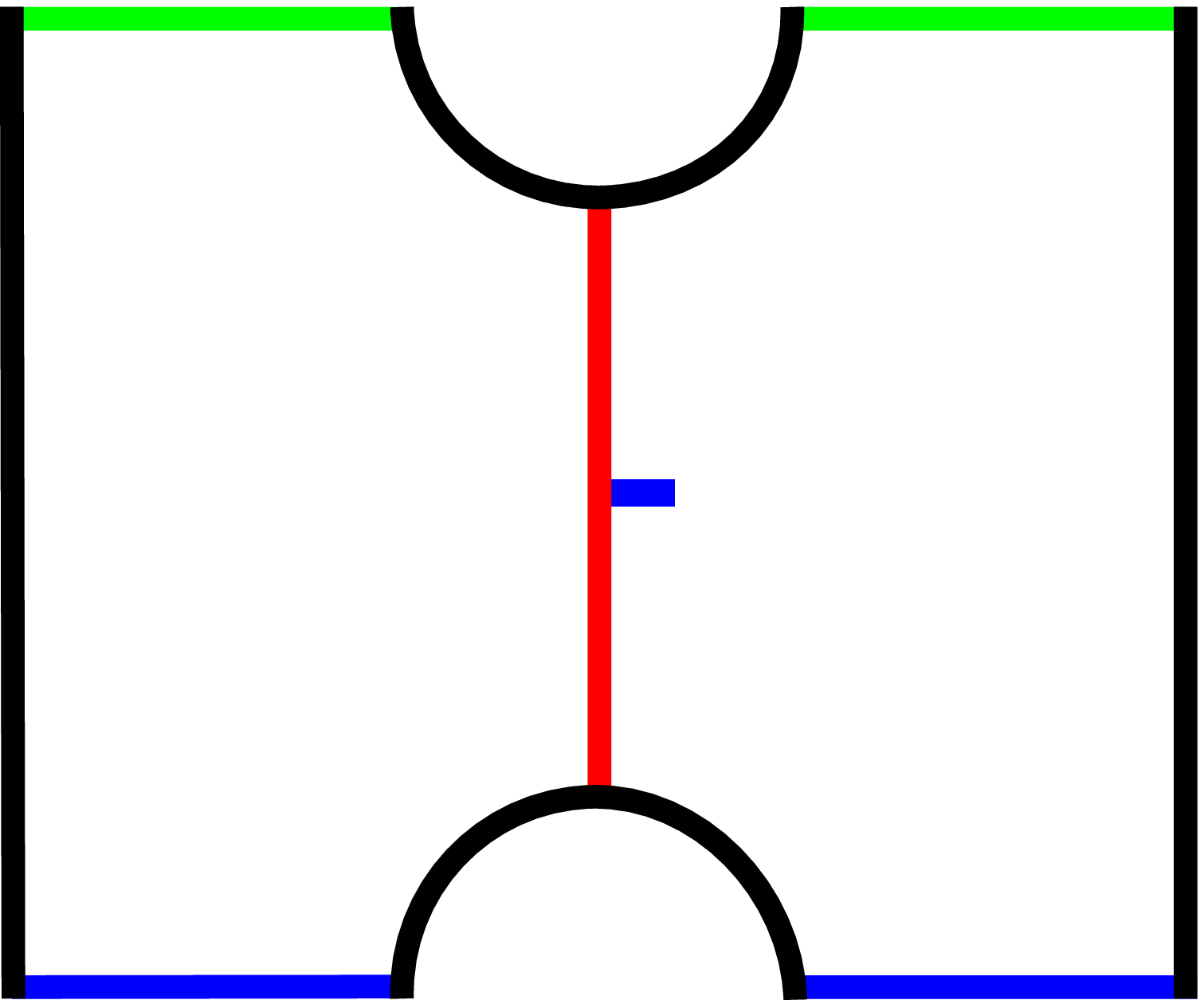}};
\node (=) at (1.5,0) {$\longleftrightarrow$};
\node (=) at (8.5,0) {$\longleftrightarrow$};
\node at (-1.5,1) {R8)};
\node at (5.5,1) {R9)};
\end{tikzpicture}
\end{center}

\begin{center}
\begin{tikzpicture}
\node (R1l) at (0,0) {\includegraphics[scale=0.12]{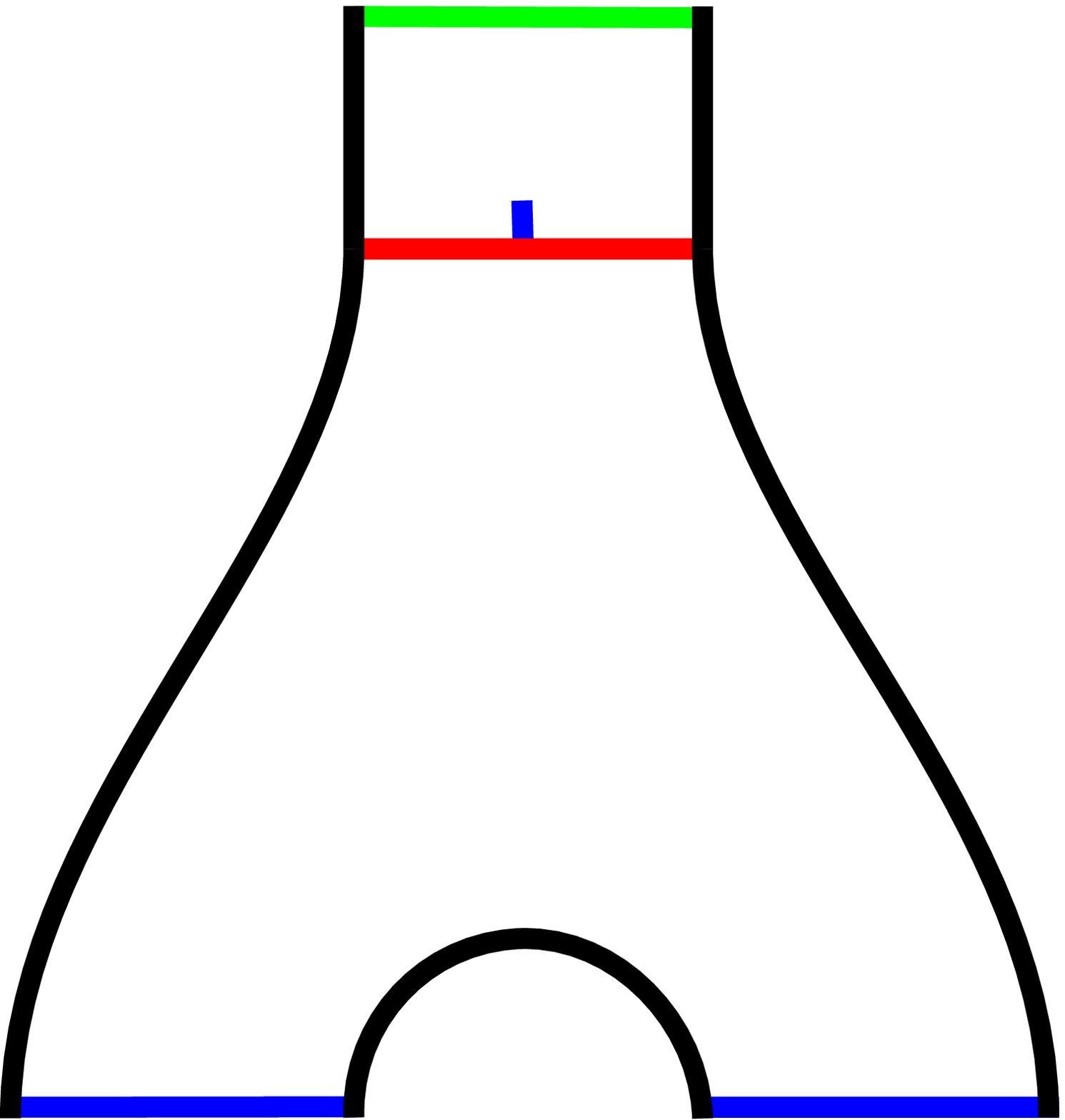}};
\node (R1r) at (3,0) {\includegraphics[scale=0.12]{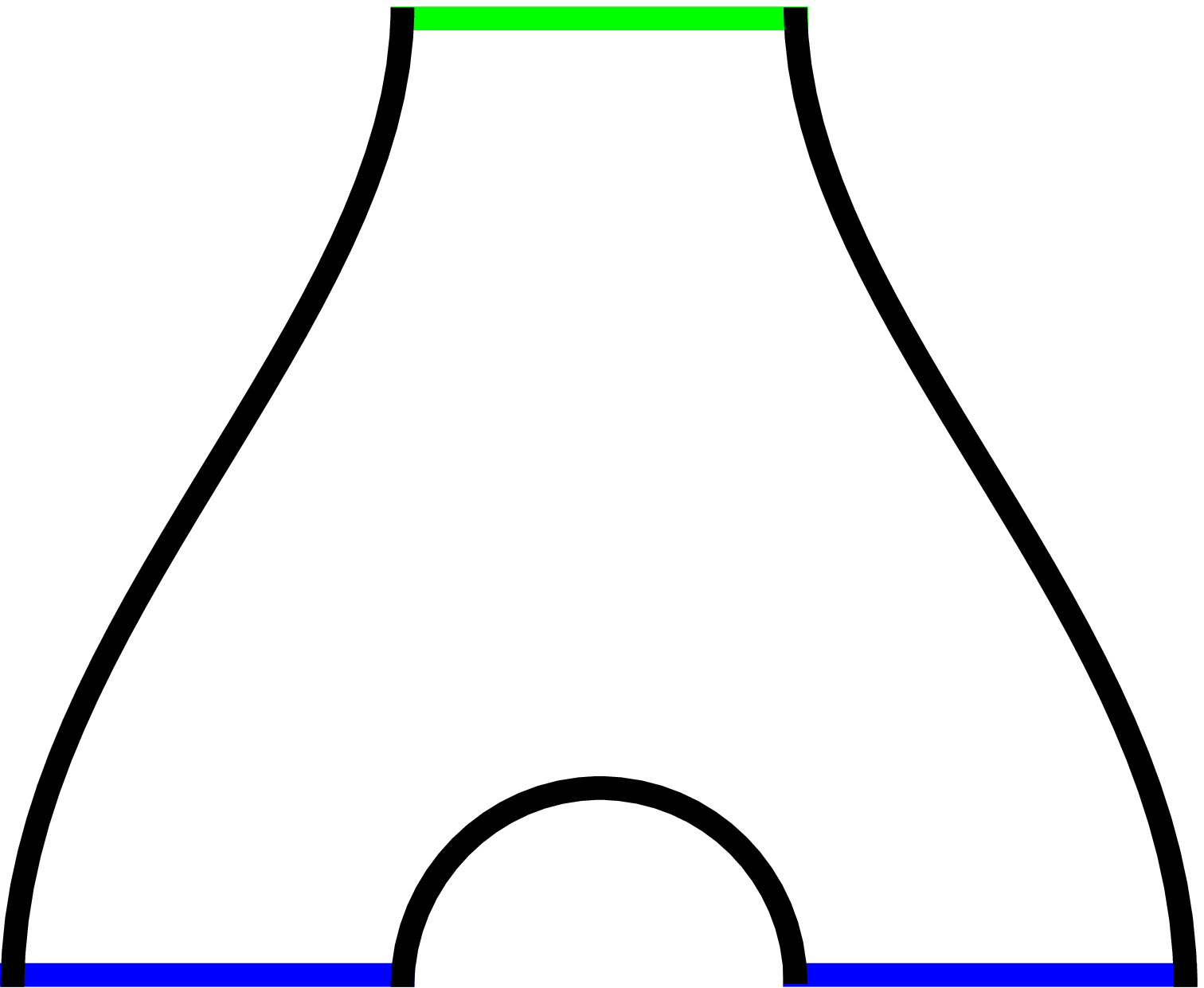}};
\node (R2l) at (7,0) {\includegraphics[scale=0.12, angle=180]{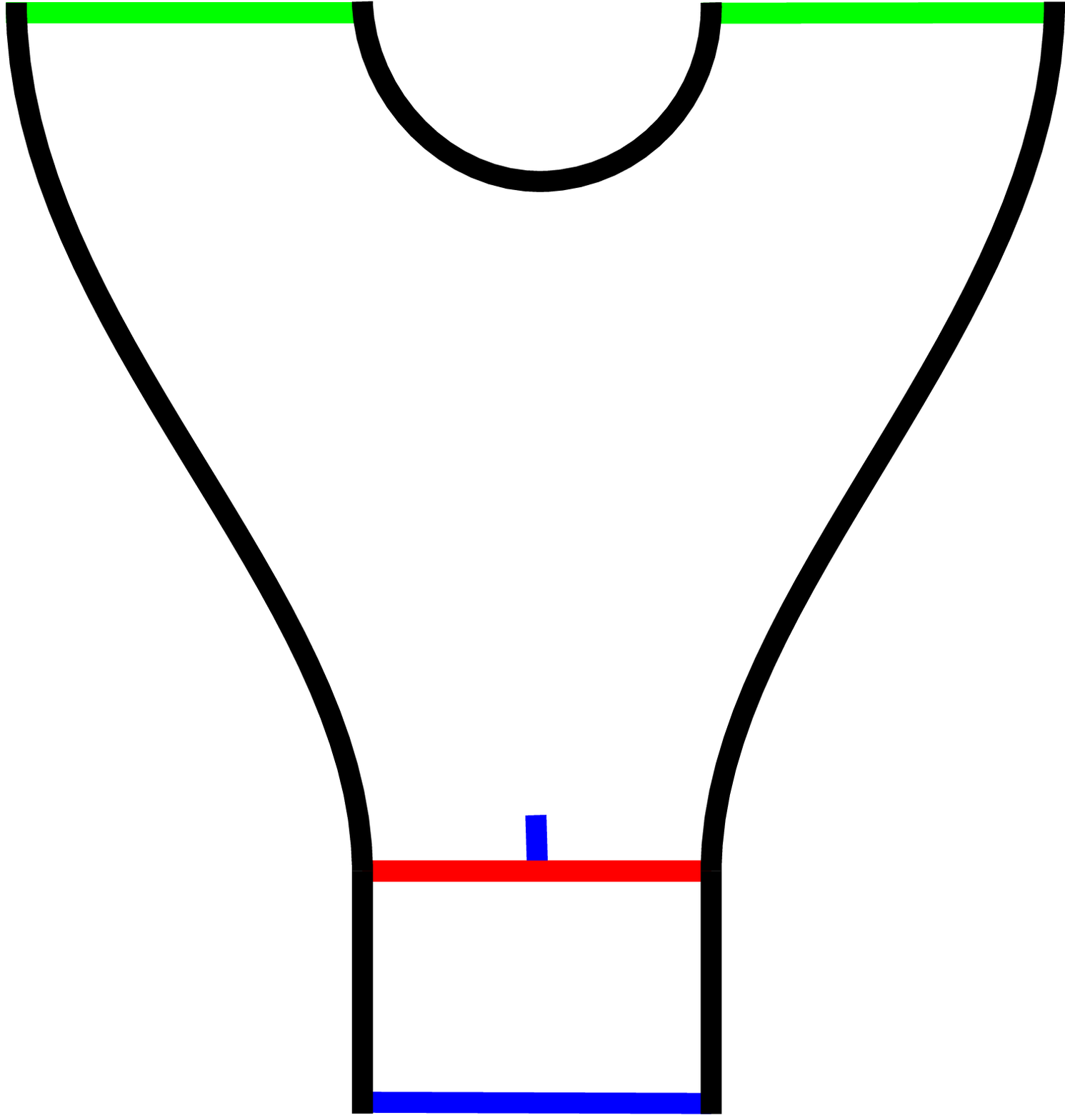}};
\node (R2r) at (10,0) {\includegraphics[scale=0.12,angle=180]{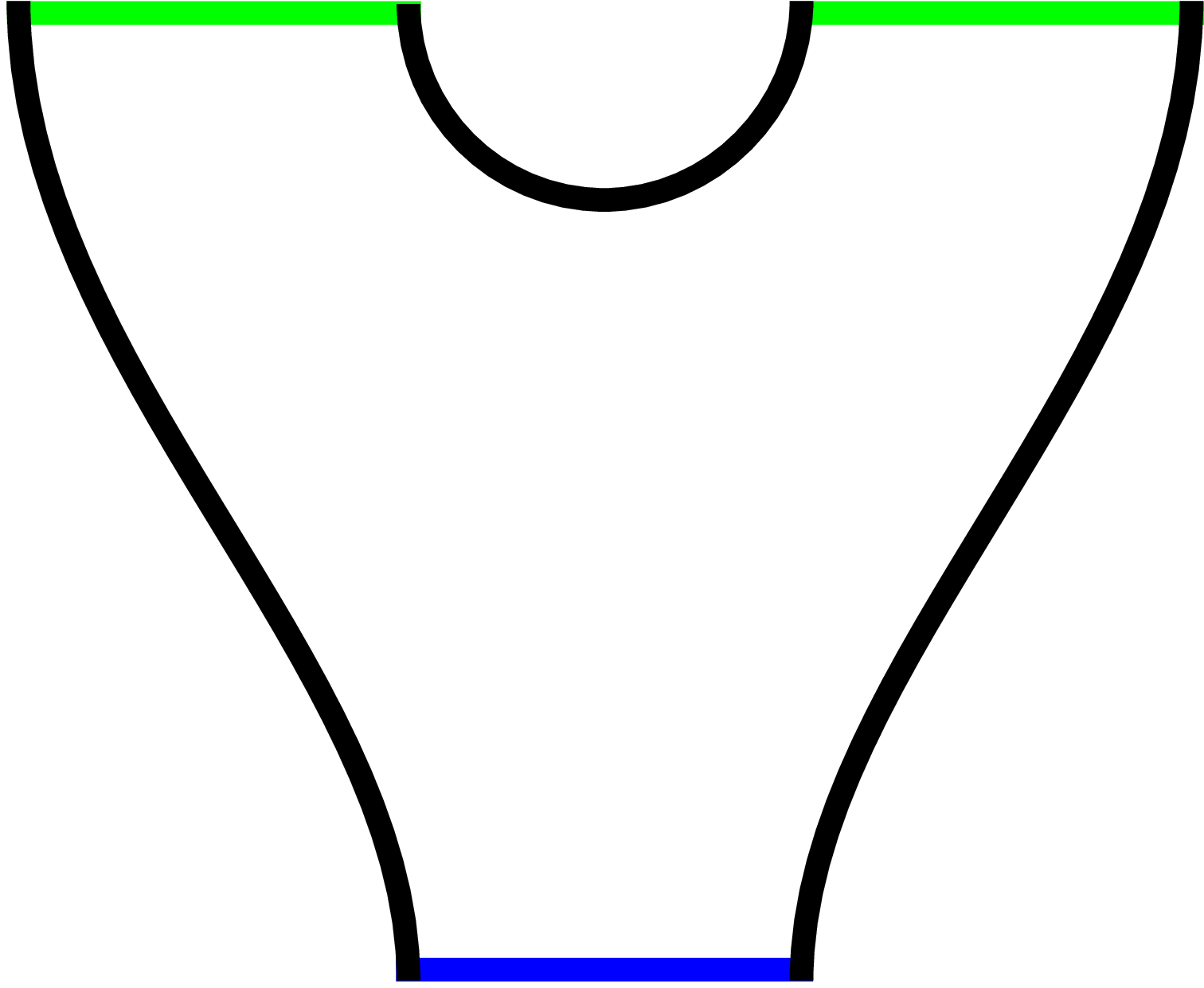}};
\node (=) at (1.5,0) {$\longleftrightarrow$};
\node (=) at (8.5,0) {$\longleftrightarrow$};
\node at (-1.5,1) {R10)};
\node at (5.2,1) {R11)};
\end{tikzpicture}
\end{center}

\begin{center}
\begin{tikzpicture}
\node (R1l) at (0,0) {\includegraphics[scale=0.12]{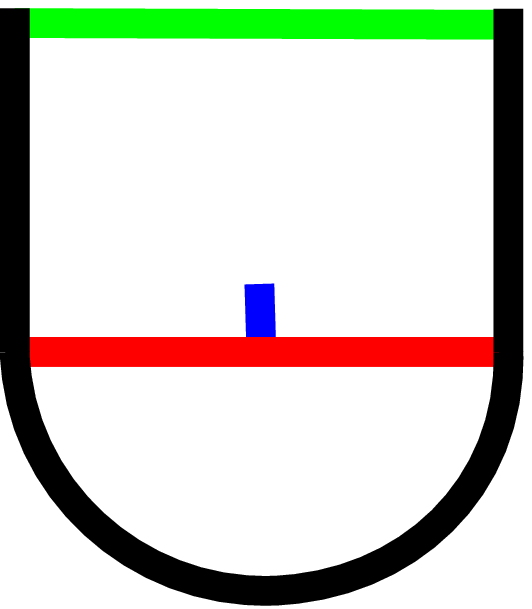}};
\node (R1r) at (3,0) {\includegraphics[scale=0.12]{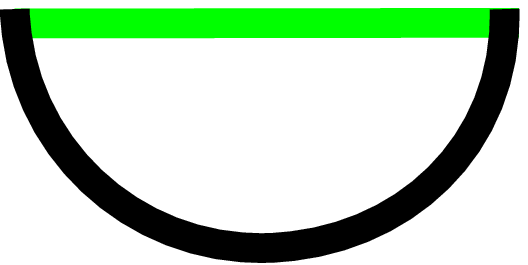}};
\node (R2l) at (7,0) {\includegraphics[scale=0.12, angle=180]{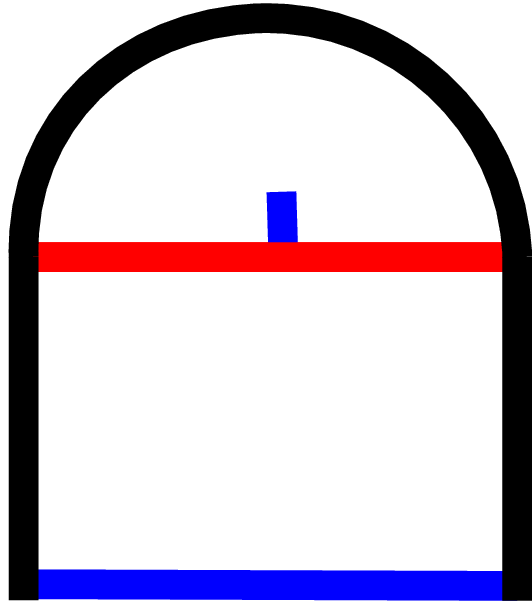}};
\node (R2r) at (10,0) {\includegraphics[scale=0.12,angle=180]{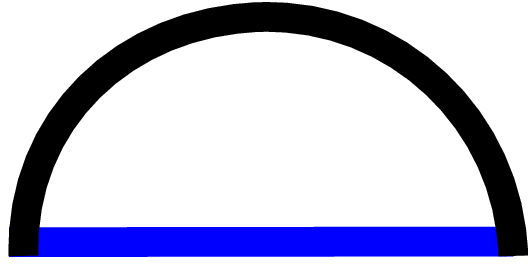}};
\node (=) at (1.5,0) {$\longleftrightarrow$};
\node (=) at (8.5,0) {$\longleftrightarrow$};
\node at (-2,1) {R12)};
\node at (5,1) {R13)};
\end{tikzpicture}
\end{center}

\item \textbf{Closed Relations:}

\begin{center}
\begin{tikzpicture}
\node (R1l) at (0,0) {\includegraphics[scale=0.12,angle=180]{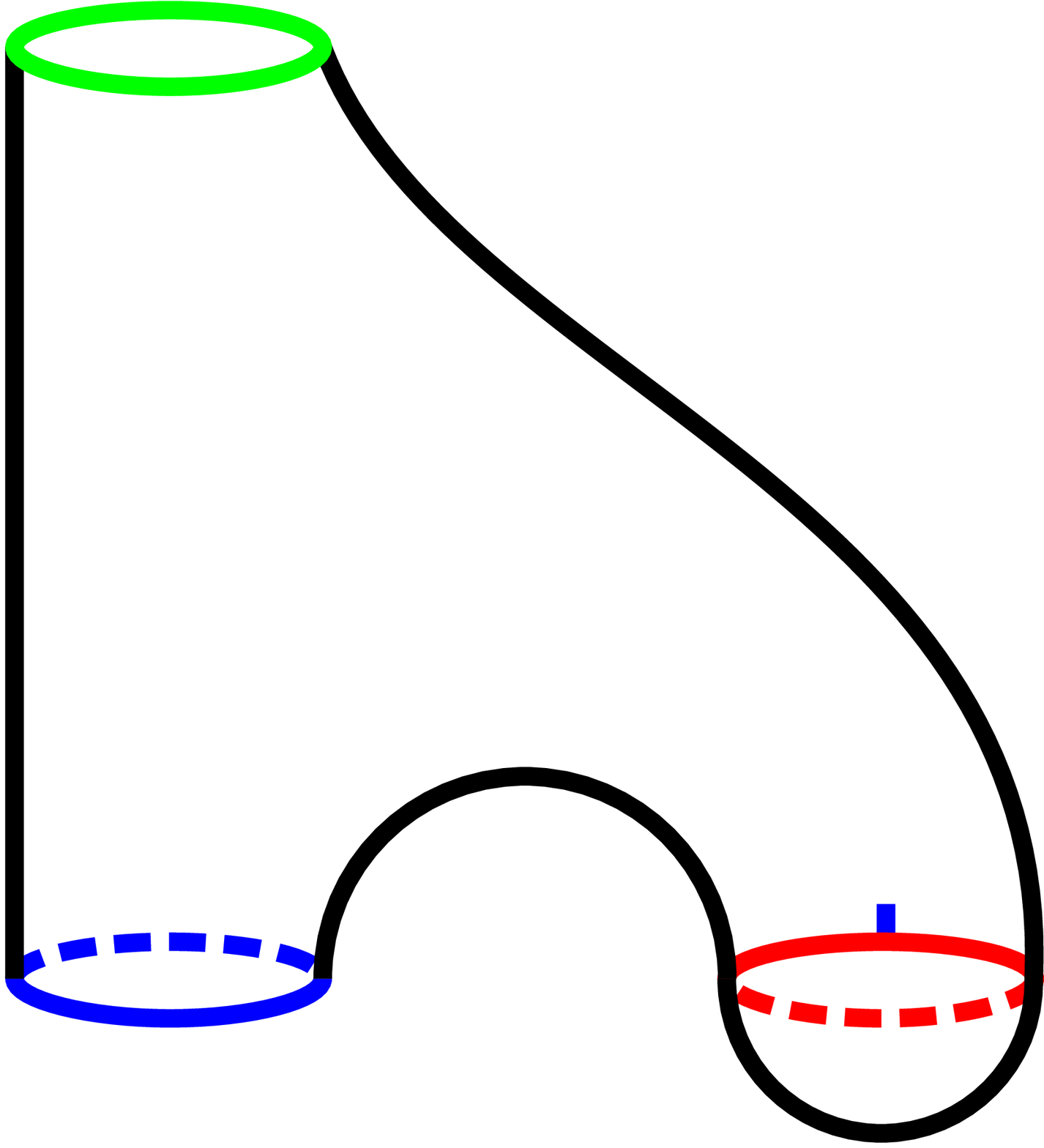}};
\node (R1r) at (3,0) {\includegraphics[scale=0.12]{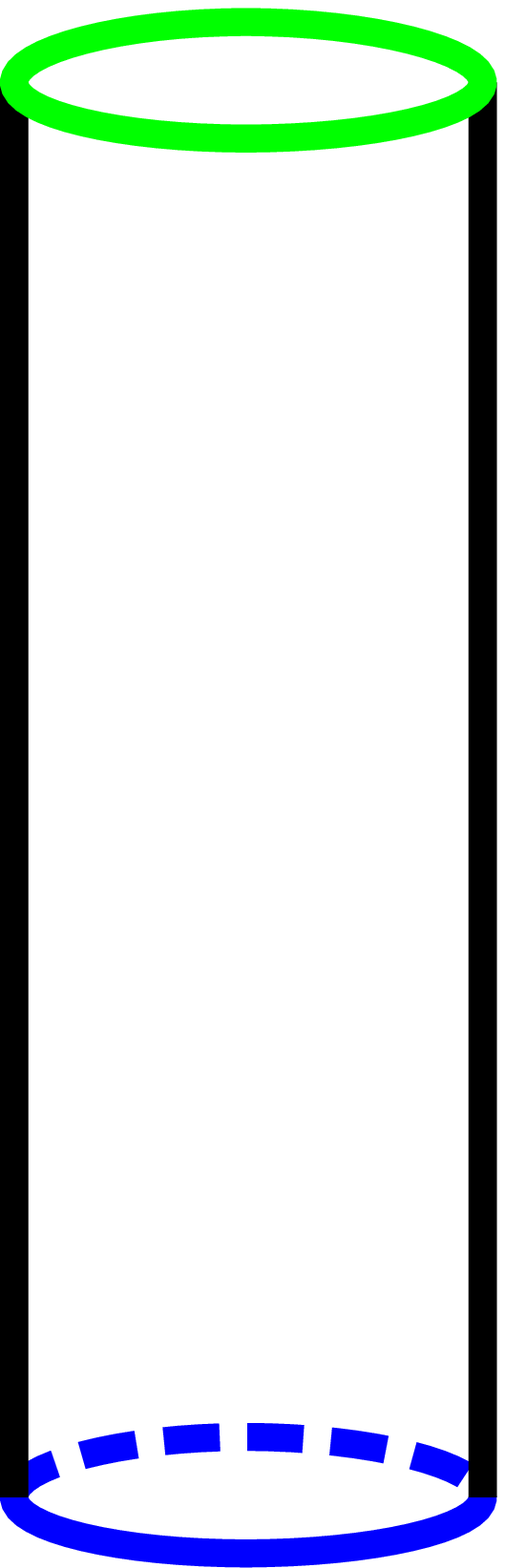}};
\node (R2l) at (7,0) {\includegraphics[scale=0.12]{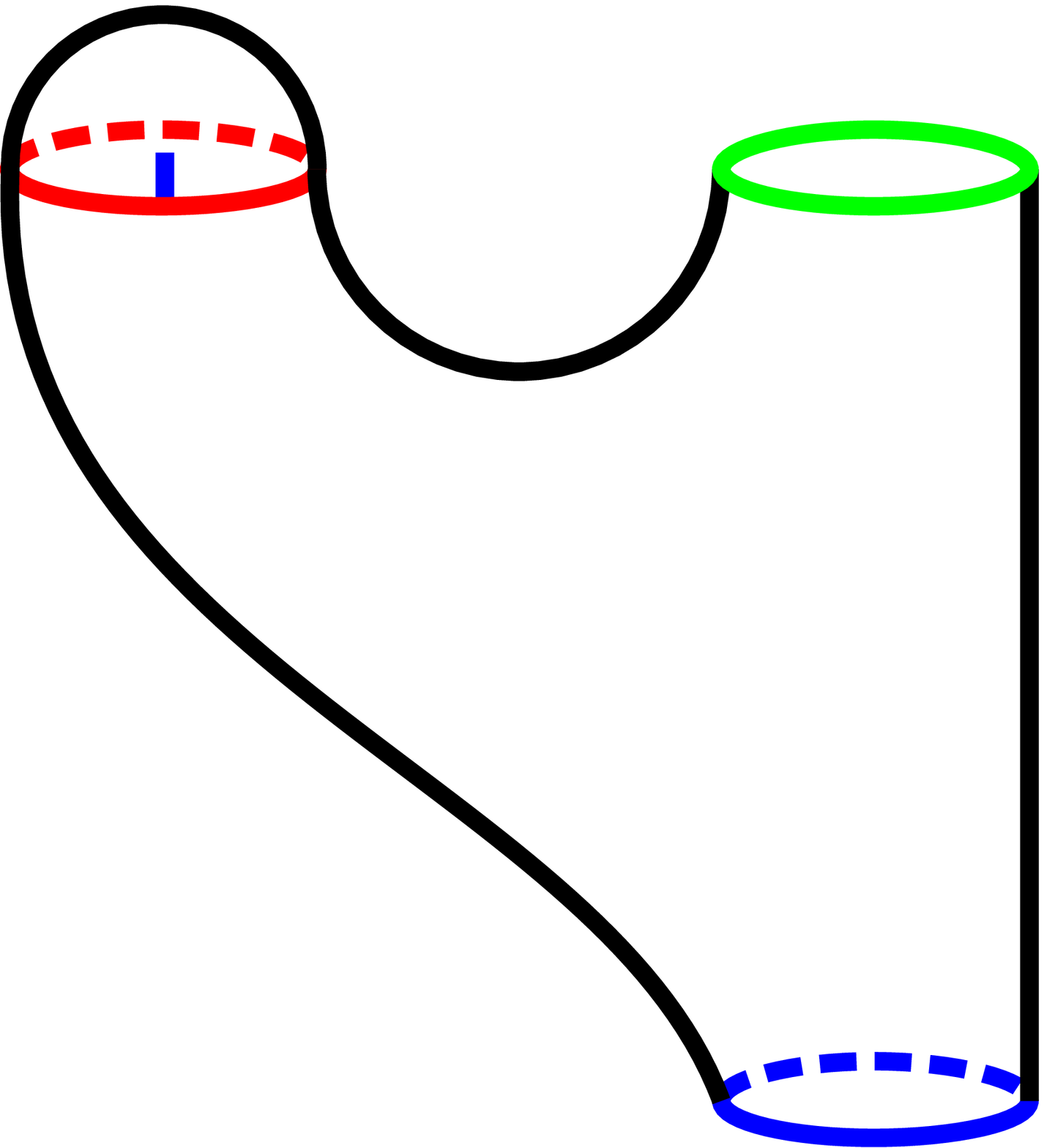}};
\node (R2r) at (10,0) {\includegraphics[scale=0.12]{figure171.eps}};
\node (=) at (1.7,0) {$\longleftrightarrow$};
\node (=) at (8.7,0) {$\longleftrightarrow$};
\node at (-1.8,1) {R14)};
\node at (5.2,1) {R15)};
\end{tikzpicture}
\end{center}

\begin{center}
\begin{tikzpicture}
\node (R1l) at (0,0) {\includegraphics[scale=0.12]{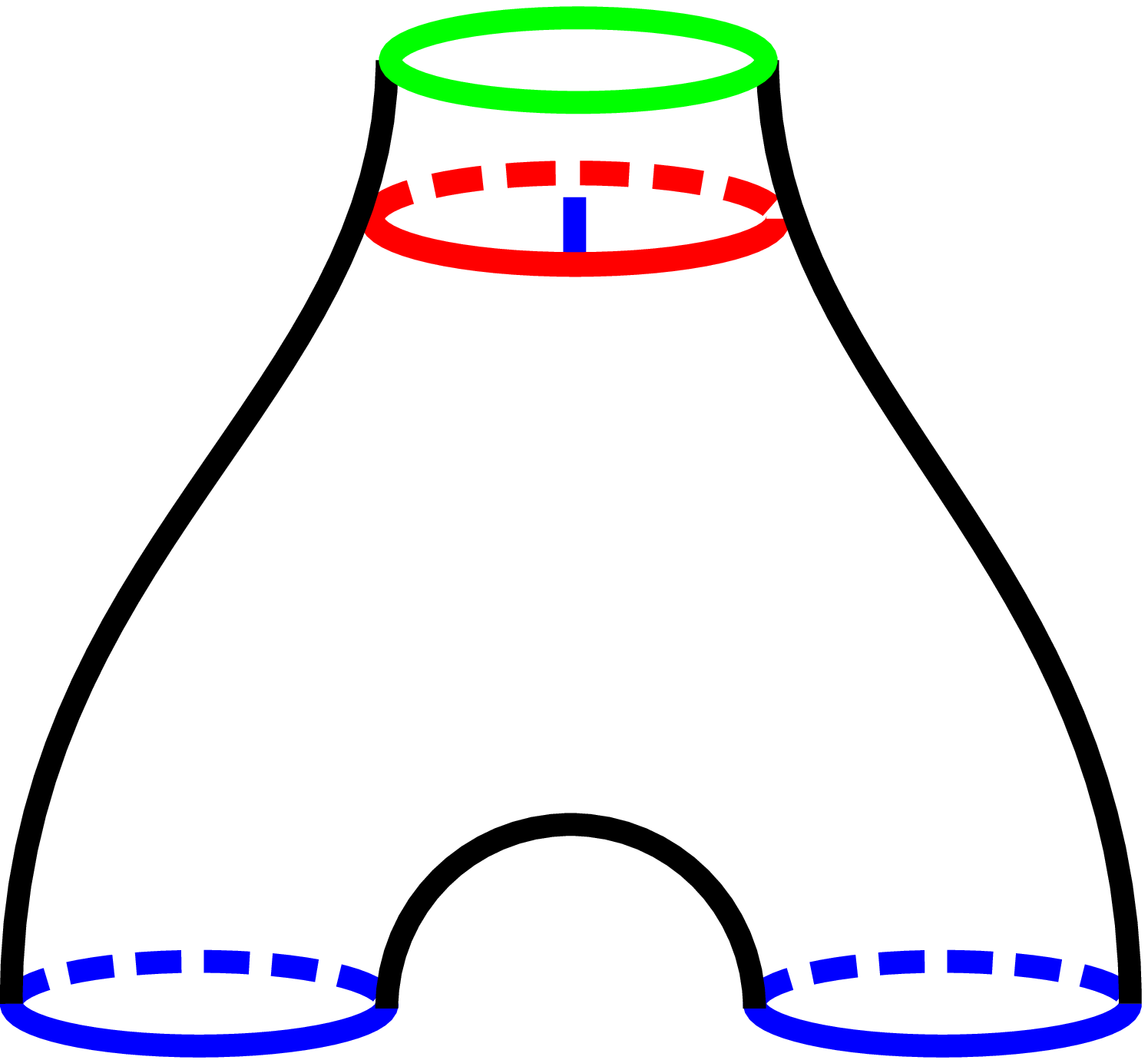}};
\node (R1r) at (3,0) {\includegraphics[scale=0.12]{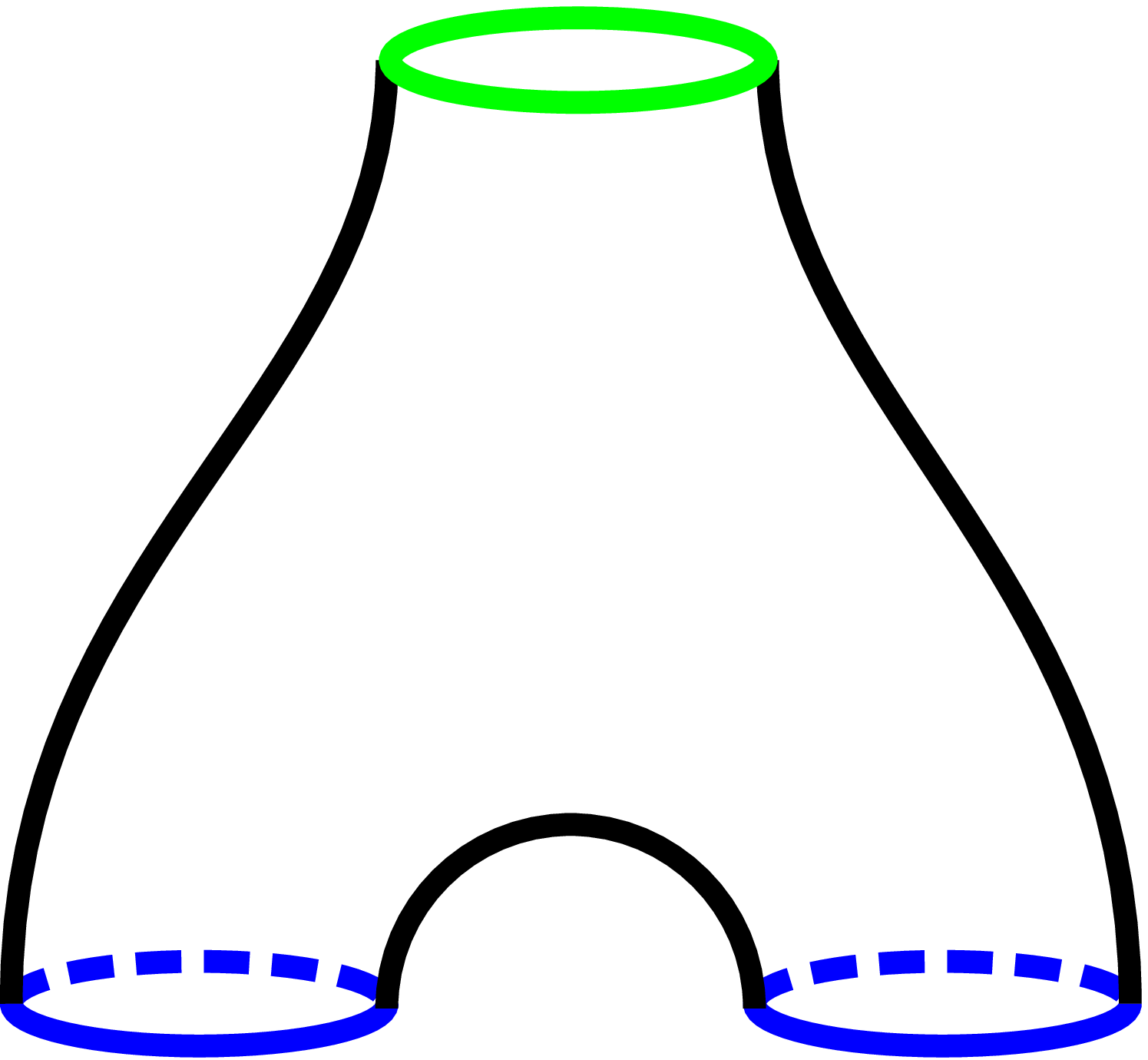}};
\node (R2l) at (7,0) {\includegraphics[scale=0.12, angle=180]{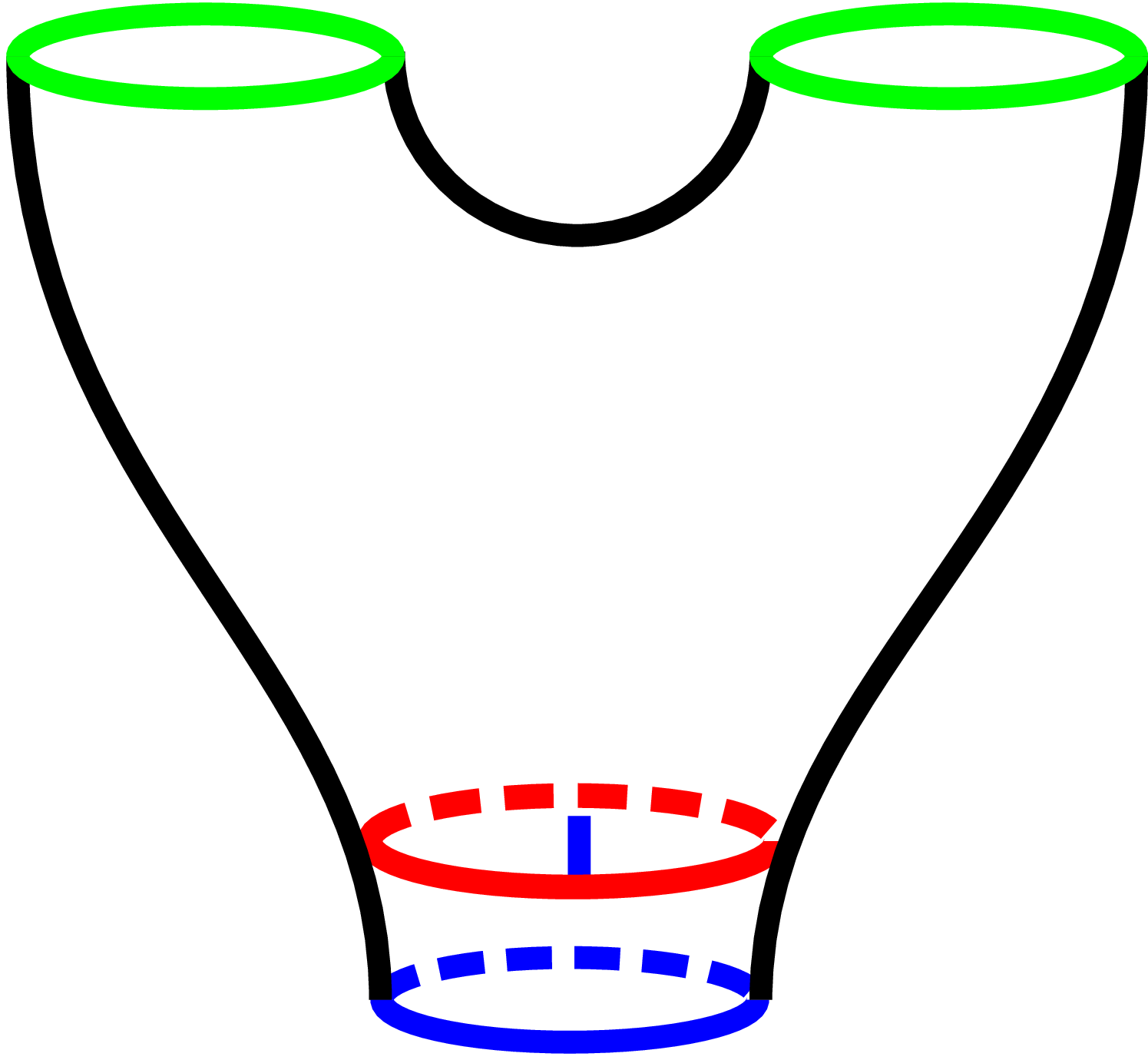}};
\node (R2r) at (10,0) {\includegraphics[scale=0.12,angle=180]{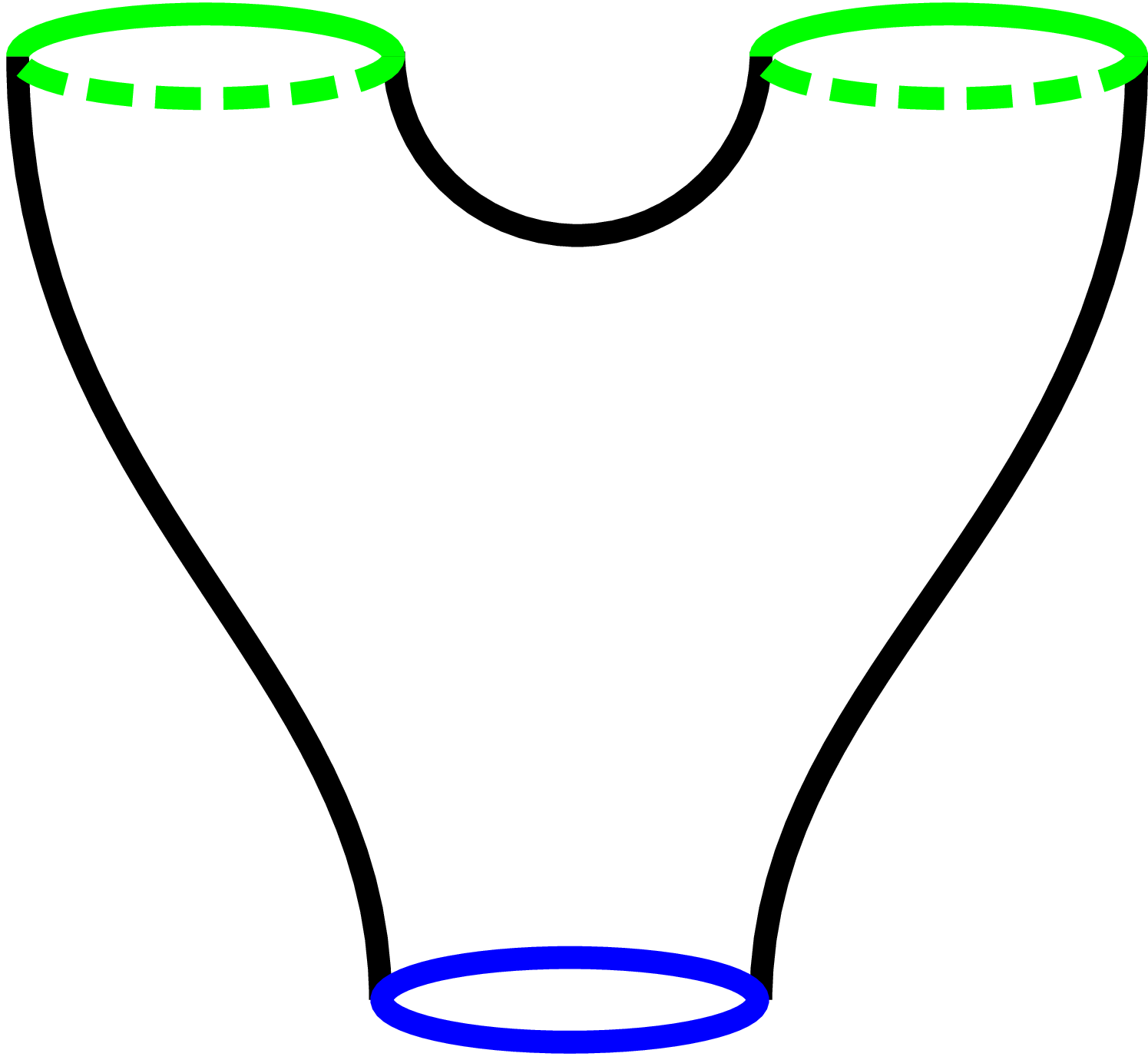}};
\node (=) at (1.7,0) {$\longleftrightarrow$};
\node (=) at (8.7,0) {$\longleftrightarrow$};
\node at (-1.8,1) {R16)};
\node at (5.2,1) {R17)};
\end{tikzpicture}
\end{center}

\begin{center}
\begin{tikzpicture}
\node (R1l) at (0,0) {\includegraphics[scale=0.12]{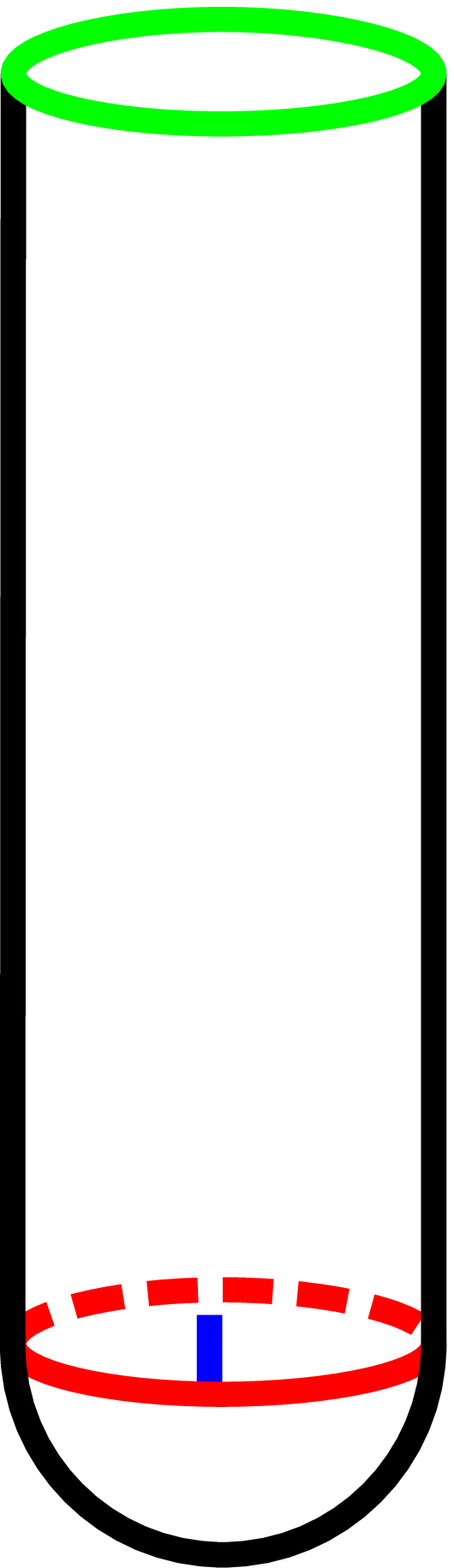}};
\node (R1r) at (3,0) {\includegraphics[scale=0.12]{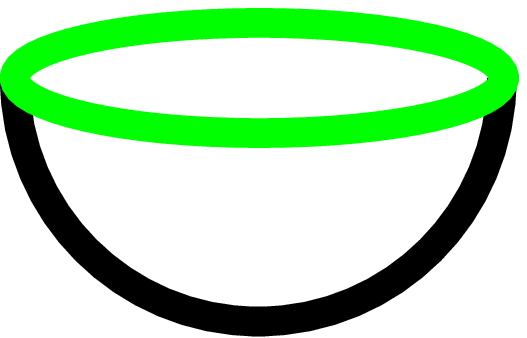}};
\node (R2l) at (7,0) {\includegraphics[scale=0.12]{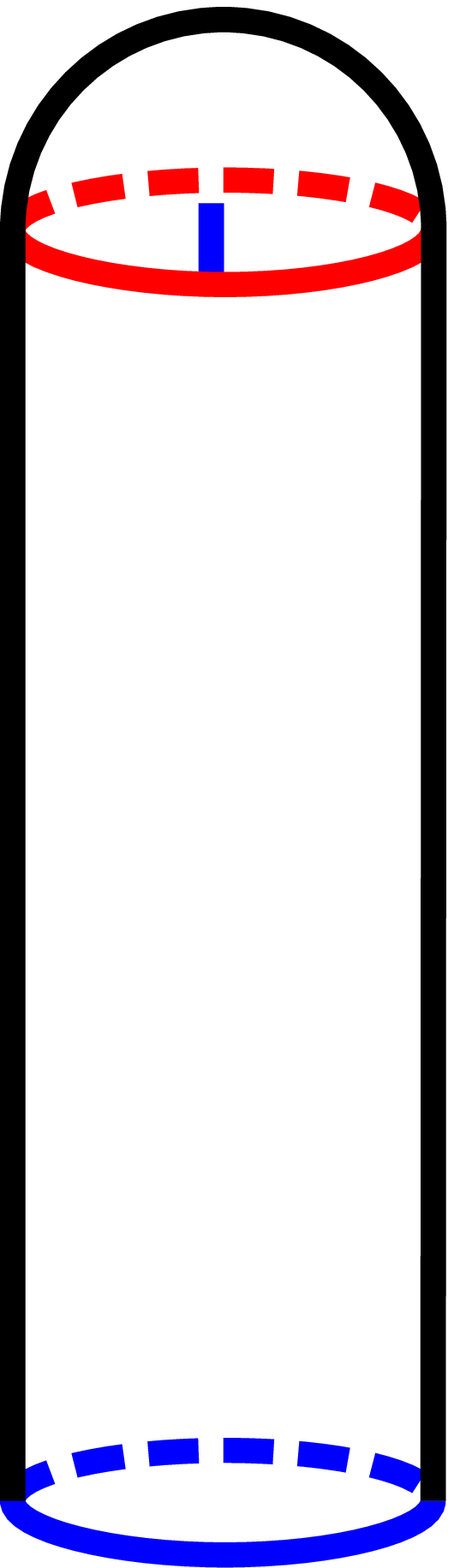}};
\node (R2r) at (10,0) {\includegraphics[scale=0.12]{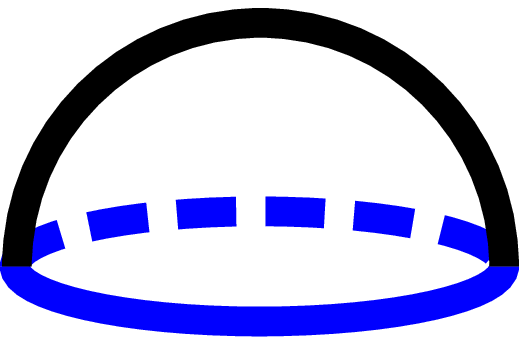}};
\node (=) at (1.7,0) {$\longleftrightarrow$};
\node (=) at (8.7,0) {$\longleftrightarrow$};
\node at (-2.3,1) {R18)};
\node at (4.7,1) {R19)};
\end{tikzpicture}
\end{center}

\begin{center}
\begin{tikzpicture}
\node (R1l) at (0,0) {\includegraphics[scale=0.12]{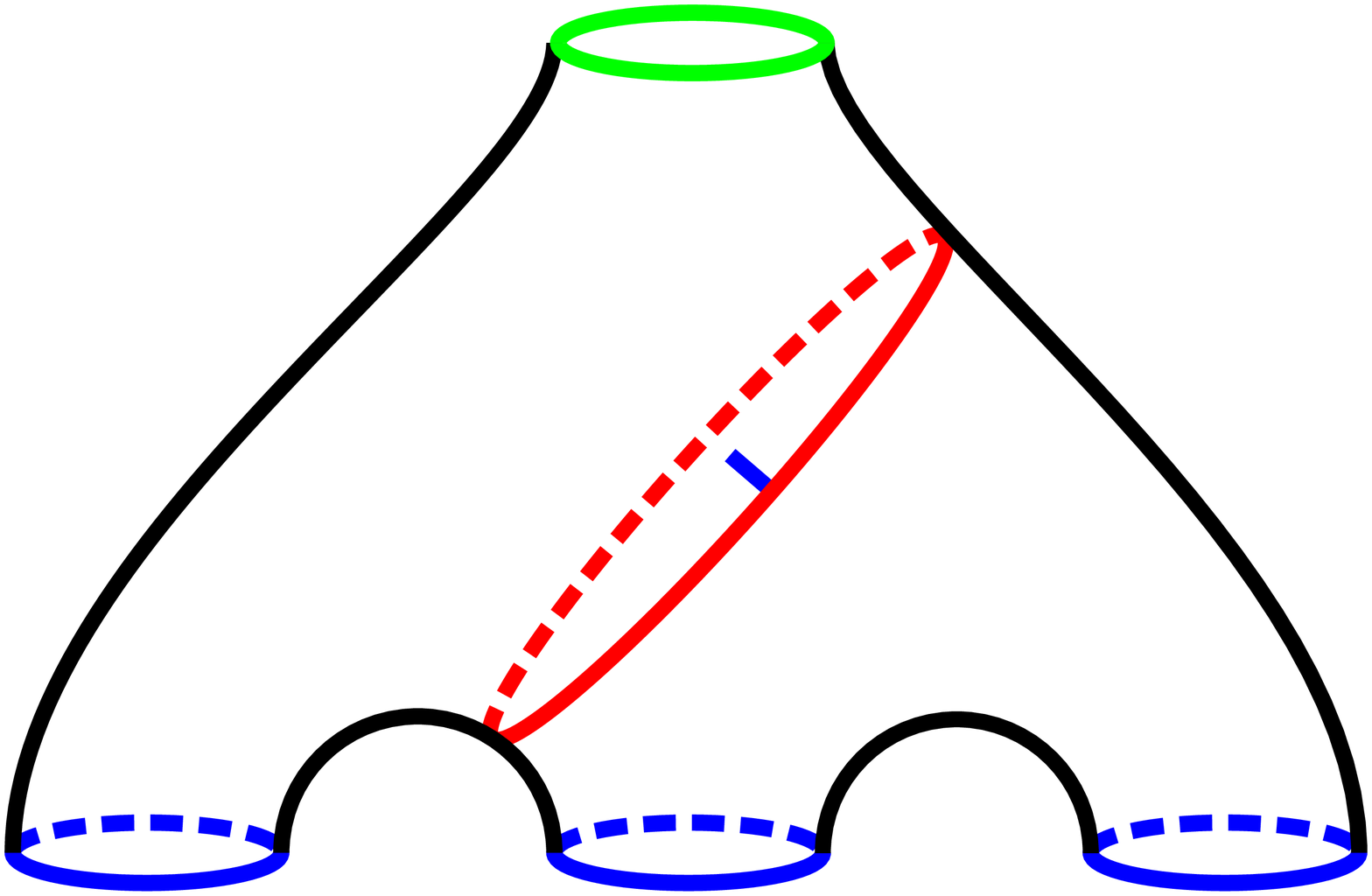}};
\node (R1r) at (4,0) {\includegraphics[scale=0.12]{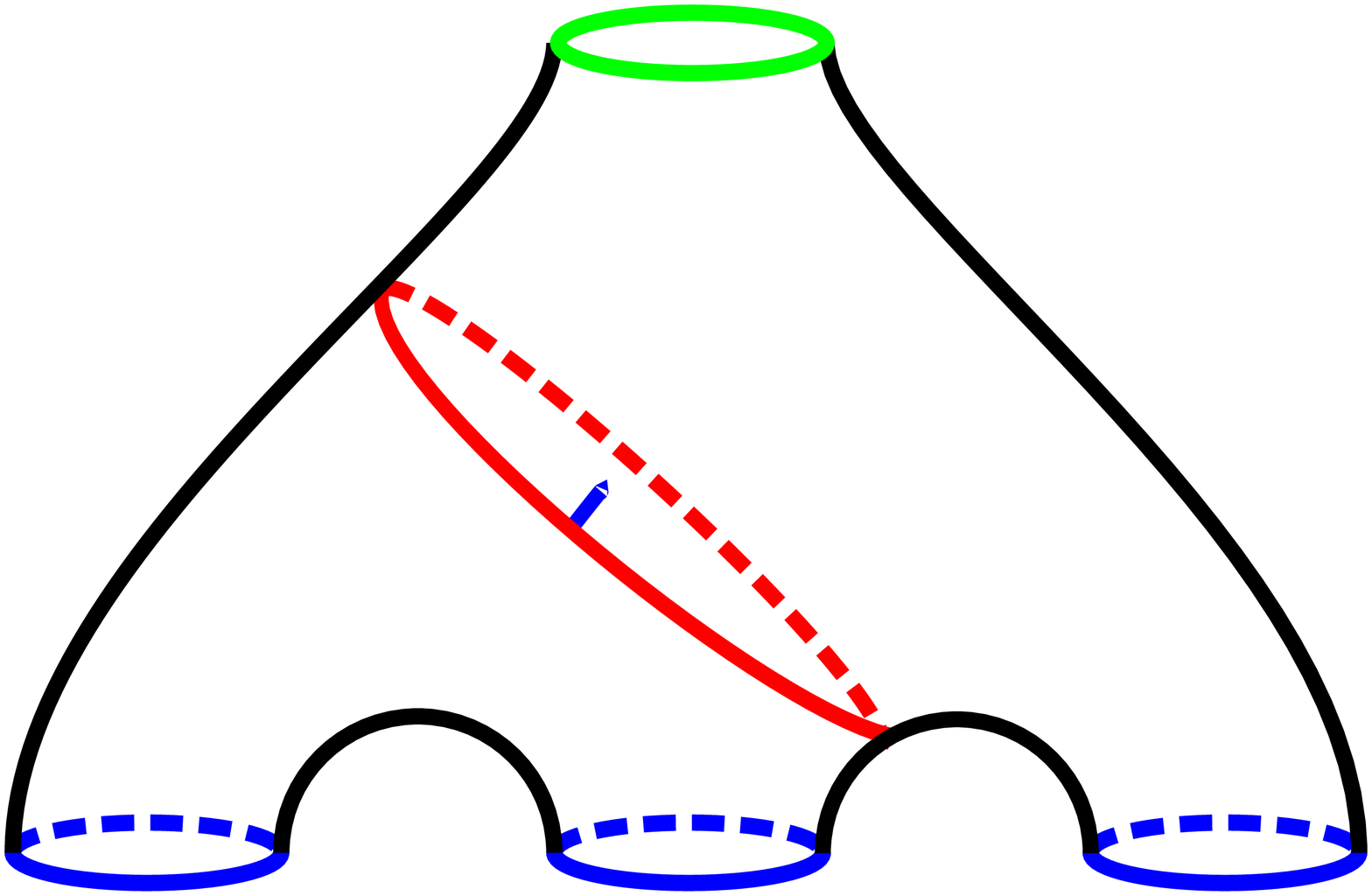}};
\node (=) at (1.7,0) {$\longleftrightarrow$};
\node at (-1.8,1) {R20)};
\end{tikzpicture}
\end{center}

\begin{center}
\begin{tikzpicture}
\node (R2l) at (0,0) {\includegraphics[scale=0.12,angle=180]{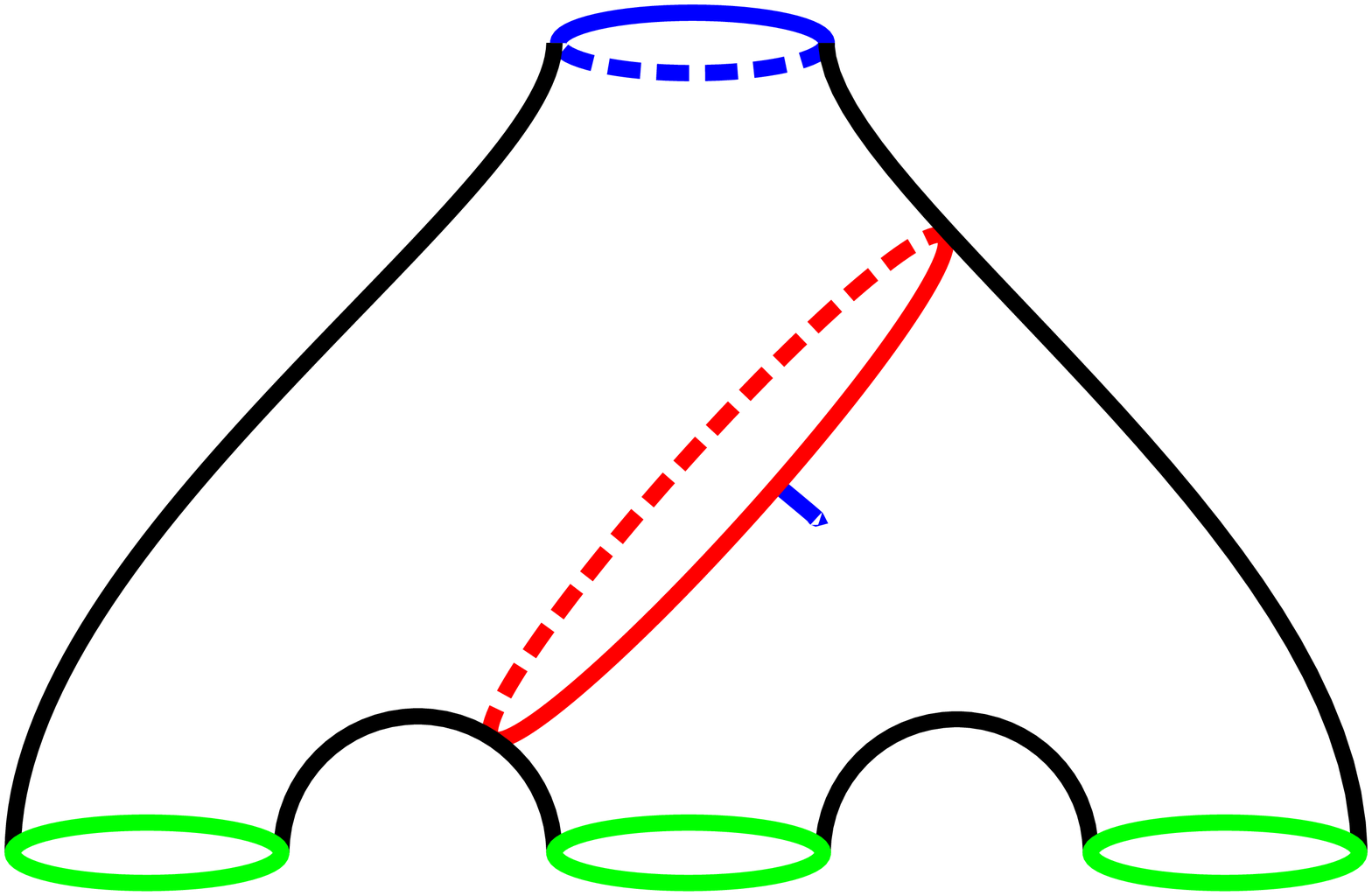}};
\node (R2r) at (4,0) {\includegraphics[scale=0.12,angle=180]{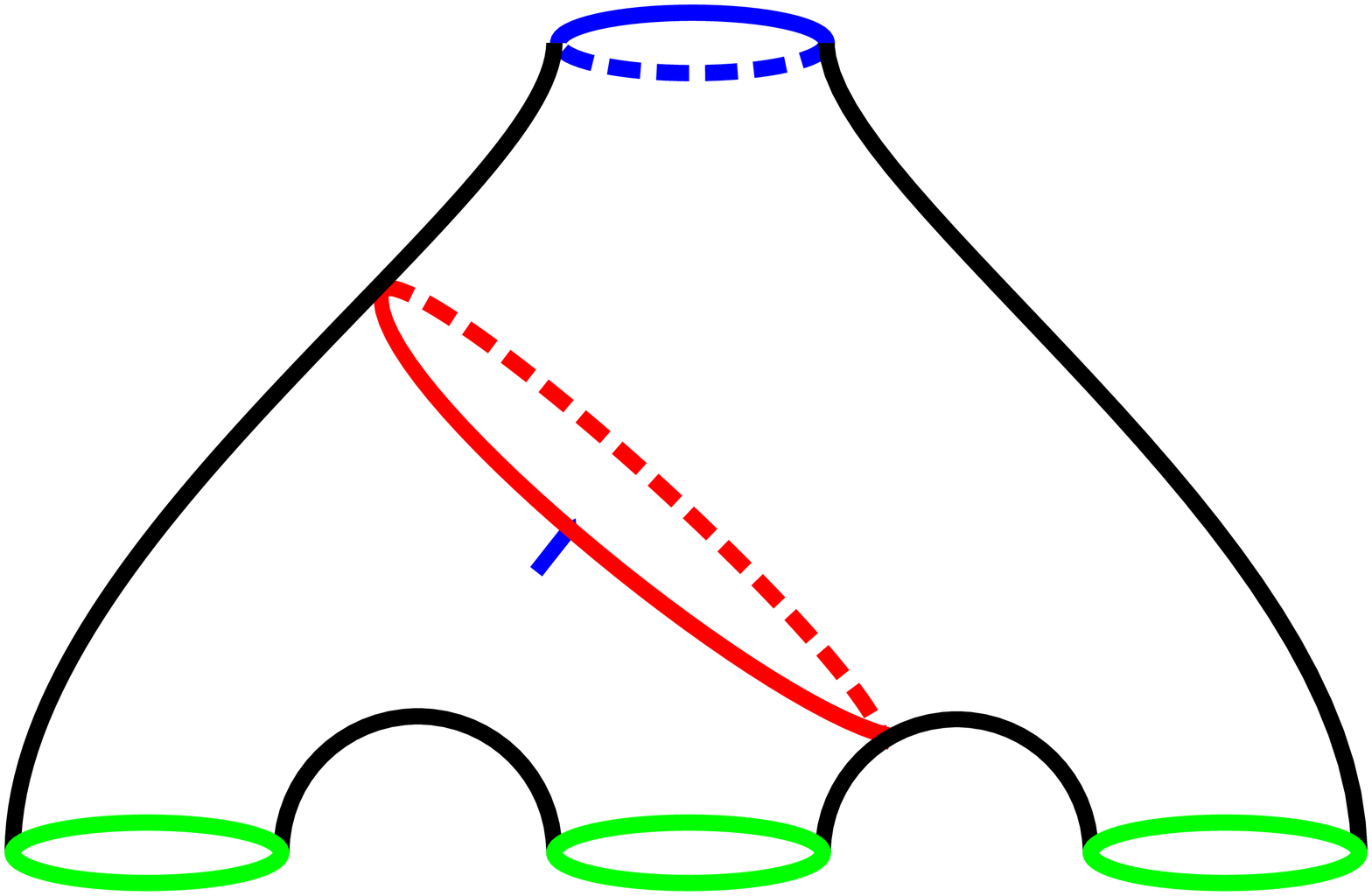}};
\node (=) at (1.7,0) {$\longleftrightarrow$};
\node at (-2,1) {R21)};
\end{tikzpicture}
\end{center}

\begin{center}
\begin{tikzpicture}
\node (R1l) at (0,0) {\includegraphics[scale=0.12]{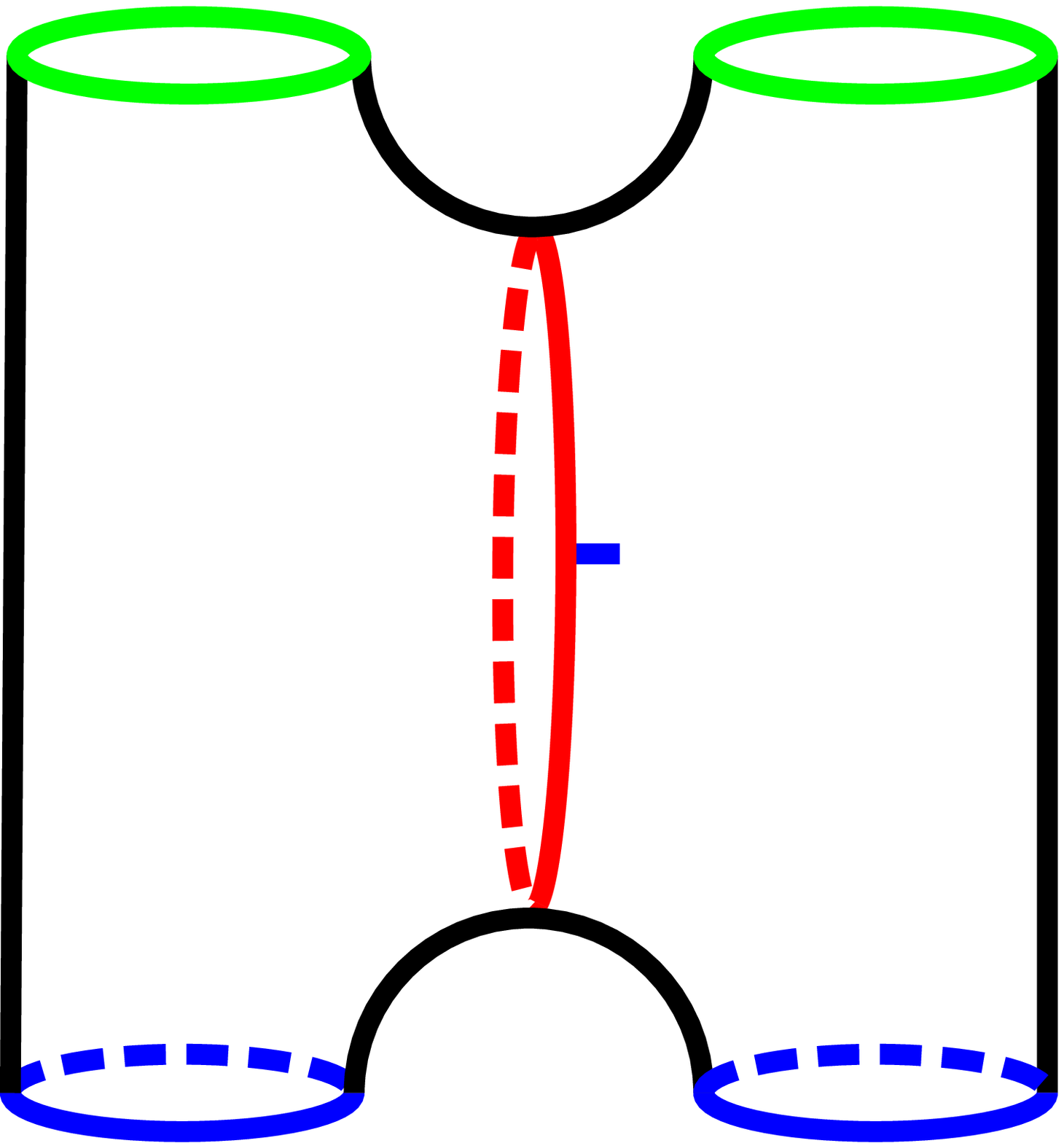}};
\node (R1r) at (3,0) {\includegraphics[scale=0.12]{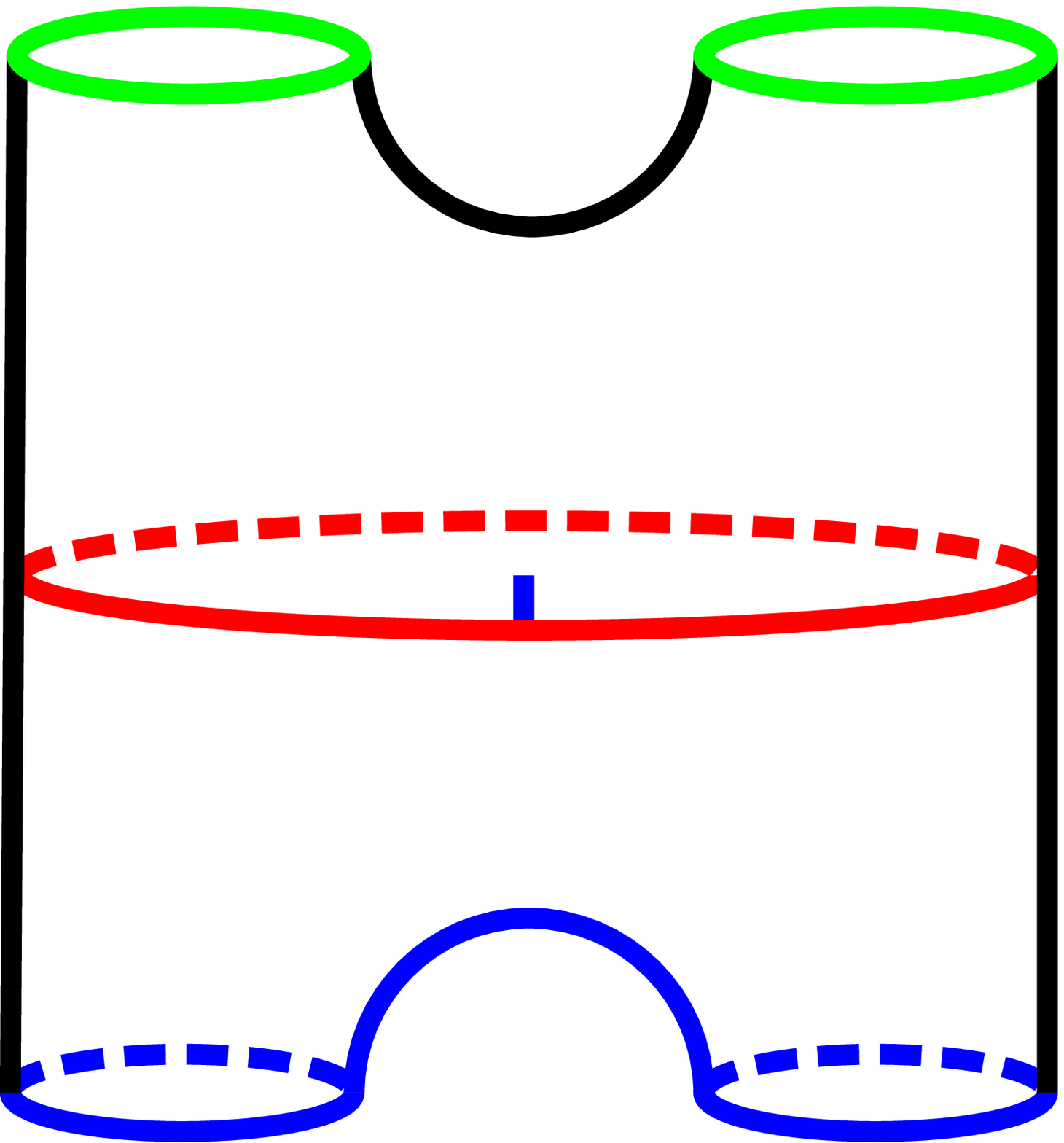}};
\node (R2l) at (7,0) {\includegraphics[scale=0.12]{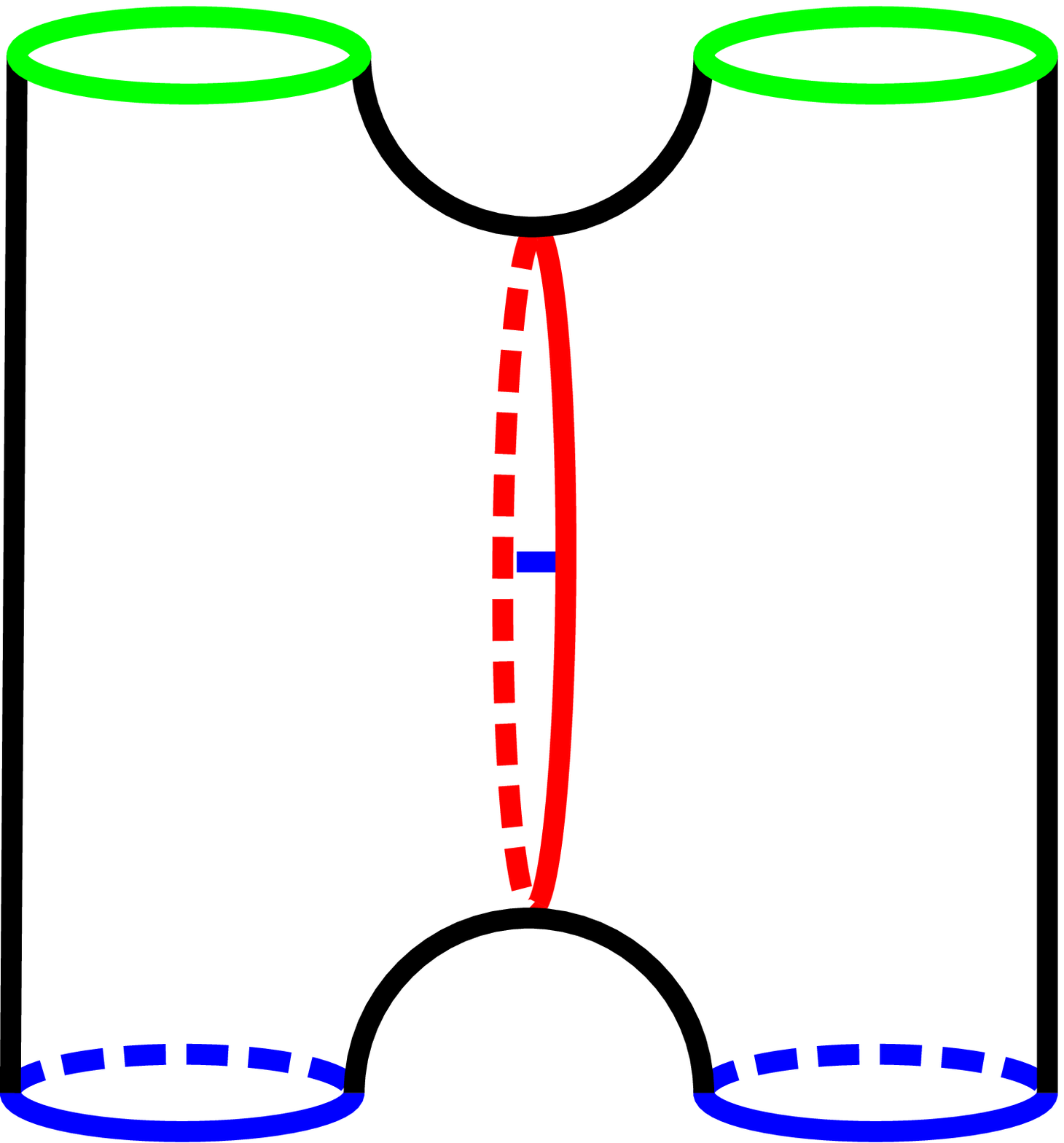}};
\node (R2r) at (10,0) {\includegraphics[scale=0.12]{figure186.eps}};
\node (=) at (1.7,0) {$\longleftrightarrow$};
\node (=) at (8.7,0) {$\longleftrightarrow$};
\node at (-1.8,1) {R22)};
\node at (5.2,1) {R23)};
\end{tikzpicture}
\end{center}
In the picture of the Dehn-twist and braid move, the red dashed lines are not glueing lines, but auxiliary curves to display the action of the elements of the mapping class group corresponding to the moves.
\begin{center}
\begin{tikzpicture}
\node (R1l) at (0,0) {\includegraphics[scale=0.12]{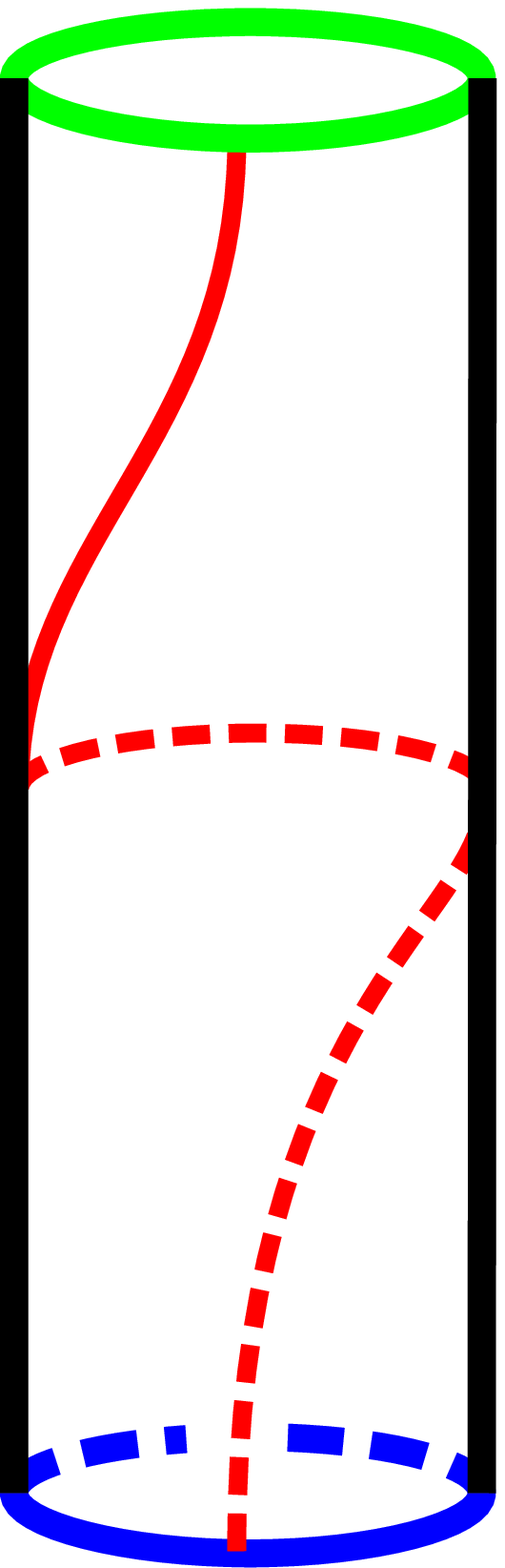}};
\node (R1r) at (3,0) {\includegraphics[scale=0.12]{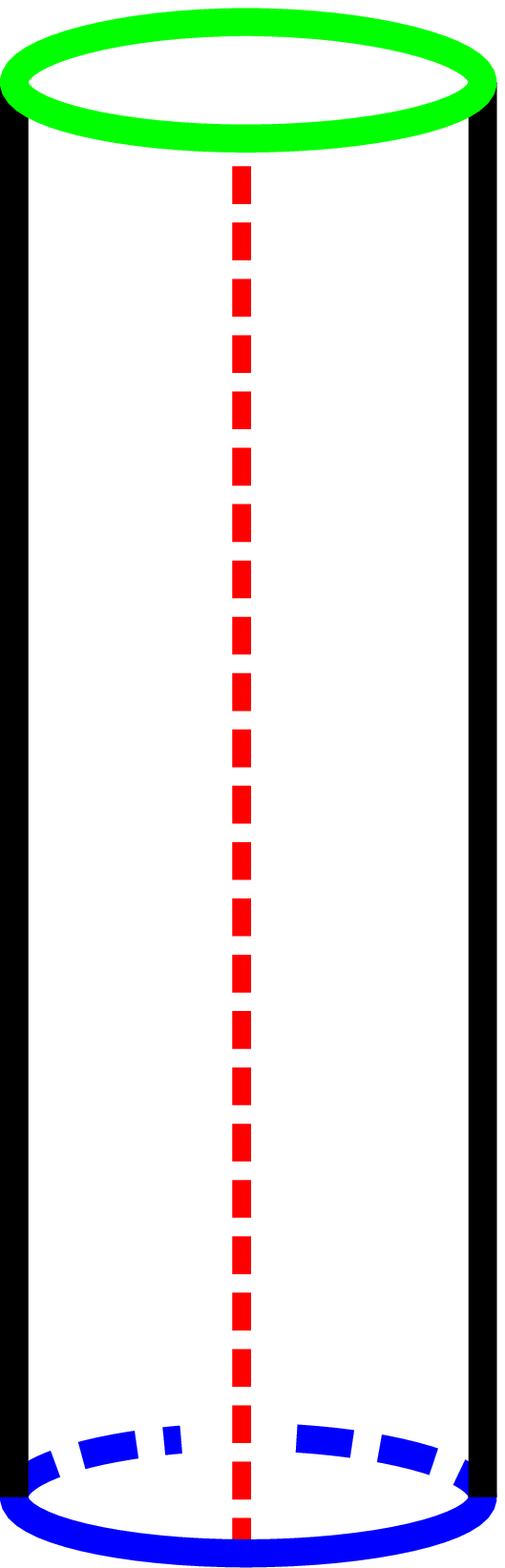}};
\node (R2l) at (7,0) {\includegraphics[scale=0.12]{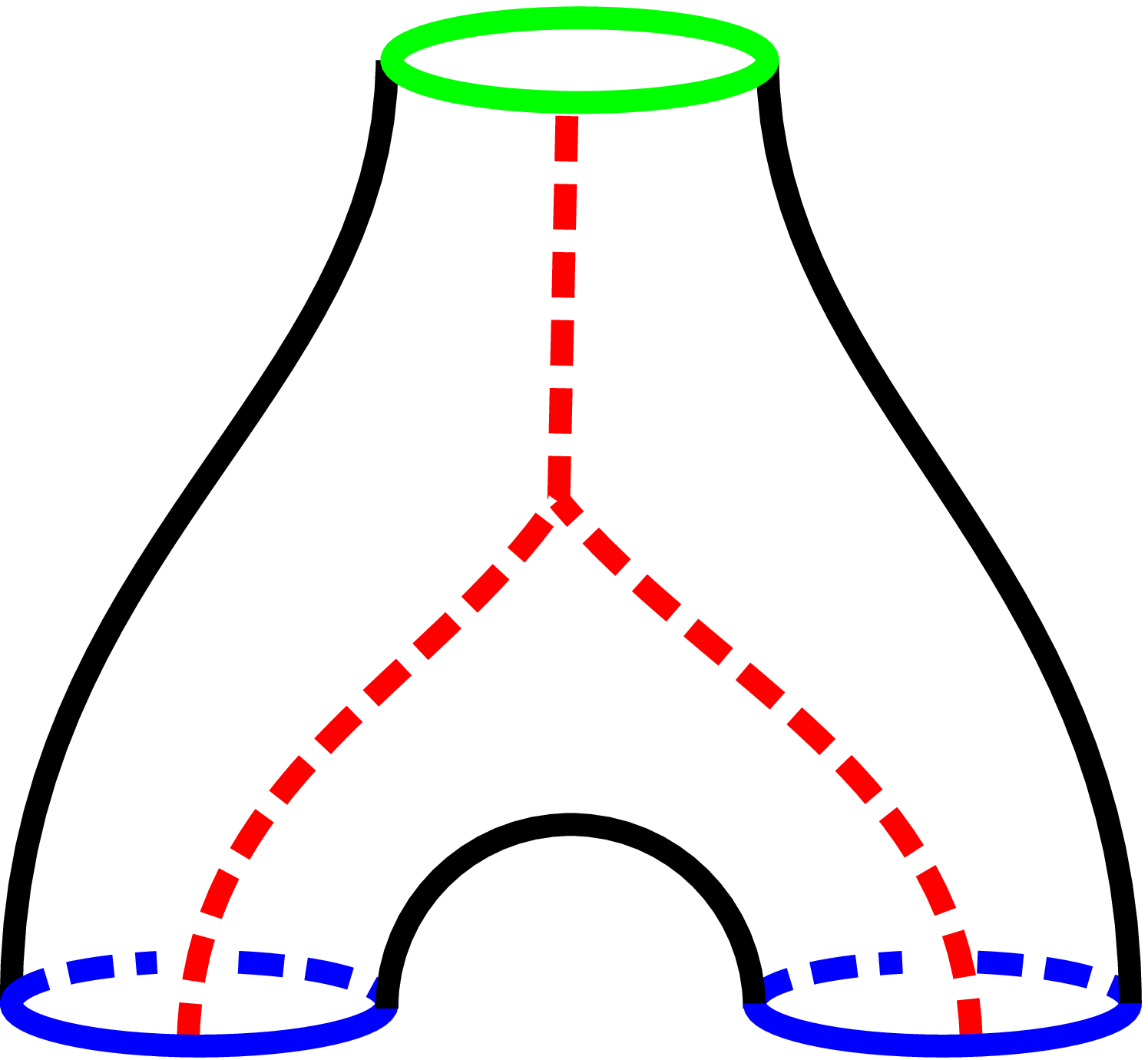}};
\node (R2r) at (10,0) {\includegraphics[scale=0.12]{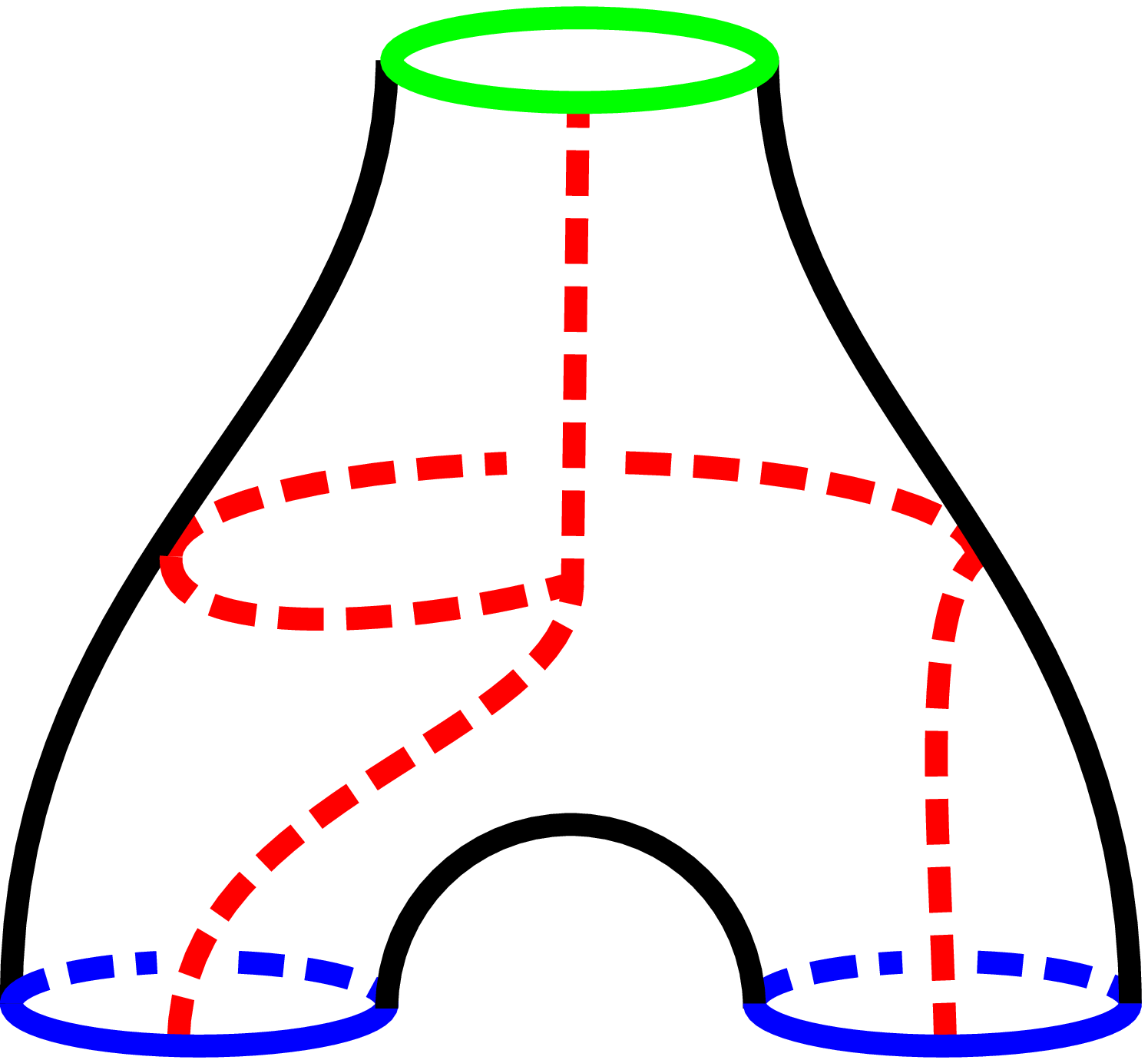}};
\node (=) at (1.7,0) {$\longleftrightarrow$};
\node (=) at (8.7,0) {$\longleftrightarrow$};
\node at (-1.8,1) {R24)};
\node at (5.2,1) {R25)};
\end{tikzpicture}
\end{center}

\item \textbf{Open-Closed Relations}

\begin{center}
\begin{tikzpicture}
\node (R1l) at (0,0) {\includegraphics[scale=0.12]{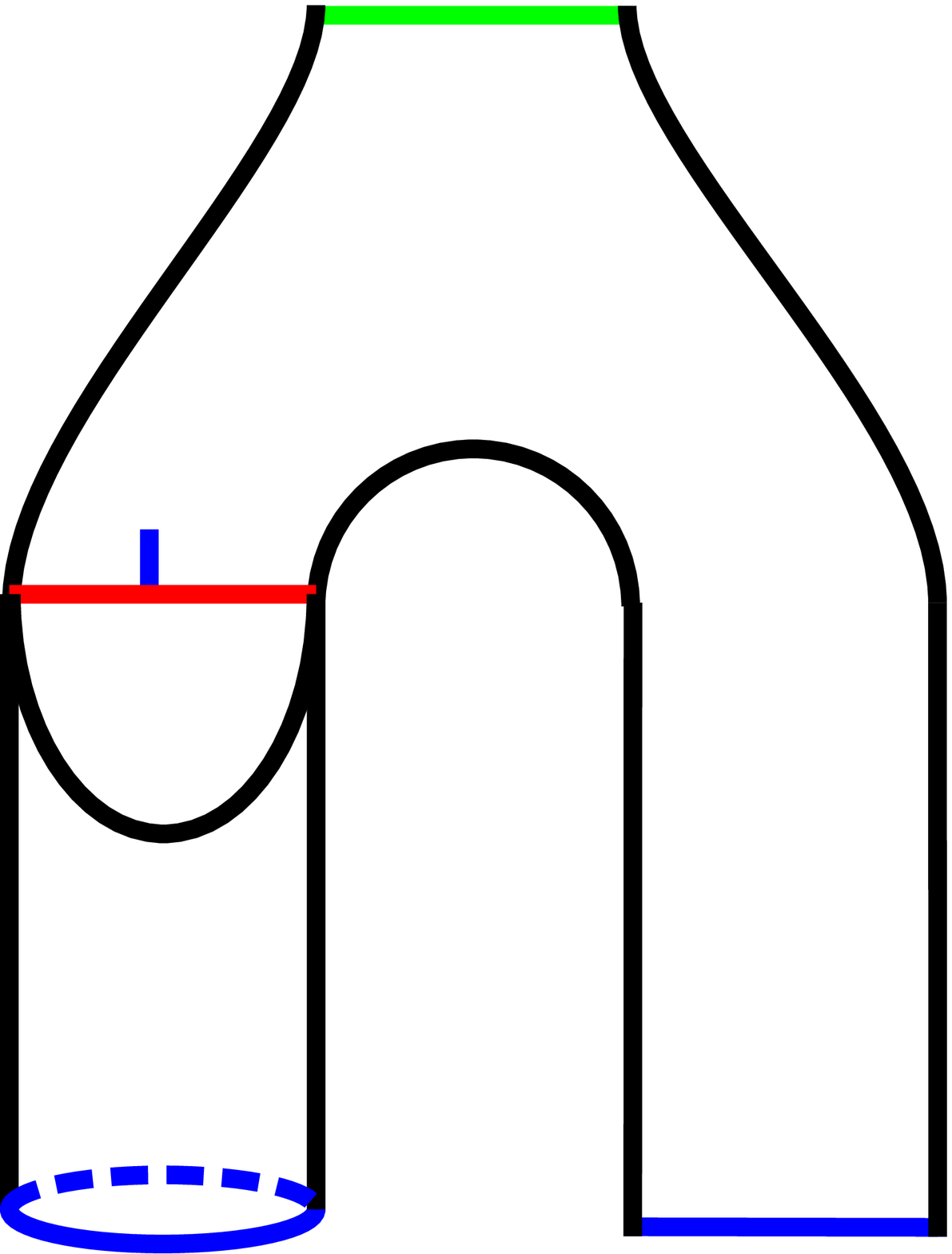}};
\node (R1r) at (3,0) {\includegraphics[scale=0.12]{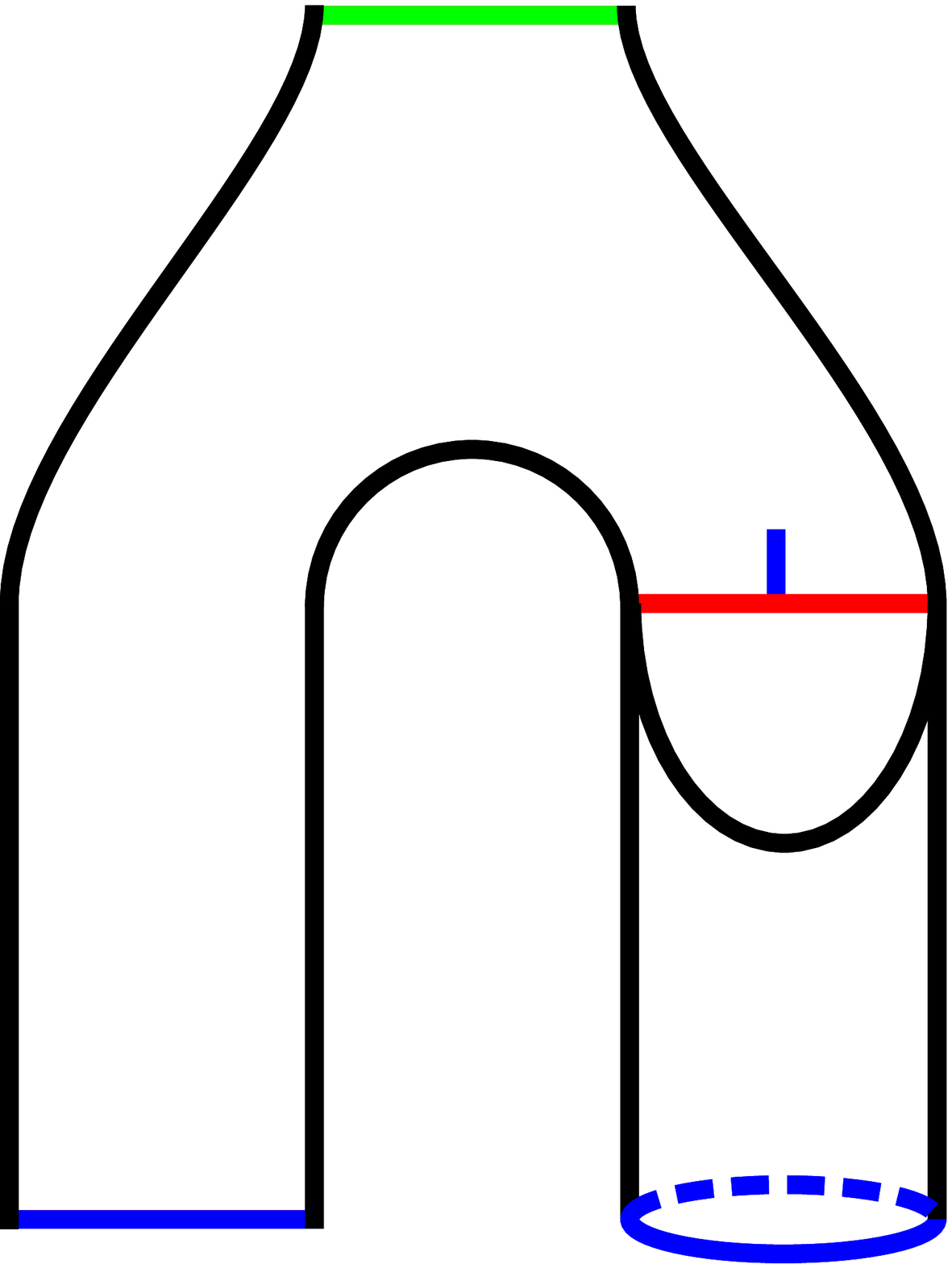}};
\node (R2l) at (7,0) {\includegraphics[scale=0.12]{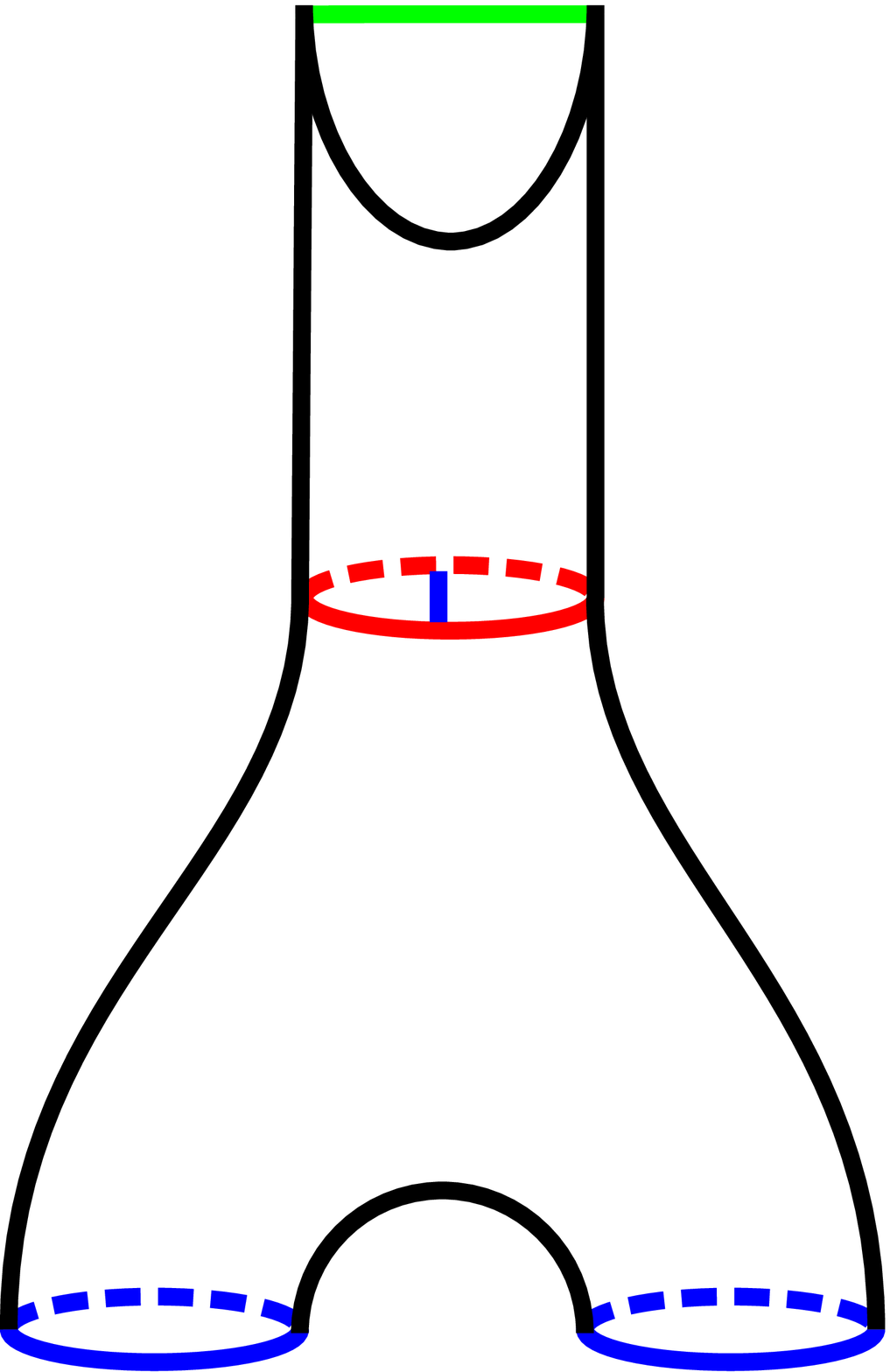}};
\node (R2r) at (10,0) {\includegraphics[scale=0.12]{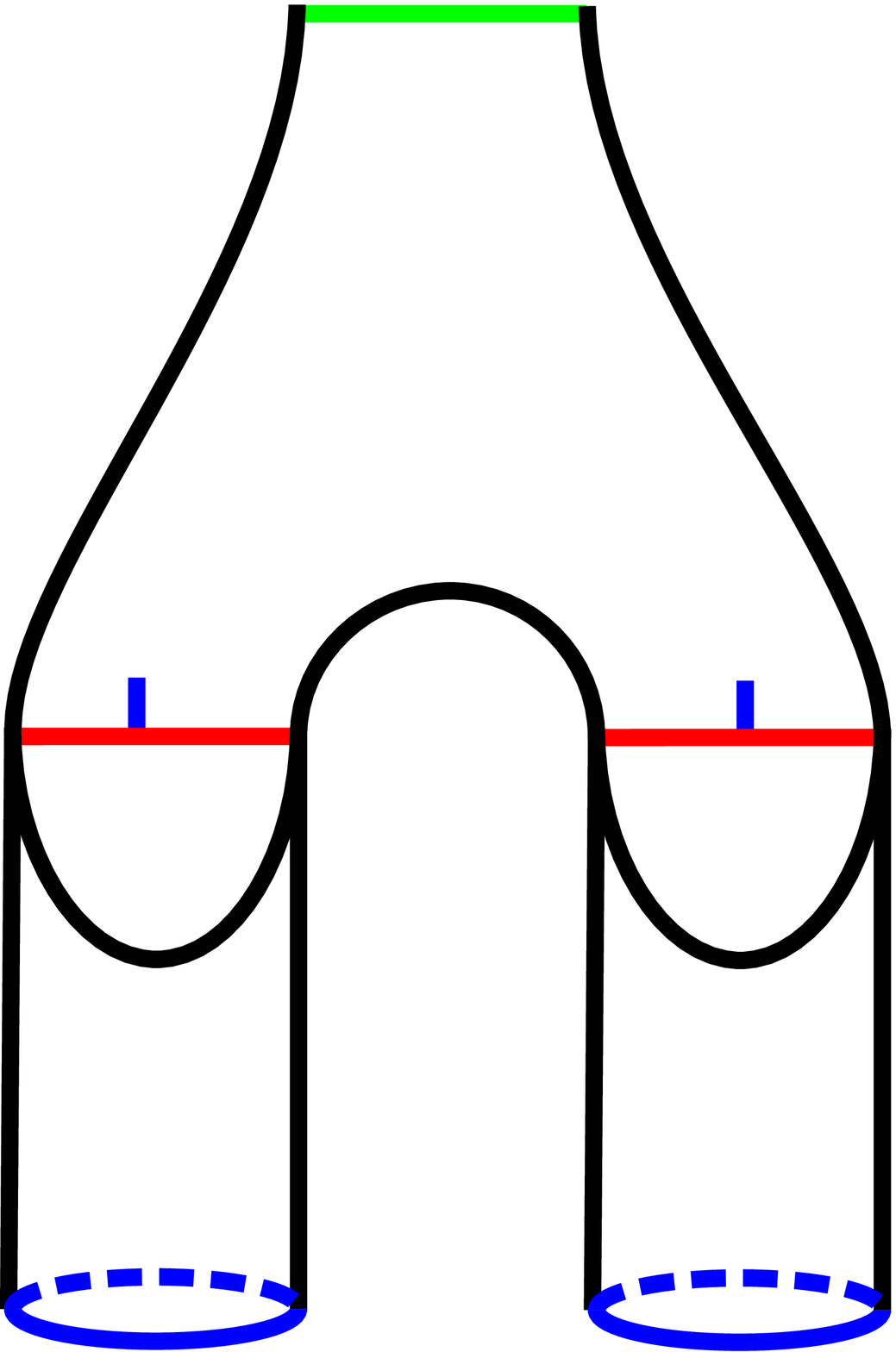}};
\node (=) at (1.7,0) {$\longleftrightarrow$};
\node (=) at (8.7,0) {$\longleftrightarrow$};
\node at (-1.8,1) {R26)};
\node at (5.2,1) {R27)};
\end{tikzpicture}
\end{center}

\begin{center}
\begin{tikzpicture}
\node at (0,0) {\includegraphics[scale=0.12]{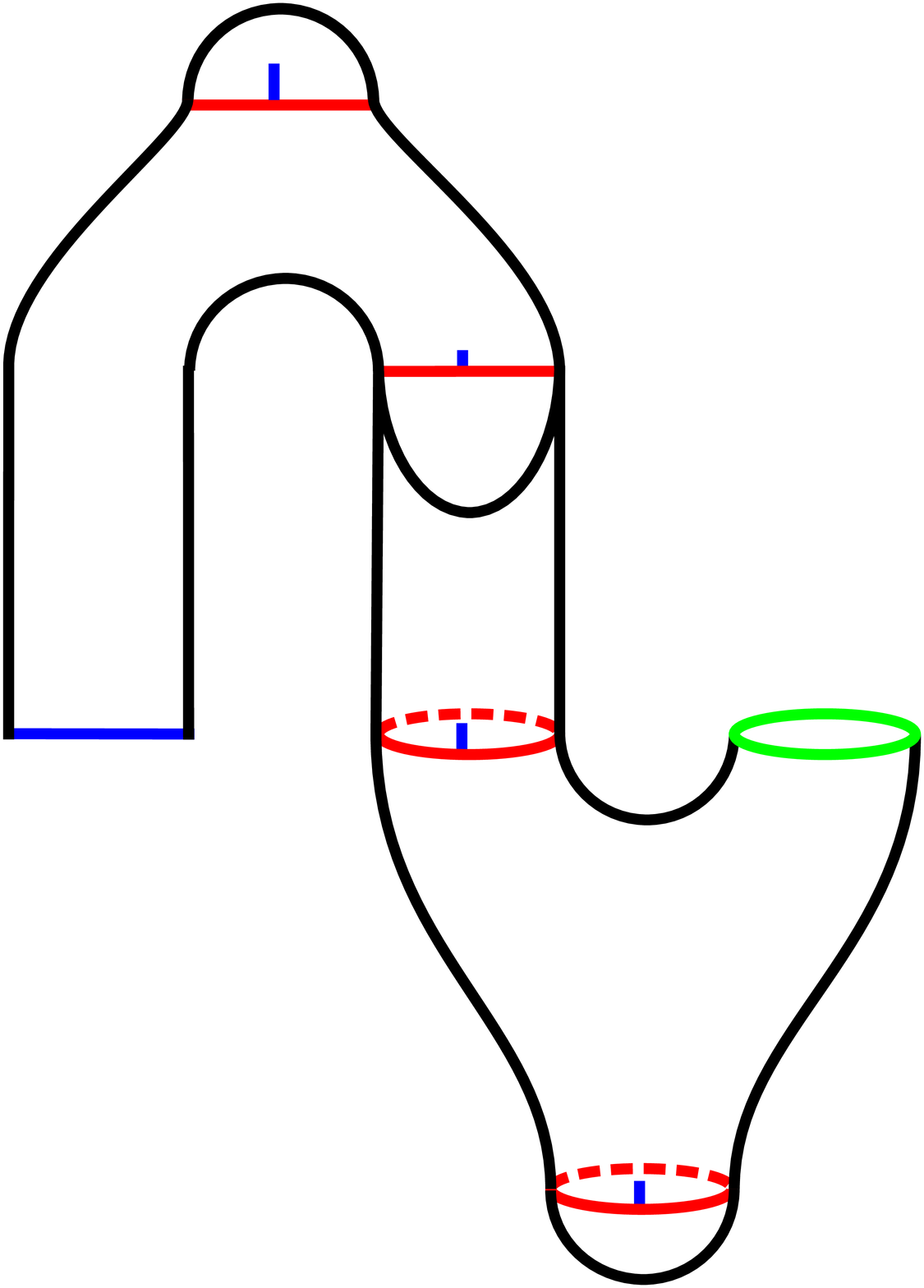}};
\node at (3,0) {\includegraphics[scale=0.12,angle=180]{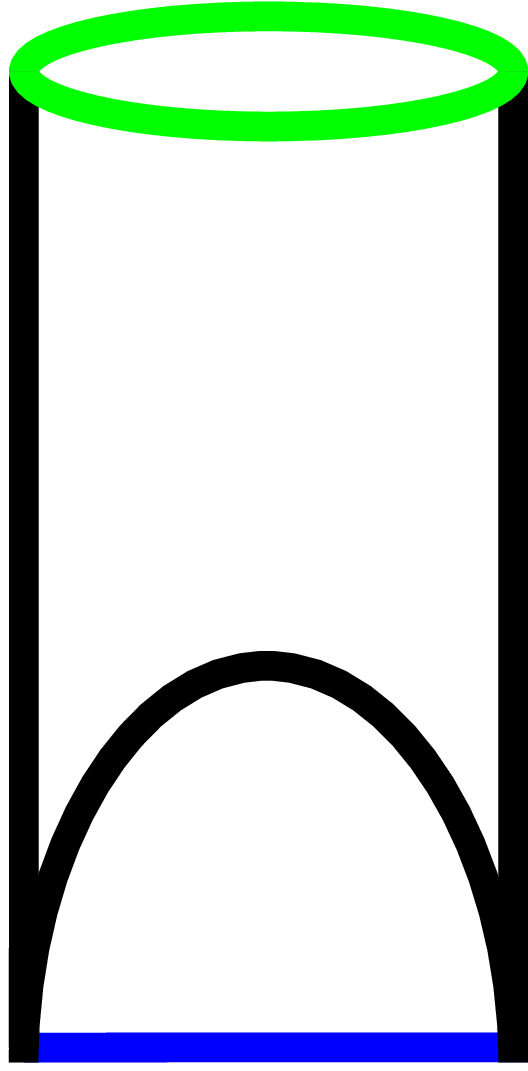}};
\node at (7,0) {\includegraphics[scale=0.12]{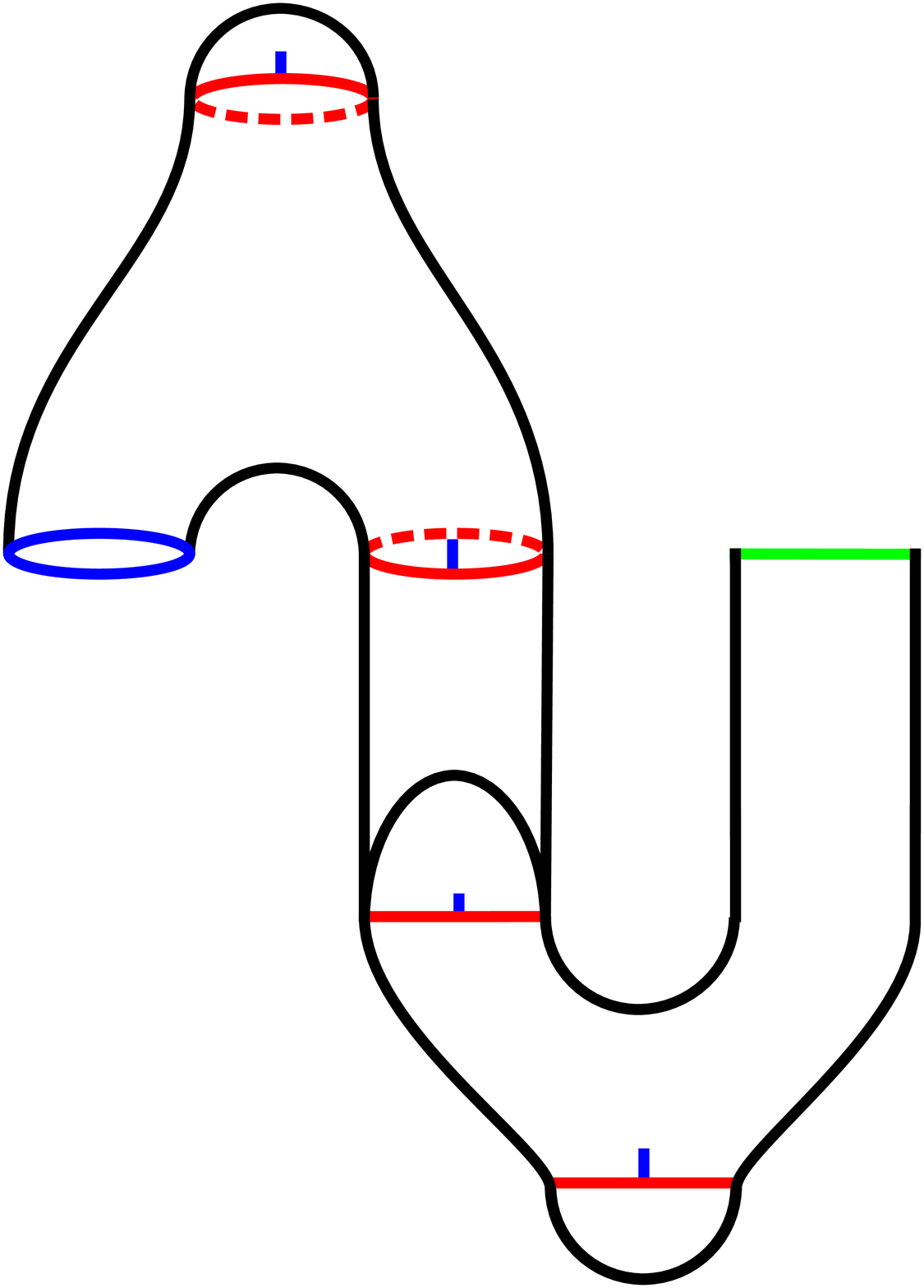}};
\node at (10,0) {\includegraphics[scale=0.12]{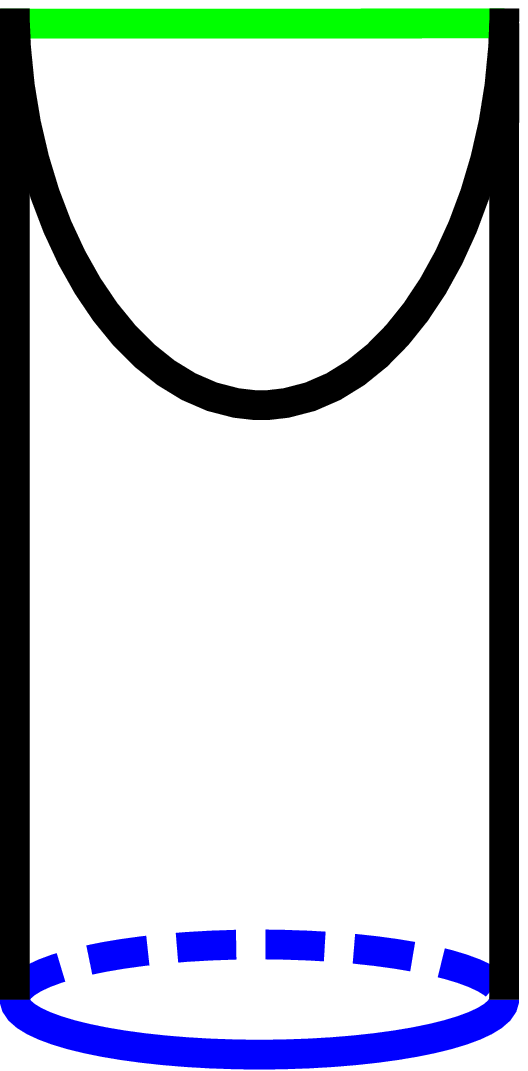}};
\node at (2,0) {$\longleftrightarrow$};
\node at (9,0) {$\longleftrightarrow$};
\node at (-2.3,1.5) {R28)};
\node at (4.7,1.5) {R29)};
\end{tikzpicture}
\end{center}
\begin{center}
\begin{tikzpicture}
\node (R1l) at (0,0) {\includegraphics[scale=0.12]{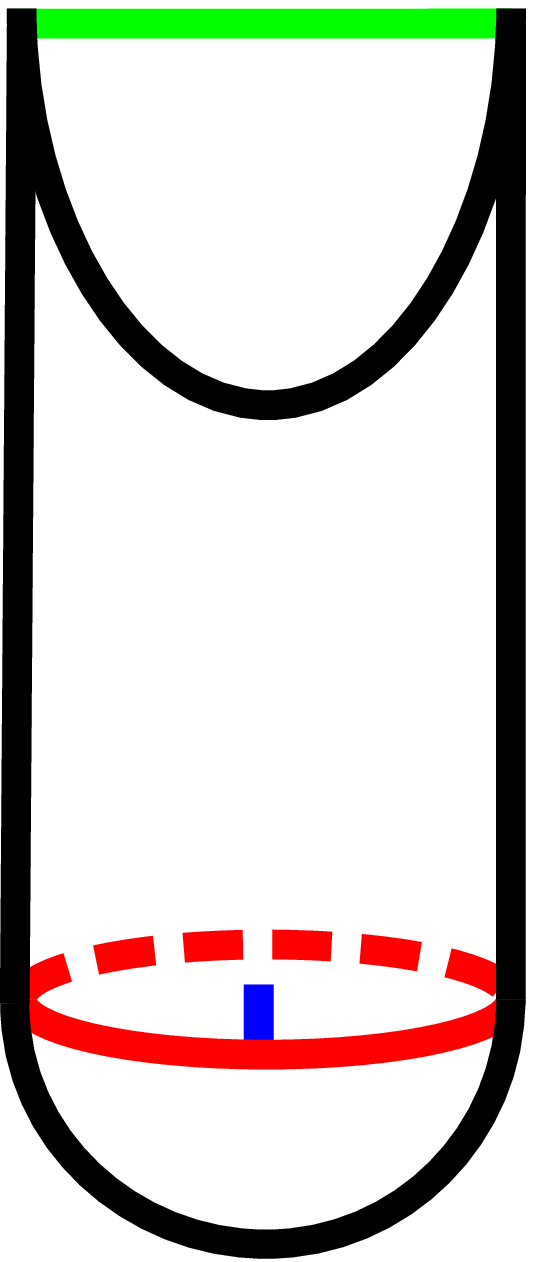}};
\node (R1r) at (3,0) {\includegraphics[scale=0.12]{figure167.eps}};
\node (R2l) at (7,0) {\includegraphics[scale=0.12]{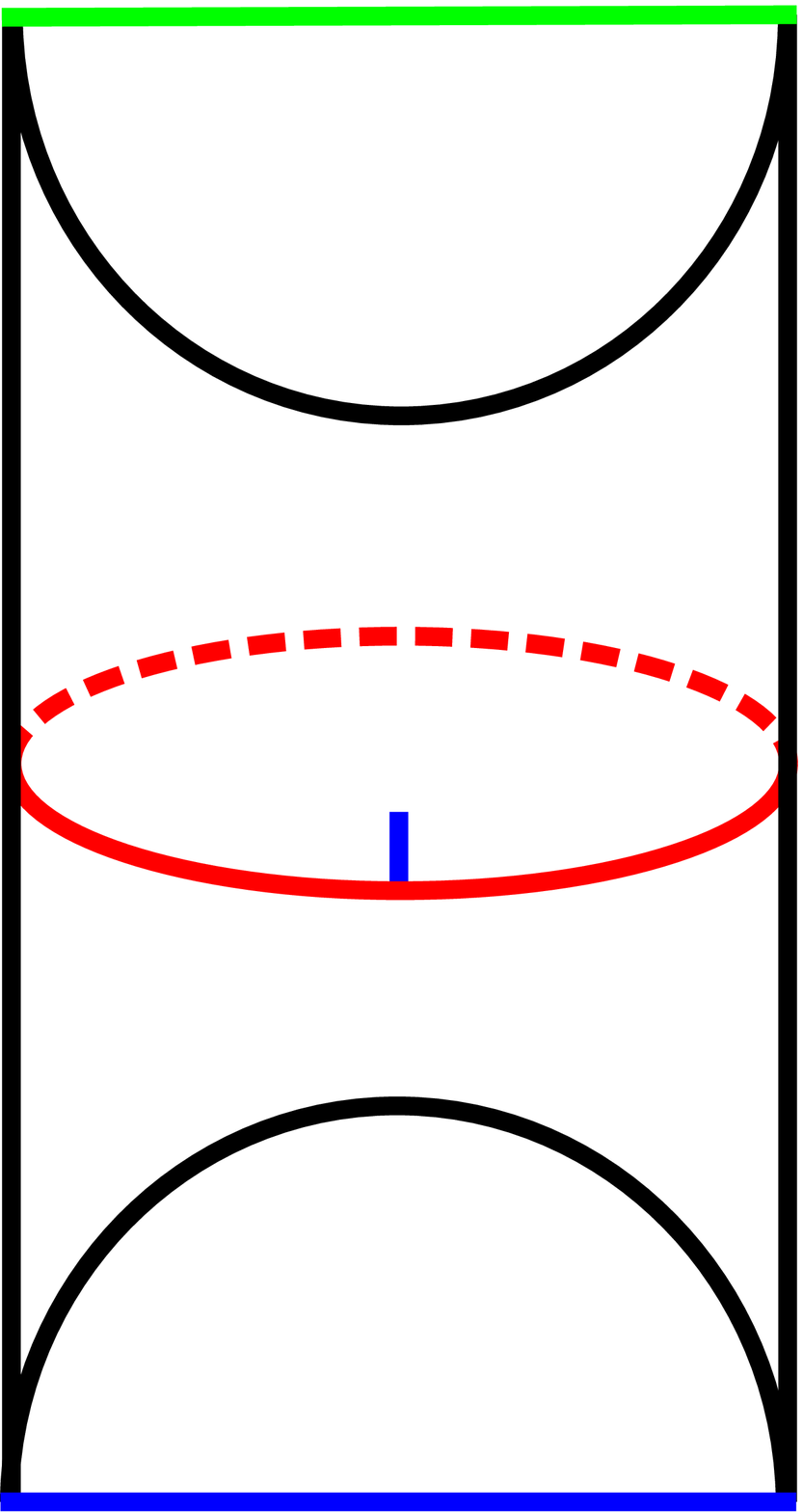}};
\node (R2r) at (10,0) {\includegraphics[scale=0.12]{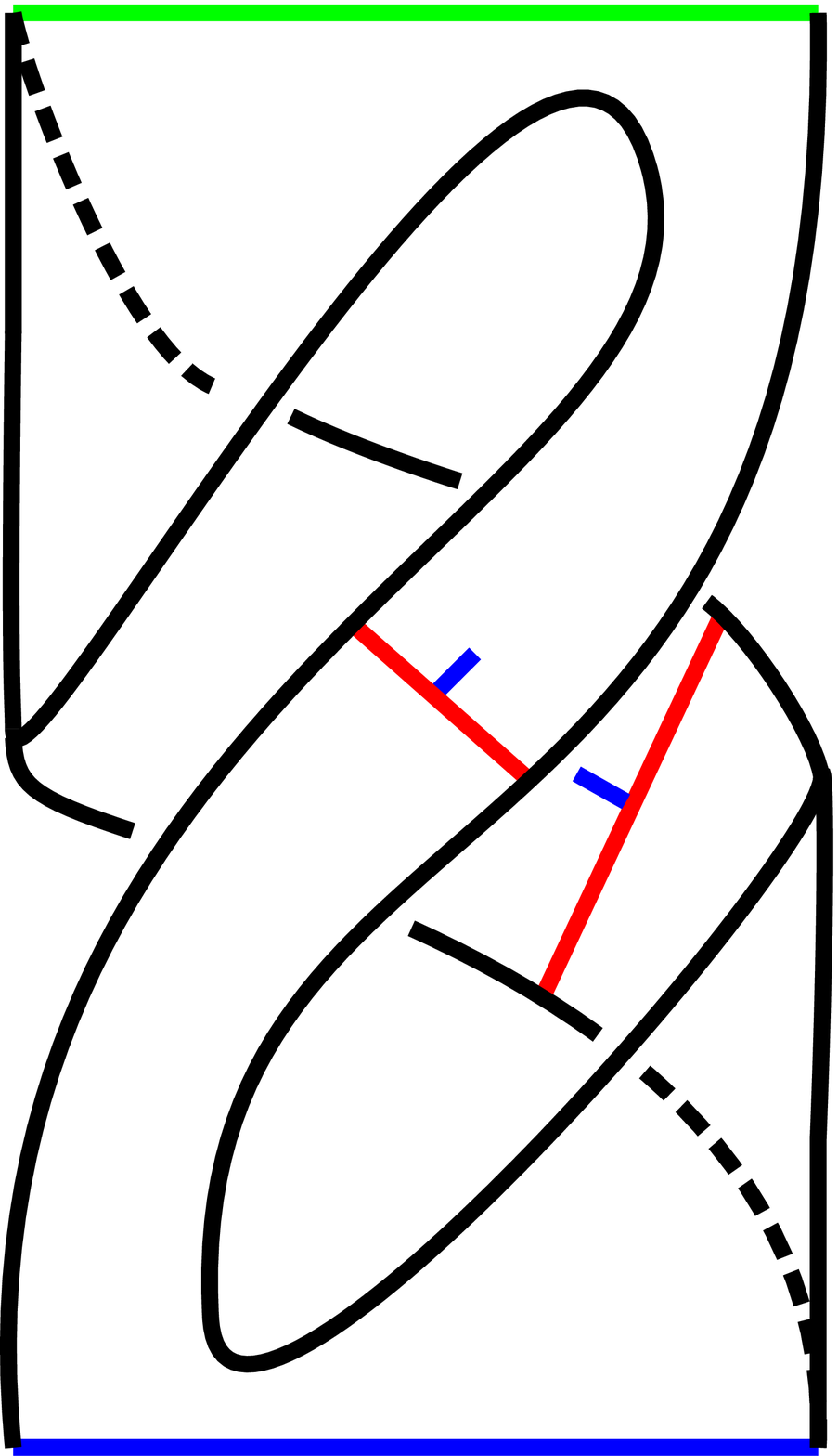}};
\node (=) at (1.7,0) {$\longleftrightarrow$};
\node (=) at (8.7,0) {$\longleftrightarrow$};
\node at (-1.8,1) {R30)};
\node at (5.2,1) {R31)};
\end{tikzpicture}
\end{center}

\item \textbf{Genus 1 Relation:} The genus one move takes place on a torus with one boundary component and interchanges $a$- and $b$-cycle of the torus as indicated by the colors. 
\begin{center}
\begin{tikzpicture}
\node at (0,0) {\includegraphics[scale=0.12]{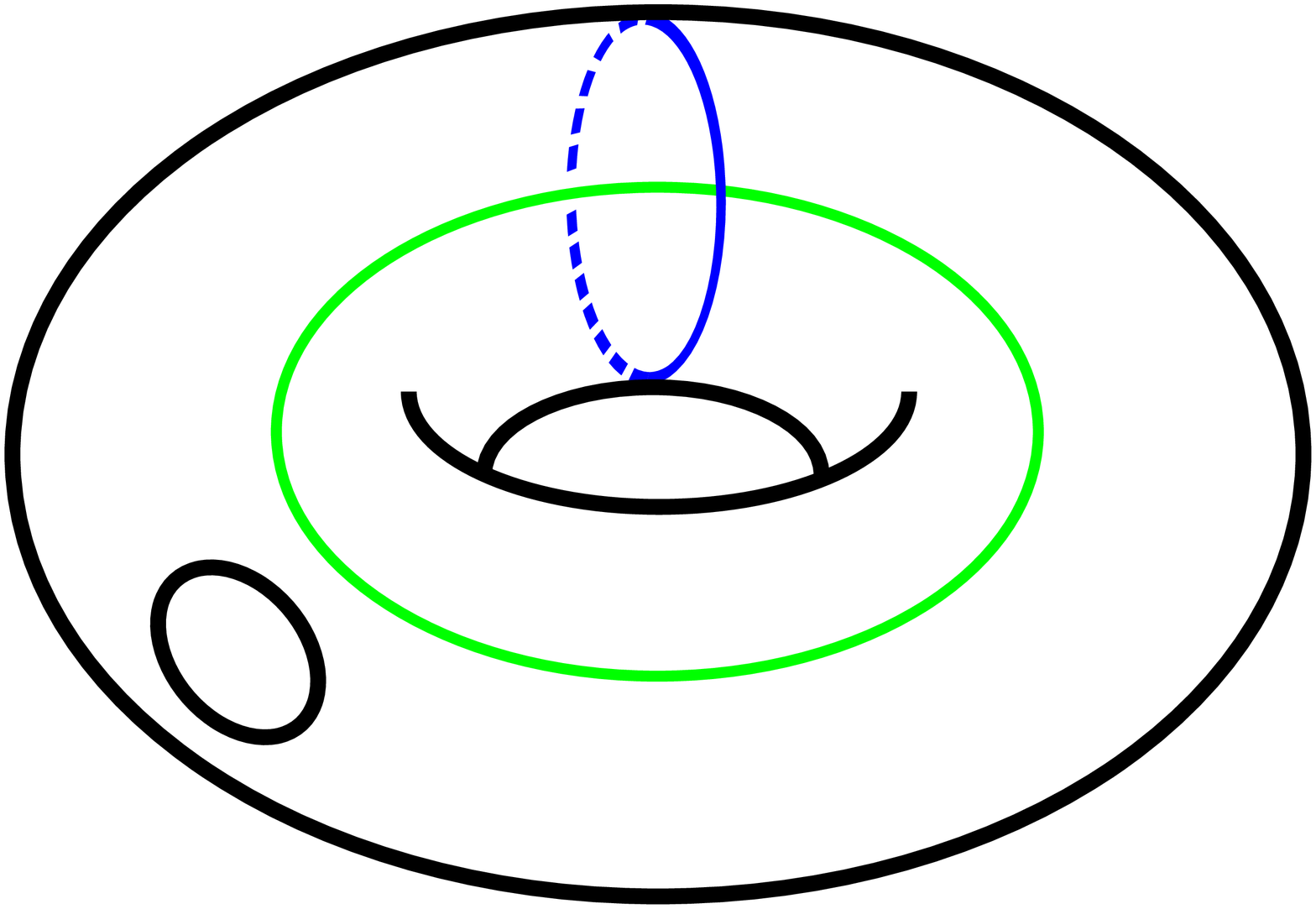}};
\node at (6,0) {\includegraphics[scale=0.12]{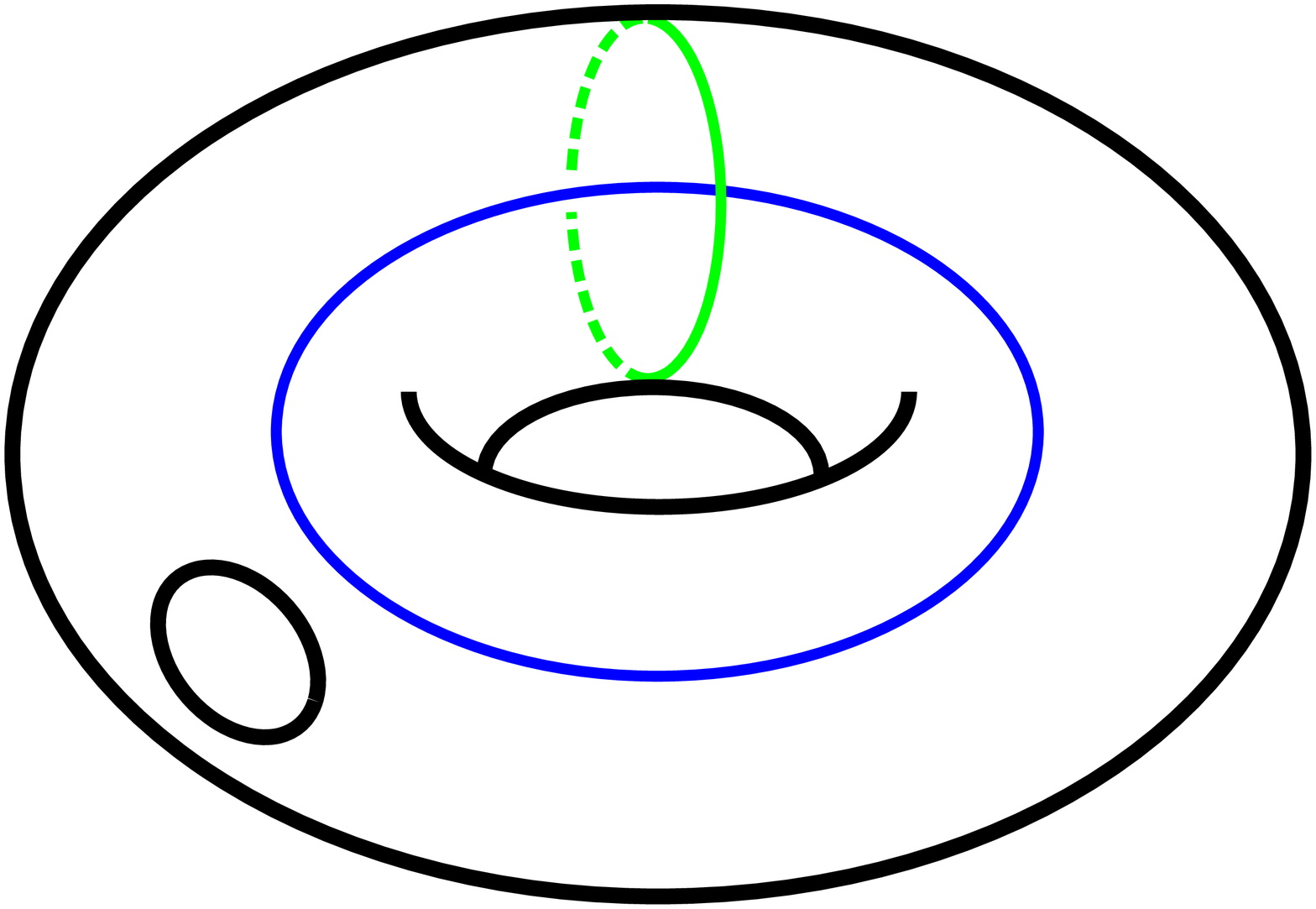}};
\node at (3,0) {$\longleftrightarrow$};
\node at (-2,1) {R32)};
\end{tikzpicture}
\end{center}
\end{enumerate}

The following theorem is the crucial simplification for the discussion of natural transformations.
\begin{theo}\cite[Theorem~2.8]{Kong_2014}\label{sewingconstraint}
Let $\Phi,\Psi\in \Funsf_\otimes(\WSsf,\Vectsf)$, $\Gcal_i$ and $\Gcal$ as above.  Then $\Gcal$ is a monoidal natural transformation if 
\eq{
\Gcal(R_{i,l})=\Gcal(R_{i,r})
}
for $\lbr R_i\rbr $ the 32 fundamental world sheet sewings given above.
\end{theo}

There is an obvious symmetric monoidal functor $\mathbf{1}:\WSsf\rightarrow \Vectsf$ called the \textit{trivial functor}. It maps any world sheet to $\Kbb$ and any morphism to the identity on $\Kbb$. The following definition of a solution to the sewing constraints is originally due to \cite{Fjelstad:2006aw}.

\begin{defn} A symmetric monoidal functor $\Theta\in \Funsf_\otimes (\WSsf,\Vectsf)$ satisfies the \textit{sewing constraints} if there is a monoidal natural transformation 
\eq{
\Delta:\mathbf{1}\Rightarrow \Theta\quad .
}
\end{defn}
We briefly explain why this is a sensible definition for a solution of the sewing constraints. First of all, the monoidal natural transformation $\Delta$ picks a vector in $\Theta(\hat{S})$ for any world sheet $\hat{S}$. In physics terms, one may call this the correlator of the surface. Recall, that correlators in CFT on any surface should be invariant under the action of the mapping class group. An element of the mapping class group gives a morphism $(\emptyset,f)$ and the functor $\mathbf{1}$ assignes the identity to it. Thus $\Theta(\emptyset,f)$ has to map the correlator onto itself by naturality. By the same argument of triviality for $\mathbf{1}$ and naturality, correlators on lower genus surfaces are sewn to correlators on higher genus surfaces. Hence this definition nicely captures all the features expected from a consistent set of correlators.

\section{Consistent Correlators from String-Nets}\label{sec5}

\subsection{Functor of Conformal Blocks}

For a consistent set of correlators we need a functor of open-closed conformal blocks $\Bcal\in Fun_\otimes(\WSsf,\Vectsf)$. This is achieved with the help of string-nets spaces. Let $(\Hcalcl,\Hcalop,\iotaclop)$ be a $(\Csf|\Zsf(\Csf))$-Cardy algebra. For a compact surface $\Sigma$ with non empty boundary we write $g_\Sigma$ for the genus of the corresponding closed surface $\Sigma^\prime$ obtained from $\Sigma$ by glueing disks to all boundary components.
\begin{enumerate}[label=\Roman*)]
\item Let $\hat{S}$ be a world sheet s.th. for $b_i\in \pi_0\left(\p \tilde{S}\right)$ it holds $\iota_
S(b_i)=b_i$ and $g_{\tilde{S}}=0$. Hence $B_{cl}=\emptyset$ and we denote $n=|B_{op}^i|$, $m=|B_{op}^o|$. The associated quotient surface $S_{n,m}$ is just a disk with $n$ incoming and $m$ outgoing open boundary components. We set 
\eq{
\Blcal(\hat{S})=\hat{H}^s\left( S,\ov{\Hcalop}\right)
}
where
\eq{
\ov{\Hcalop}=\underbrace{\widetilde{\Hcalop}\otimes \widetilde{\Hcalop}}_{n+m}, \, \quad
\widetilde{\Hcalop}=\begin{cases} &\Hcalop, \, \text{for outgoing boundary}\\
&\Hcalop^\ast,\, \text{for incoming boundary}\end{cases}\quad .
}
The tensor product $\ov{\Hcalop}$ is ordered according to the ordering function $\ordrm$.
\item Next we consider world sheets $\hat{S}$ with $B_{cl}\neq \emptyset$, $g_S=0$, with all open boundary components on a single boundary circle and no connected component of the boundary of the quotient surface is a physical boundary. These are world sheets, whose quotient surfaces are of the form

\begin{figure}[H]
\includegraphics[scale=0.12]{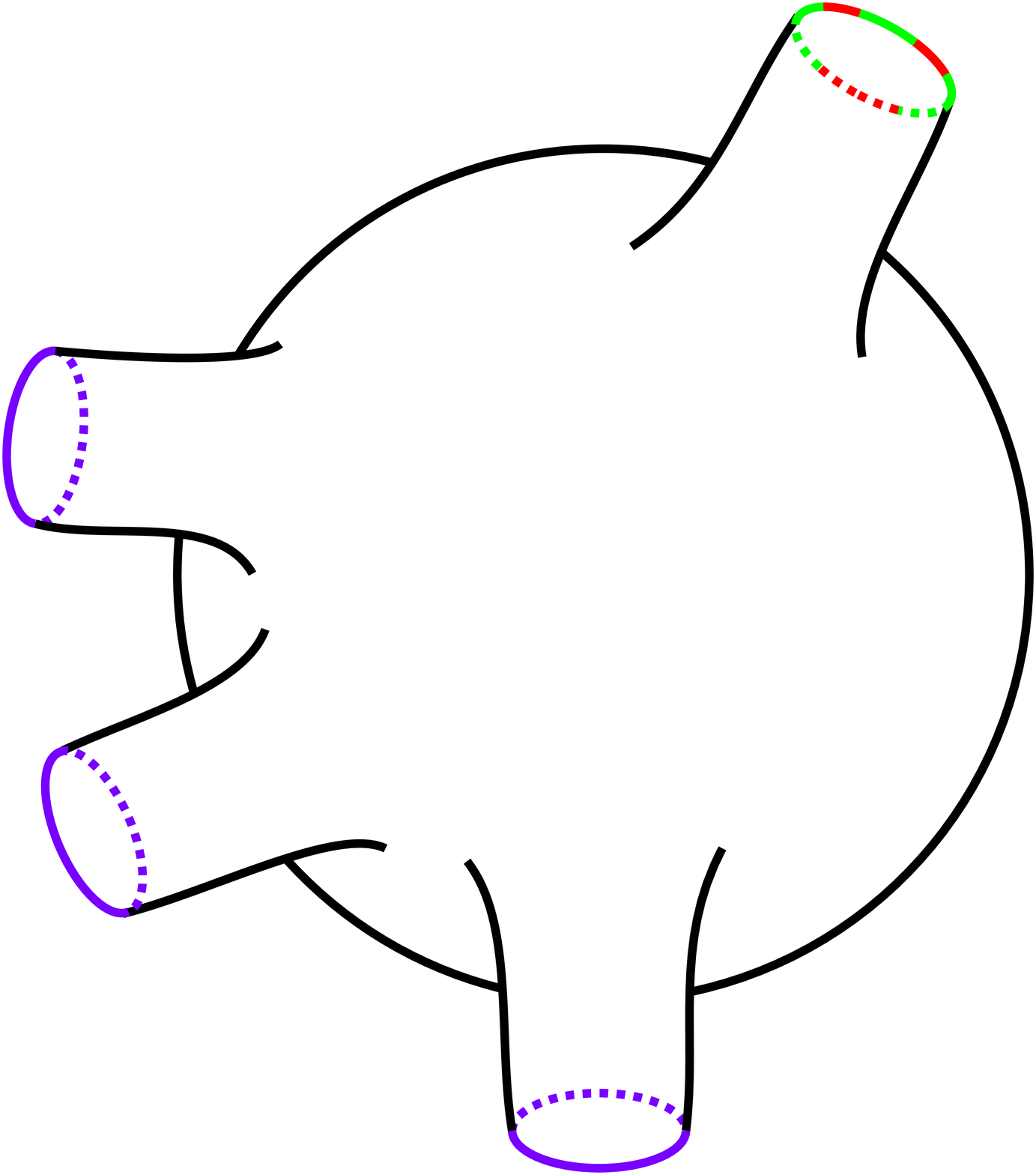}
\caption{Example of type II) world sheet quotient surface.}
\label{typeII}
\end{figure}
In figure \ref{typeII}, there are three closed boundary components, shown in purple. In addition, there are three open boundary components, sitting on the same connected component of the boundary. These are colored red. The green part of the boundary shows physical boundary components. We label closed boundary components with $\Hcalcl$ and open boundary components with $L(\Hcalop)$. The string-net space on the quotient surface is then given by $\hat{H}^s(S,\ov{\Hcalcl},\ov{L(\Hcalop)})=\hom_{\Zsf(\Csf)}(\mathbf{1},\ov{\Hcalcl}\otimes \ov{L(\Hcalop)})$. In this case a subspace of the string-net space has to be chosen in order to get a well defined functor for composition of morphisms. This is due to the fact that $L$, though being a Frobenius functor, is not a tensor functor, hence $\hom_{\Zsf(\Csf)}(\mathbf{1},\ov{L(\Hcalop)})\nsimeq \hom_{\Csf}(\mathbf{1},\ov{\Hcalop})$. But $L$ is lax and colax tensor functor thus there are morphisms $\phi^L_\mathbf{1}:\mathbf{1}_{\Zsf(\Csf)}\rightarrow L(\mathbf{1}_\Csf)$, $\phi^L:L(A)\otimes L(B)\rightarrow L(A\otimes B)$ and $\psi^l_\mathbf{1}:L(\mathbf{1})\rightarrow \mathbf{1}_{\Zsf(\Csf)}$, $\psi^L:L(A\otimes B)\rightarrow L(A)\otimes L(B)$. They are given explicitly by \cite{Kong_2008}

\begin{center}
\begin{tikzpicture}
\node at (-4.5,0) {\includegraphics[scale=0.15]{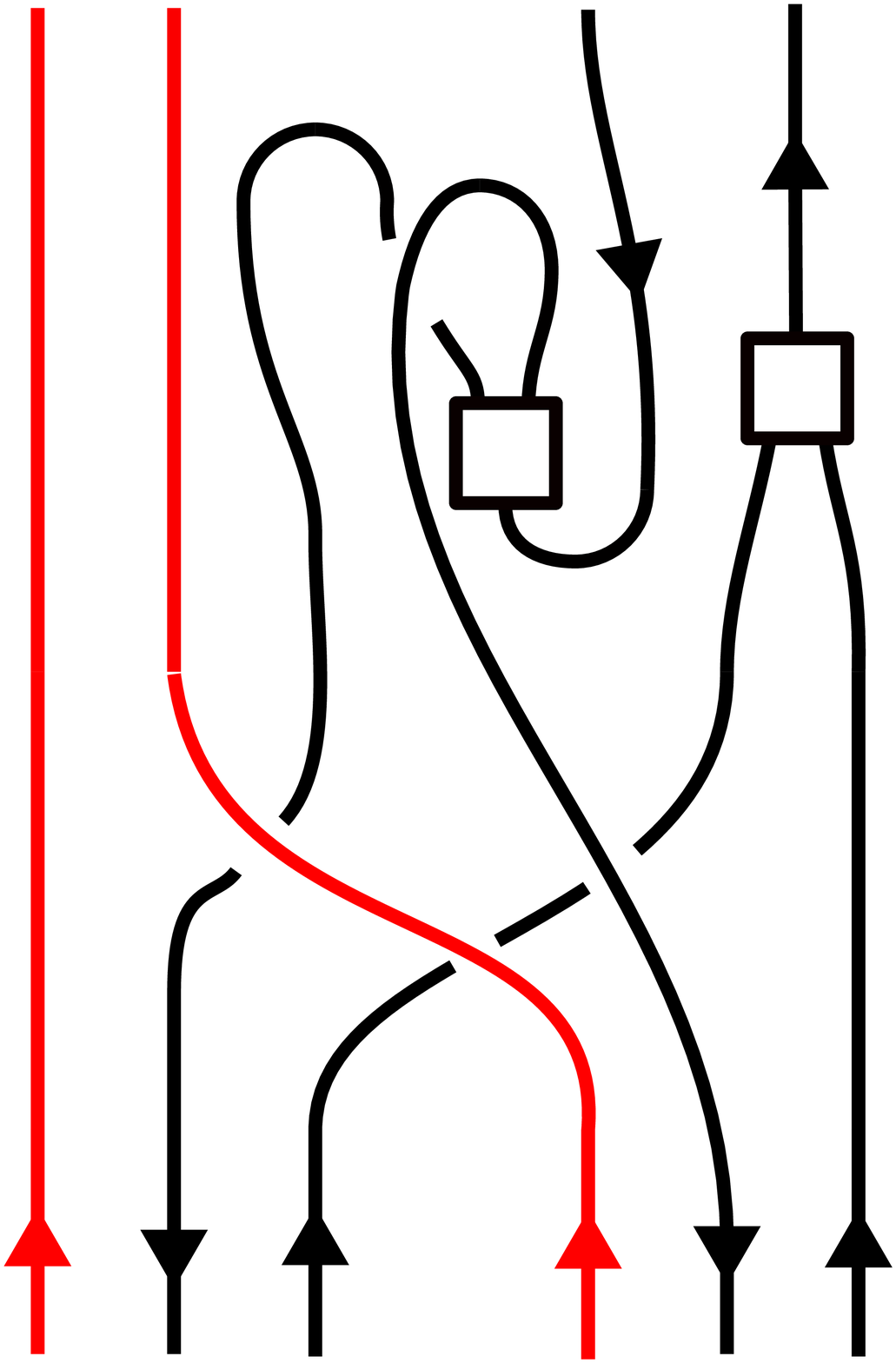}};
\node at (-6.7,0) {$\underset{i,j,k \in I}{\bigoplus}\, \underset{\alpha}{\sum}$};
\node at (-7.5,0.1) {$=$};
\node at (-7.9,0.1) {$\phi^L$};
\node at (-6.1,-2.3) {\scriptsize $\color{red}A$};
\node at (-5.2,-2.3) {\scriptsize $i$};
\node at (-4.8,-2.3) {\scriptsize $i$};
\node at (-4.2,-2.3) {\scriptsize $\color{red}B$};
\node at (-3.4,-2.3) {\scriptsize $j$};
\node at (-3.0,-2.3) {\scriptsize $j$};
\node at (-4.3,0.75) {\scriptsize $\alpha$};
\node at (-3.35,1) {\scriptsize $\alpha$};
\node at (-3.8,2.1) {\scriptsize $k$};
\node at (-3.1,2.1) {\scriptsize $k$};
\node at (3,0) {\includegraphics[scale=0.15]{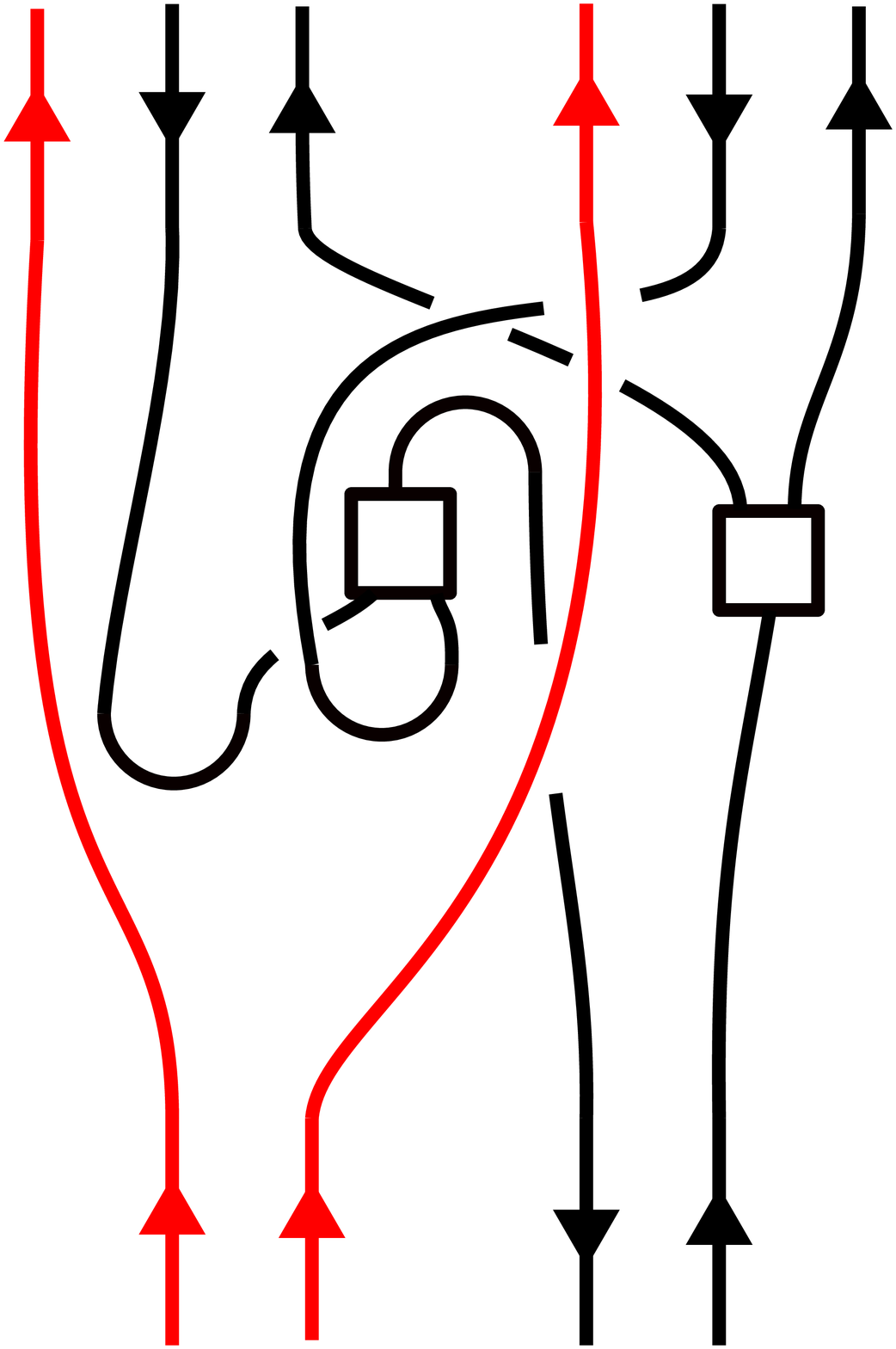}};
\node at (-1.4,0.1) {$\psi^L$};
\node at (-1.0,0.1) {$=$};
\node at (0,0) {$\underset{i,j,k\in \Isf}{\bigoplus}\,\underset{\alpha}{\sum}\,\frac{d_id_j}{d_k\Dsf^2}$};
\node at (1.9,-2.3) {\scriptsize $\color{red} A$};
\node at (2.4,-2.3) {\scriptsize $\color{red} B$};
\node at (3.6,-2.3) {\scriptsize $k$};
\node at (4.1,-2.3) {\scriptsize $k$};
\node at (2.3,2.1) {\scriptsize $i$};
\node at (2.7,2.1) {\scriptsize $i$};
\node at (4.1,2.1) {\scriptsize $j$};
\node at (4.6,2.1) {\scriptsize $j$};
\node at (2.85,0.4) {\scriptsize $\alpha$};
\node at (4.1,0.35) {\scriptsize $\alpha$}; 
\end{tikzpicture}
\end{center}

and
\begin{center}
\begin{tikzpicture}
\node at (0,0) {\includegraphics[scale=0.15]{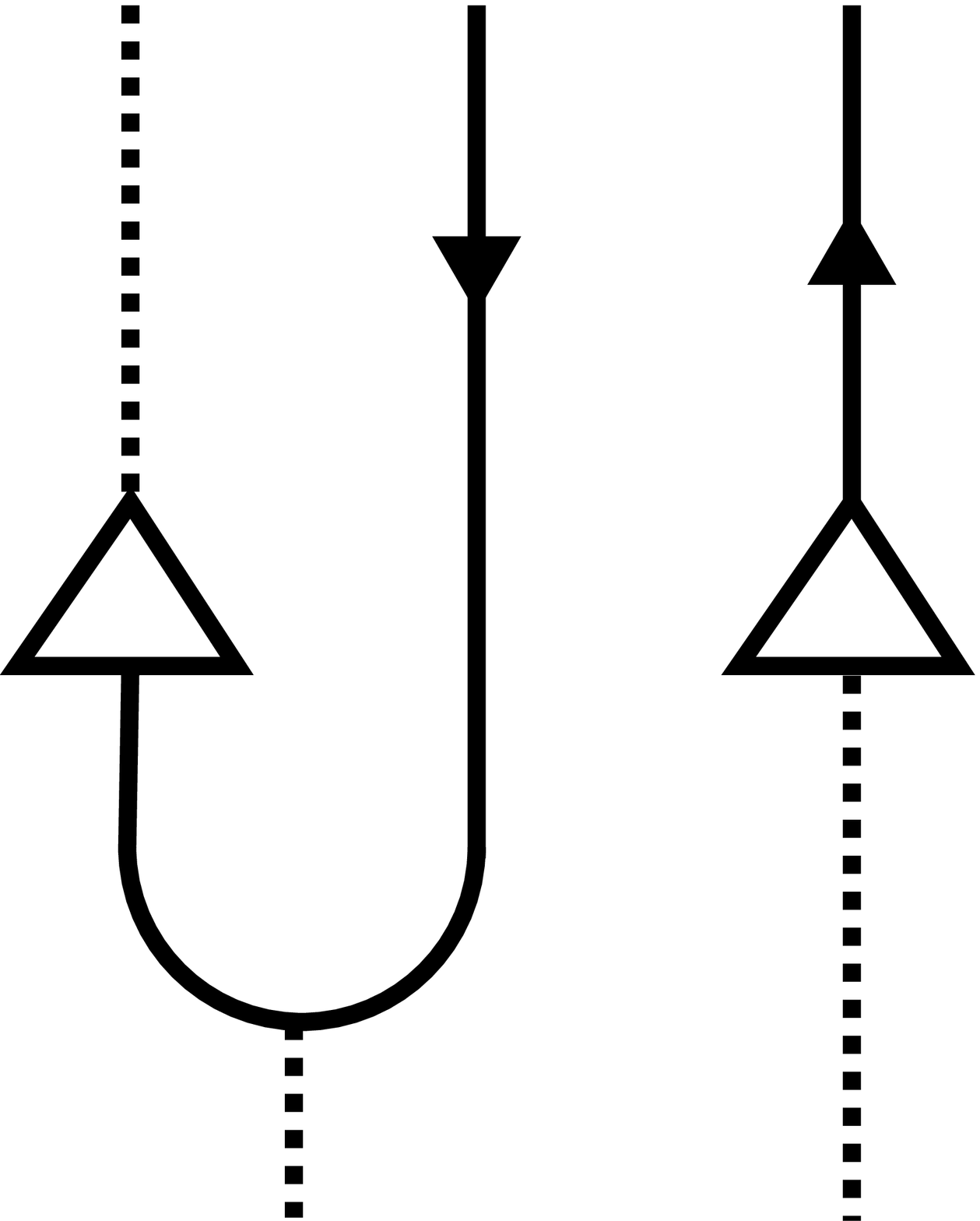}};
\node at (-3.5,0) {$\phi^L_\mathbf{1}$};
\node at (-3,0) {$=$};
\node at (-2,0) {$\begin{aligned} \bigoplus_{i\in \Isf} \sum_{\alpha} \end{aligned}$};
\node at (-0.9,0) {\scriptsize $\alpha$};
\node at (0.9,0) {\scriptsize $\alpha$};
\node at (0.2,1.2) {\scriptsize $i$};
\node at (1.1,1.2) {\scriptsize $i$};
\end{tikzpicture}
\end{center}

where we depict the monoidal unit by a dashed line. Since we don't need it in the following we just give the formula for $\psi^L_\mathbf{1}=\Dsf^2\,\id_{\mathbf{1}_{\Zsf(\Csf)}}$ and not its explicit graphical representation. We define linear maps
\eq{
Z: \hom_{\Zsf(\Csf)}\left(\mathbf{1},\ov{\Hcalop}\right)&\rightarrow \hom_{\Zsf(\Csf)}\left(\mathbf{1},\ov{L(\Hcalop)}\right)\\
f&\mapsto Z(f)\equiv (\psi^L\otimes \id\otimes \dots \otimes \id)\circ\dots \circ(\psi^L\otimes \id)\circ\psi^L\circ L(f)\circ \phi_\mathbf{1}^L
}
\eq{
Y:\hom_{\Zsf(\Csf)}\left(\mathbf{1},\ov{L(\Hcalop)}\right)&\rightarrow \hom_{\Csf}\left(\mathbf{1},\ov{\Hcalop}\right)\simeq \hom_{\Zsf(\Csf)}\left(\mathbf{1},\ov{\Hcalop}\right)\\
g&\mapsto Y(g)\equiv \dsf\circ F\left[\phi^L\circ (\id\otimes \phi^L)\circ \dots \circ (\id\otimes \dots\otimes \id\otimes \phi^L)\circ g\right]
}
where $\dsf$ is the map

\begin{center}
\begin{tikzpicture}
\node at (0,0) {\includegraphics[scale=0.15]{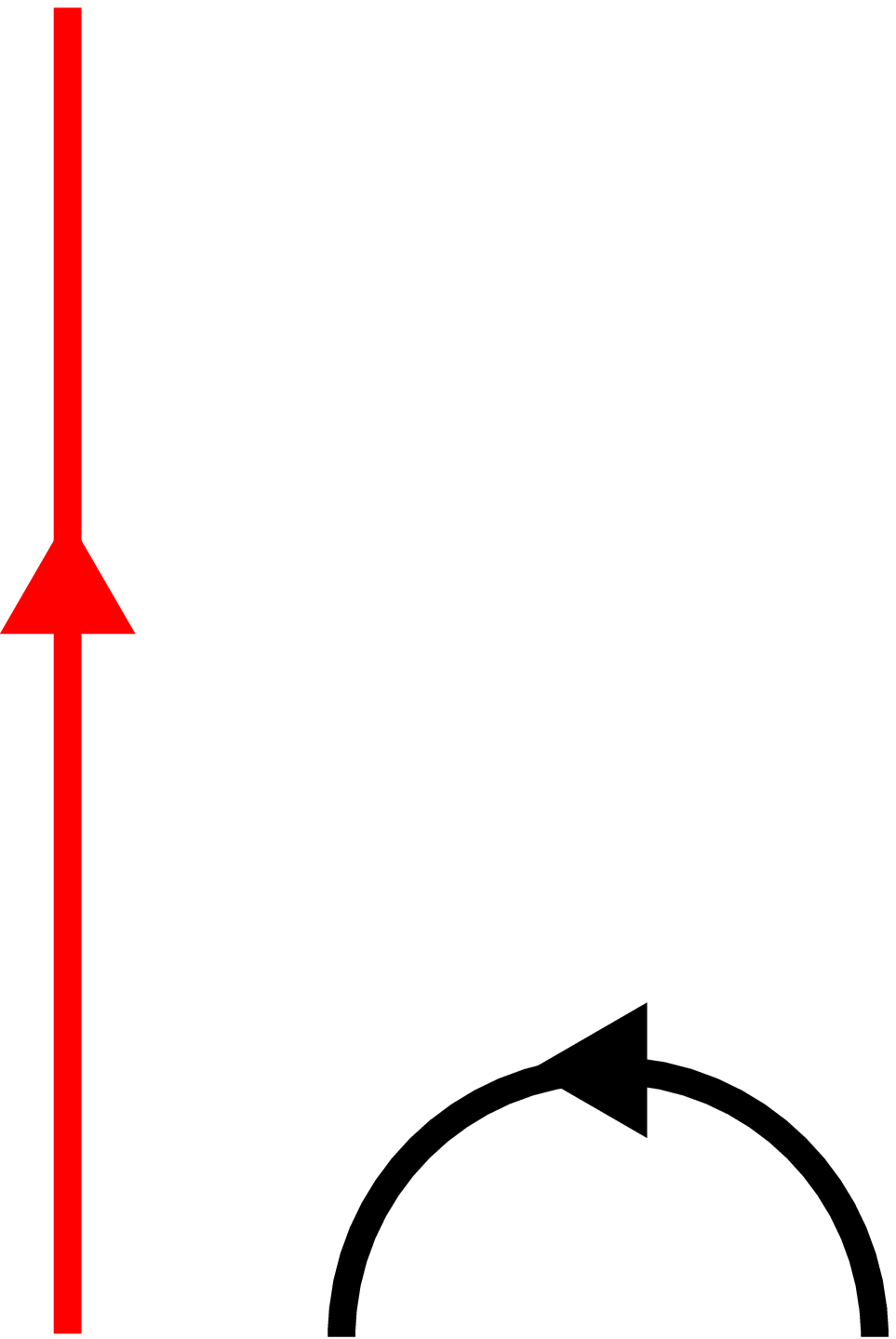}};
\node at (-2.5,0) {$\dsf$};
\node at (-2.2,0) {$=$};
\node at (-1.5,0) {$\begin{aligned} \bigoplus_{i\in \Isf} \end{aligned}$};
\node at (-0.95,-1) {\scriptsize $\color{red} \ov{\Hcalop}$};
\node at (0,-1) {\scriptsize $i$};
\end{tikzpicture}.
\end{center}

The map $Y$ is a left inverse to $Z$.

\begin{lem}\label{restricitionstatespacelemma} $Y\circ Z=\id_{\hom_{\Zsf(\Csf)}\left(\mathbf{1},\ov{\Hcalop}\right)}$.
\end{lem}
\begin{proof}
First note that 
\begin{center}
\begin{tikzpicture}
\node at (0,0) {\includegraphics[scale=0.15]{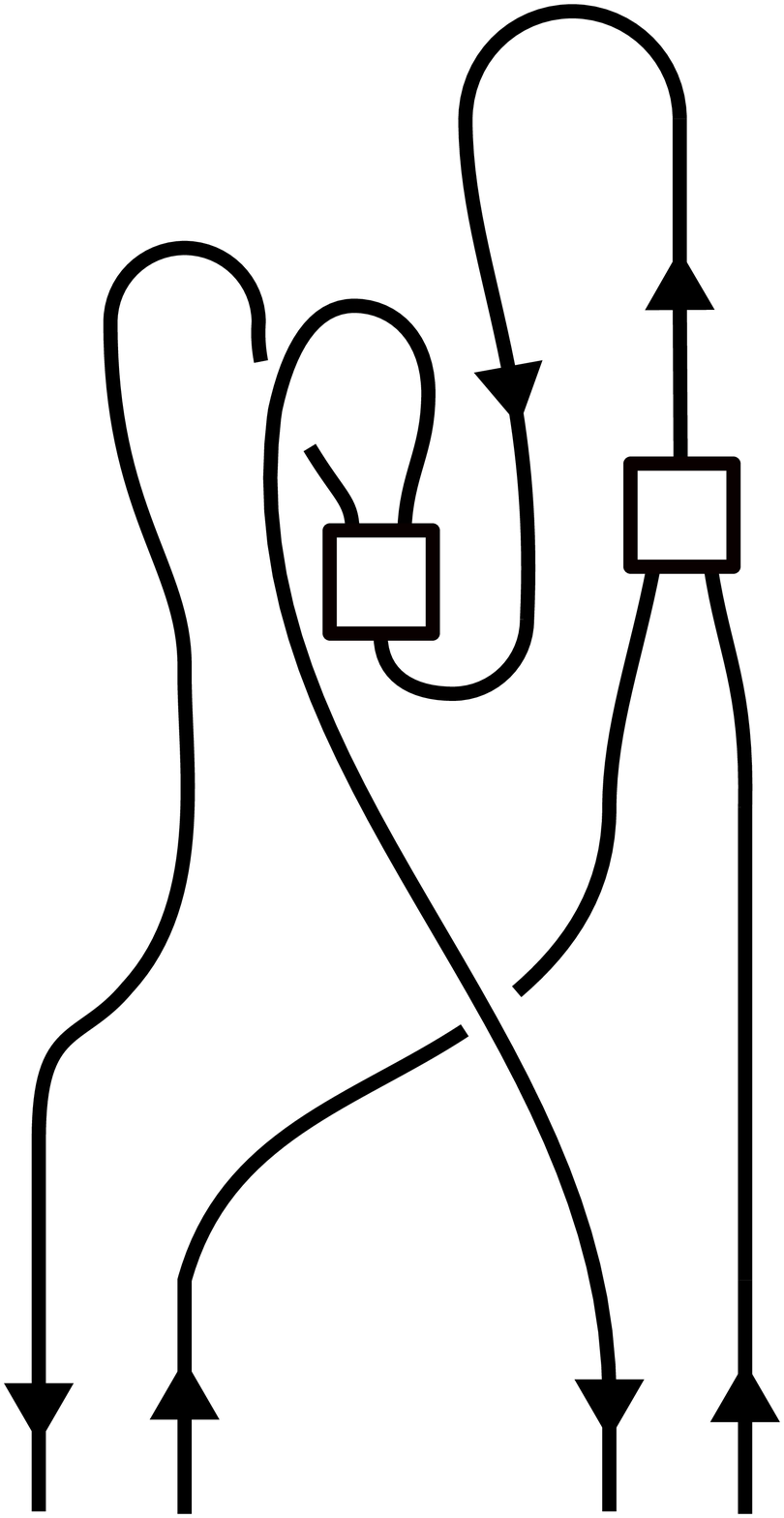}};
\node at (4,0) {\includegraphics[scale=0.15]{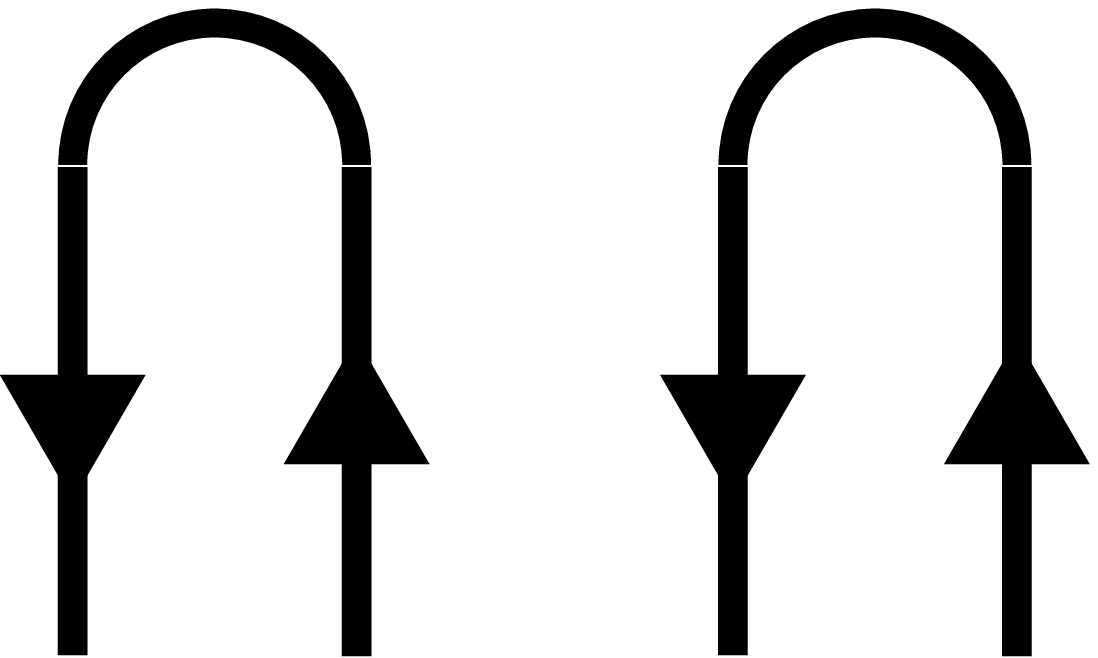}};
\node at (-2,0) {$\begin{aligned} \bigoplus_{i,j,k\, \in\,  \Isf} \, \sum_{\alpha} \end{aligned}$};
\node at (-0.95,-2.3) {\scriptsize $i$};
\node at (-0.5,-2.3) {\scriptsize $i$};
\node at (0.85,-2.3) {\scriptsize $j$};
\node at (1.25,-2.3) {\scriptsize $j$};
\node at (1.1,2) {\scriptsize $k$};
\node at (-0.05,0.55) {\scriptsize $\alpha$};
\node at (0.95,0.75) {\scriptsize $\alpha$};
\node at (1.75,0) {$=$};
\node at (2.5,0) {$\begin{aligned}\bigoplus_{i,j\, \in \,\Isf}\end{aligned}$};
\node at (3.4,-0.5) {\scriptsize $i$};
\node at (4.4,-0.5) {\scriptsize $j$};
\end{tikzpicture}
\end{center}
thus by induction it holds 
\begin{center}
\begin{tikzpicture}
\node at (0,0) {$Y(g)$};
\node at (0.5,0) {$=$};
\node at (1.2,0) {$\begin{aligned} \bigoplus_{i_1,\dots, i_n\, \in \, \Isf}\end{aligned}$};
\node at (4,0) {\includegraphics[scale=0.15]{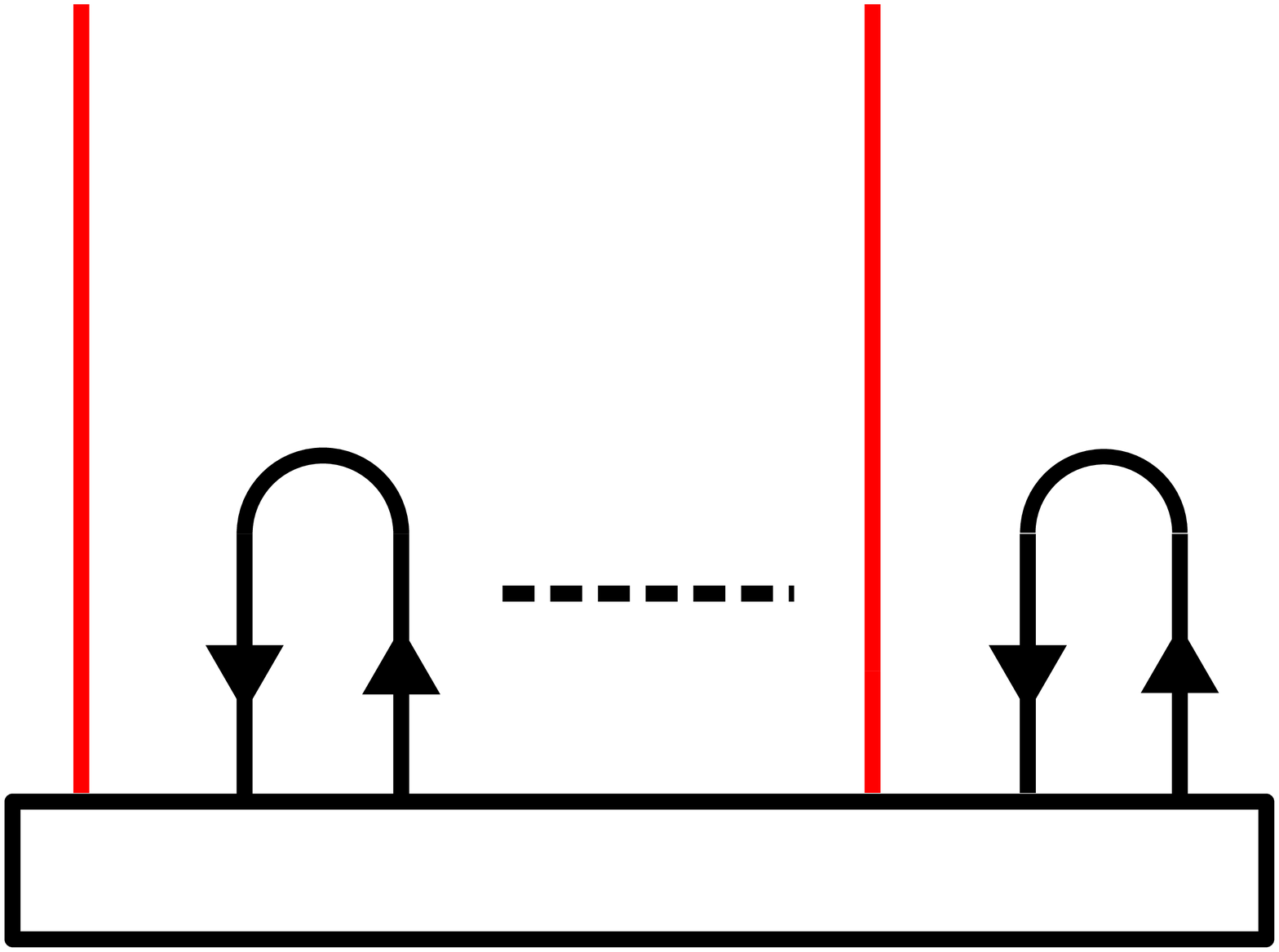}};
\node at (4,-1.1) {$g$};
\node at (3.6,-0.7) {\scriptsize $i_1$};
\node at (5.9,-0.7) {\scriptsize $i_n$};
\end{tikzpicture}
\end{center}
where red strands are labeled $\Hcalop$ or $\Hcalop^\ast$. Next it holds 
\begin{center}
\begin{tikzpicture}
\node at (0,0) {\includegraphics[scale=0.15]{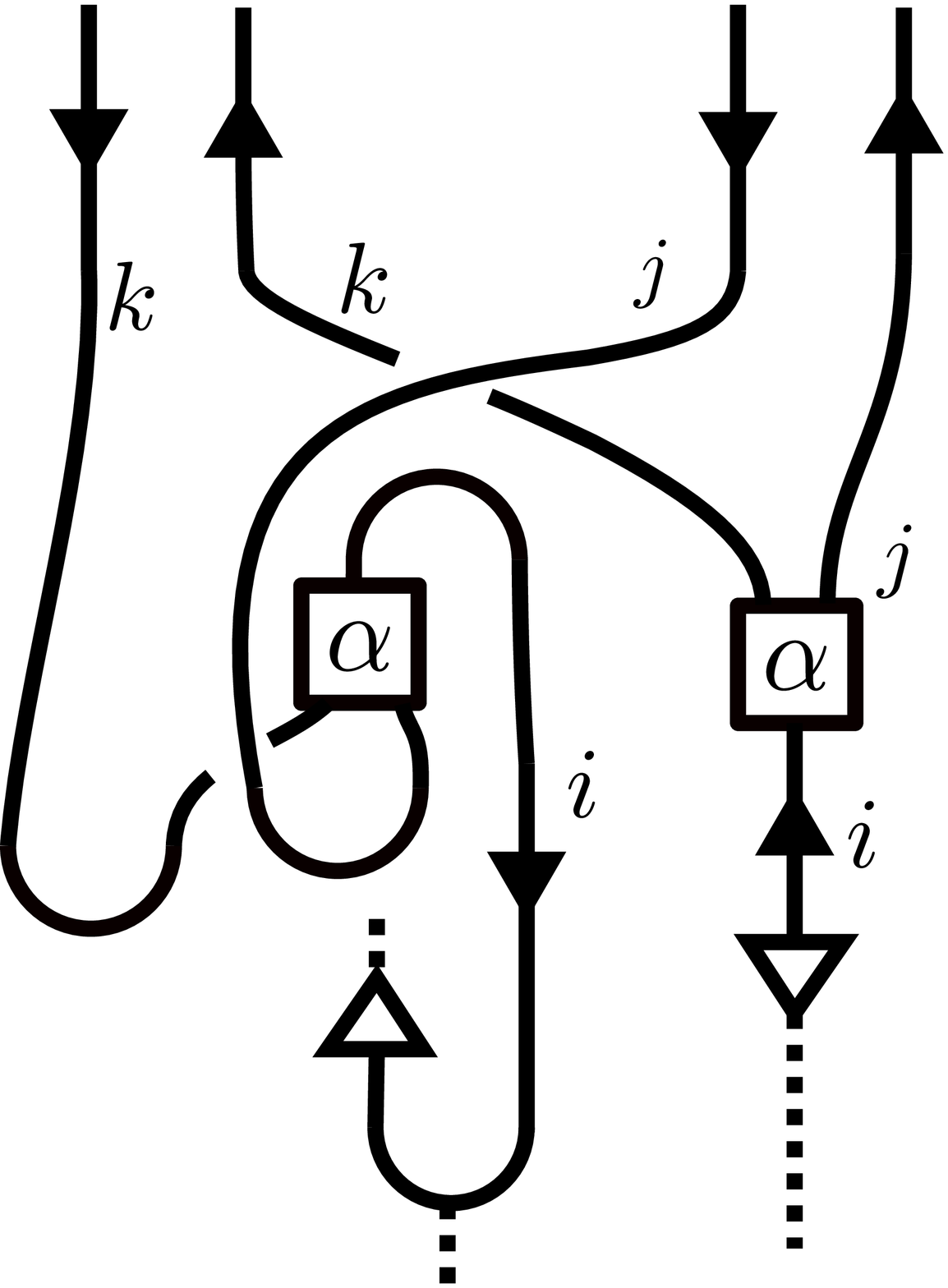}};
\node at (5,0) {\includegraphics[scale=0.15]{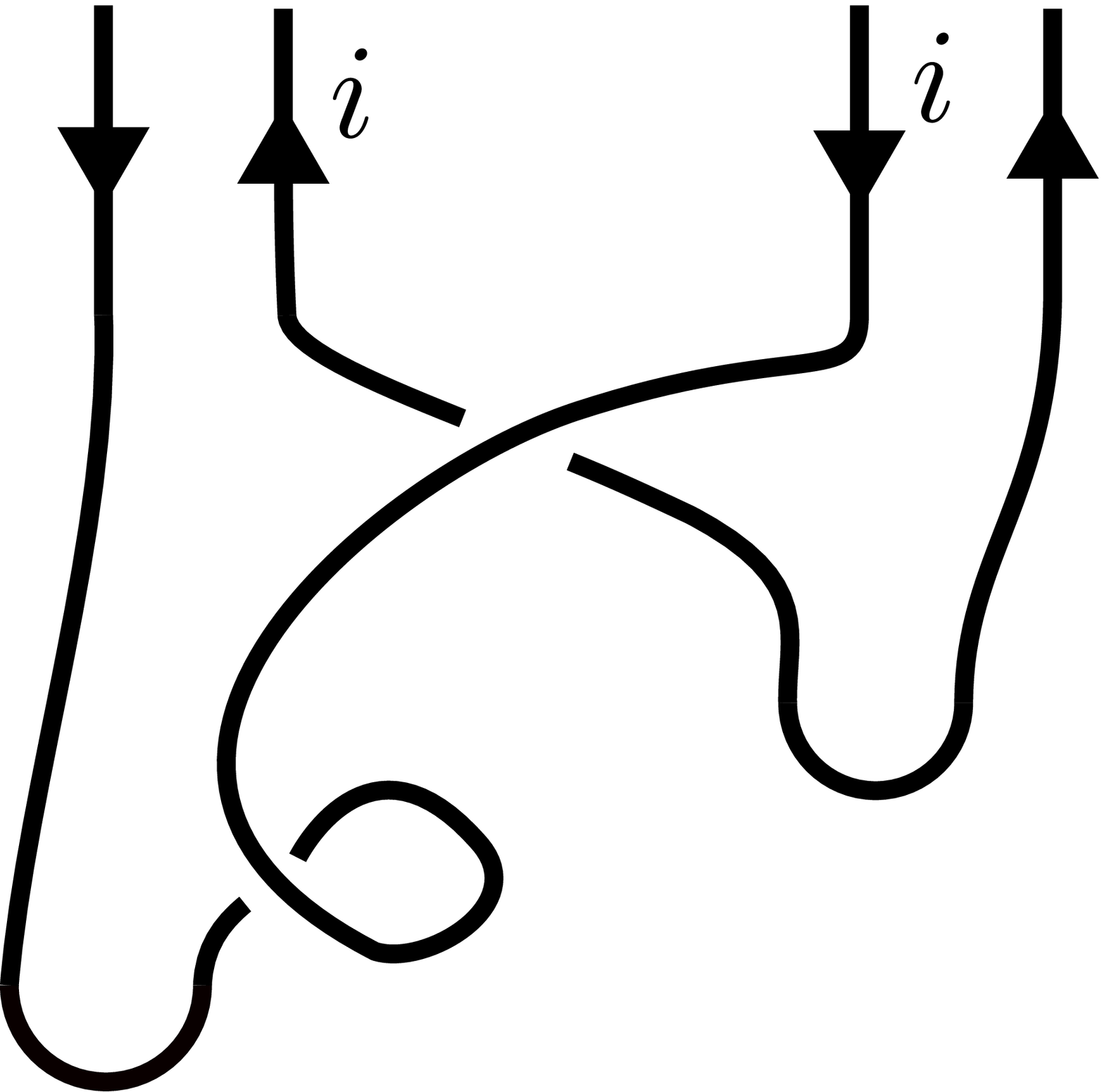}};
\node at (-3,0) {$\begin{aligned} \bigoplus_{i,j,k\, \in \, \Isf}\, \sum_\alpha\, \frac{d_kd_j}{d_i\Dsf^2}\end{aligned}$};
\node at (2,0) {$=$};
\node at (3,0) {$\begin{aligned}\bigoplus_{i\in \,\Isf}\, \frac{d_i}{\Dsf^2}\end{aligned}$};
\end{tikzpicture}
\end{center}
and composing with $(\id\otimes \dsf)\circ(\psi^L\otimes \id)$ gives
\begin{center}
\begin{tikzpicture}
\node at (0,0) {\includegraphics[scale=0.15]{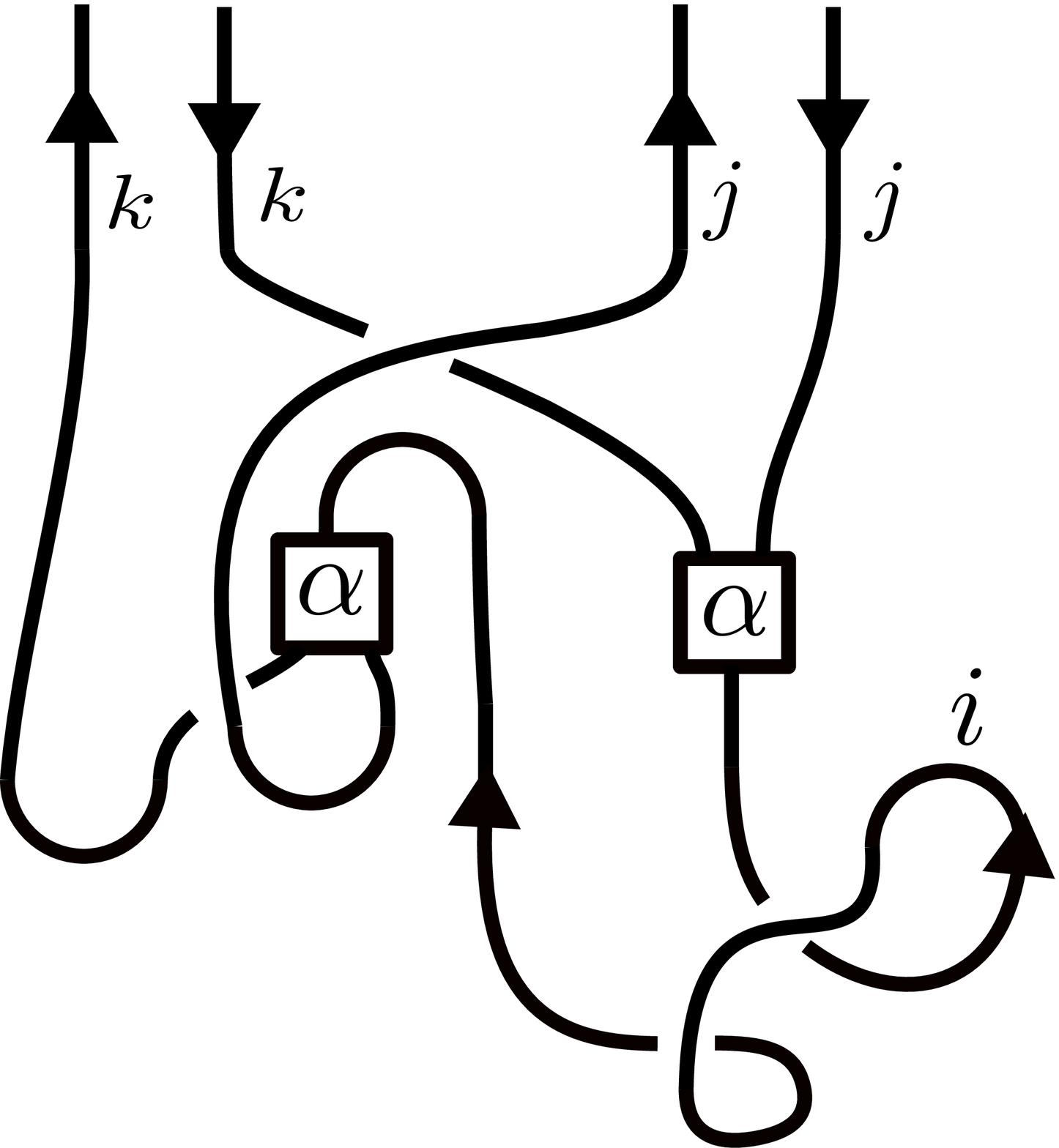}};
\node at (5,0) {\includegraphics[scale=0.15]{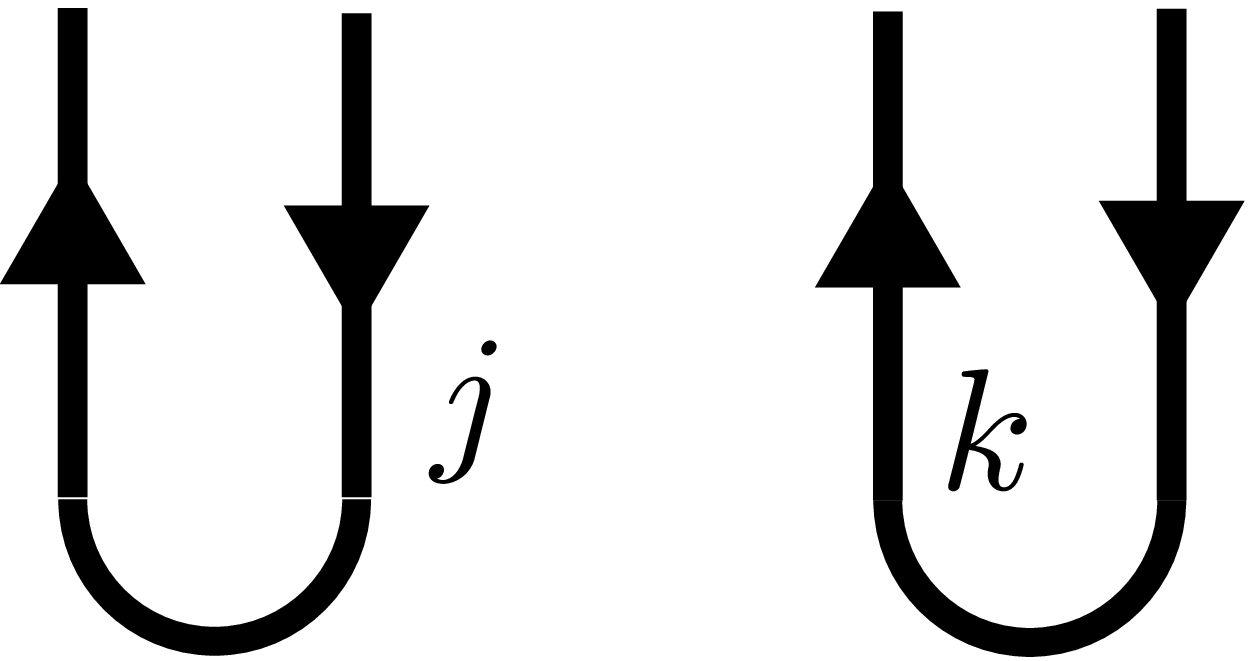}};
\node at (-3,0) {$\begin{aligned} \bigoplus_{i,j,k}\, \sum_{\alpha} \, \frac{d_jd_kd_i}{d_i\Dsf^2}\end{aligned}$};
\node at (1.5,0) {$=$};
\node at (3,0) {$\begin{aligned} \bigoplus_{j,k\in \Isf} \, \frac{d_jd_k}{\Dsf^2}\end{aligned}$};
\end{tikzpicture}.
\end{center}
We dropped and will drop red strands in the next picture as they are irrelevant to the argument and clutter pictures. Applying $(\id\otimes \dsf)\circ(\psi^L\otimes \id)$ again gives 
\begin{center}
\begin{tikzpicture}
\node at (0,0) {\includegraphics[scale=0.15]{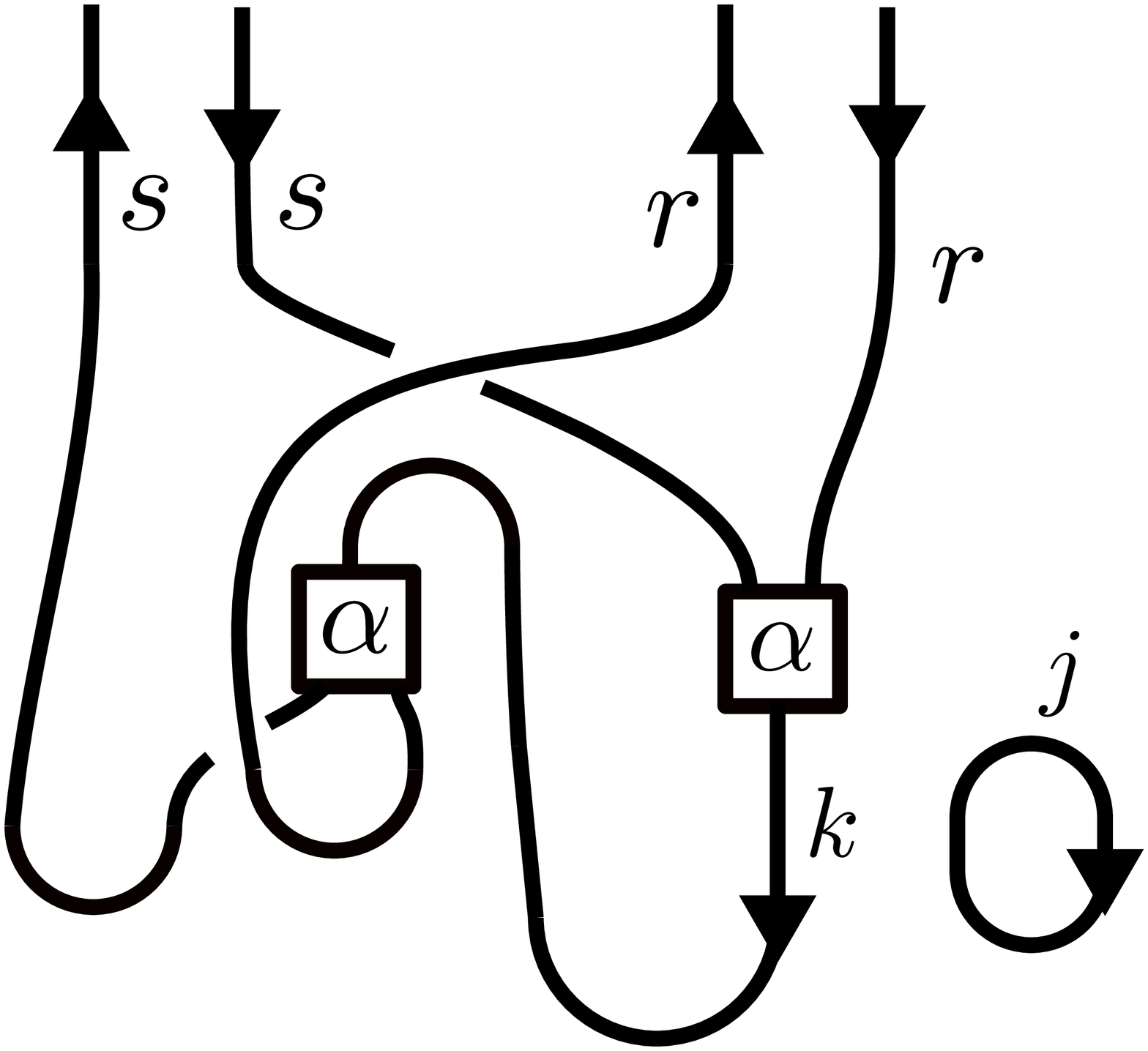}};
\node at (4.5,0) {\includegraphics[scale=0.15]{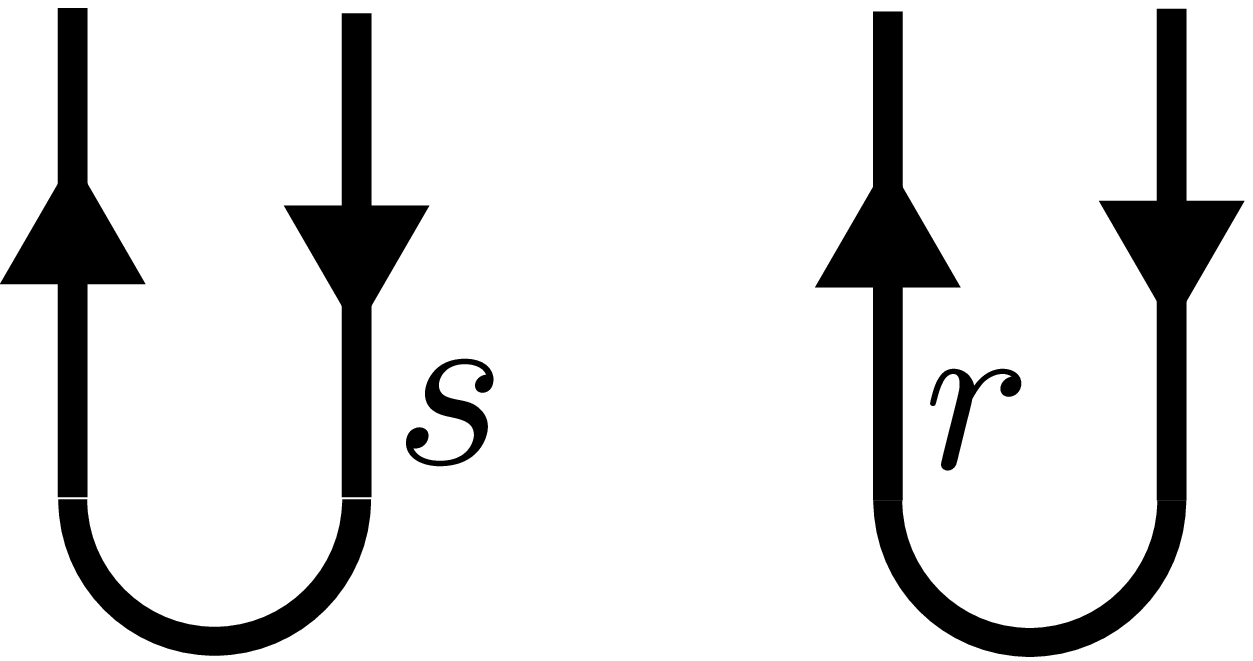}};
\node at (-3,0) {$\begin{aligned} \bigoplus_{s,r,j,k\, \in \, \Isf} \, \sum_{\alpha}\, \frac{d_sd_rd_jd_k}{d_j\Dsf^4}\end{aligned}$};
\node at (1.5,0) {$=$};
\node at (2.5,0) {$\begin{aligned}\bigoplus_{r,s\, \in \, \Isf}\,\frac{d_sd_r}{\Dsf^2}\end{aligned}$};
\end{tikzpicture}.
\end{center} 
By induction we get 

\begin{tabular}{m{1.5cm} m{0.5cm} m{1cm}}
{$Y\circ Z(f)$} & $=$ & 
\begin{tikzpicture}
\node at (0,0) {\includegraphics[scale=0.1]{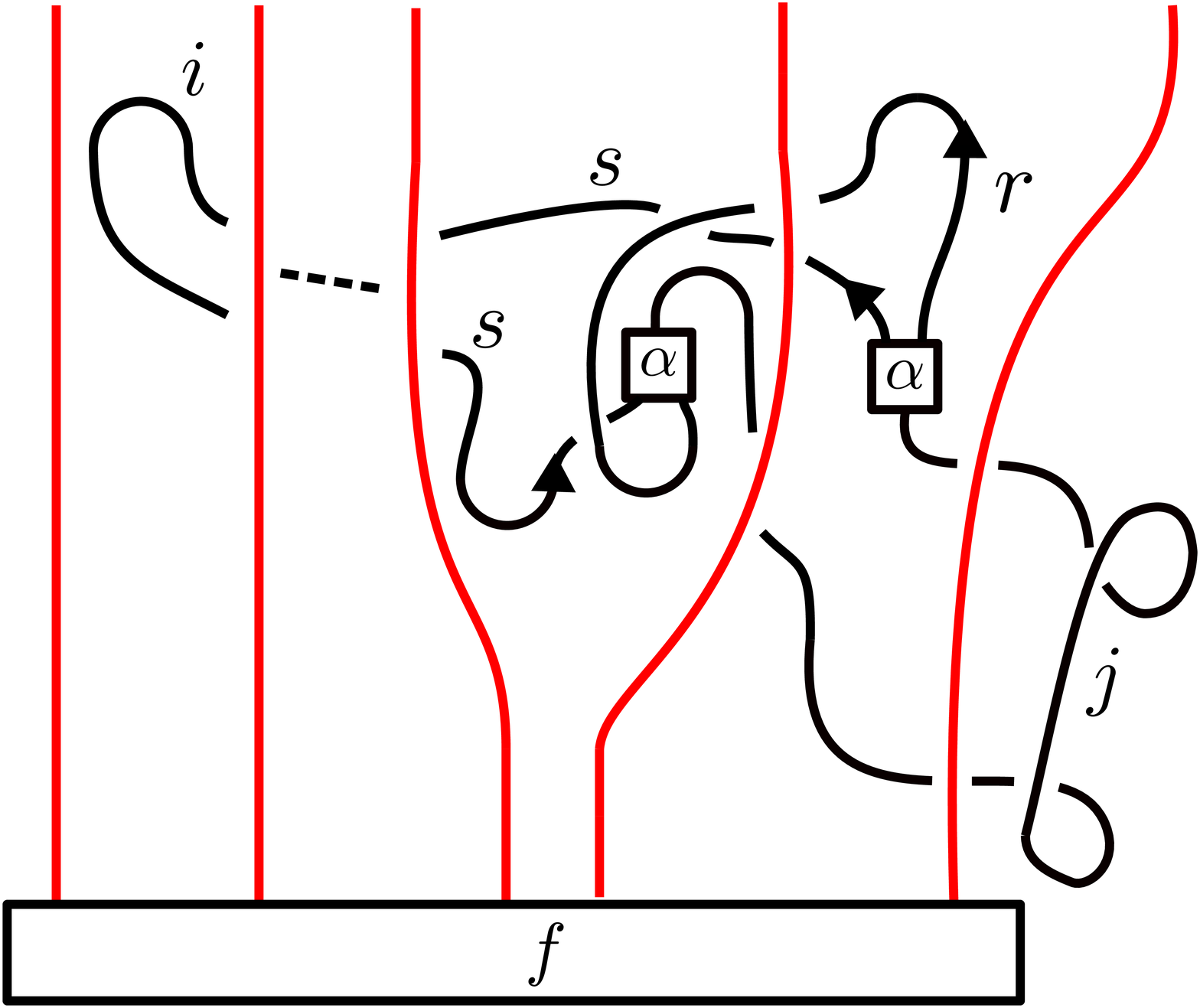}};
\node at (-4,0) {$\begin{aligned}\bigoplus_{i,\dots, s,r,j\, \in \, \Isf}\,\sum_{\alpha_1,\dots}\, C_1\frac{d_sd_r}{\Dsf^2}\end{aligned}$};
\end{tikzpicture}\\
 &=  &
\begin{tikzpicture}
\node at (0,0) {\includegraphics[scale=0.1]{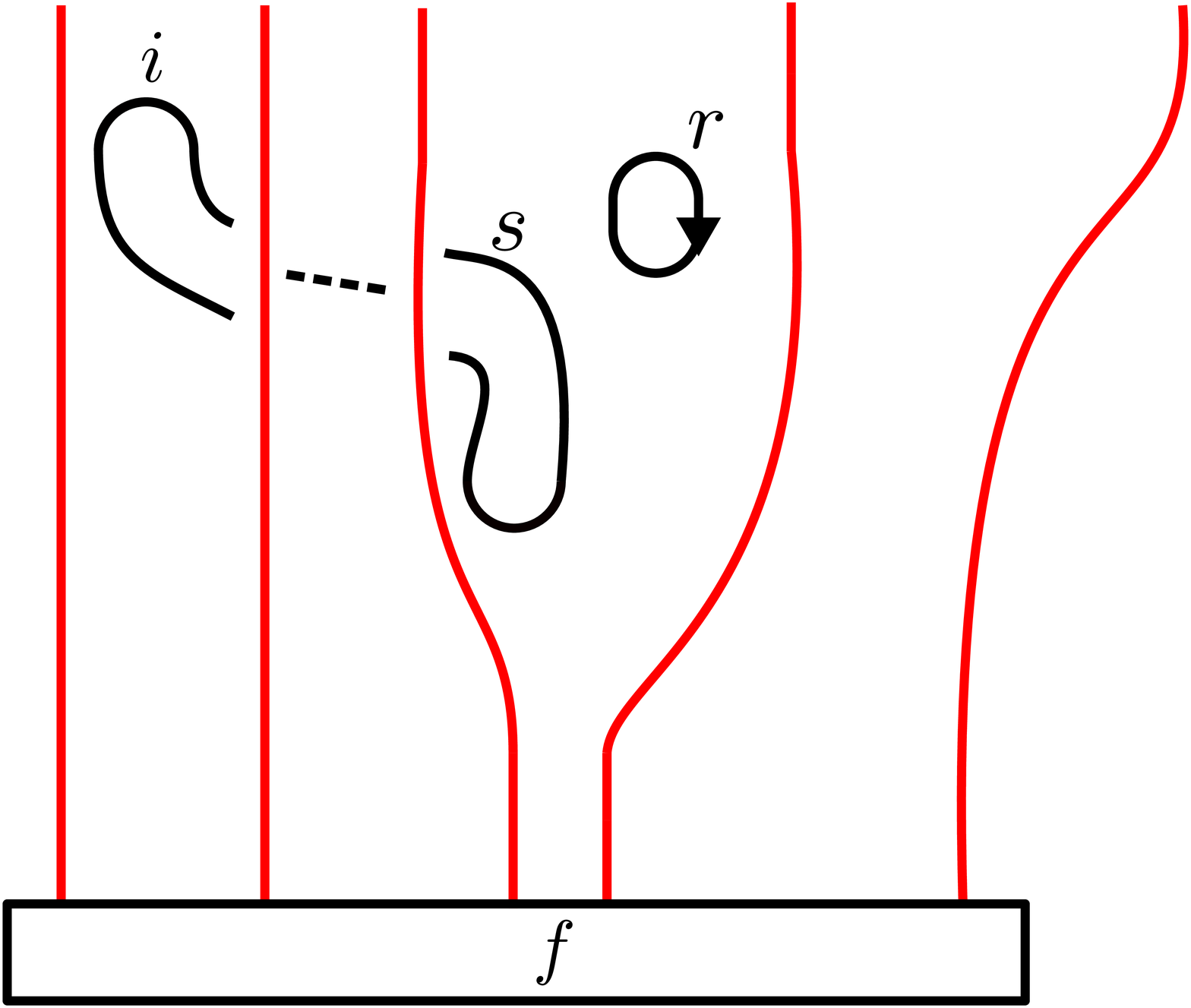}};
\node at (-3.5,0) {$\begin{aligned}\bigoplus_{i,\dots, s,r\, \in \, \Isf}\, \sum\, C_2 \frac{d_r}{\Dsf^2}\end{aligned}$};
\end{tikzpicture}\\
& = & 
\begin{tikzpicture}
\node at (0,0) {\includegraphics[scale=0.1]{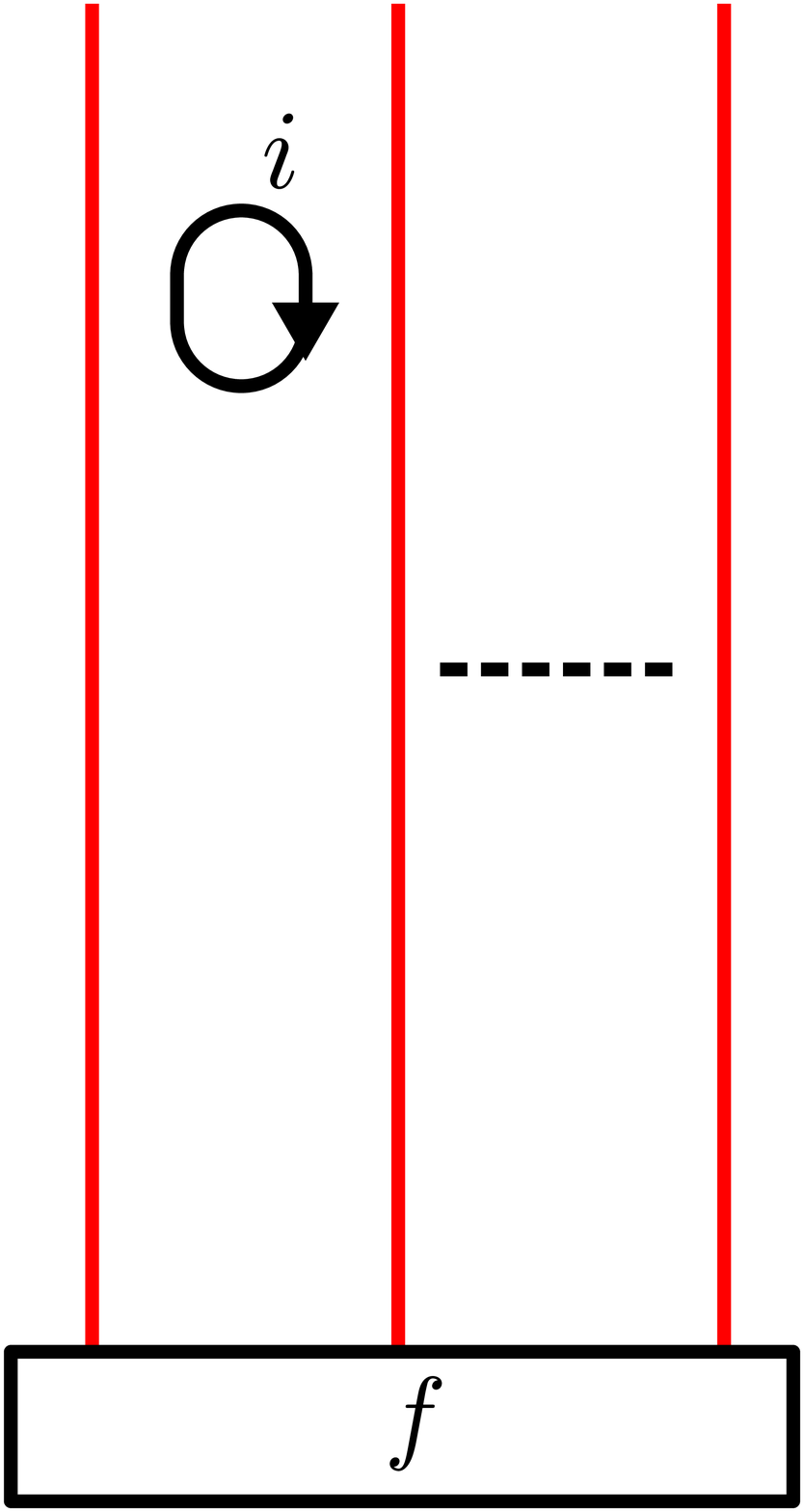}};
\node at (4,0) {\includegraphics[scale=0.1]{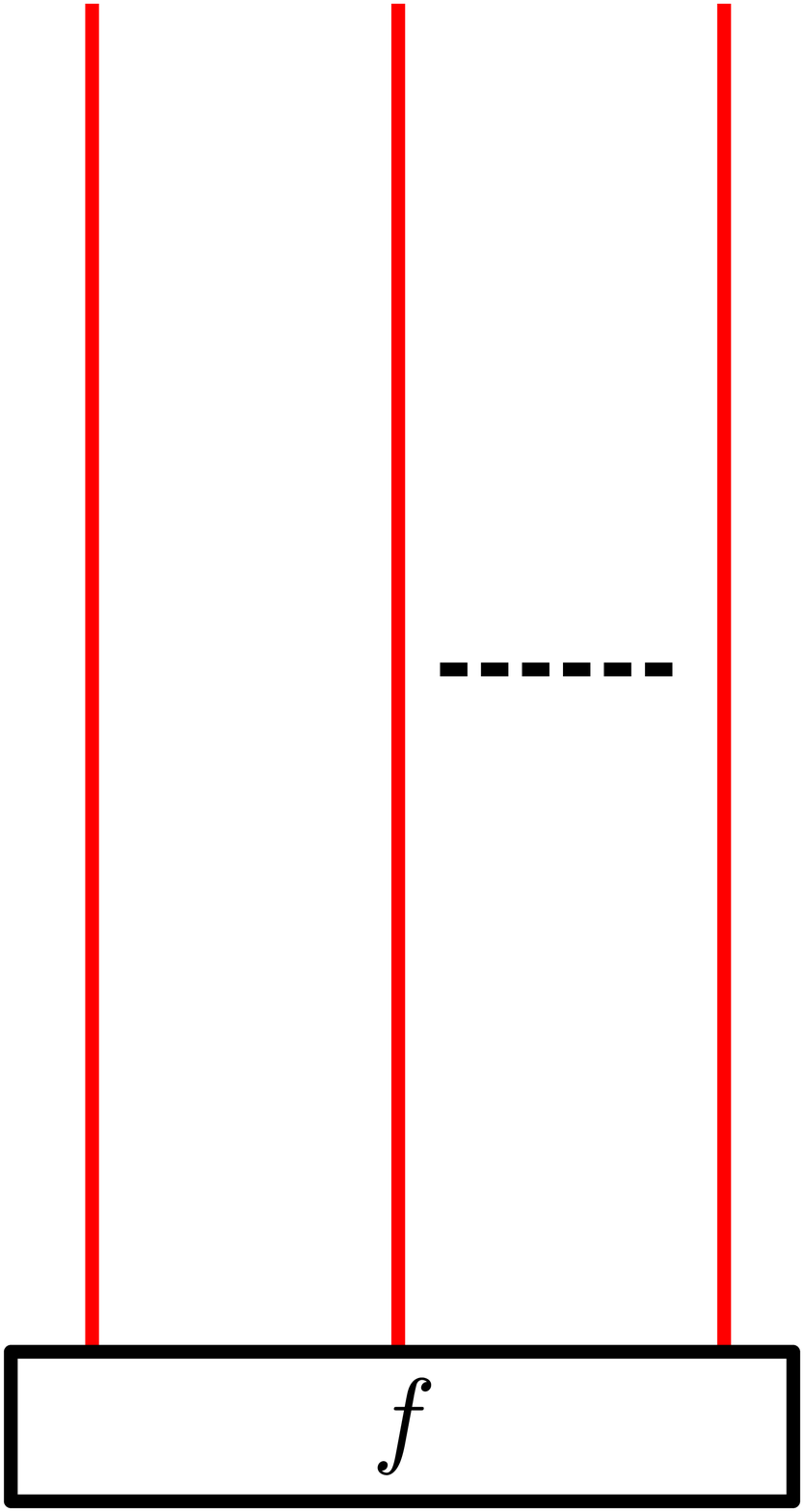}};
\node at (-2,0) {$\begin{aligned} \bigoplus_{i\, \in \, \Isf}\,\frac{d_i}{\Dsf^2}\end{aligned}$};
\node at (2,0) {$=$};
\end{tikzpicture}.
\end{tabular}

The dots in the first and second row indicate further applications of $\psi^L$, hence we have written a summation in the second row indicating summing over basis elements $\lbr b_\alpha\rbr $ similar to the summation in the first row. In addition all prefactors from $\psi^L$ are collected in the coefficients $C_1$, $C_2$.
\end{proof}

The vector space $\Blcal(\hat{S})$ can now be defined as follows. Denote $\ov{\Hcalcl^\ast}$ for the object dual to $\ov{\Hcalcl}$ (the closed labels of $\hat{S}$). Let $f\in \hom_{\Zsf(\Csf)}\left(\mathbf{1},\ov{\Hcalcl^\ast}\right)$, using evaluation morphisms, there is a map 
\eq{
\tilde{\circ}:\hom_{\Zsf(\Csf)}\left(\mathbf{1},\ov{\Hcalcl^\ast}\right)\otimes \hom_{\Zsf(\Csf)}\left(\mathbf{1},\ov{\Hcalcl}\otimes \ov{L(\Hcalop)}\right)\rightarrow \hom_{\Zsf{(\Csf)}}\left(\mathbf{1},\ov{L(\Hcalop)}\right)\quad .
}
We define 
\eq{
\Blcal(\hat{S})=\lbr g\in  \hom_{\Zsf(\Csf)}\left(\mathbf{1},\ov{\Hcalcl}\otimes \ov{L(\Hcalop)}\right)\, |\, \forall f\in \hom_{\Zsf(\Csf)}\left(\mathbf{1},\ov{\Hcalcl^\ast}\right), \, \exists h\in \hom_{\Zsf(\Csf)}\left(\mathbf{1},\ov{\Hcalop}\right)\right.\\ 
\left.\phantom{================} \text{ s.th. } f\tilde{\circ}g=Z(h) \rbr \quad .
} 
\item To any other world sheet $\hat{S}$ with $n_{op}=|B_{op}^i|$, $n_{cl}=|B_{cl}^i|$ incoming open/ closed boundaries and $m_{op}=|B_{op}^o|$, $m_{cl}=|B_{cl}^o|$ outgoing open / closed boundaries we associate the vector space 
\eq{
\Blcal(\hat{S})=\hat{H}^s\left(S,\ov{L(\Hcalop)},\ov{\Hcalcl}\right)
}
Factors of incoming and outgoing insertions in the tensor product are inserted in the order given by the ordering function on the world sheet. 
\end{enumerate}
Hence, modulo minor technicalities we assign to a world sheet the string-net space on its quotient surface decorated by the ingredients of the Cardy algebra. Tree level open world sheets (case I)) are singled out to properly disentangle the open and closed theory. We have to treat the case II) separatly, as the adjoint functor $L$ is not a tensor functor, but only a Frobenius functor. Hence, if one simply concatenates string-nets, one doesn't land in the right vector space, when glueing disks to the closed boundary components of a type II) world sheet. This is cured by choosing the restricted vector space for world sheets of type II). This defines $\Blcal$ on objects. To define it on morphisms, we first note that any homeomorphism of world sheets gives a homeomorphism of the quotient surfaces. This induces a linear map in string-net spaces. Thus we really only have to define $\Blcal$ for the sewing part of a morphism $(\Scal,F)$ in $\WSsf$. 
To define the gluing we need a small lemma.
\begin{lem} $L(A^\ast)\simeq L(A)^\ast$.
\end{lem}
\begin{proof}
Evaluation and coevaluation morphisms are defined by 
\eq{
\evrm_{L(A)}&=\psi_\mathbf{1}^L\circ L(\evrm_A)\circ\phi^L:L(A^\ast)\otimes L(A)\rightarrow \mathbf{1}\\\coevrm_{L(A)}&=\psi^L\circ L(\coevrm_A)\circ \phi_\mathbf{1}^L:\mathbf{1}\rightarrow L(A)\otimes L(A^\ast)
}
and their graphical representation reads
\begin{center}
\begin{tikzpicture}
\node at (-4,0) {\includegraphics[scale=0.15]{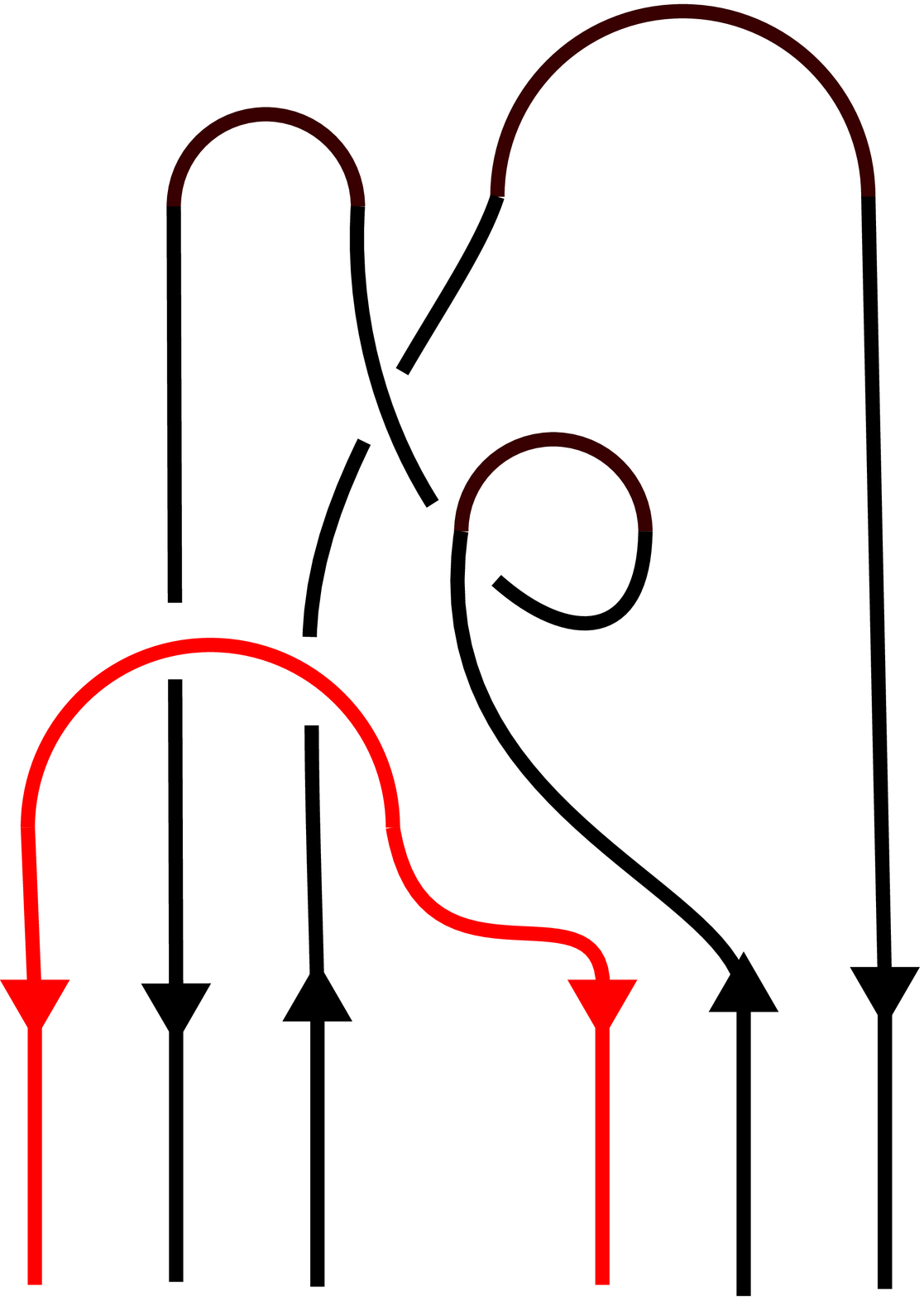}};
\node at (-7.7,0.1) {$\evrm_{L(A)}$};
\node at (-7,0.1) {$=$};
\node at (-6.5,0) {$\underset{i\in \Isf}{\bigoplus}$};
\node at (-6,0.1) {$\frac{\Dsf^2}{d_i}$};
\node at (-5.55,-2) {\scriptsize $\color{red} A$};
\node at (-4.7,-2) {\scriptsize $i$};
\node at (-4.2,-2) {\scriptsize $i$};
\node at (4,0) {\includegraphics[scale=0.15]{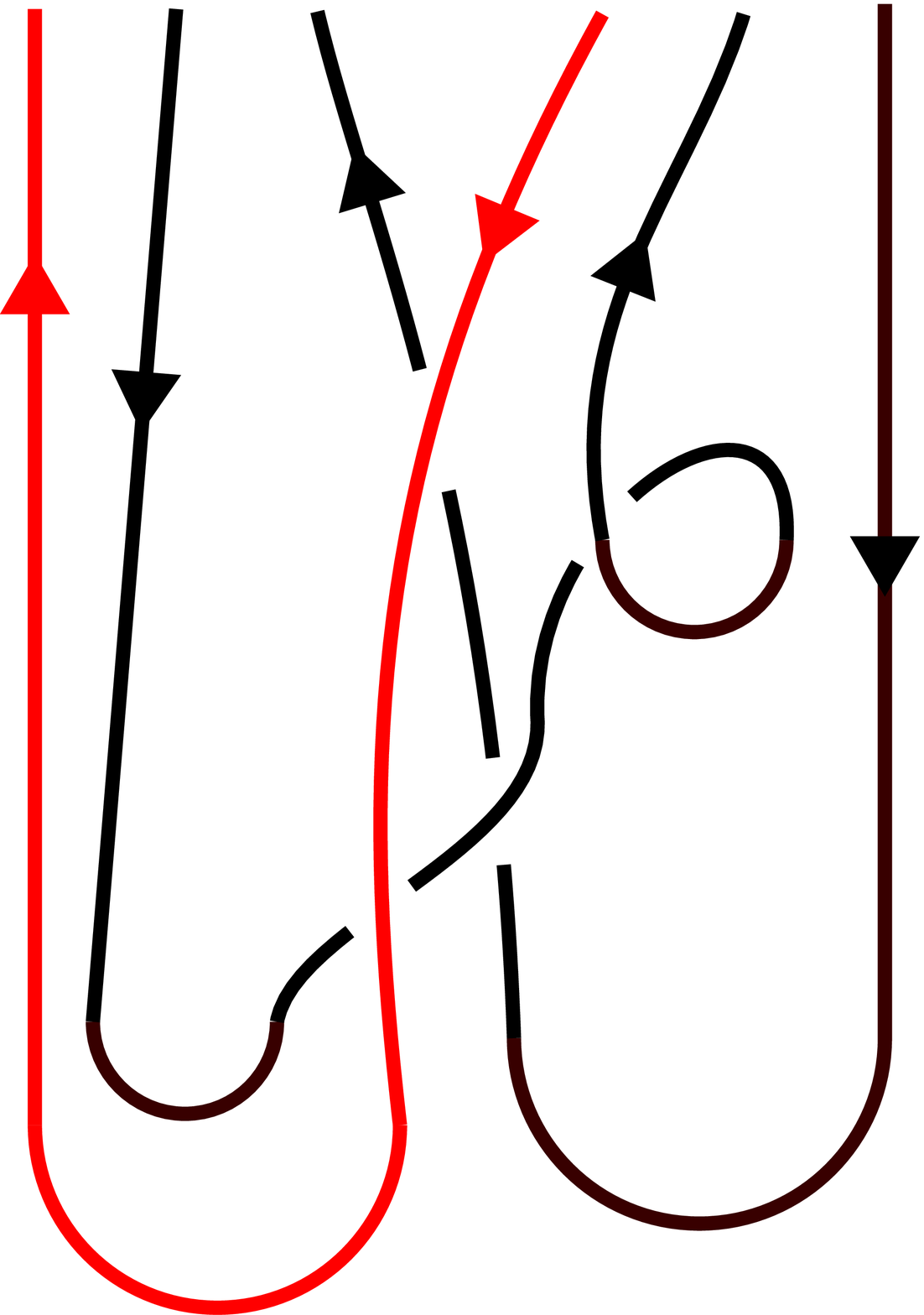}};
\node at (0,0.1) {$\coevrm_{L(A)}$};
\node at (1,0.1) {$=$};
\node at (1.5,0) {$\underset{i\in \Isf}{\sum}$};
\node at (2,0) {$\frac{d_i}{\Dsf^2}$};
\node at (2.4,2) {\scriptsize $\color{red}A$};
\node at (3.3,2) {\scriptsize $i$};
\node at (3.7,2) {\scriptsize $i$};
\end{tikzpicture}.
\end{center}
One easily checks that the straightening relations hold. 
\end{proof}

We have to distinguish five cases when defining sewing in terms of string-nets.
\begin{enumerate}[label=\roman*)]
\item Gluing two world sheets of type I) such that the resulting world sheet is again of type I). Then the linear map is just concatenation of string-nets.
\item Gluing two world sheets of type I) s.th. the result is of type III). In this case we first apply the map $Z$ to the string-nets on the surfaces, then concatenate and add projector circles to new boundary components.
\item Gluing a world sheet of type III) and type II) s.th. the resulting world sheet is of type I). This is the case if the quotient surface of the first world sheet is a sphere with closed boundary components, open boundary components sitting on a single connected component of the boundary and no connected component of the boundary is a physical boundary. The other world sheet is a disjoint union of disks with a single closed boundary. Gluing the disks to the sphere gives a disk. In this case, string-nets are concatenated and the result is post-composed with $Y$.
\item Gluing a world sheet of type I) and II) s.th. the result is of type II). Again we first replace the string-net on the type I) surface by applying $Z$ followed by stacking string-nets. The same definition applies to the case where we replace Type II) with type III) world sheets.
\item Any other gluing is just concatenation of string-nets across glued boundaries (see e.g. figure \ref{randomcorrel}). 
\end{enumerate}
This defines $\Bcal$ on morphisms of $\WSsf$.
\begin{prop} $\Blcal:\WSsf\rightarrow \Vectsf$ is a symmetric monoidal functor. 
\end{prop}
\begin{proof}
We only need to check that the linear maps for gluings respect compositions. Firstly we check that first gluing of type i) and then applying $Z$ in fact gives the same as applying $Z$ and then stacking. Recall from section \ref{sec1} that concatenation of string-nets is defined using the evaluation morphism. Let $f\in \hom_{\Zsf(\Csf)}(\mathbf{1},\ov{\Hcalop}\otimes \Hcalop)$ and $g\in \hom_{\Zsf(\Csf)}(\mathbf{1},\Hcalop^\ast\otimes \ov{\Hcalop})$, then it holds $Z(f)\in \hom_{\Zsf(\Csf)}(\mathbf{1},\ov{L(\Hcalop)}\otimes L(\Hcalop)^\ast)$, $Z(g)\in \hom_{\Zsf(\Csf)}(\mathbf{1},L(\Hcalop)\otimes \ov{L(\Hcalop)})$. For the gluing of $Z(f),\, Z(g)$ we compute  

\begin{center}
\begin{table}[H]
\begin{tabular}{m{0.5cm} m{9cm}}
= &
\begin{tikzpicture}
\node at (0,0) {\includegraphics[scale=0.2]{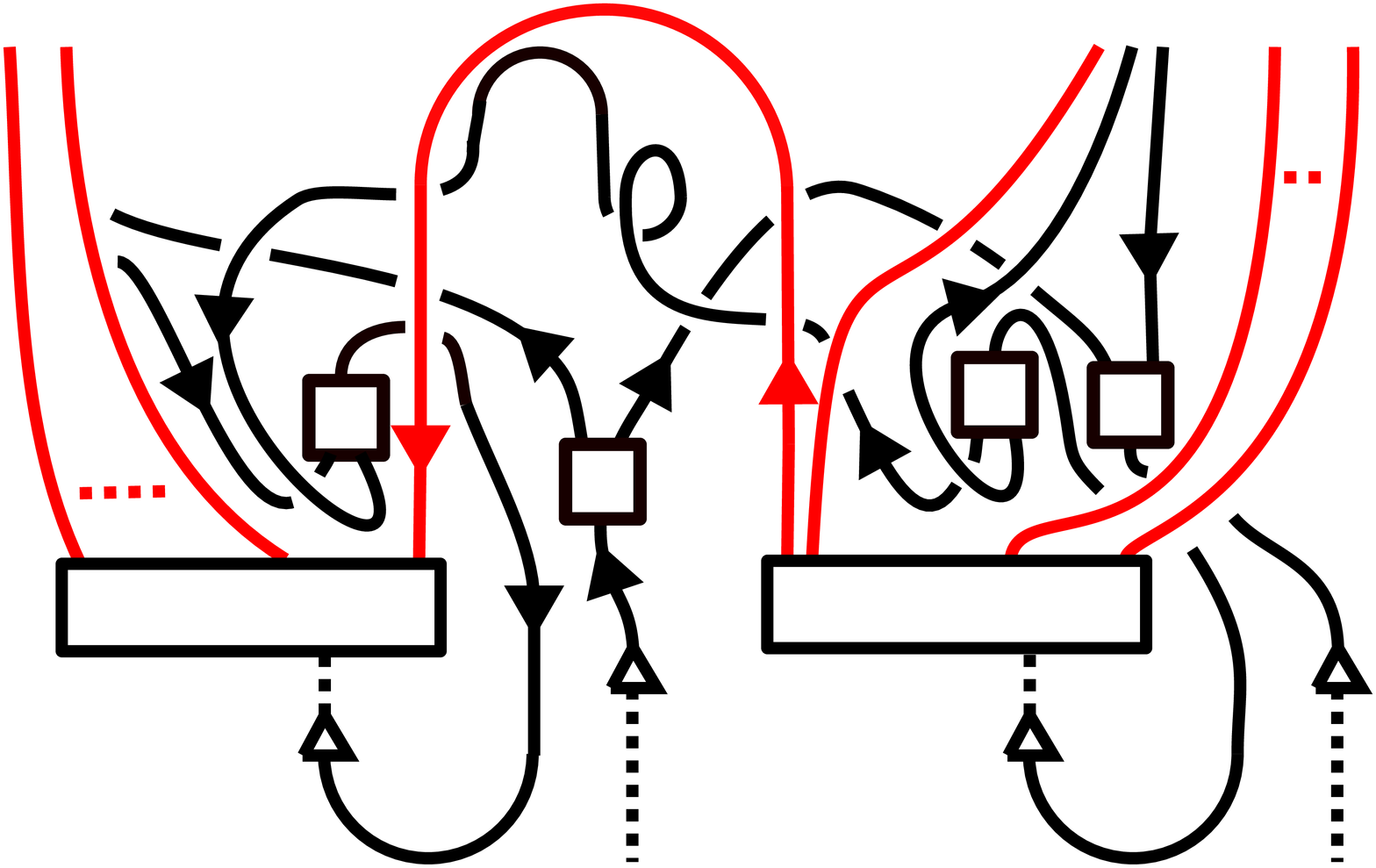}};
\node at (-2.1,-0.85) {\scriptsize $f$};
\node at (1.4,-0.85) {\scriptsize $g$};
\node at (-0.6,-1.2) {\scriptsize $i$};
\node at (-0.1,-0.9) {\scriptsize $i$};
\node at (-2,0.3) {\scriptsize $j$};
\node at (-2.55,0.8) {\scriptsize $k$};
\node at (-0.75,0.7) {\scriptsize $k$};
\node at (0.1,0.3) {\scriptsize $j$};
\node at (0.95,-0.4) {\scriptsize $j$};
\node at (1,0.6) {\scriptsize $\ell$};
\node at (2.5,-1.2) {\scriptsize $r$};
\node at (3.3,-0.8) {\scriptsize $r$};
\node at (2,0.7) {\scriptsize $j$};
\node at (2.45,1.4) {\scriptsize $\ell$};
\node at (-1.68,0.1) {\scriptsize $\alpha$};
\node at (-0.42,-0.23) {\scriptsize $\alpha$};
\node at (1.52,0.2) {\scriptsize $\tau$ };
\node at (2.19,0.12) {\scriptsize $\tau$};
\node at (-5,0) {$\begin{aligned} \bigoplus_{i,k,j,\ell,r\in \Isf}\sum_{\alpha,\tau}\frac{d_kd_\ell d_j}{\Dsf^2 d_id_r}\end{aligned}$};
\end{tikzpicture}
\end{tabular}
\end{table}
\end{center}

\begin{table}[H]
\begin{tabular}{m{0.5cm} m{9cm}}
= & 
\begin{tikzpicture}
\node at (0,-5) {\includegraphics[scale=0.2]{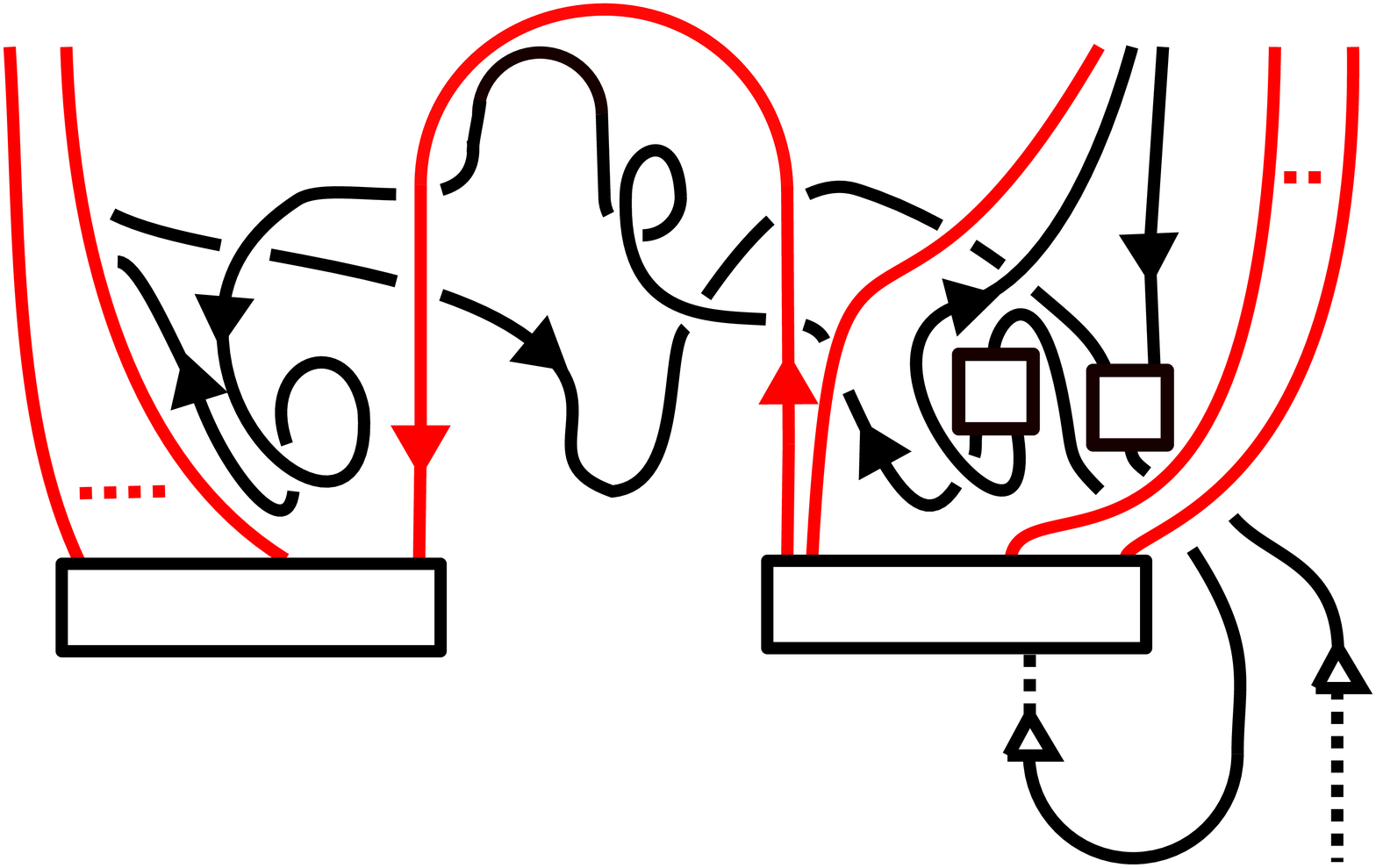}};
\node at (-2.1,-5.85) {\scriptsize $f$};
\node at (1.4,-5.85) {\scriptsize $g$};
\node at (0.95,-5.4) {\scriptsize $j$};
\node at (1,-4.4) {\scriptsize $\ell$};
\node at (2.5,-6.2) {\scriptsize $r$};
\node at (3.3,-5.8) {\scriptsize $r$};
\node at (2,-4.3) {\scriptsize $j$};
\node at (2.45,-3.6) {\scriptsize $\ell$};
\node at (1.52,-4.78) {\scriptsize $\tau$ };
\node at (2.19,-4.85) {\scriptsize $\tau$};
\node at (-5,-5) {$\begin{aligned} \bigoplus_{j,\ell,r\in \Isf}\sum_{\alpha,\tau}\frac{d_\ell d_j}{\Dsf^2 d_r}\end{aligned}$};
\end{tikzpicture}
\end{tabular}
\end{table}

\begin{table}[H]
\begin{tabular}{m{0.5cm} m{9cm}}
= & 
\begin{tikzpicture}
\node at (0,-10) {\includegraphics[scale=0.2]{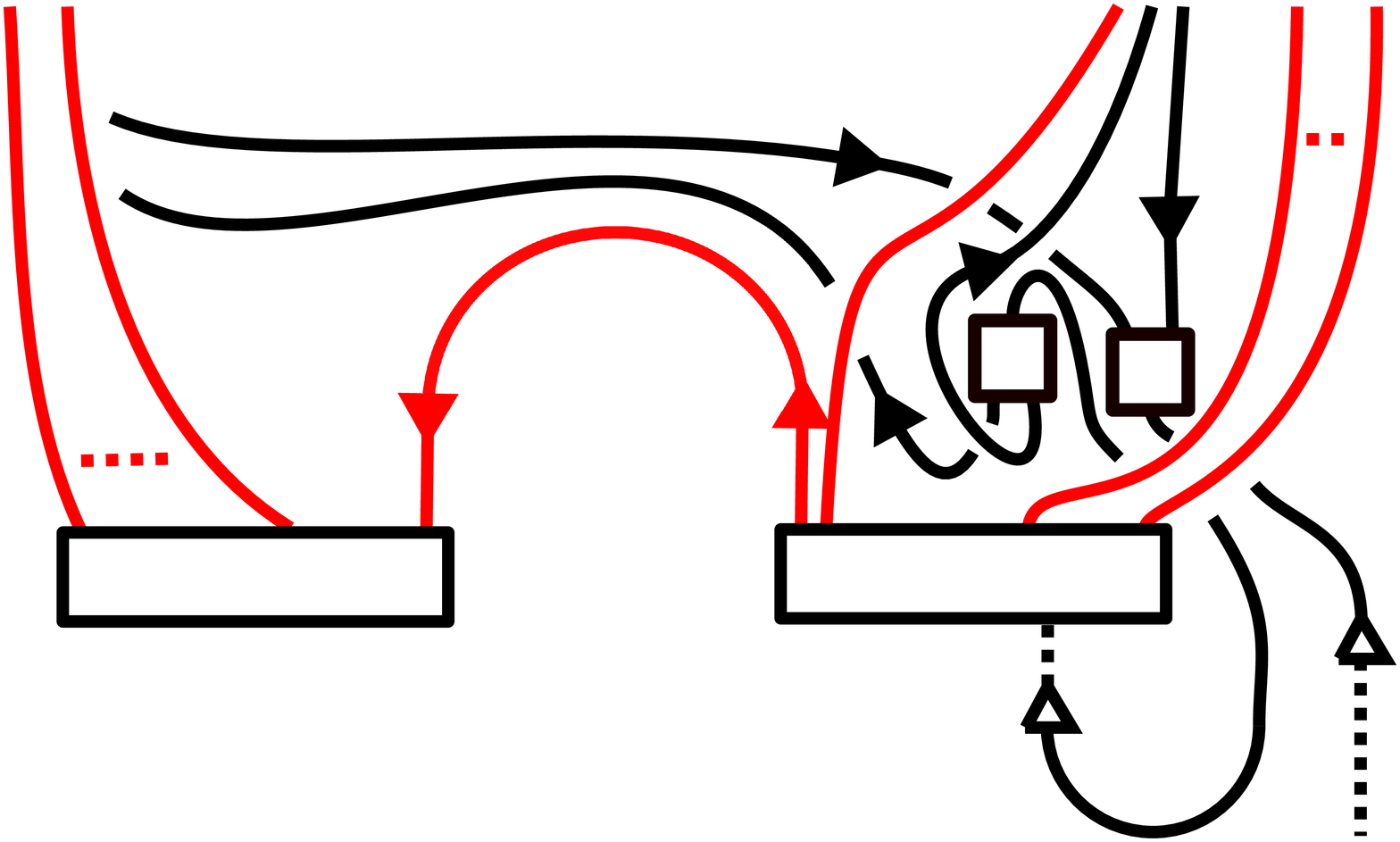}};
\node at (-2.1,-10.8) {\scriptsize $f$};
\node at (1.4,-10.75) {\scriptsize $g$};
\node at (0.95,-10.3) {\scriptsize $j$};
\node at (1,-9.4) {\scriptsize $\ell$};
\node at (2.5,-11.2) {\scriptsize $r$};
\node at (3.3,-10.8) {\scriptsize $r$};
\node at (2,-9.2) {\scriptsize $j$};
\node at (2.45,-8.6) {\scriptsize $\ell$};
\node at (1.52,-9.7) {\scriptsize $\tau$ };
\node at (2.19,-9.8) {\scriptsize $\tau$};
\node at (-5,-10) {$\begin{aligned} \bigoplus_{j,\ell,r\in \Isf}\sum_{\alpha,\tau}\frac{d_\ell d_j}{\Dsf^2 d_r}\end{aligned}$};
\end{tikzpicture}.
\end{tabular}
\end{table}

Red strands are colored with $A$ or $A^\ast$, depending on the orientation. In the first equality we use the definition of the basis elements $\lbr \theta_\alpha^{(ij);k}\rbr $ and their duals. But the final picture is nothing else than $Z$ applied to the gluing of $f$ and $g$. Next, we note that compositions of gluing disks to closed boundary components of type II) world sheets and gluing along open boundary components of the same world sheet is well defined due to Lemma \ref{restricitionstatespacelemma}. Hence, compositions including open boundary components are well defined. Compositions of gluing along closed boundary compositions are obviously well defined. Finally, a disjoint union of world sheets $\coprod \hat{S}$ gets mapped to $\bigotimes \Bcal\left(\hat{S}\right)$. Thus, $\Bcal$ is a symmetric monoidal functor.
\end{proof}

Note that $\Blcal(\hat{S})$ for type II) world sheets is non-zero and contains all interesting cases. Since \linebreak $\hom_{\Zsf(\Csf)}\left(\mathbf{1},L(\Hcalop)\right) \simeq \hom_{\Csf}\left(\mathbf{1}, \Hcalop\right)$, $\Blcal(\hat{S})$ contains all string-nets obtained from gluing type I) world sheets to a type II) world sheet with a single open boundary component.
It is worthwhile analyzing which vector spaces the functor assigns to generating world sheets.
\begin{enumerate}[label=\Roman*)]
\item \textbf{Open World Sheets:} The quotient surfaces of open generating world sheets are all homeomorphic to a disk, though with different numbers of incoming and outgoing open boundaries. For a disk $D(n_i,n_o)$ with $n_i$ incoming open boundaries and $n_o$ outgoing open boundaries we get the vector space
\eq{
\Blcal(D(n_i,n_0))=\hat{H}^s(D,\ov{\Hcalop})=\hom_{\Zsf(\Csf)}(\mathbf{1},\Hcalop)\simeq \hom_\Csf(\mathbf{1},\ov{\Hcalop})\quad .
}

\item \textbf{Closed World Sheets:} In this case the quotient surface is topologically a sphere with $n_i$ incoming closed boundaries and $n_o$ outgoing closed boundaries. With the same notation as in the open case we get the vector space
\eq{
\Blcal(S^2(n_i,n_o))=\hat{H}^s(S^2,\ov{\Hcalcl})=\hom_{\Zsf(\Csf)}(\mathbf{1},\ov{\Hcalcl})\quad .
}
\item \textbf{Open-Closed World Sheets:} Finally open-closed generating world sheets get mapped
\eq{
\Blcal(I)&=\hat{H}^s(I,L_{op}(\Hcalop)\otimes \Hcalcl^\ast)\simeq \hom_{\Zsf(\Csf)}(\mathbf{1},L_{op}(\Hcalop)\otimes \Hcalcl^\ast)\\
\Blcal(I^\dagger)&=\hat{H}^s(I^\dagger,L_{op}(\Hcalop)^\ast\otimes \Hcalcl)=\simeq \hom_{\Zsf(\Csf)}(\mathbf{1},L_{op}(\Hcalop)^\ast\otimes \Hcalcl)\quad .
}
\end{enumerate}
Thus $\Blcal$ associates to generating world sheets the vector spaces expected from the categorical description of RCFTs.
\subsection{Fundamental Correlators on Generating World Sheets}\label{subsec61}

For a consistent system of correlators we have to give fundamental correlators on generating world sheets and show the sewing constraints for this set of correlators. We start by defining fundamental correlators for Cardy algebra $(\Hcalcl,\Hcalop,\iotaclop)$. In the following red curves correspond to edges colored with $\Hcalop$, purple curves always denote edges of graphs colored by $L(\Hcalop)$ and orange curves are edges colored by $\Hcalcl$. Incoming and outgoing boundaries should be clear from the orientation of edges. Trivalent disk-shaped vertices either denote multiplication or comultiplication in Frobenius algebras $\Hcalcl,\,L(\Hcalop),\, \Hcalop$. Similarly one-valent disk-shaped vertices are either unit or counit. In both cases orientation of edges fix the kind of morphism assigned, so we suppress another graphical distinction between the two.
\begin{enumerate}[label=\Roman*)]\label{worldsheetstringnets}
\item \textbf{Open World Sheets:}

\begin{figure}[H]
\centering
\begin{tikzpicture}
\node (opprop) at (0,0) {\includegraphics[scale=0.2]{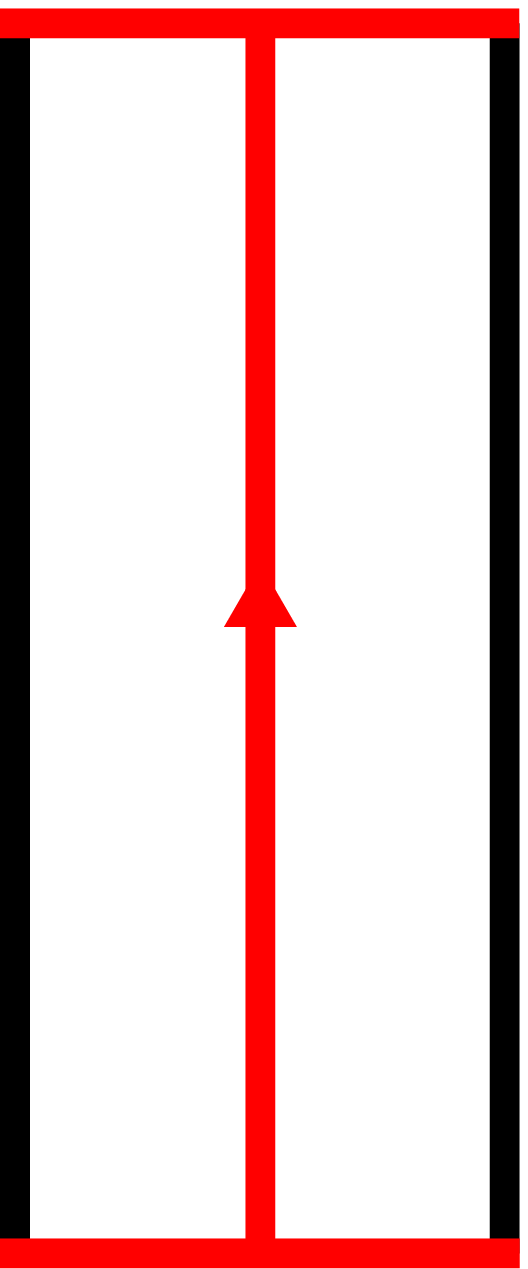}};
\node (mop) at (3,0) {\includegraphics[scale=0.2]{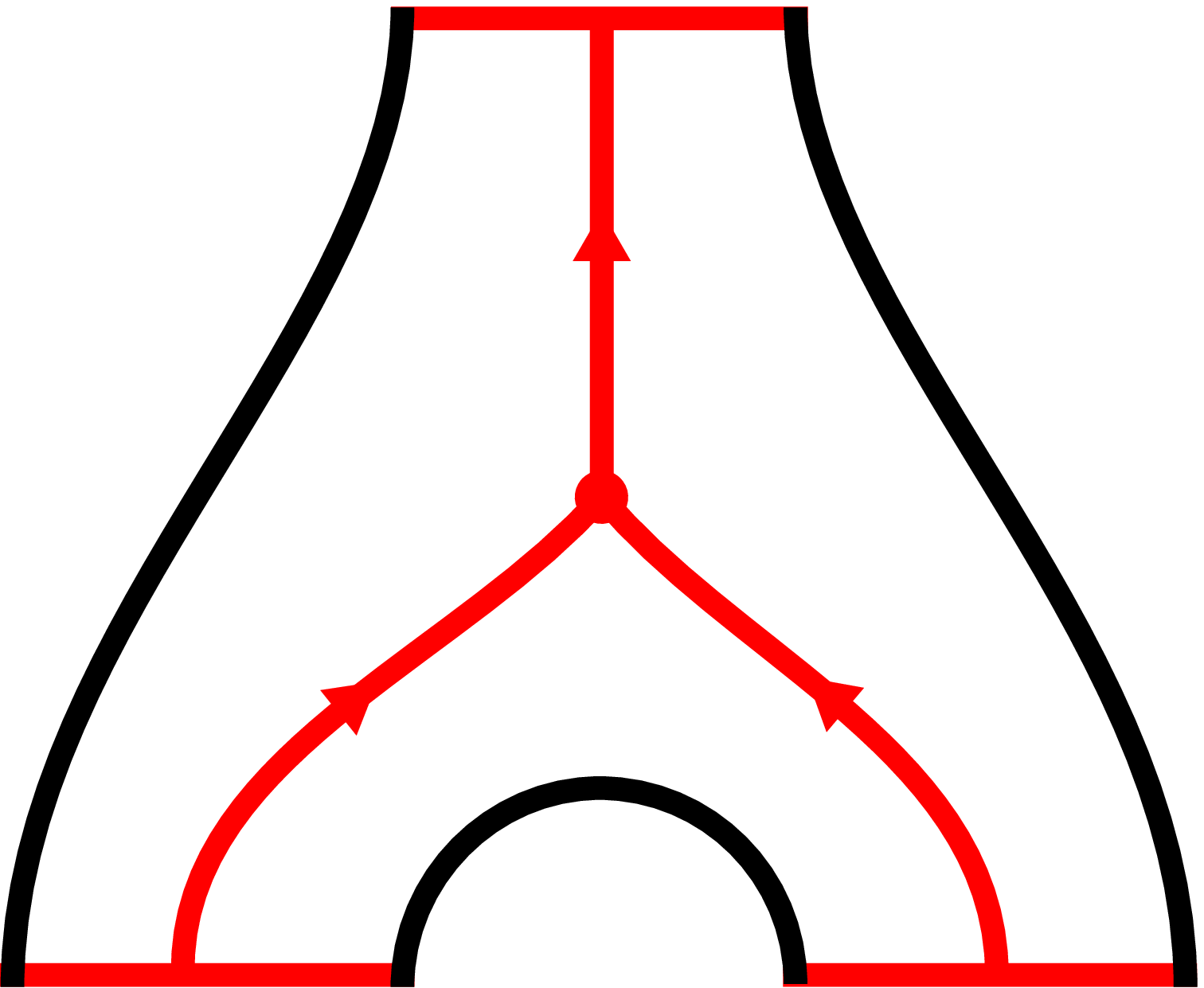}};
\node (Deltaop) at (6,0) {\includegraphics[scale=0.2, angle=180]{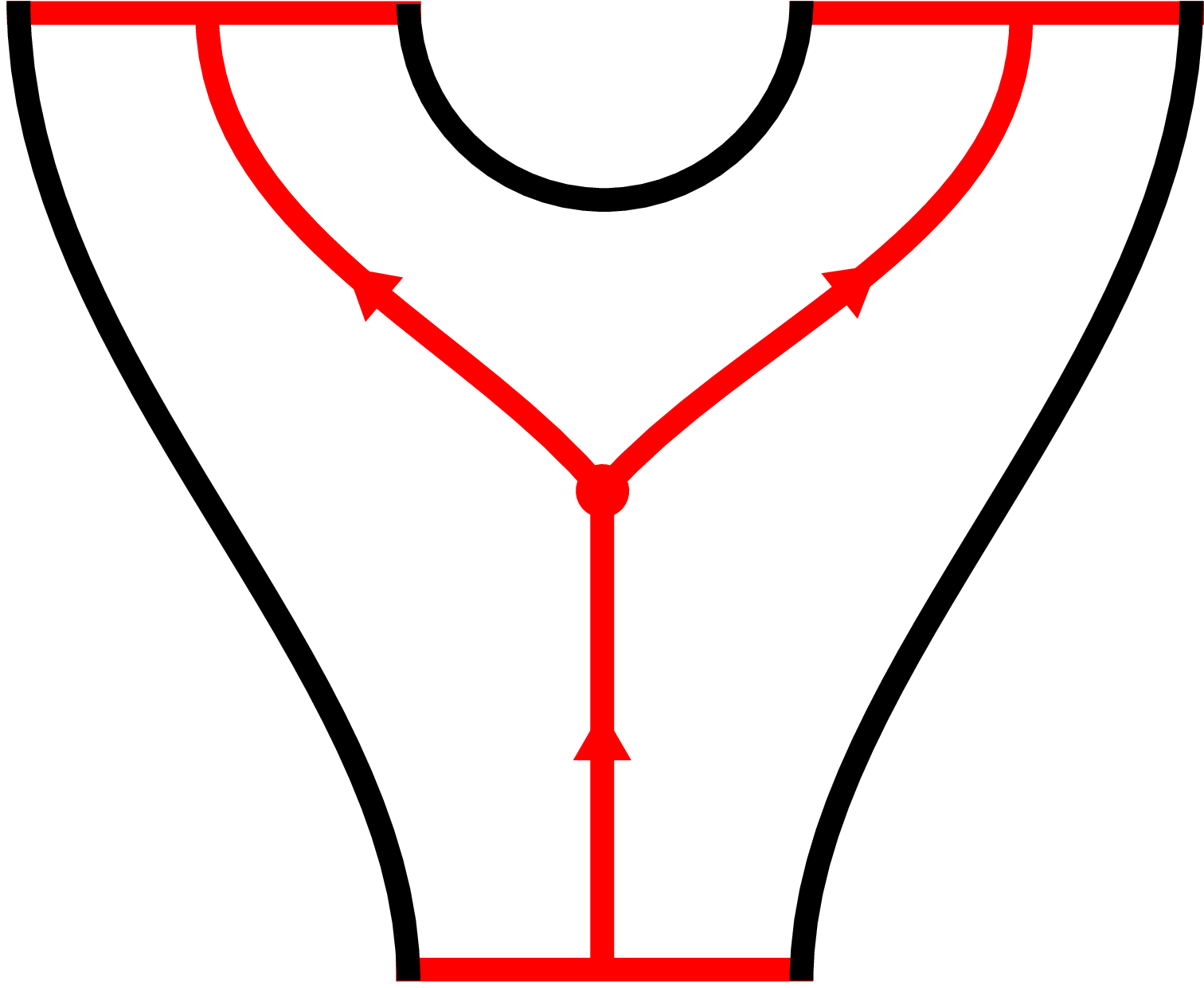}};
\node (unitop) at (9,0) {\includegraphics[scale=0.23]{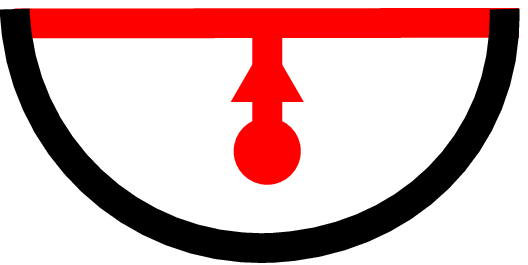}};
\node (counit) at (12,0) {\includegraphics[scale=0.23, angle=180]{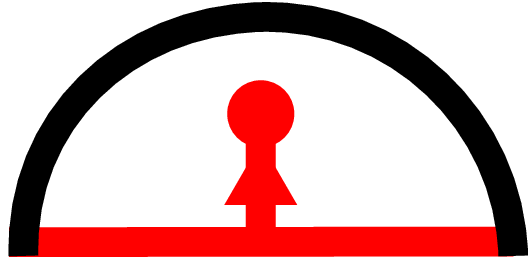}};
\node (Oprop) at (0,-2.2) {$\corr^{op}_{prop}$};
\node (Om) at (3,-2.2) {$\corr^{op}_m$};
\node (Odelta) at (6,-2.2) {$\corr^{op}_\Delta$};
\node (Ounit) at (9,-2.2) {$\corr^{op}_\eta$};
\node (Ocounit) at (12,-2.2) {$\corr^{op}_\epsilon$};
\end{tikzpicture}
\caption{Open fundamental correlators.}
\label{openfundamentalcorr}
\end{figure}
Note that for these world sheets the inserted projector $P$ (see figure \ref{projector}) is absent as it can be isotoped to a point and therefore vanishes.

\item \textbf{Closed World Sheets:}

\begin{figure}[H]
\begin{tikzpicture}
\node (opprop) at (0,0) {\includegraphics[scale=0.2]{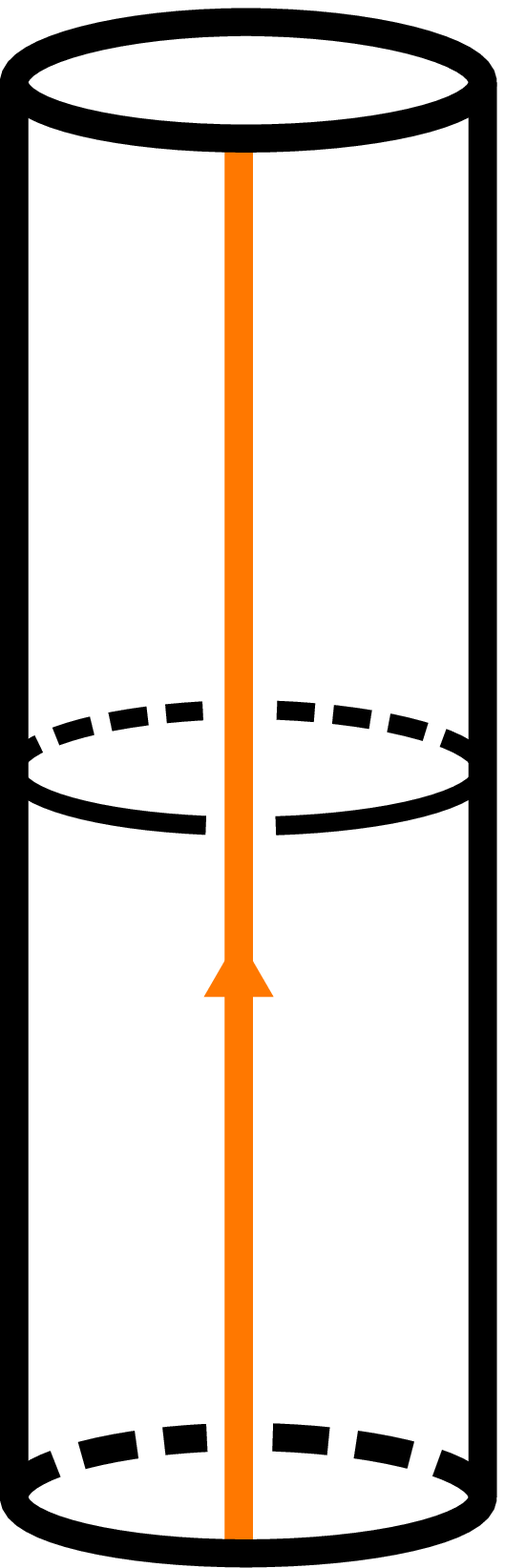}};
\node (mop) at (3,0) {\includegraphics[scale=0.2]{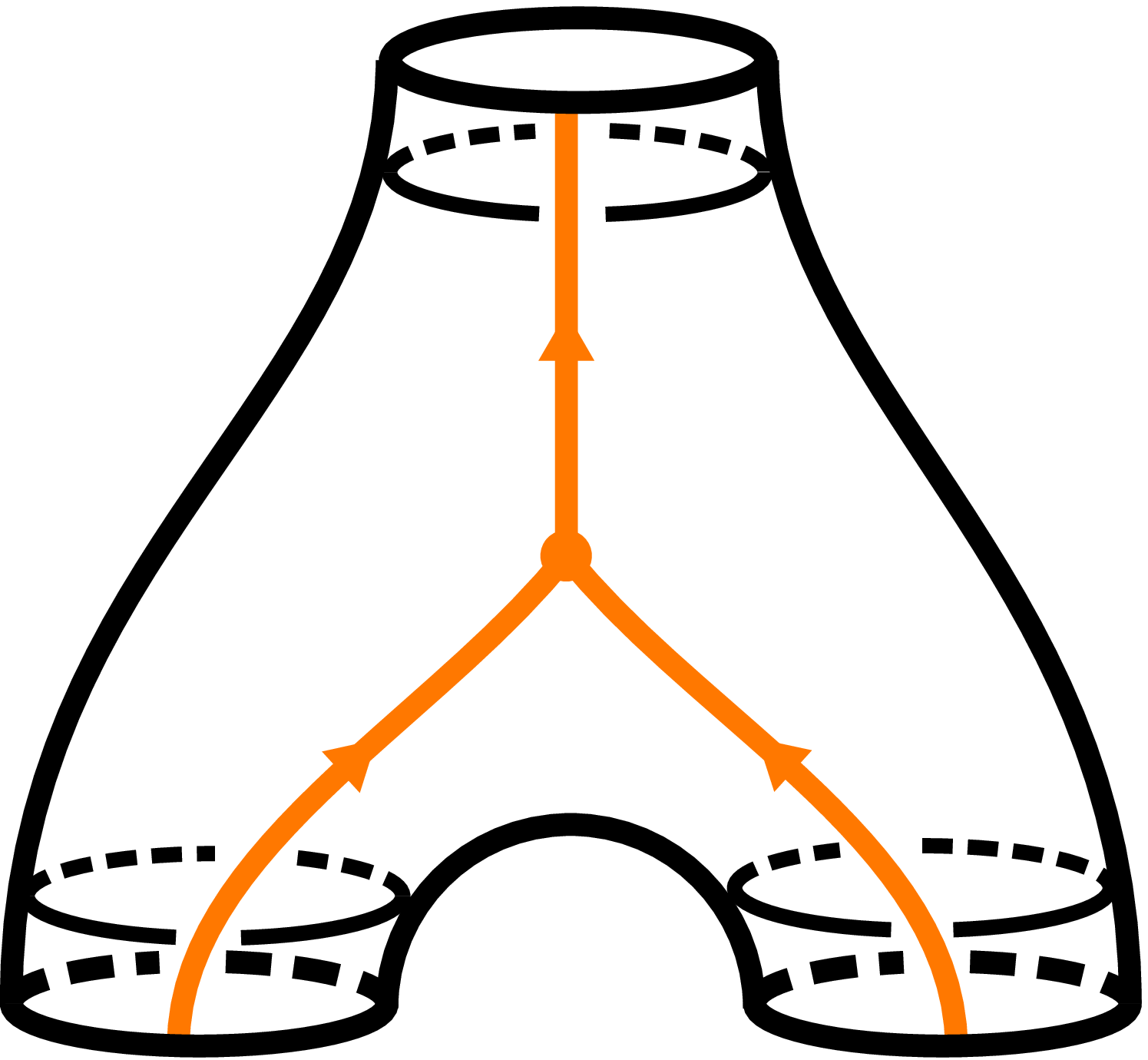}};
\node (Deltaop) at (6,0) {\includegraphics[scale=0.2, angle=180]{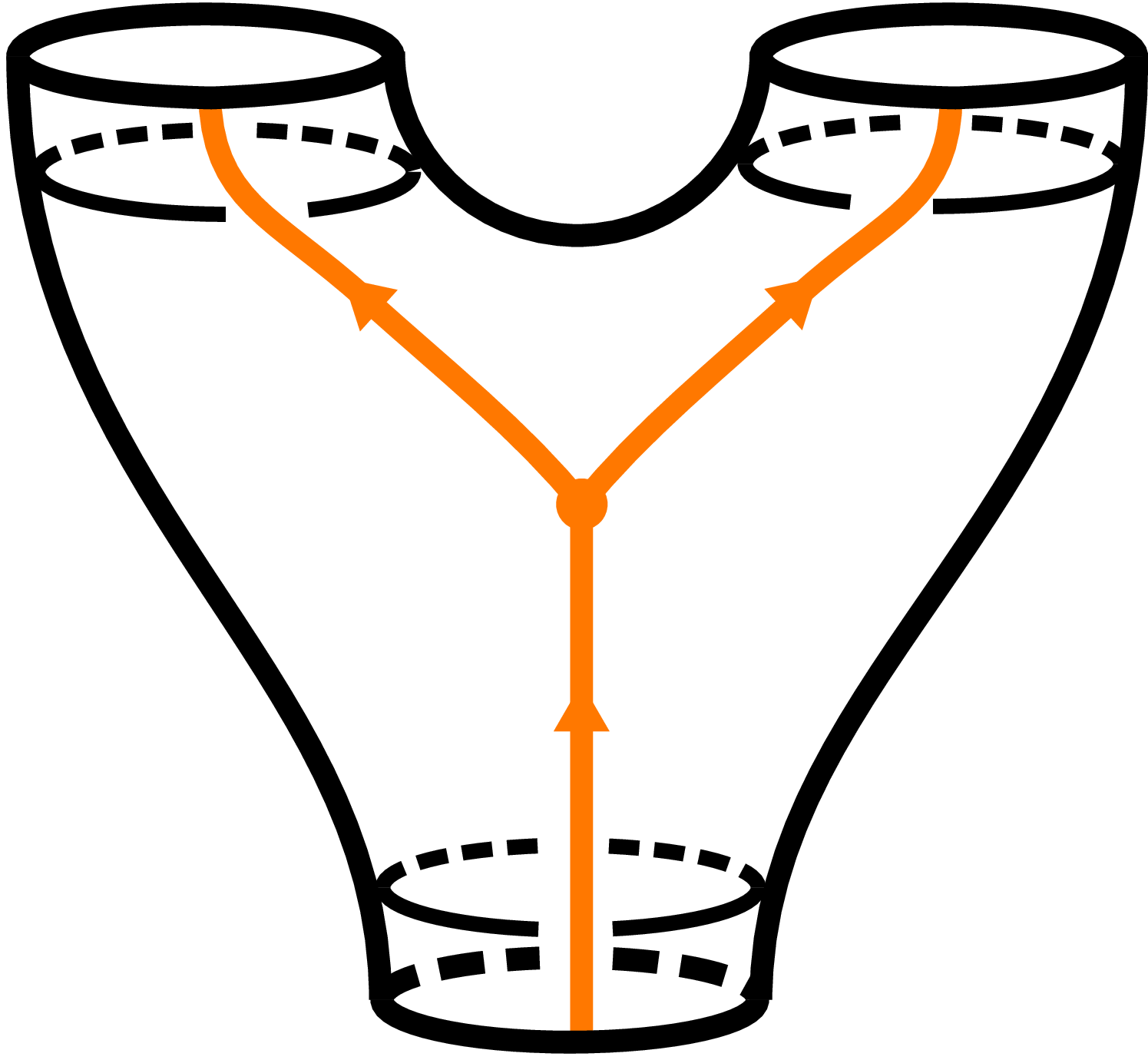}};
\node (unitop) at (9,0) {\includegraphics[scale=0.23]{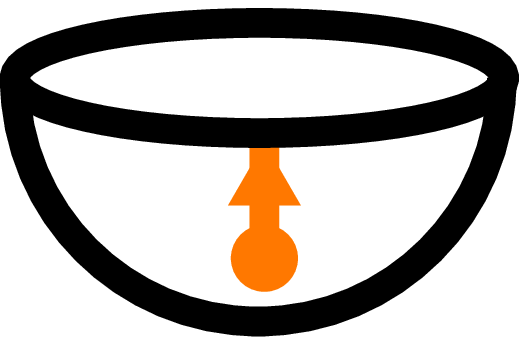}};
\node (counit) at (12,0) {\includegraphics[scale=0.23]{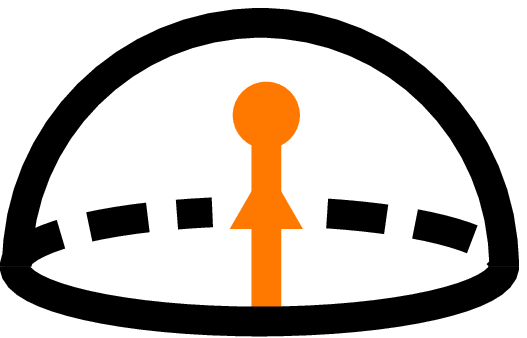}};
\node (Oprop) at (0,-2.2) {$\corr^{cl}_{prop}$};
\node (Om) at (3,-2.2) {$\corr^{cl}_m$};
\node (Odelta) at (6,-2.2) {$\corr^{cl}_\Delta$};
\node (Ounit) at (9,-2.2) {$\corr^{cl}_\eta$};
\node (Ocounit) at (12,-2.2) {$\corr^{cl}_\epsilon$};
\end{tikzpicture}
\caption{Closed fundamental correlators.}
\label{closedfundamentalcorr}
\end{figure}

\item \textbf{Open-Closed World Sheets:}

World sheets $I$, $I^\dagger$ are topologically cylinders and a single projector line is inserted. Boxes denote the morphisms $\iotaclop$ or $\iotaclop^\dagger$, where again the orientation of edges displayed fixes the type of morphism.

\begin{figure}[H]
\centering
\begin{tikzpicture}
\node (iota) at (0,0) {\includegraphics[scale=0.2]{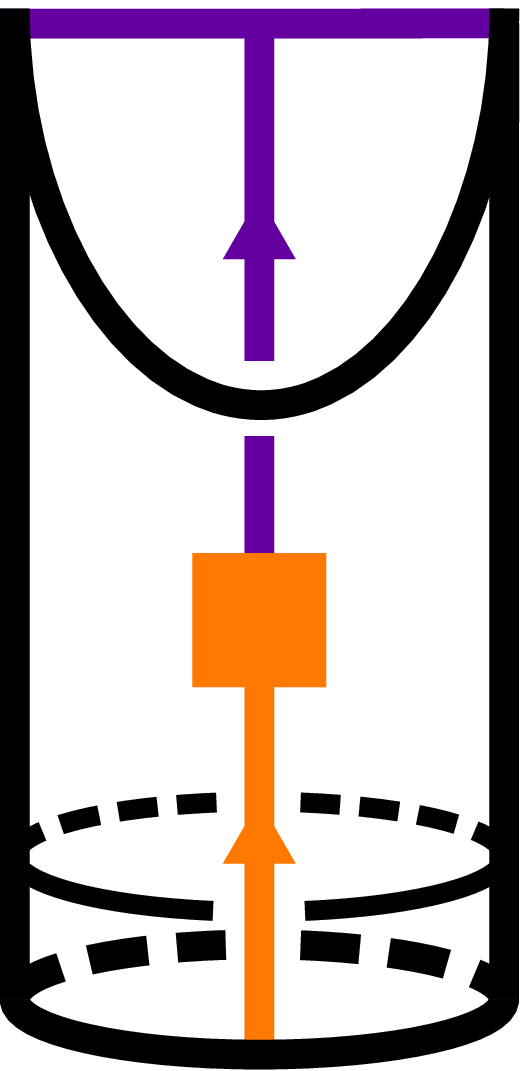}};
\node (iotadagger) at (3,0) {\includegraphics[scale=0.2, angle=180]{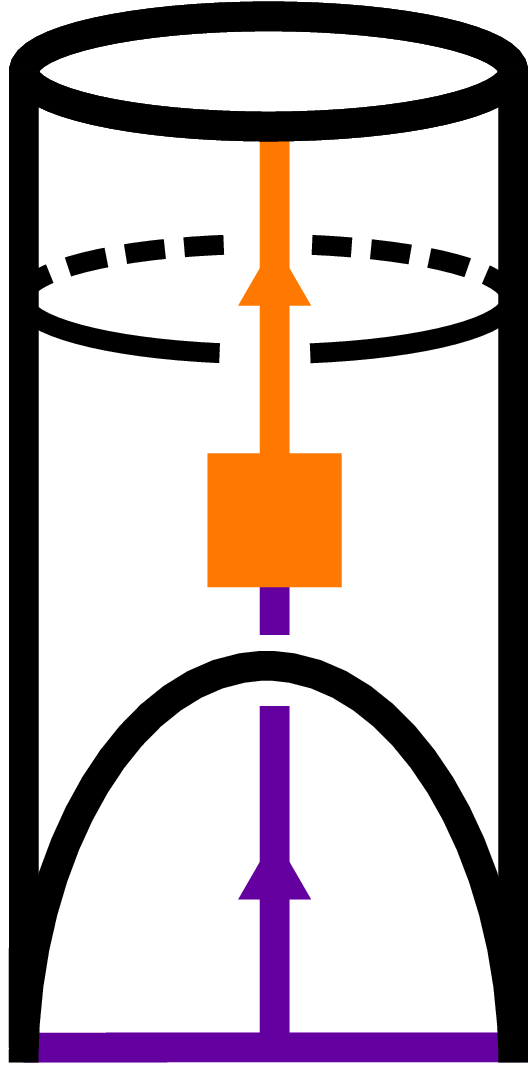}};
\node (I) at (0,-2.2) {$\corr_{I}$};
\node (Idagger) at (3,-2.2) {$\corr_{I^\dagger}$};
\end{tikzpicture}
\caption{Open-closed fundamental correlators.}
\label{openclosedfundamentalcorr}
\end{figure}
\end{enumerate}

We are now ready to state and prove the first main result of the paper.
\begin{theo} \label{maintheo1} The correlators 
\eq{
\lbr \corr^{op}_{prop},\corr^{op}_m,\corr^{op}_\Delta,\corr^{op}_\eta,\corr^{op}_\epsilon,\corr^{cl}_{prop},\corr^{cl}_{m},\corr^{cl}_\Delta,\corr^{cl}_{\eta},\corr^{cl}_{\epsilon},\corr_I,\corr_{I^\dagger} \rbr 
}
satisfy the sewing constraints.
\end{theo}

In order to show the theorem, we have to show that the correlators on both world sheets for all 32 relations in section \ref{relations} agree. We split the proof in several lemmas.
\begin{lem} The correlators
\eq{
\lbr  \corr^{op}_{prop},\corr^{op}_m,\corr^{op}_\Delta,\corr^{op}_\eta,\corr^{op}_\epsilon \rbr 
}
satisfy all open relations.
\end{lem}
\begin{proof}
This is the easiest part of the theorem as string-nets on disks can be manipulated according to the graphical calculus of its coloring category. It is immediate that the relations directly follow from the fact that $\Hcalop$ is a symmetric Frobenius algebra in $\Csf$, which is fully faithfully embedded in $\Zsf(\Csf)$. Relations R1)-R4) are unit and counit properties. R5) is satisfied as $\Hcalop$ is a symmetric Frobenius algebra. Relations R6) and R7) are (co-)associativity for (co-)multiplications. Next, R8) and R9) are the Frobenius property and finally the last four relations R10)-R13) are just the fact that composing with $\corr_{prop}^{op}$ leaves any morphism invariant in the graphical calculus.
\end{proof}

\begin{lem}\cite[Lemma~3.8]{Schweigert:2019zwt} The correlators 
\eq{
\lbr \corr^{cl}_{prop},\corr^{cl}_{m},\corr^{cl}_\Delta,\corr^{cl}_{\eta},\corr^{cl}_{\epsilon} \rbr 
}
satisfy all closed relations.
\end{lem}

\begin{lem} All open-closed relations are satisfied.
\end{lem}
\begin{proof}
We start with relation R26). As the picture suggests this will follow from the center condition of $\iotaclop$.
\begin{center}
\begin{tikzpicture}
\node at (0,0) {\includegraphics[scale=0.2]{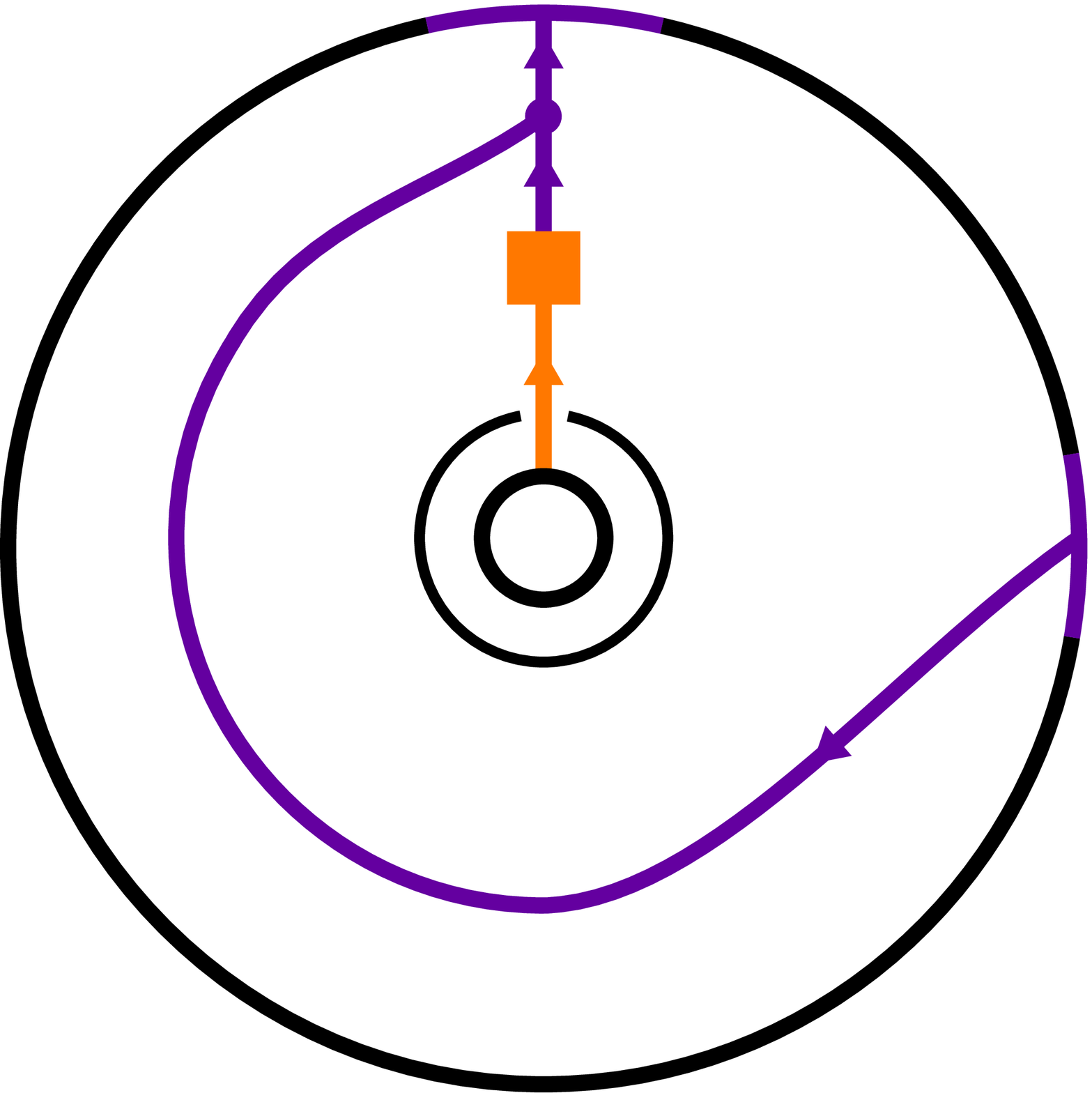}};
\node at (7,0) {\includegraphics[scale=0.2]{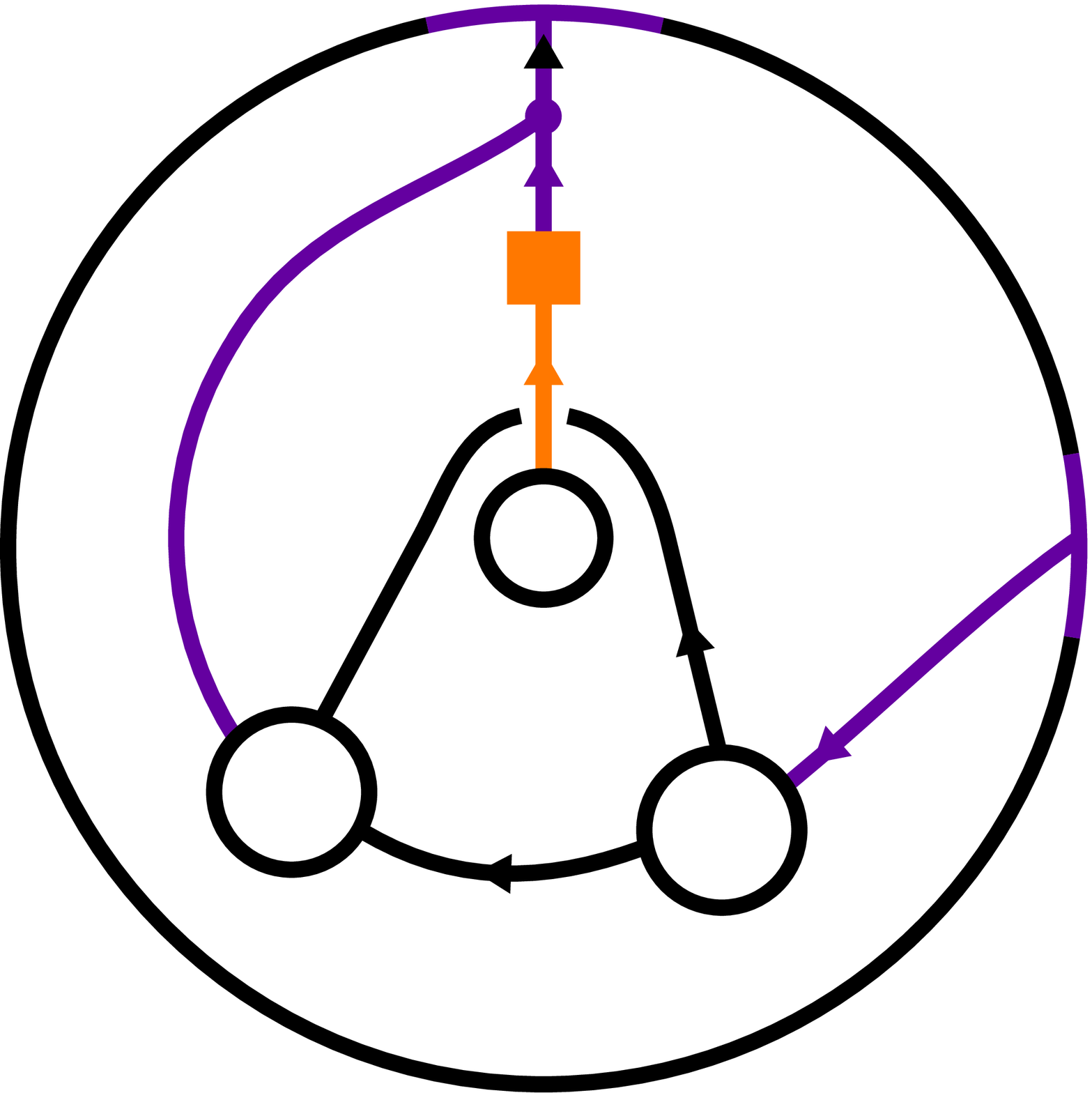}};
\node at (0,-5) {\includegraphics[scale=0.2]{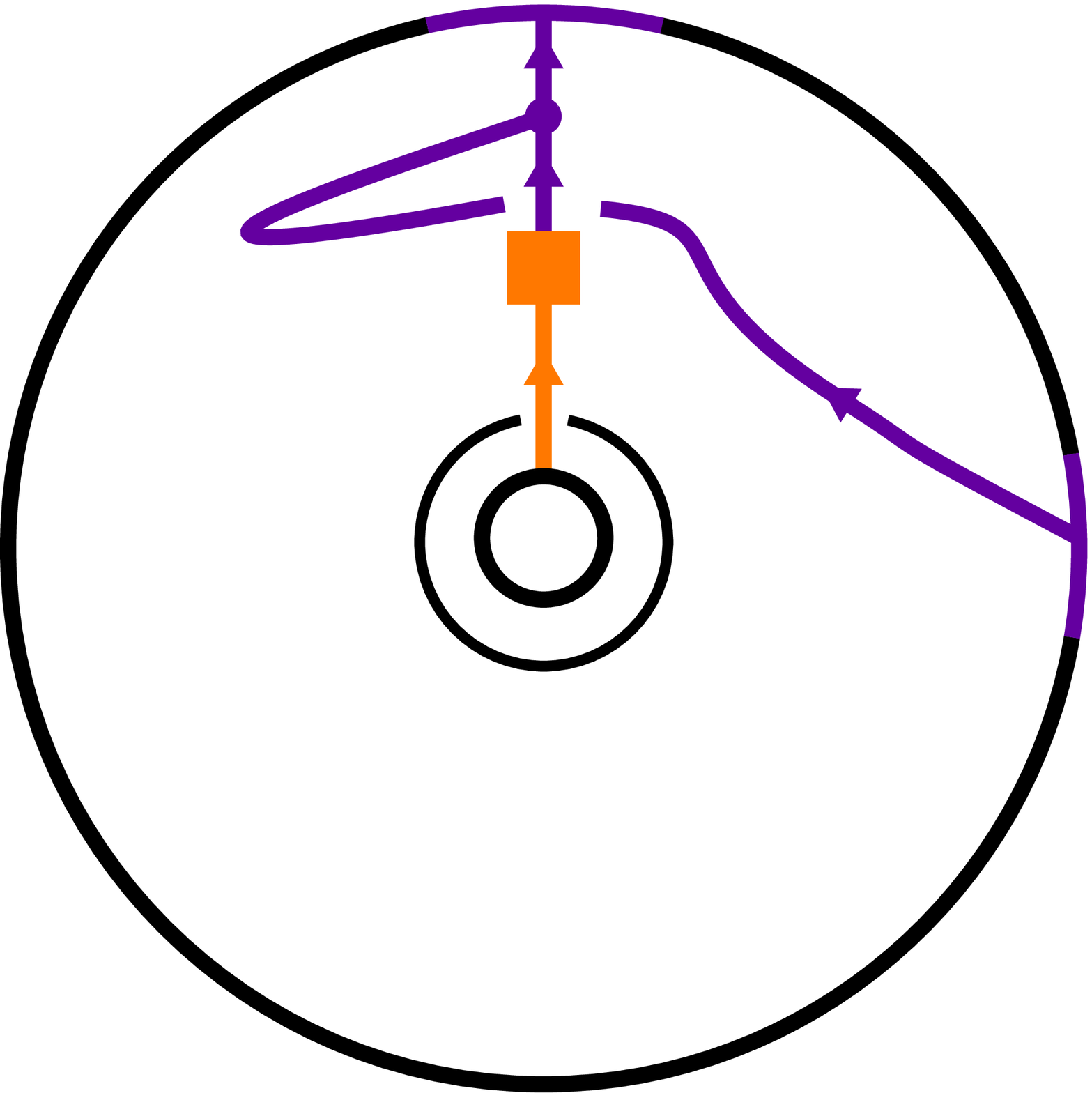}};
\node at (7,-5) {\includegraphics[scale=0.2]{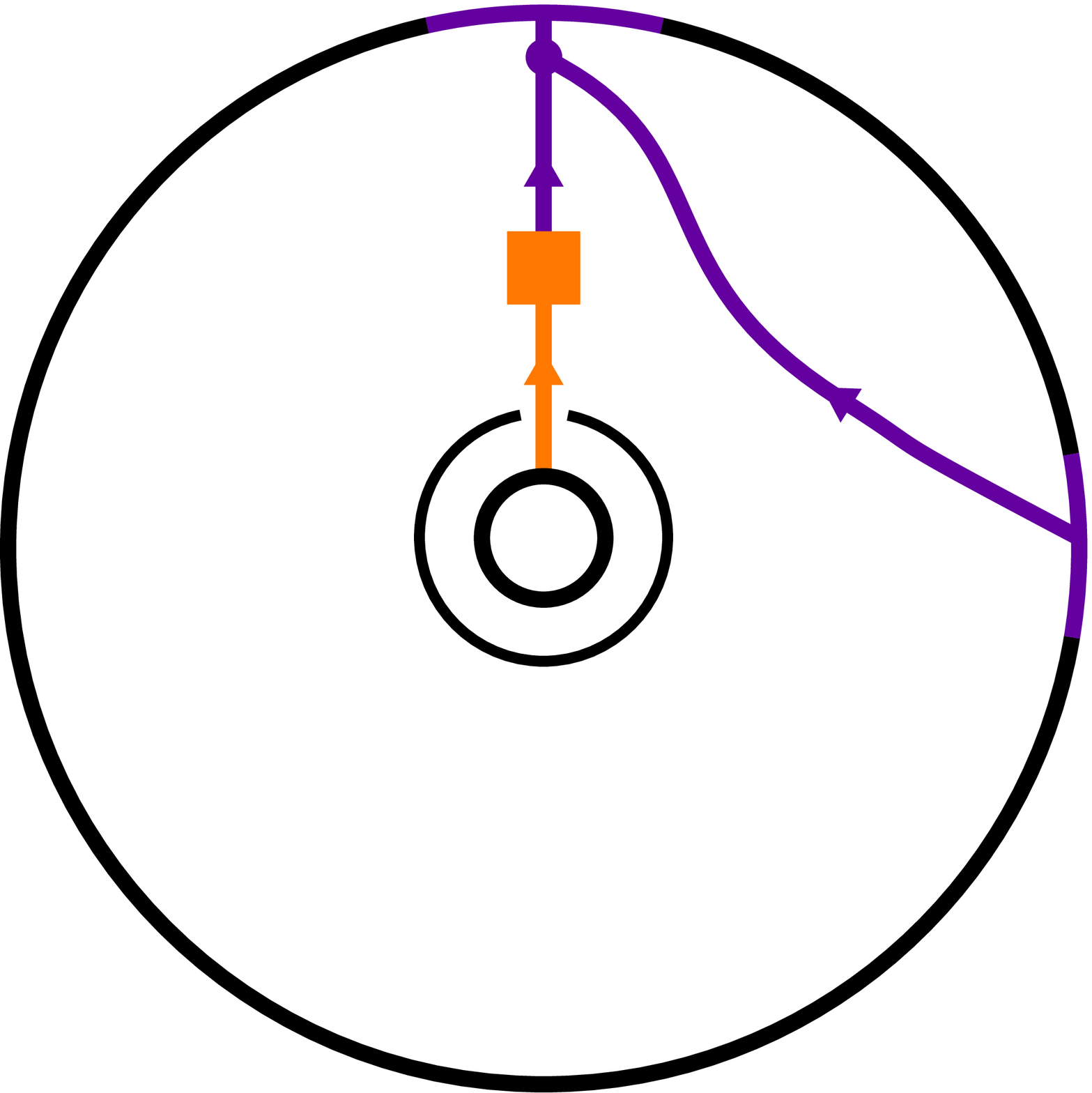}};
\node at (3,0) {$=$};
\node at (4,0) {$\begin{aligned}\sum_{i,k\in \Isf} \frac{d_id_k}{\Dsf^2}\end{aligned}$};
\node at (-3,-5) {$=$};
\node at (3.5,-5) {$=$};
\node at (6.05,-0.9) {$\alpha$};
\node at (7.67,-1.05) {$\alpha$};
\node at (6.7,-1.5) {$i$};
\node at (7.5,0.5) {$k$};
\end{tikzpicture}
\end{center}
Where we again use completeness followed by the center condition. 

Next we prove relation R27):

\begin{center}
\begin{tikzpicture}
\node at (0,0) {\includegraphics[scale=0.2]{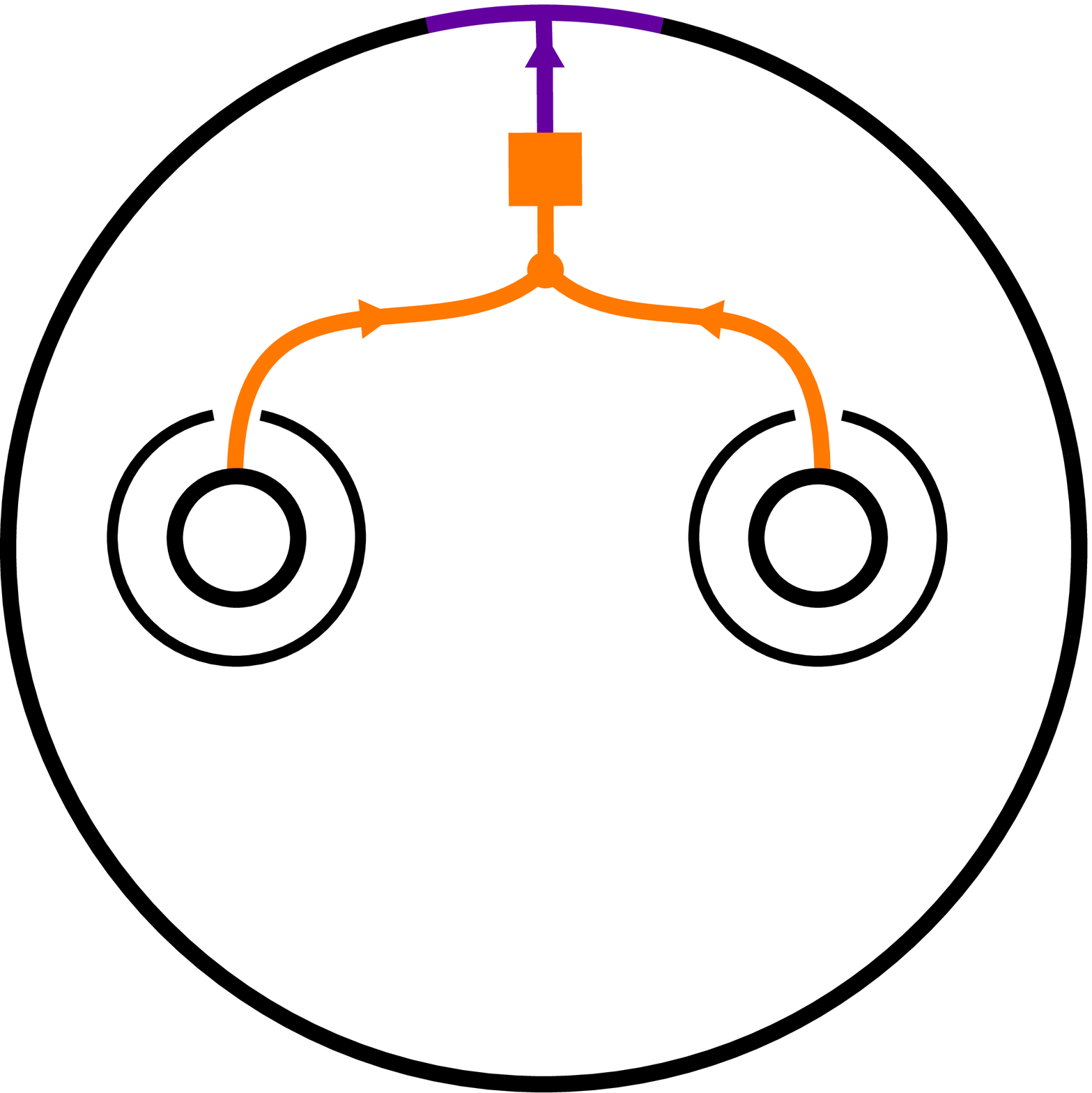}};
\node at (6,0) {\includegraphics[scale=0.2]{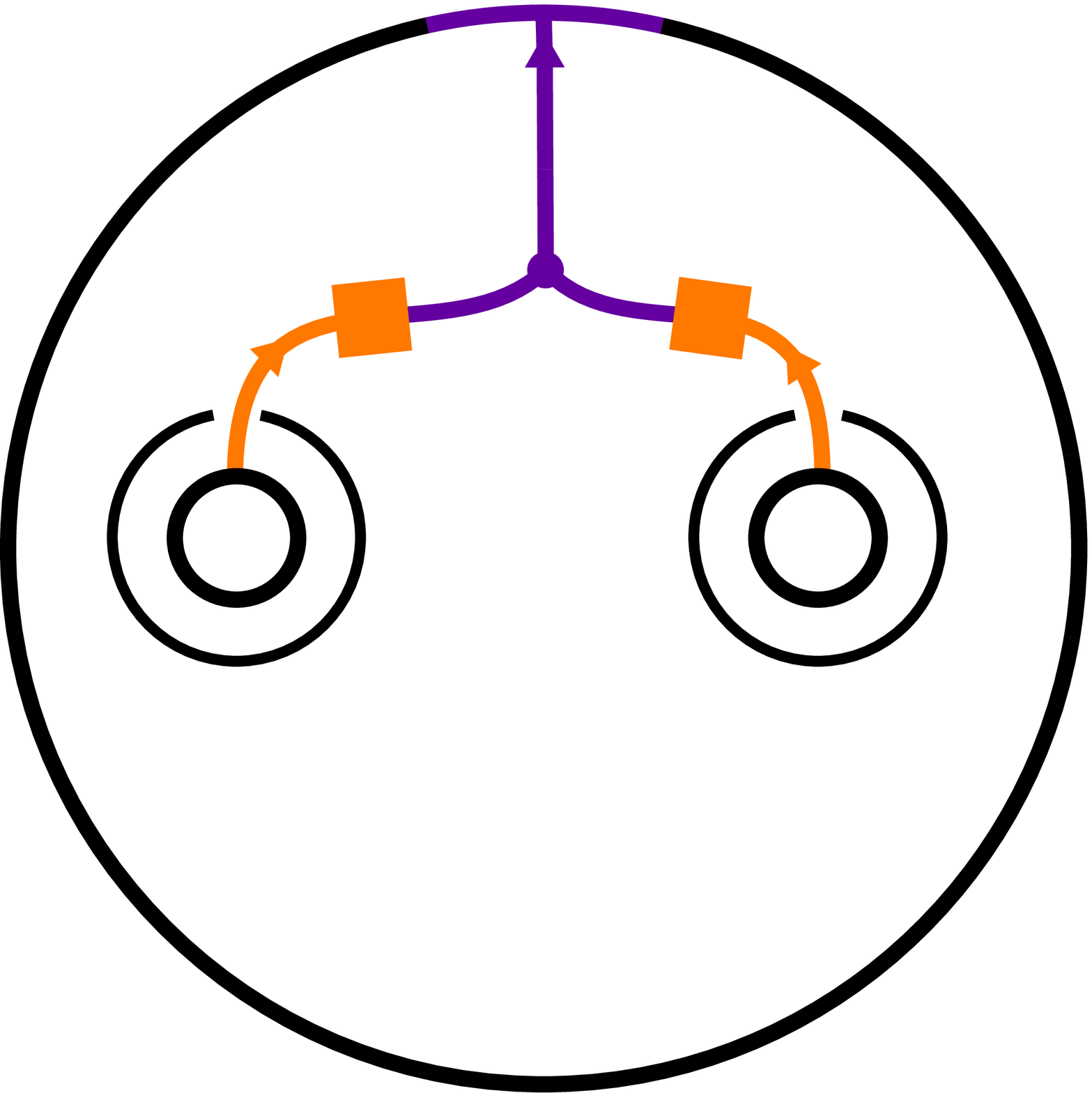}};
\node at (3,0) {$=$};
\end{tikzpicture}
\end{center}
where we see the obvious equality from $\iotaclop$ being an algebra homomorphism. Form the same reasoning it follows that relation R30) is satisfied. Relations R29) and R28) are consistency checks for the definition of $\iotaclop$ and $\iotaclop^\dagger$. We show R28), the other one goes exactly the same. 

\begin{center}
\begin{tikzpicture}
\node at (0,0) {\includegraphics[scale=0.2]{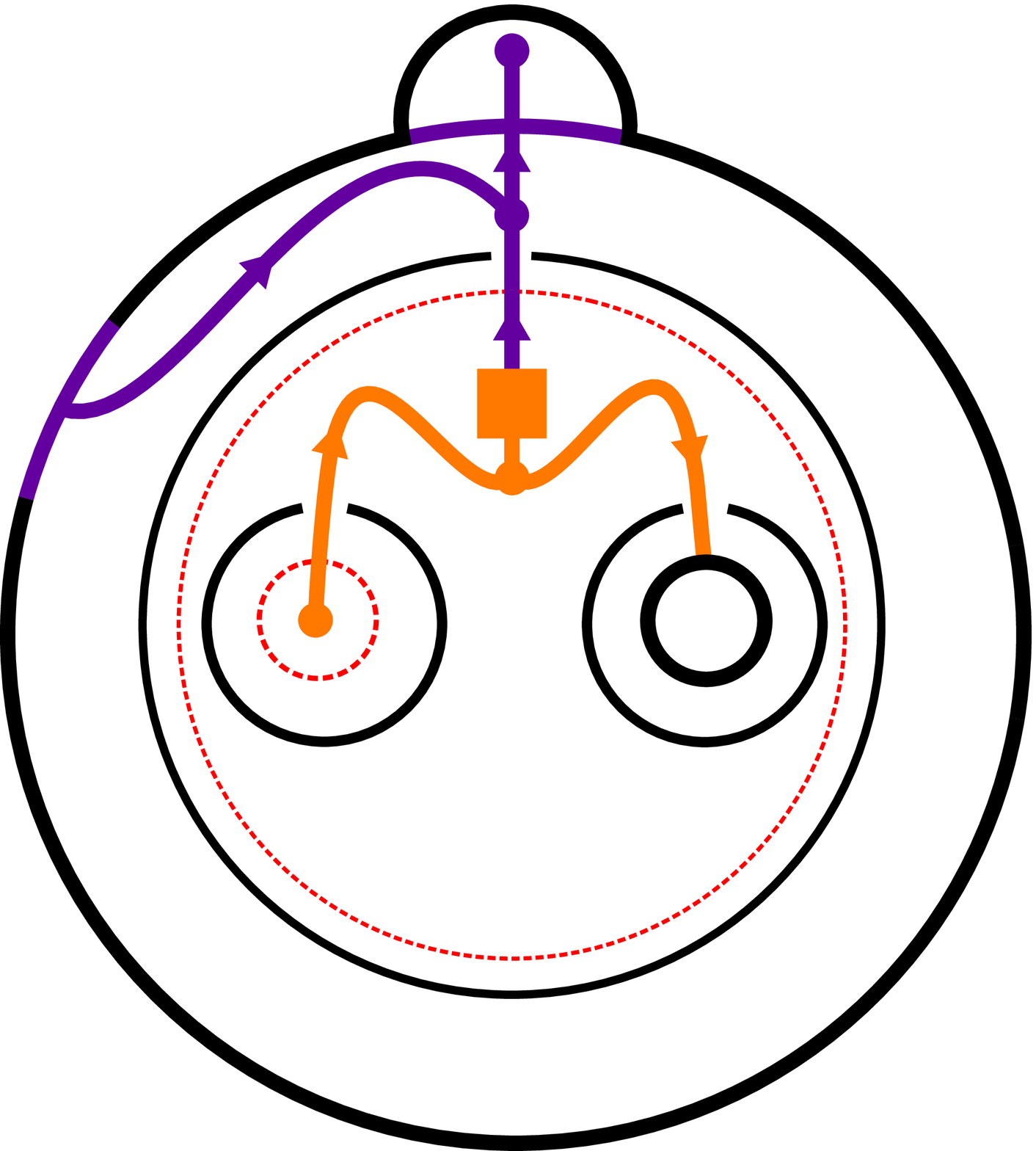}};
\node at (6,0) {\includegraphics[scale=0.2]{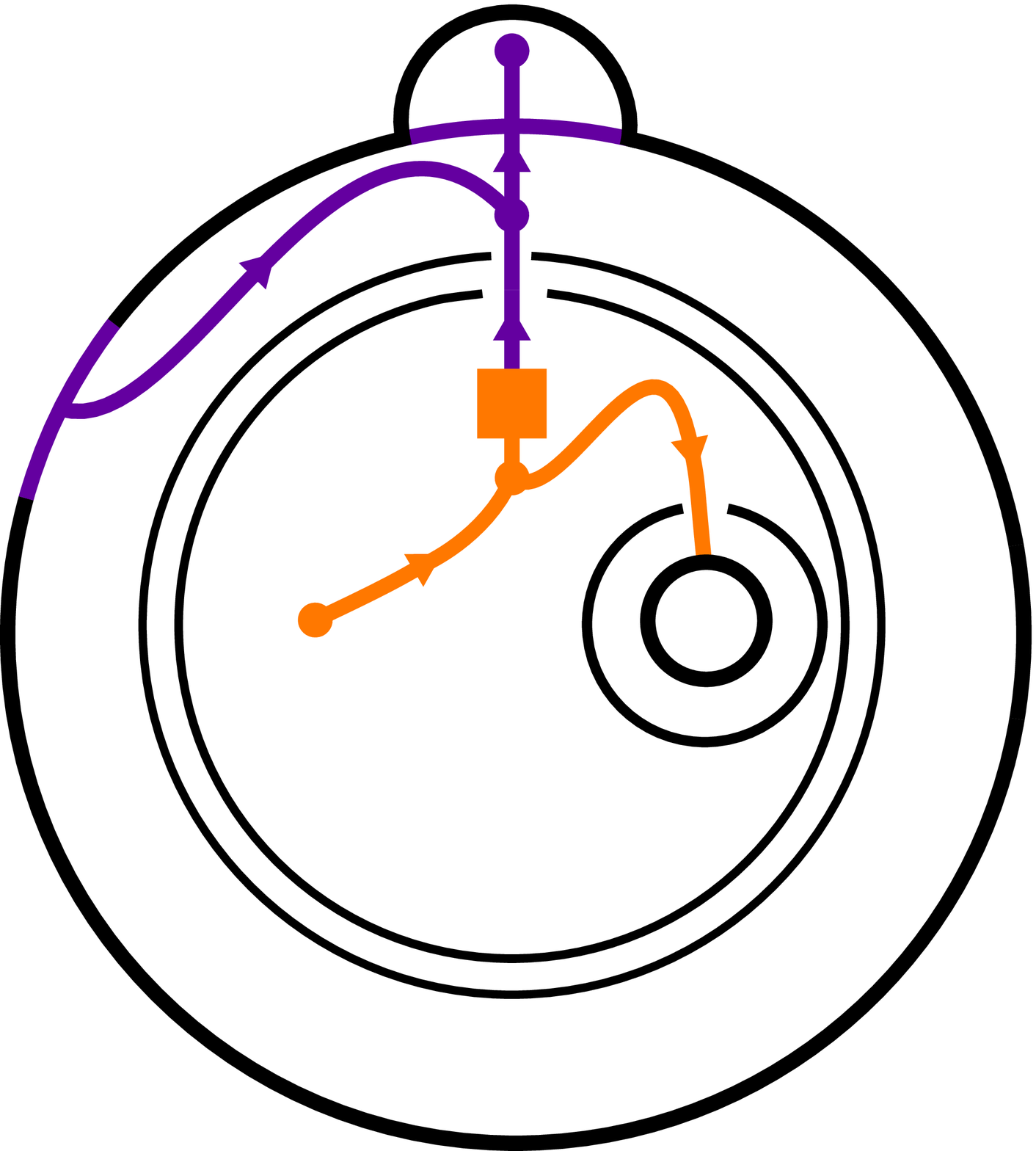}};
\node at (0,-5) {\includegraphics[scale=0.2]{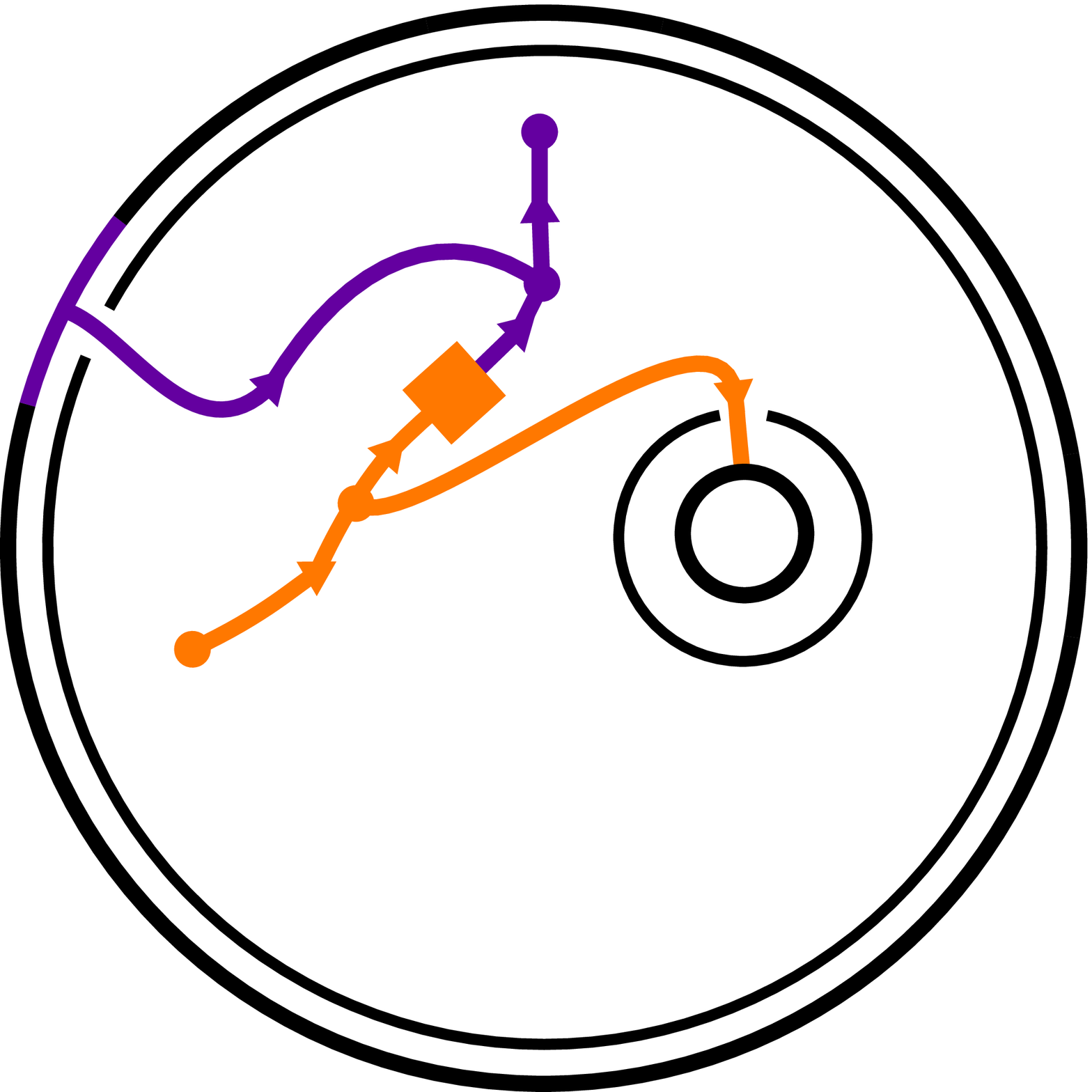}};
\node at (6,-5) {\includegraphics[scale=0.2]{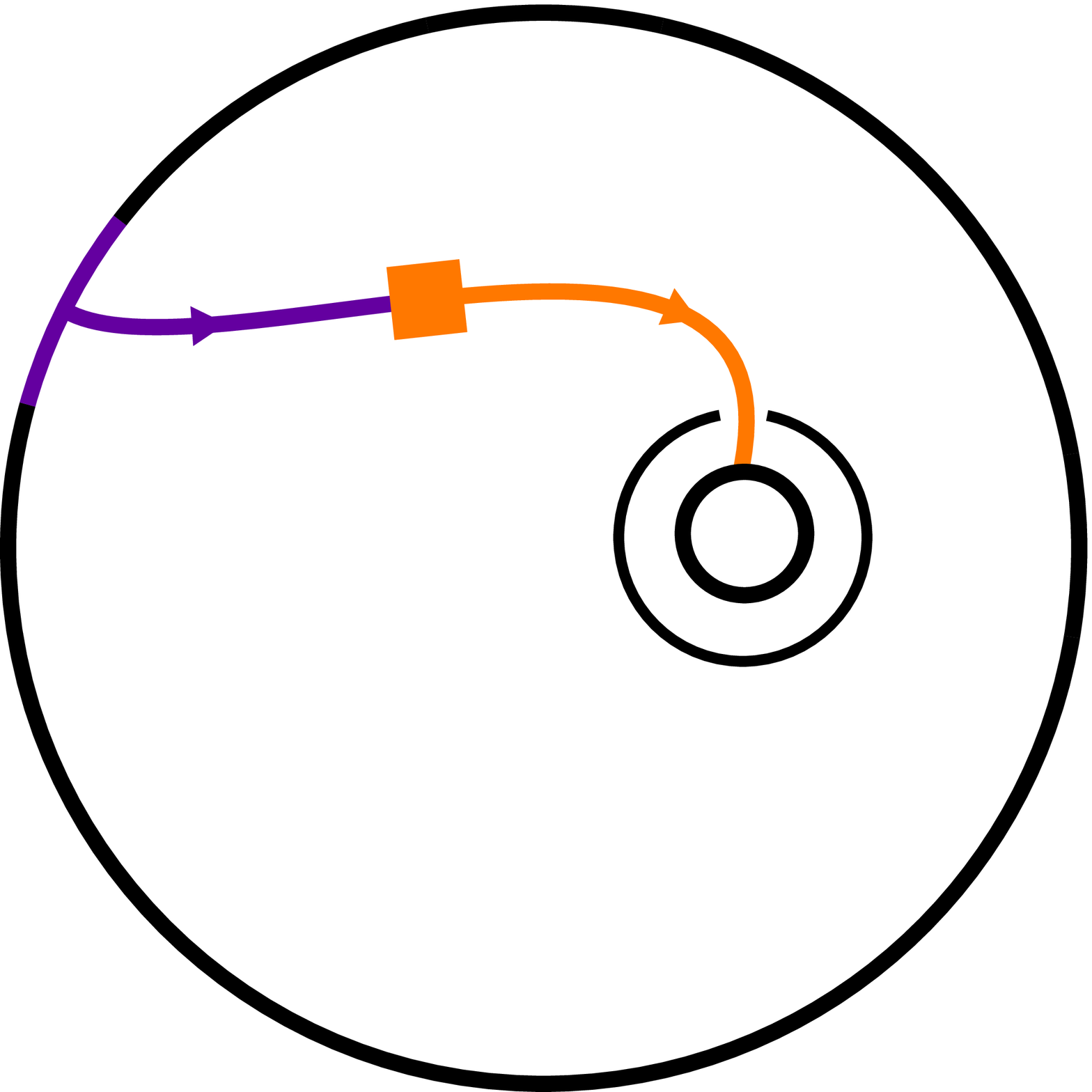}};
\node at (3,0) {$=$};
\node at (-3,-5) {$=$};
\node at (3,-5) {$=$};
\end{tikzpicture}.
\end{center}
In the first picture red dashed lines indicate where we glued world sheets. In the second equality we used the projector property and in the third equality we inserted the definition of the morphism $\iota^\dagger$.
For the Cardy condition R31) we again have to drag along projection circles:

\begin{center}
\begin{tikzpicture}
\node at (0,0) {\includegraphics[scale=0.2]{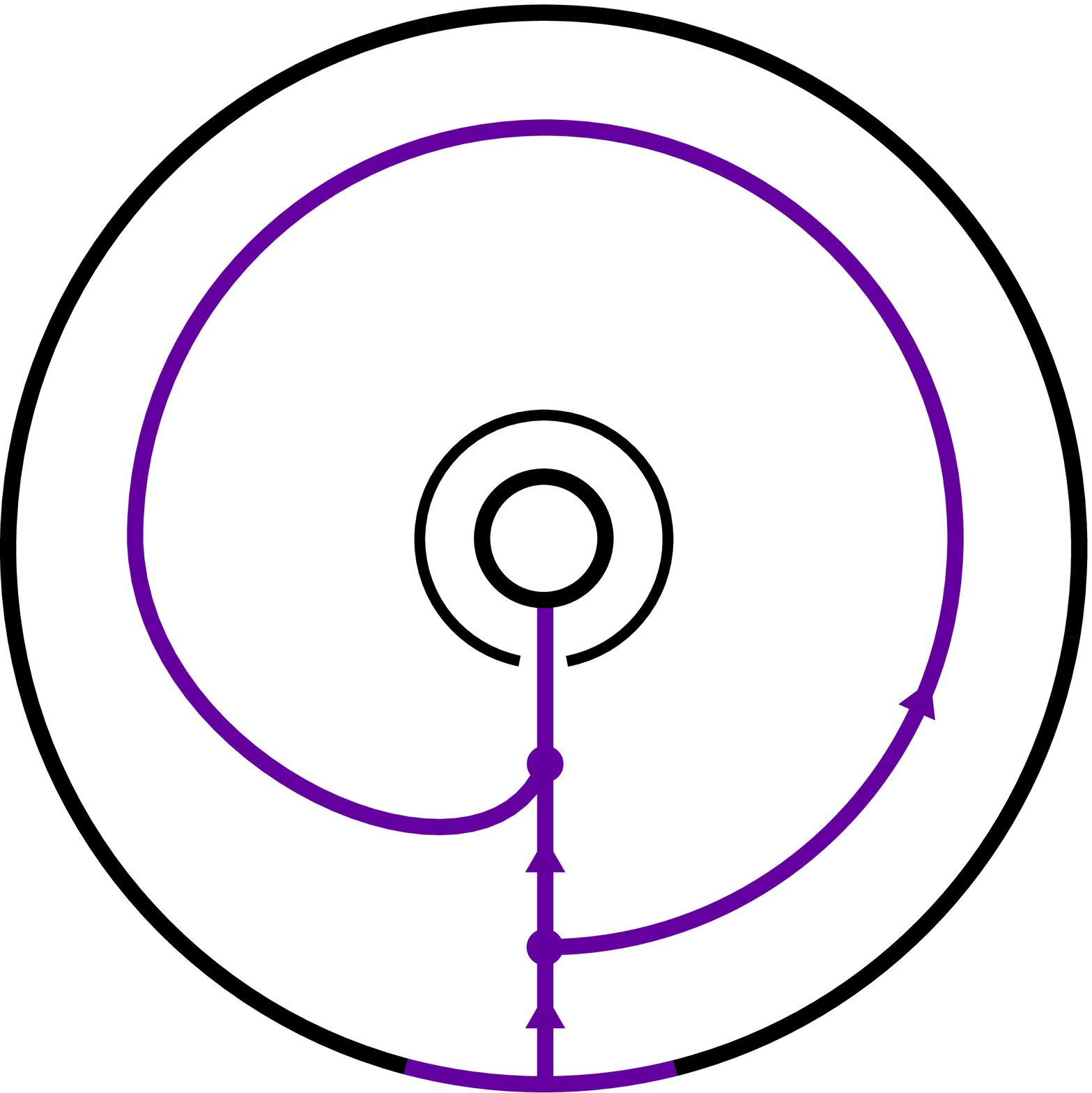}};
\node at (7,0) {\includegraphics[scale=0.2]{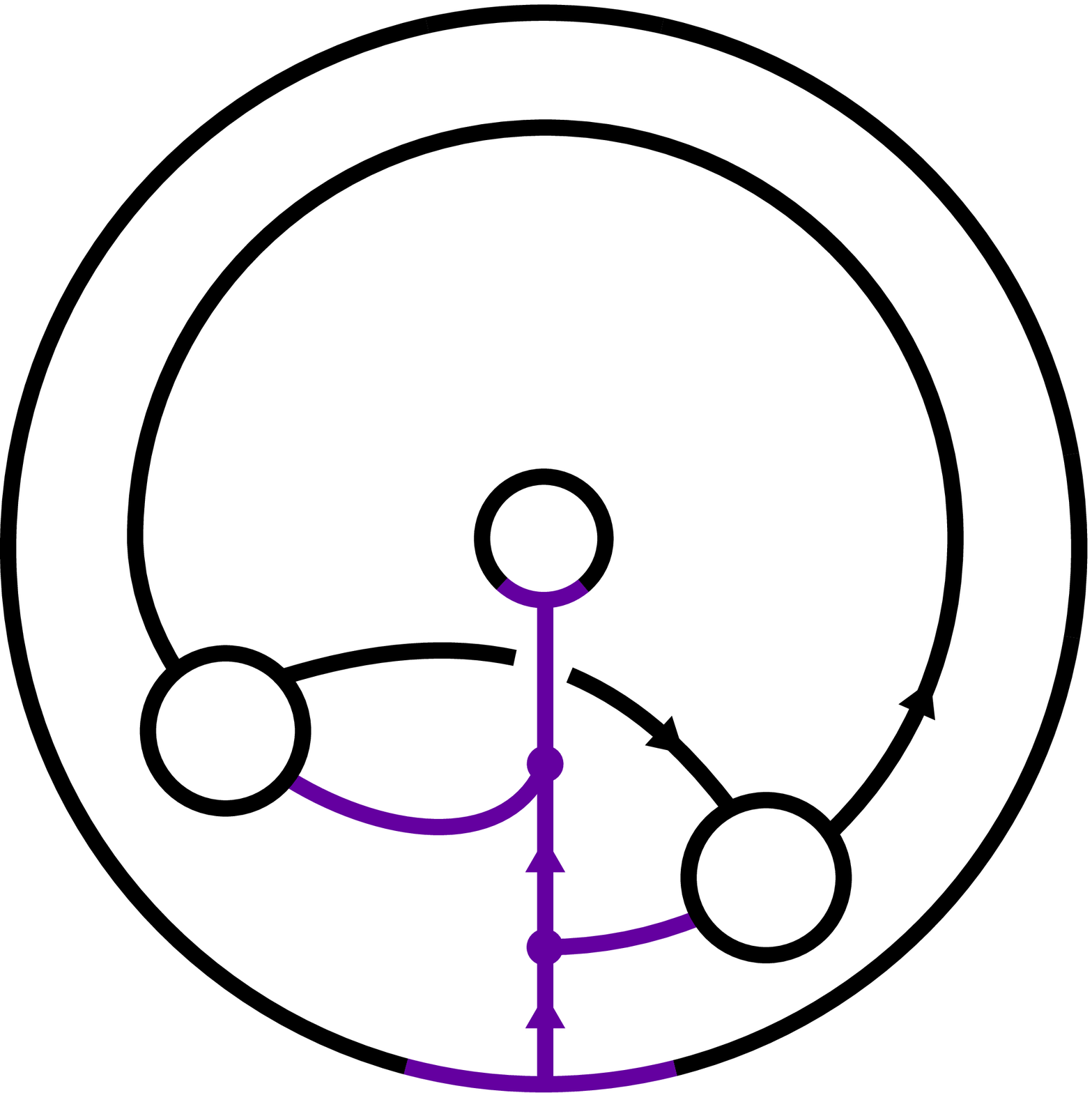}};
\node at (0,-5) {\includegraphics[scale=0.2]{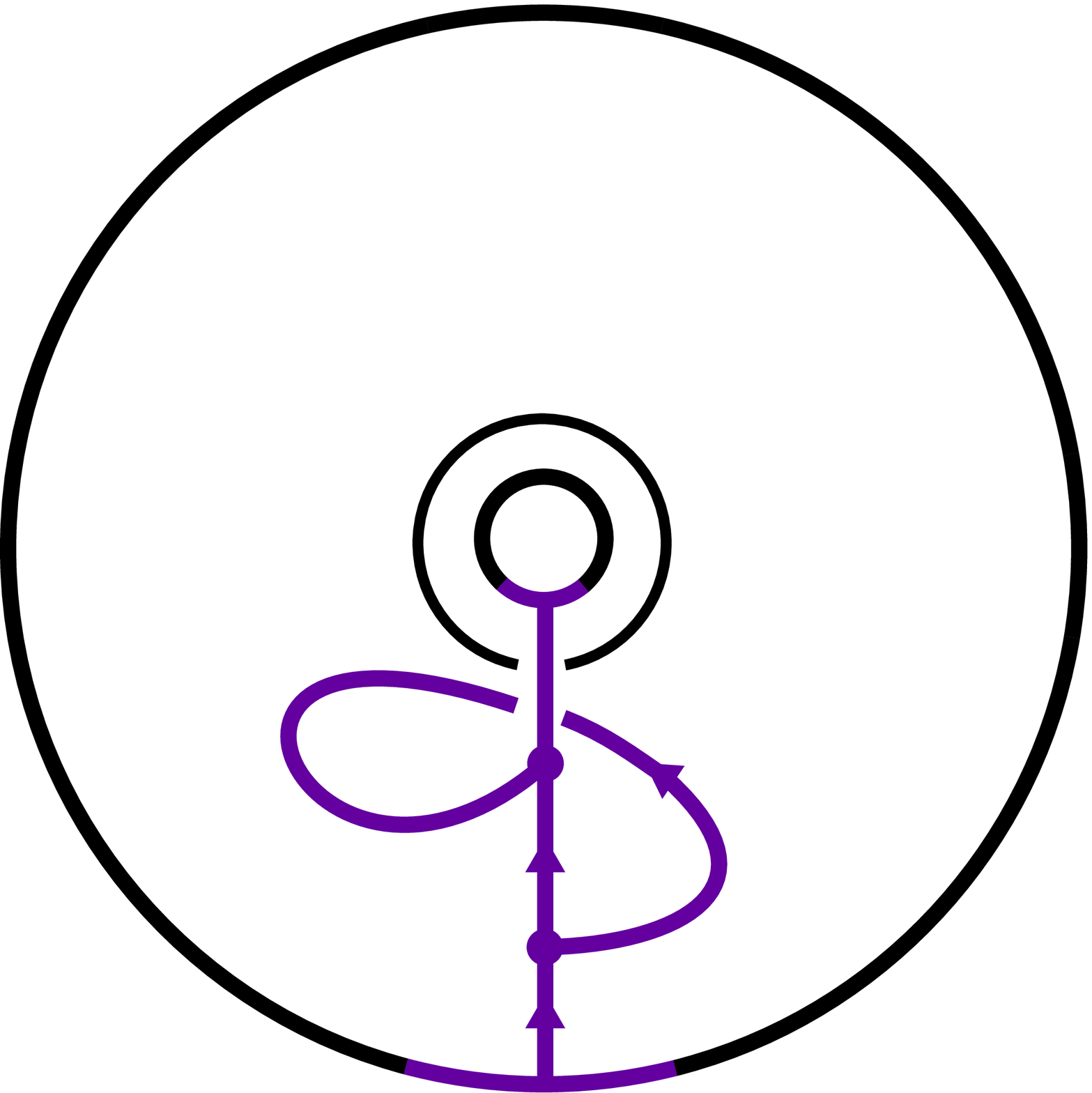}};
\node at (7,-5) {\includegraphics[scale=0.2]{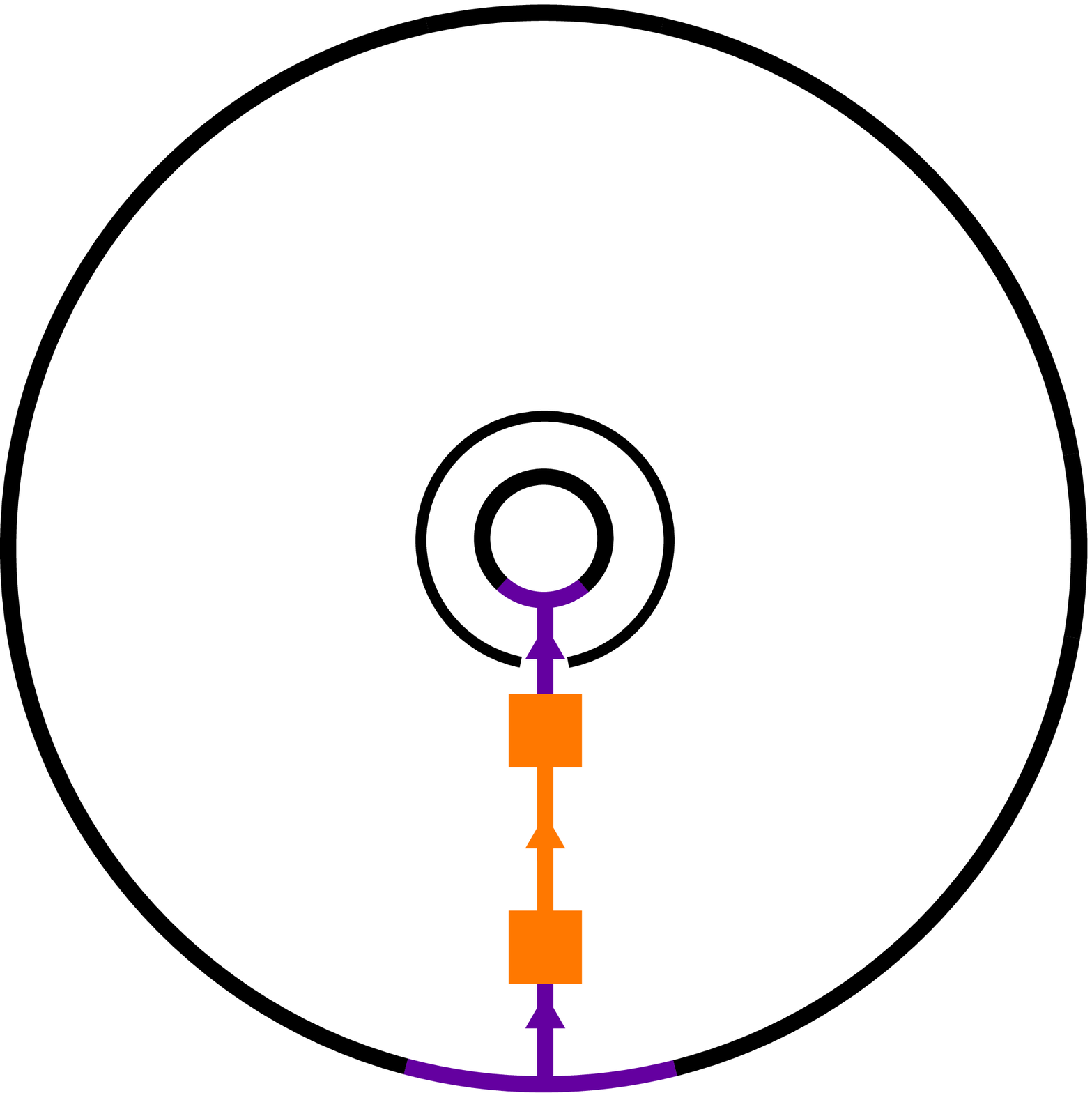}};
\node at (4,0) {$\begin{aligned}\sum_{i,k} \frac{d_id_k}{\Dsf^2}\end{aligned}$};
\node at (-3,-5) {$=$};
\node at (3.5,-5) {$=$};
\node at (3,0) {$=$};
\node at (8.2,0) {$k$};
\node at (7.6,-0.5) {$i$};
\node at (5.83,-0.69) {$\alpha$};
\node at (7.84,-1.22) {$\alpha$};
\end{tikzpicture}
\end{center}
\end{proof}

\begin{lem} The genus one one point correlator is invariant under the $S$-move.
\end{lem}
\begin{proof}
The proof of R32) is again graphical and given by the following steps.
\begin{table}[H]
\resizebox{12cm}{!}{\begin{tabular}{ m{0.5cm} m{5cm} m{0.5cm} m{5cm}}
 & 
\begin{tikzpicture}
\node at (0,0) {\includegraphics[scale=0.2]{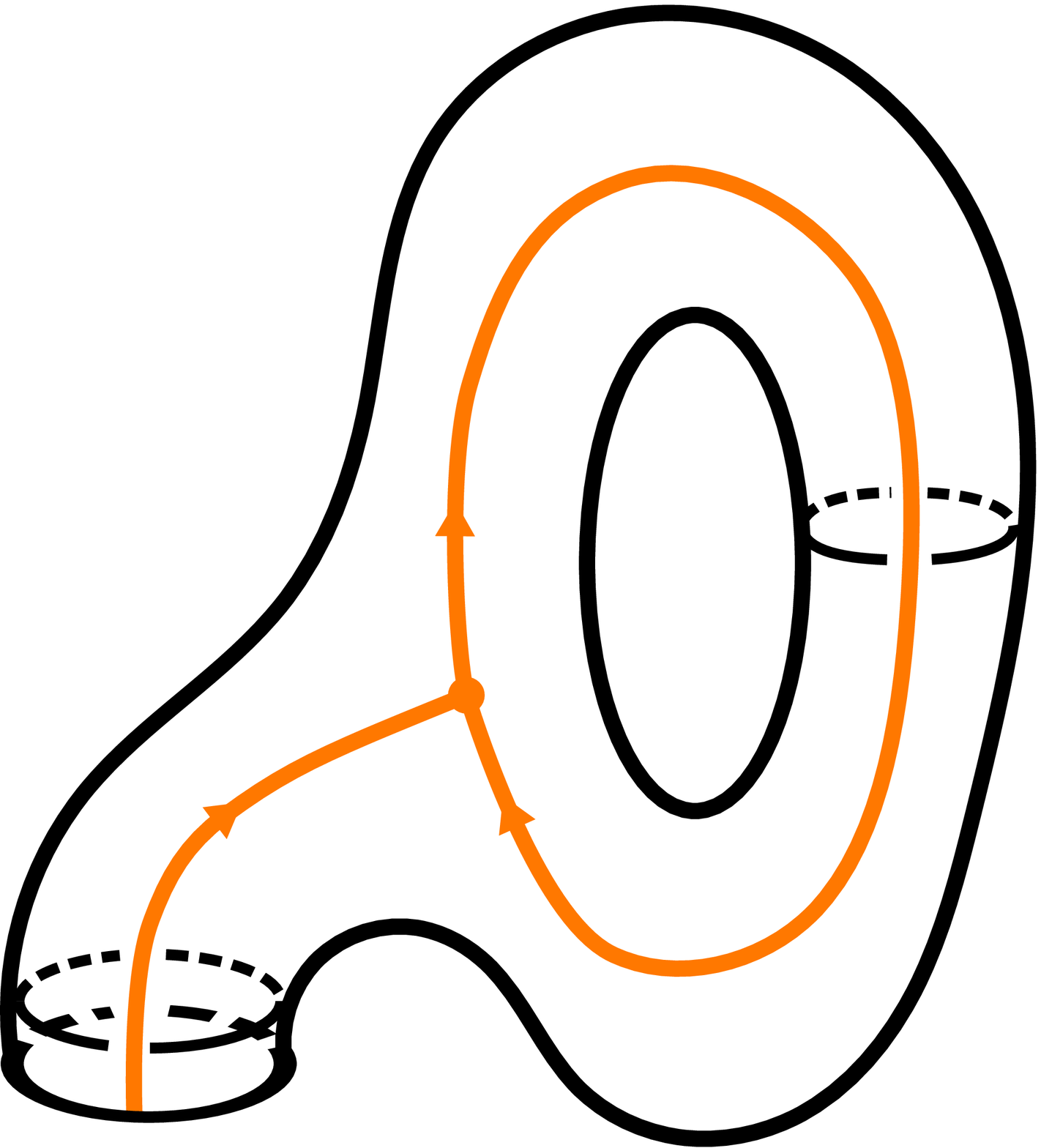}};
\end{tikzpicture} 
& = & \begin{tikzpicture}
\node at (0,0) {\includegraphics[scale=0.25]{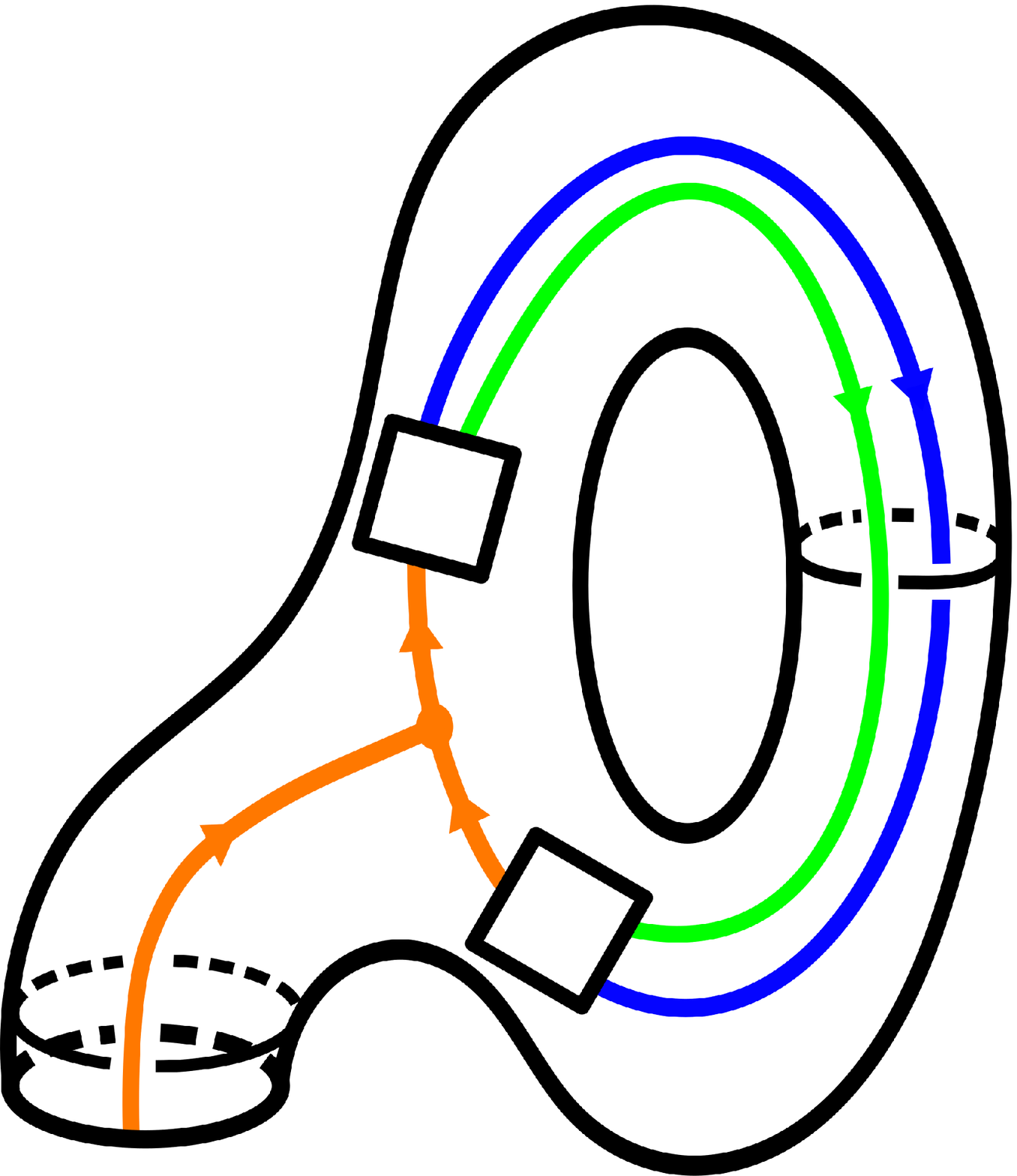}};
\node at (-2.5,0) {$\begin{aligned}\sum_{i,j\in \Isf} \end{aligned}$};
\node at (0.5,2) {$\color{blue} j$};
\node at (0.6,1.3) {$\color{green} i$};
\node at (-0.3,0.35) {$\beta$};
\node at (0.2,-1.4) {$\beta$};
\end{tikzpicture}
\end{tabular}}
\end{table}

\begin{table}[H]
\begin{tabular}{ m{0.5cm} m{5cm} m{0.5cm} m{4.5cm}}
= &
\begin{tikzpicture} 
\node at (0,0) {\includegraphics[scale=0.2]{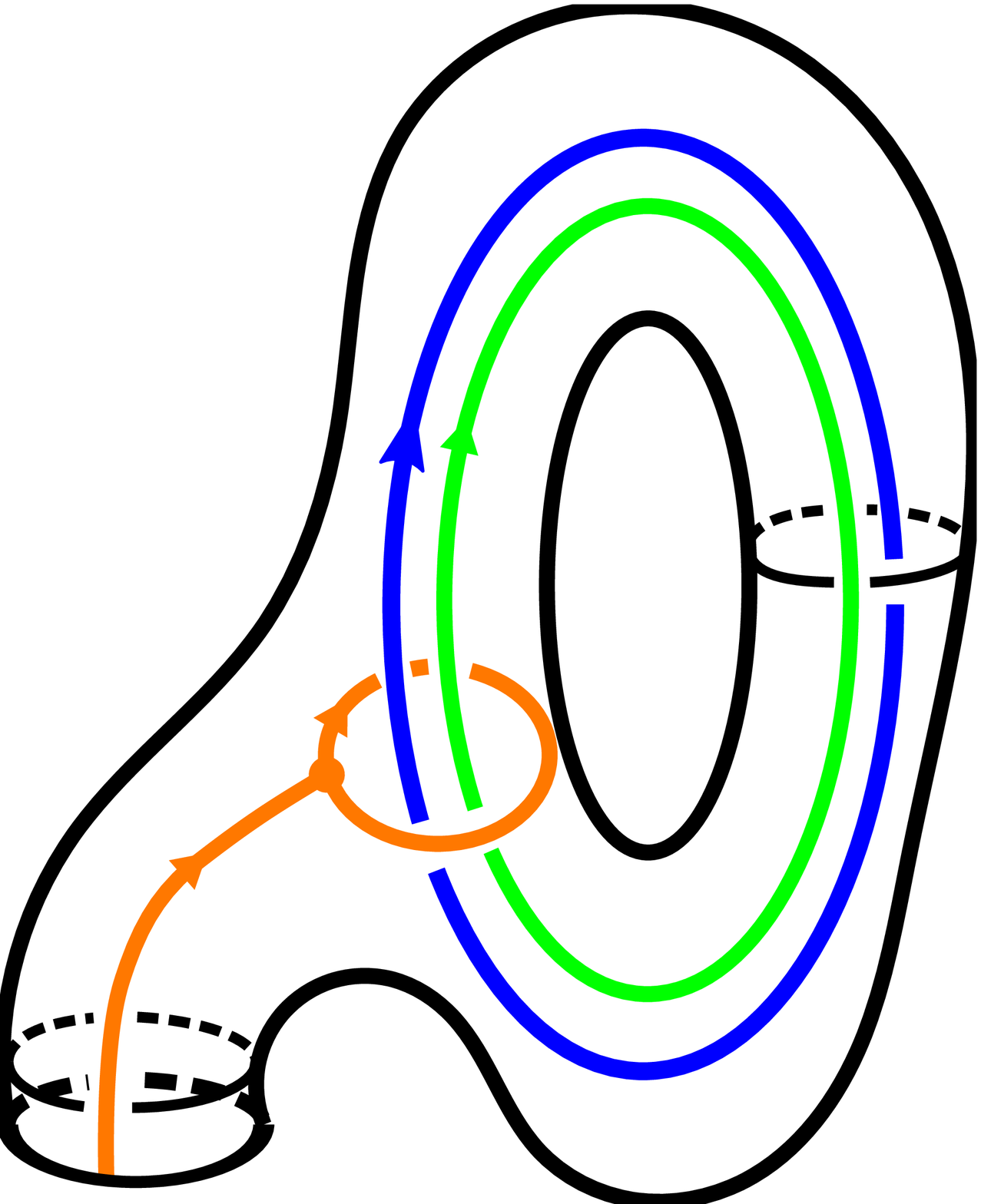}};
\node at (-2,0) {$\begin{aligned}\sum_{i,j\in \Isf} \frac{d_id_j}{\Dsf^2}\end{aligned}$};
\node at (0.5,2) {$\color{blue} j$};
\node at (0.6,1.3) {$\color{green} i$};
\end{tikzpicture}
& = & \begin{tikzpicture}
\node at (0,0) {\includegraphics[scale=0.2]{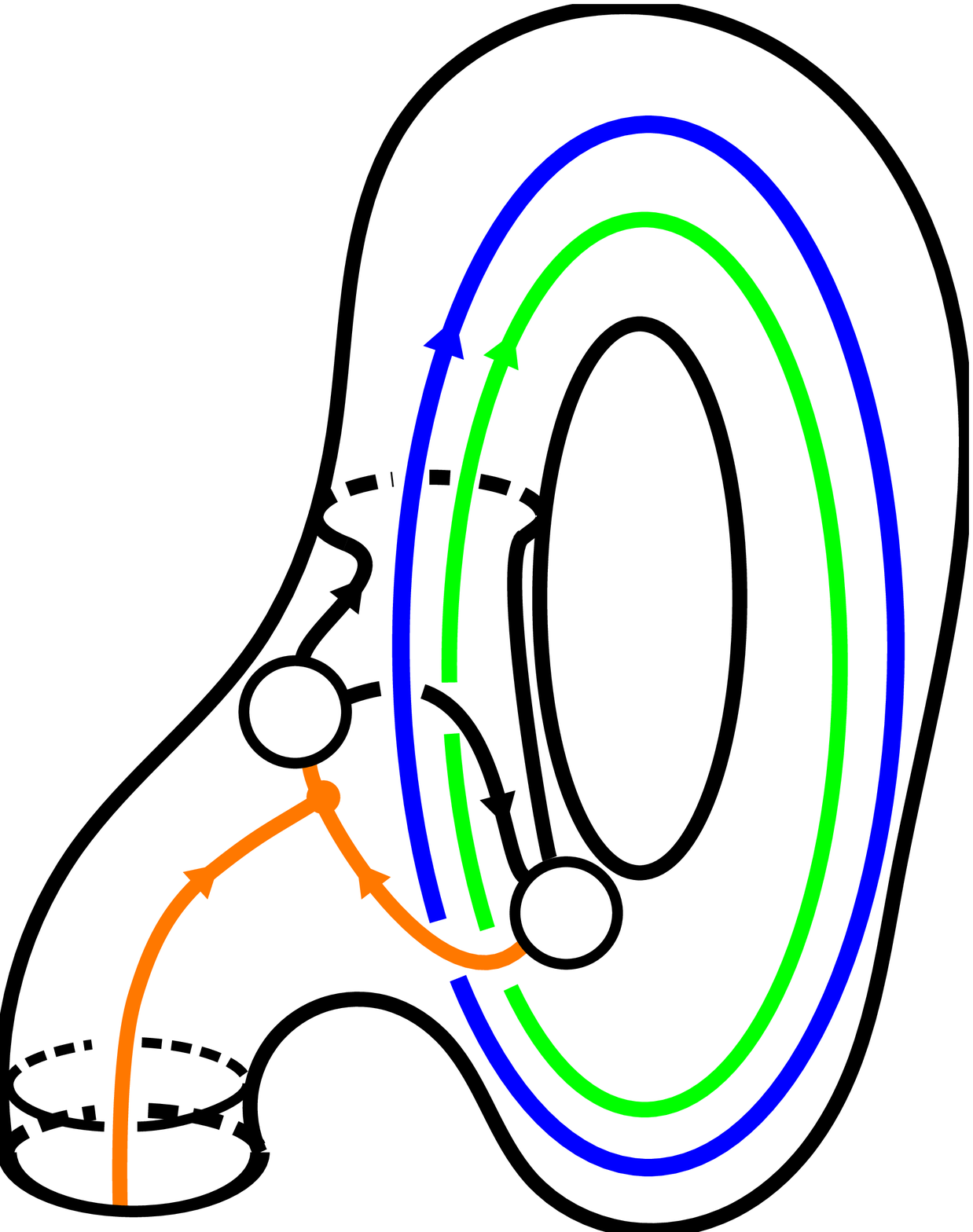}};
\node at (-3,0) {$\begin{aligned} \sum_{i,j,k,l\in \Isf} \frac{d_id_jd_kd_l}{\Dsf^4}\end{aligned}$};
\node at (0.5,2.2) {$\color{blue} j$};
\node at (0.6,1.4) {$\color{green} i$};
\node at (-0.75,-0.38) {$\alpha$};
\node at (0.35,-1.2) {$\alpha$};
\node at (0,-0.25) {$l$};
\node at (0,0.2) {$k$};
\end{tikzpicture}
\end{tabular}
\end{table}

\begin{table}[H]
\begin{tabular}{ m{0.5cm} m{5cm} m{0.5cm} m{5cm}}
=& \begin{tikzpicture}
\node at (0,0) {\includegraphics[scale=0.2]{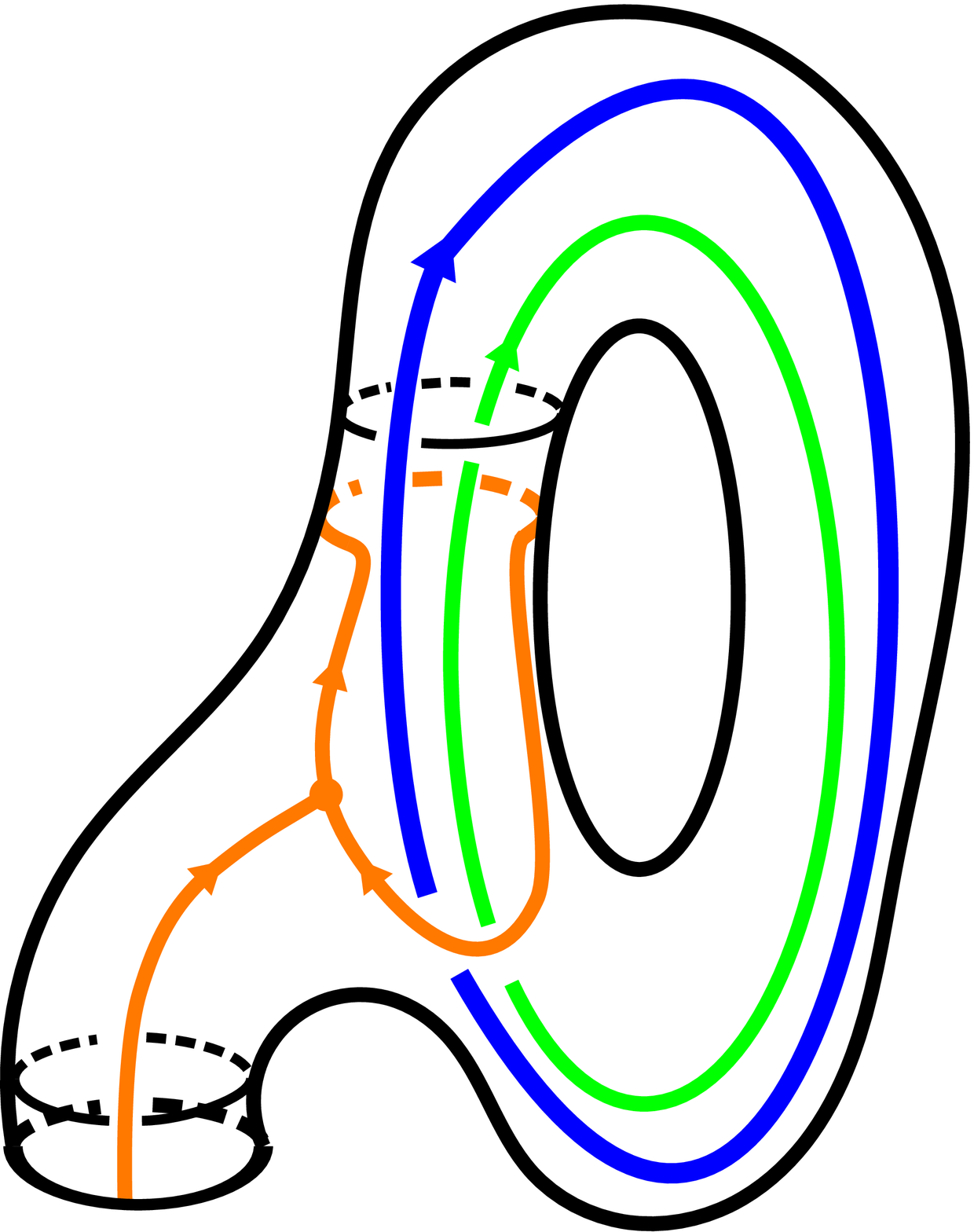}};
\node at (-2,0) {$\begin{aligned}\sum_{i,j\in \Isf}\frac{d_id_j}{\Dsf^2}\end{aligned}$};
\node at (0.2,2.2) {$\color{blue} j$};
\node at (0.6,1.4) {$\color{green} i$};
\end{tikzpicture}
& = & \begin{tikzpicture}
\node at (0,0) {\includegraphics[scale=0.2]{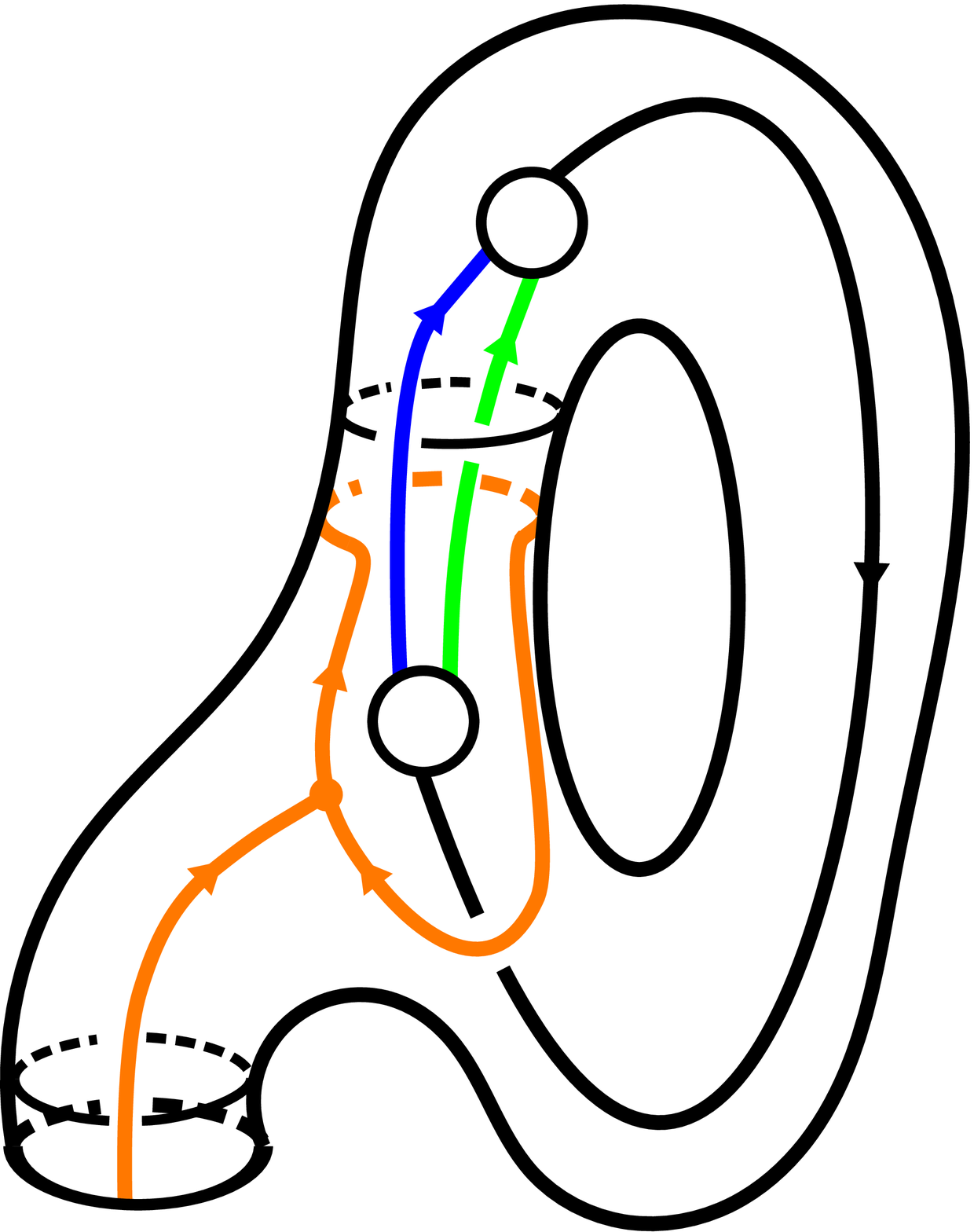}};
\node at (-3,0) {$\begin{aligned}\sum_{i,j,r\in \Isf}\frac{d_id_jd_r}{\Dsf^2}\end{aligned}$};
\node at (-0.3,1.5) {$\color{blue} j$};
\node at (0.3,1.2) {$\color{green} i$};
\node at (1.3,0) {$r$};
\node at (0.2,1.6) {$\alpha$};
\node at (-0.25,-0.4) {$\alpha$};
\end{tikzpicture}
\end{tabular}
\end{table}

\begin{table}[H]
\begin{tabular}{ m{0.5cm} m{5cm} m{0.5cm} m{5cm}}
= & \begin{tikzpicture}
\node at (0,0) {\includegraphics[scale=0.2]{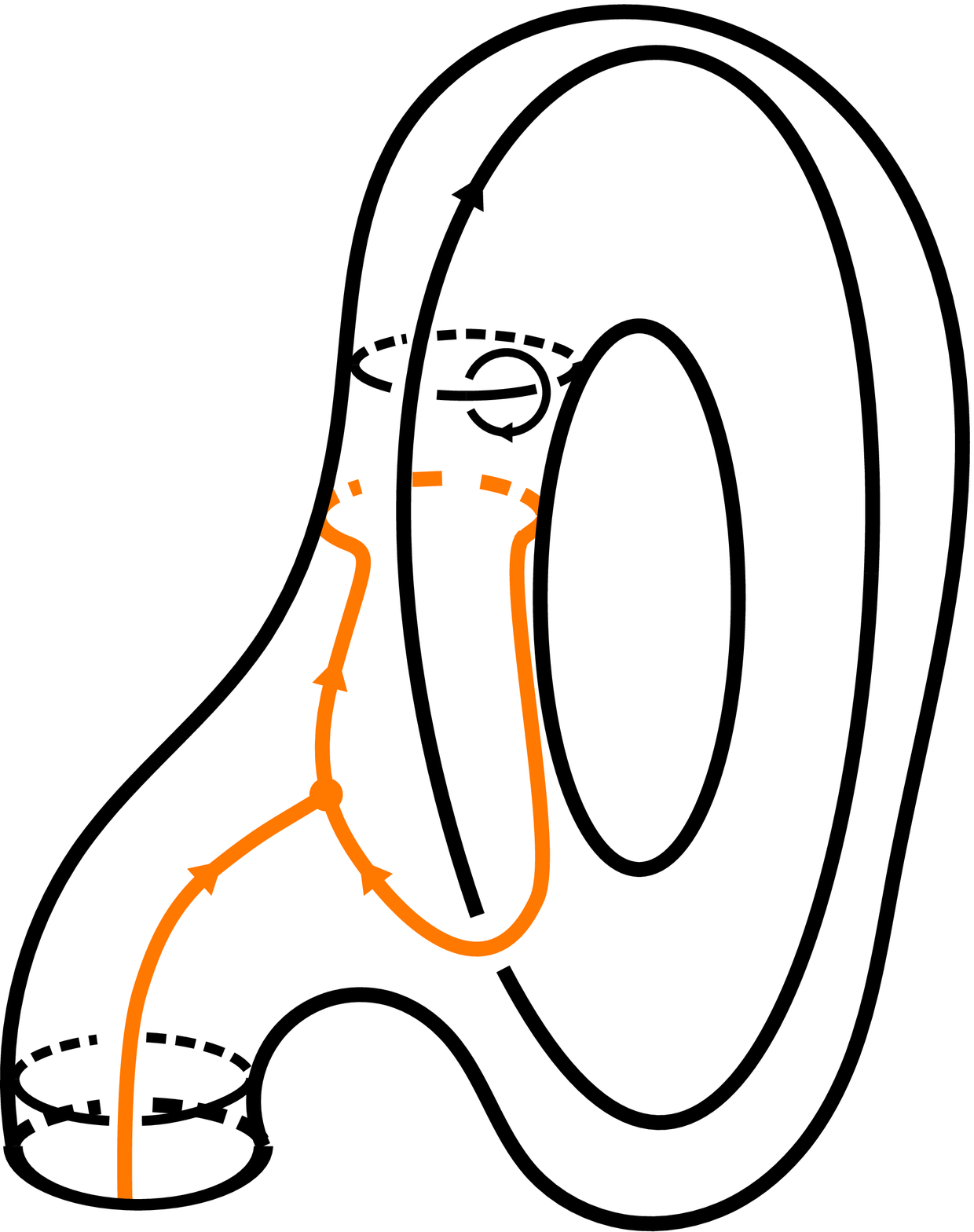}};
\node at (-2,0) {$\begin{aligned}\sum_{r,i\in \Isf} \frac{d_id_r}{\Dsf^2}\end{aligned}$};
\node at (1.4,0) {$r$};
\node at (-0.1,0.7) {$i$};
\end{tikzpicture}
& = & \begin{tikzpicture}
\node at (7,-18) {\includegraphics[scale=0.2]{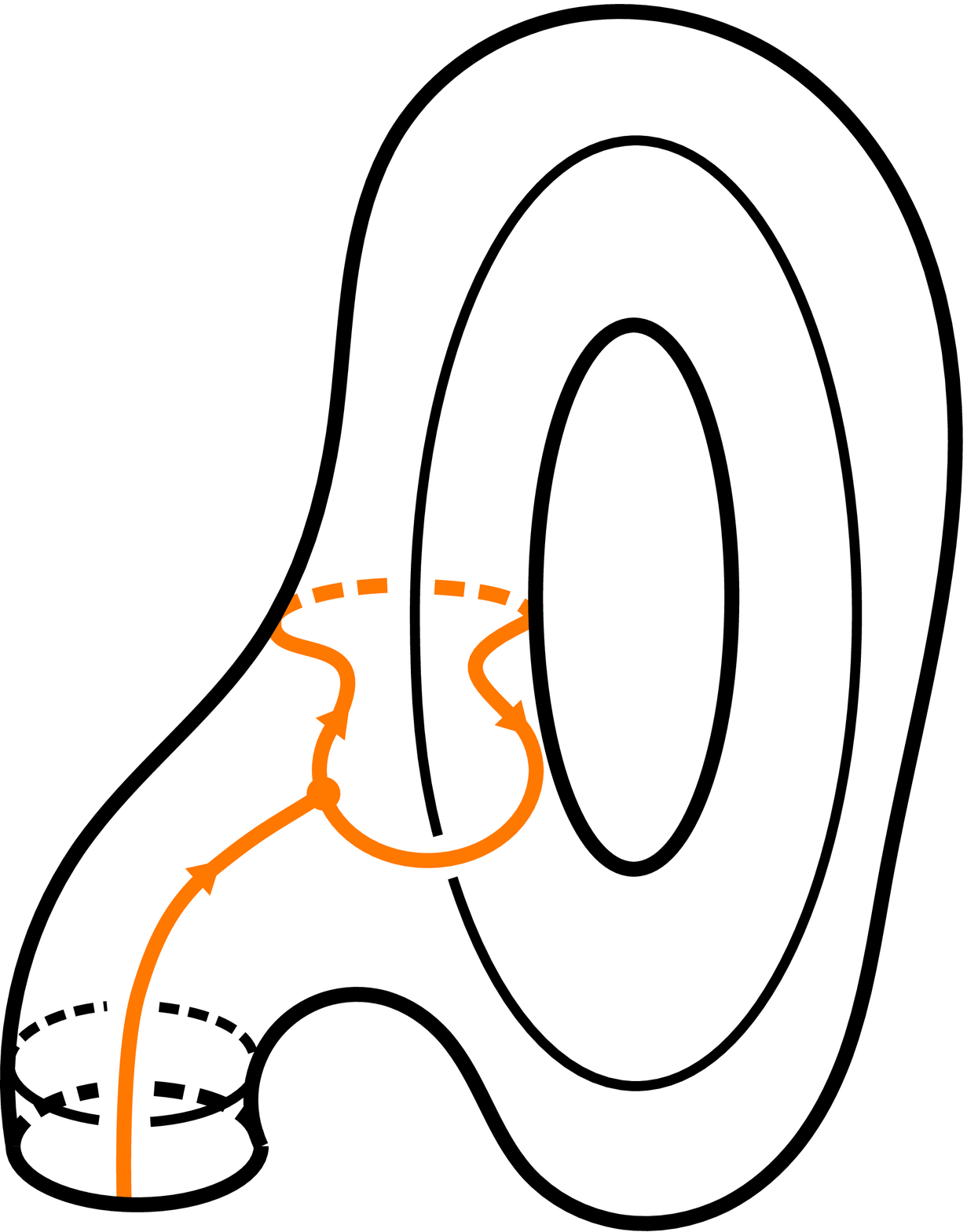}};
\end{tikzpicture}
\end{tabular}
\end{table}

The first equality uses the complete basis for elements in the Drinfeld-center but in a different normalization. Graphically this is indicated by using squares instead of round coupons. Note that there is no extra factor for the quantum dimension in this case. The second step is the modular property for $\Hcalcl$. In the third step we transported the projector circle along the torus and inserted the completeness relation in order to drag the $\Hcalcl$-colored curve along the circle in the fourth step. Using again completeness and finally lemma \ref{circleline} yields the result.
\end{proof}

This completes the proof of Theorem \ref{maintheo1}.


\subsection{From Sewing Constraints on String-Net Spaces to Cardy Algebras}

In the previous section we defined a fixed set of fundamental correlators and showed that the properties of a Cardy algebra leads to a solution of the sewing constraints for these fundamental correlators. In this section we go the other way round, i.e. we assume that a solution to the sewing constraints for the functor $\Blcal$ exists and show that this gives in fact a $(\Csf|\Zsf(\Csf))$-Cardy algebra. A string-net on a surface of genus $g$ with $n$ boundary components can have finitely many connected components winding non-contractible 1-cycles on the surface. However, on projector decorated surfaces things simplify considerably. Recall that the first homology group of a compact surface $S_{g,n}$ of genus $g$ with $n$ boundary components has $2g+n$ generators. The first $2g$-generators are the usual $a$- and $b$-cycle running around holes of tori. The other $n$ generators are simple closed curves homotopic to the boundary (see figure \ref{homologysurface}.).

\begin{figure}[H]
\centering
\includegraphics[scale=0.2]{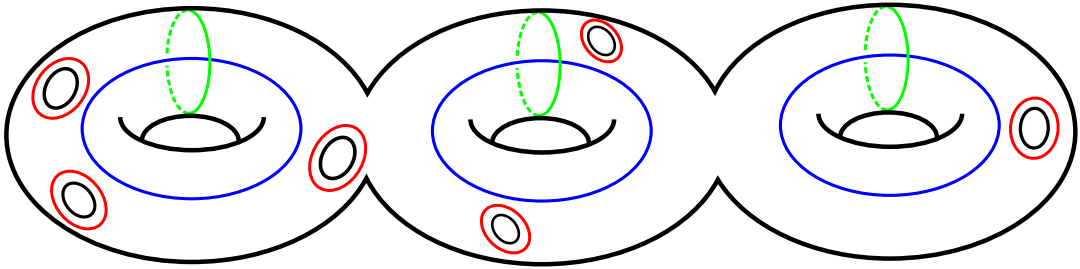}
\caption{A genus 3 surface $S_{3,6}$ with a-cycles respectively b-cycles shown in green and blue. Red circle show boundary generators in $H_1(S_{3,6})$.}
\label{homologysurface}
\end{figure}

\begin{prop}\label{trivialbdycompo} Let $S_{g,n}$ be a compact surface of genus $g$ with $n$ boundary components and $A_1,\dots,A_n\in \Zsf(\Csf)$. Any element in $\hat{H}^s(S_{g,n},A_1,\dots, A_n)$ is equivalent to a string-net with trivial winding around boundary generators of $H_1(S_{g,n})$.
\end{prop}

\begin{Cor} For $S_{0,n}$ a sphere with $n$ boundary components any string-net in \linebreak $\hat{H}^{s}(S_{0,n},A_1,\dots, A_n)$ is equivalent to a string-net with a single coupon. 
\end{Cor}

\begin{proof}
This is the same argument as \cite[Lemma~3.7]{Schweigert:2019zwt}. Assume a string-net has non-trivial winding along a boundary component. In an annular neighborhood of the boundary the string-net can be manipulated as follows
\begin{center}

\begin{tikzpicture}
\node at (0,0) {\includegraphics[scale=0.2]{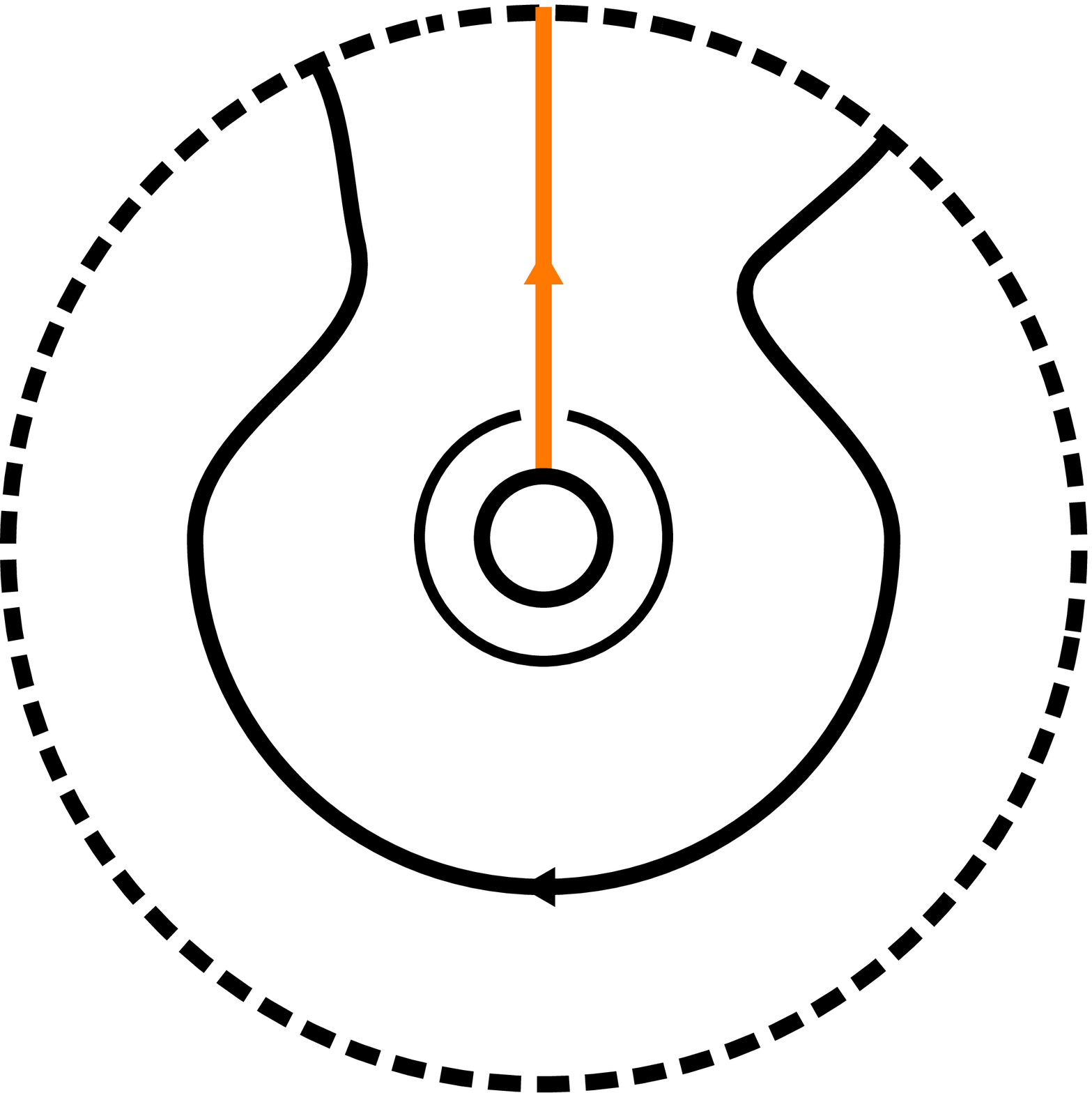}};
\node at (7,0) {\includegraphics[scale=0.2]{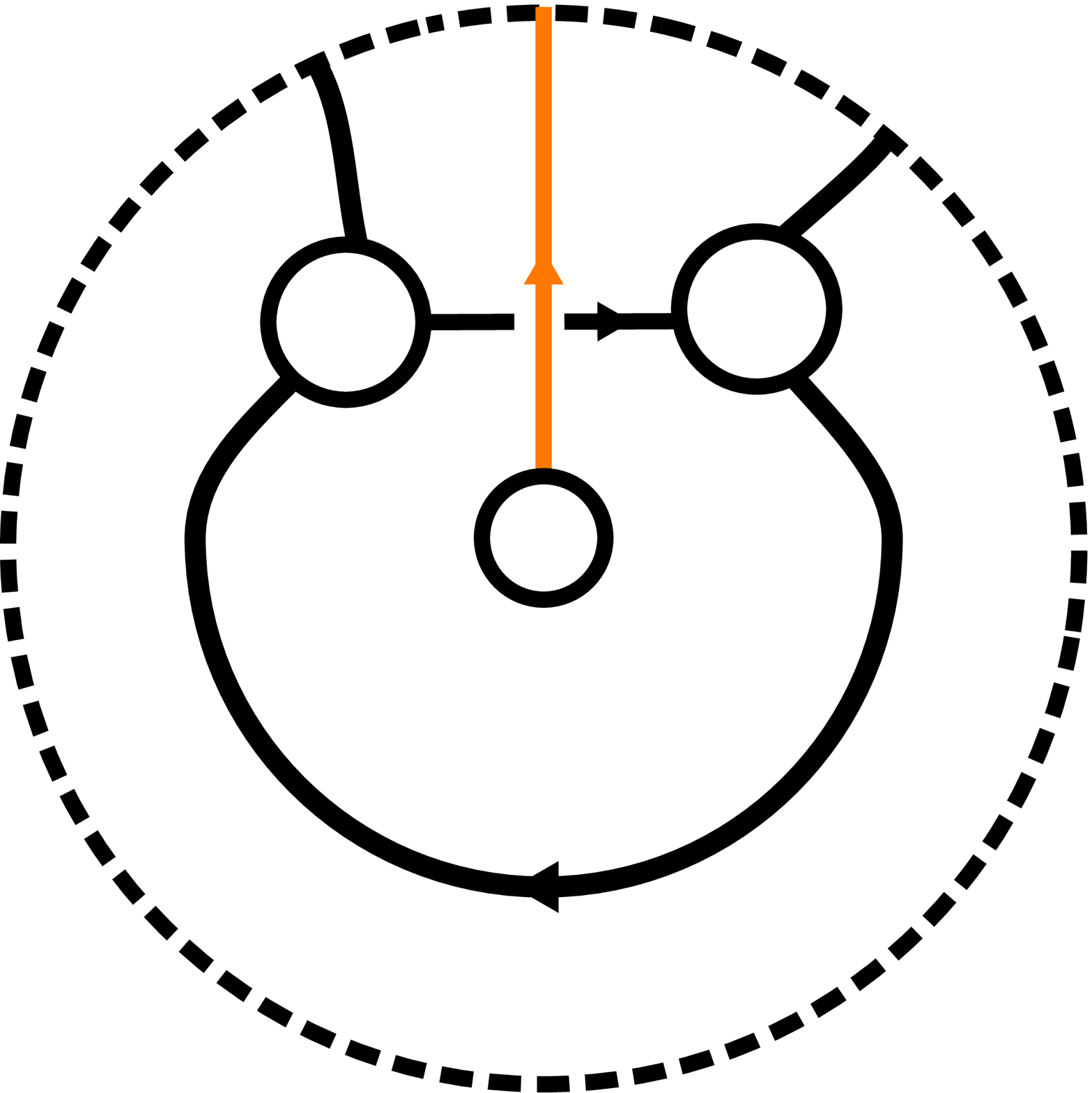}};
\node at (3.5,-5) {\includegraphics[scale=0.2]{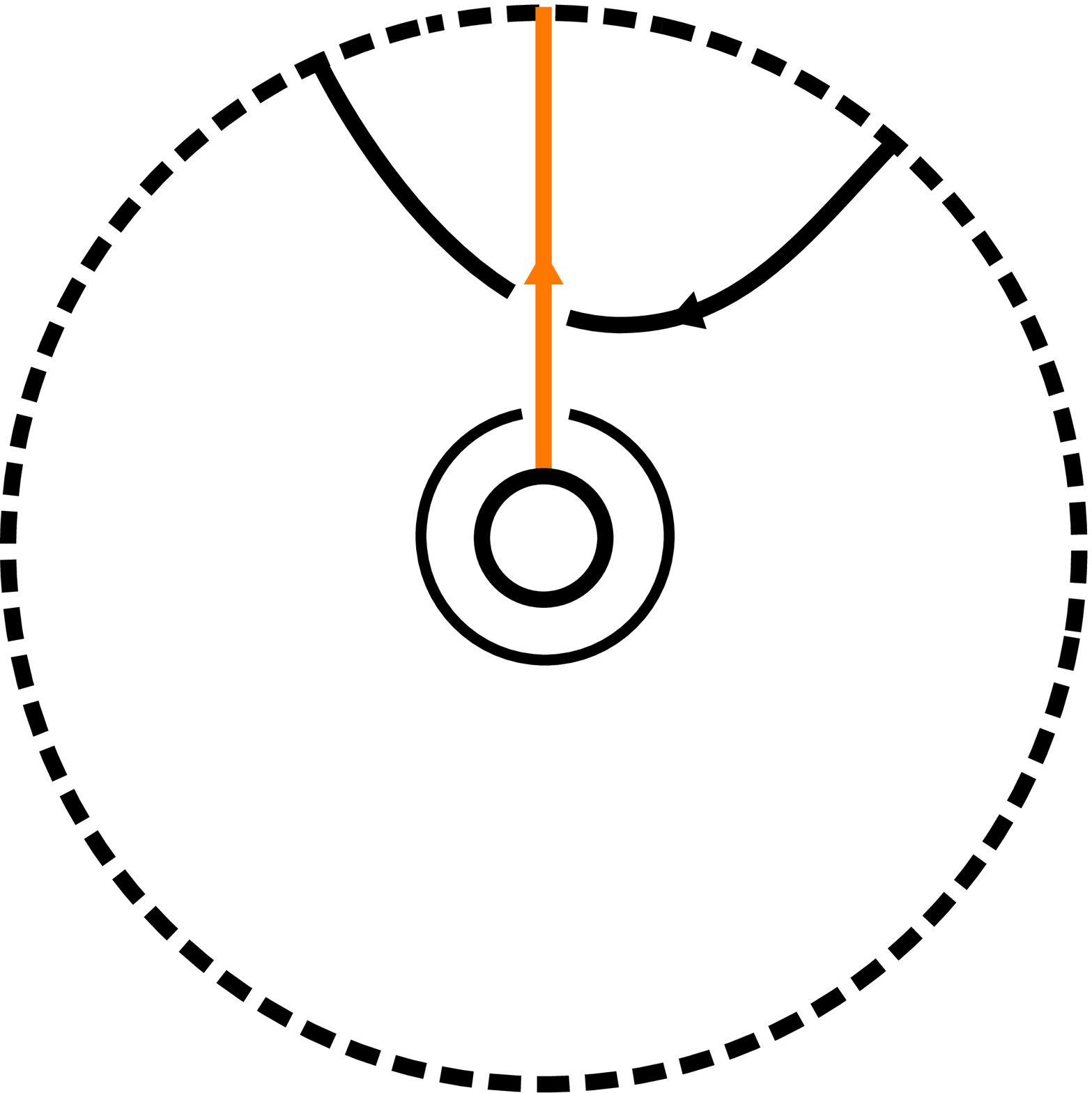}};
\node at (3,0) {$=$};
\node at (4,0) {$\begin{aligned}\sum_{k,l} \frac{d_kd_l}{\Dsf^2}\end{aligned}$};
\node at (6.25,0.85) {$\alpha$};
\node at (7.8,0.85) {$\alpha$};
\node at (6.8,0.6) {$k$};
\node at (5.8,-1) {$l$};
\node at (1,-5) {$=$};
\end{tikzpicture}.
\end{center}
Orientation of curves is chosen arbitrarily in the picture. For any other orientation the computation is exactly the same.
\end{proof}

The proposition and its corollary imply that any string-net on a generating world sheet is of the form shown in Figures \ref{openfundamentalcorr}, \ref{closedfundamentalcorr}, \ref{openclosedfundamentalcorr} where disk-shaped vertices are now fixed morphisms of the right type. Assume a boundary coloring by the closed object $\widehat{\Gcalcl}$ and open objects $\widehat{\Gcalop},\, L(\widehat{\Gcalop})$, i.e. closed boundaries of a world sheet have boundary value $\widehat{\Gcalcl}$ and open ones $\widehat{\Gcalop},\, L(\widehat{\Gcalop})$, depending on the type of world sheet. For example, on the world sheet $C_m$ the vertex corresponds to a morphism $\widehat{m}_{cl}:\widehat{\Gcalcl}\otimes \widehat{\Gcalcl}\rightarrow \widehat{\Gcalcl}$ in $\Zsf(\Csf)$. On $C_\Delta$ the vertex is a morphism $\widehat{\Delta}_{cl}:\widehat{\Gcalcl}\rightarrow \widehat{\Gcalcl}\otimes \widehat{\Gcalcl}$ and so on. We have to take special care of world sheets $O_{prop}$ and $C_{prop}$. Those give rise to maps 
\eq{
\hat{p}_{op}\in \hom_{\Zsf(\Csf)}(\widehat{\Gcalop},\widehat{\Gcalop}),\qquad \hat{p}_{cl}\in \hom_{\Zsf(\Csf)}(\widehat{\Gcalcl},\widehat{\Gcalcl}).
}
Assuming the sewing constraints hold it readily follows that $\hat{p}_{op}$ and $\hat{p}_{cl}$ are idempotent maps. In the previous discussion these maps were fixed to be the identity maps. Thus we may make the additional assumption that $\hat{p}_{op}$, $\hat{p}_{cl}$ are invertible, which by finiteness of the morphism spaces implies that $\hat{p}_{op}=\id$ and $\hat{p}_{cl}=\id$. But we don't have to. Since $\Zsf(\Csf)$ is abelian we can choose a retract $(\Gcalcl,e_{cl},r_{cl})$ for $\hat{p}_{cl}$, i.e. 
\eq{
e_{cl}:\Gcalcl\rightarrow \widehat{\Gcalcl},\quad  r_{cl}:\widehat{\Gcalcl}\rightarrow \Gcalcl\\
e_{cl}\circ r_{cl}=\hat{p}_{cl},\quad r_{cl}\circ e_{cl}=\id_{\Gcalcl}
}
and similar a retract $(\Gcalop,e_o,r_o)$ for $p_o$ in $\Csf$. On their images $\Gcalcl$, $\Gcalop$ the propagator morphisms act as the identity and sewing constraints realize a $(C|\Zsf(\Csf))$-Cardy algebra on $\Gcalcl$, $\Gcalop$ rather than on $\widehat{\Gcalcl}$, $\widehat{\Gcalop}$.
\begin{theo}\label{maintheo3}
Any set of fundamental string-nets on generating world sheets with closed boundary values $\widehat{\Gcalcl}$ and open boundary values $\widehat{\Gcalop},\,L(\widehat{\Gcalop})$, which satisfy the sewing constraints, defines a $(\Csf|\Zsf(\Csf))$-Cardy algebra $(\Gcalcl,\Gcalop,\iotaclop)$, which is unique up to isomorphism.
\end{theo}

\begin{proof}
As discussed above from fundamental world sheets we get the following ten maps

\begin{table}[h]
\centering
\begingroup
\setlength{\tabcolsep}{10pt}
\begin{tabular}{c |c| c| c| c |c| c| c| c| c}
$O_m$ & $O_\Delta$ & $O_\eta$ & $O_\epsilon$ & $C_m$ & $C_\Delta$ & $C_\eta$ & $C_\eta$ & $I$ & $I^\dagger$ \\
\hline 
$m_{op}$ & $\Delta_{op}$ & $\eta_{op}$ & $\epsilon_{op}$ & $m_{cl}$ & $\Delta_{cl}$ & $\eta_{cl}$ & $\epsilon_{cl}$&  $ \iota$ & $ \iota^\dagger$ 
\end{tabular}
\endgroup
\caption{The first row states the type of world sheet and the second row the corresponding maps.}
\end{table}
where the morphisms are
\eq{
\mop&=r_o\circ \widehat{\mop}\circ (e_o\otimes e_o):\Gcalop\otimes \Gcalop\rightarrow \Gcalop\\
\Dop&=(r_o\otimes r_o)\circ \widehat{\Dop}\circ e_o:\Gcalop\rightarrow \Gcalop\otimes \Gcalop\\
\etaop&= r_o\circ \widehat{\etaop}:\mathbf{1}\rightarrow \Gcalop\\
\epop&= \widehat{\epop}\circ e_o:\Gcalop\rightarrow \mathbf{1}\\
\mcl&=r_{cl}\circ \widehat{\mcl}\circ (e_{cl}\otimes e_{cl}):\Gcalcl\otimes \Gcalcl\rightarrow \Gcalcl\\
\Dcl&=(r_{cl}\otimes r_{cl})\circ \widehat{\Dcl}\circ e_{cl}:\Gcalcl\rightarrow\Gcalcl\otimes \Gcalcl\\
\etacl&=r_{cl}\circ \widehat{\etacl}:\mathbf{1}\rightarrow \Gcalcl\\
\epcl &= \widehat{\epcl}\circ e_{cl}:\Gcalcl\rightarrow \mathbf{1}\\
\iotaclop &=L(r_o)\circ \widehat{\iotaclop}\circ e_{cl}:\Gcalcl\rightarrow L(\Gcalop)\\
\iotaclop^\dagger & = r_{cl}\circ \widehat{\iotaclop^\dagger}\circ L(e_o):L(\Gcalop)\rightarrow \Gcalcl.
}
The hatted morphisms are the maps appearing in the coupons for the string-nets. Since the graphical representation of fundamental correlators stays the same, the proofs in section \ref{subsec61} can be just run backwards giving the defining relations of a $(\Csf|\Zsf(\Csf))$-Cardy algebra for these morphisms. In order to show that this carries over to a Cardy algebra on $(\Gcalcl,\Gcalop, \iotaclop)$, we first note that relations R3), R8) and R10) give
\eq{\label{mopprop}
\hat{p}_{op}\circ \widehat{m_{op}}=\widehat{m_{op}}\circ(\hat{p}_{op}\otimes \id)=\widehat{m_{op}}\circ (\id\otimes \hat{p}_{op})=\widehat{m_{op}}\quad. 
}
Similar relations hold for $\widehat{\Delta_{op}}$, $\widehat{\eta_{op}}$, $\widehat{\epsilon_{op}}$, $\widehat{m_{cl}}$, $\widehat{\Delta_{cl}}$, $\widehat{\eta_{cl}}$ and $\widehat{\epsilon_{cl}}$. Furthermore sewing relation R29) and the preceding discussion yields 
\eq{
L(\hat{p}_{op})\circ \widehat{\iotaclop}=\widehat{\iotaclop}\circ \hat{p}_{cl}=\widehat{\iotaclop}\quad .
}
Using these relations one readily checks that $(\Gcalop,m_{op},\Delta_{op},\eta_{op},\epsilon_{op})$ is symmetric Frobenius algebra, $(\Gcalcl, m_{cl},\Delta_{cl},\eta_{cl},\epsilon_{cl})$ is a symmetric, commutative Frobenius algebra and $\iotaclop$ is an algebra homomorphism. To check modularity we introduce the graphical notation

\begin{center}
\begin{tikzpicture}
\node at (-5.5,0) {\includegraphics[scale=0.1]{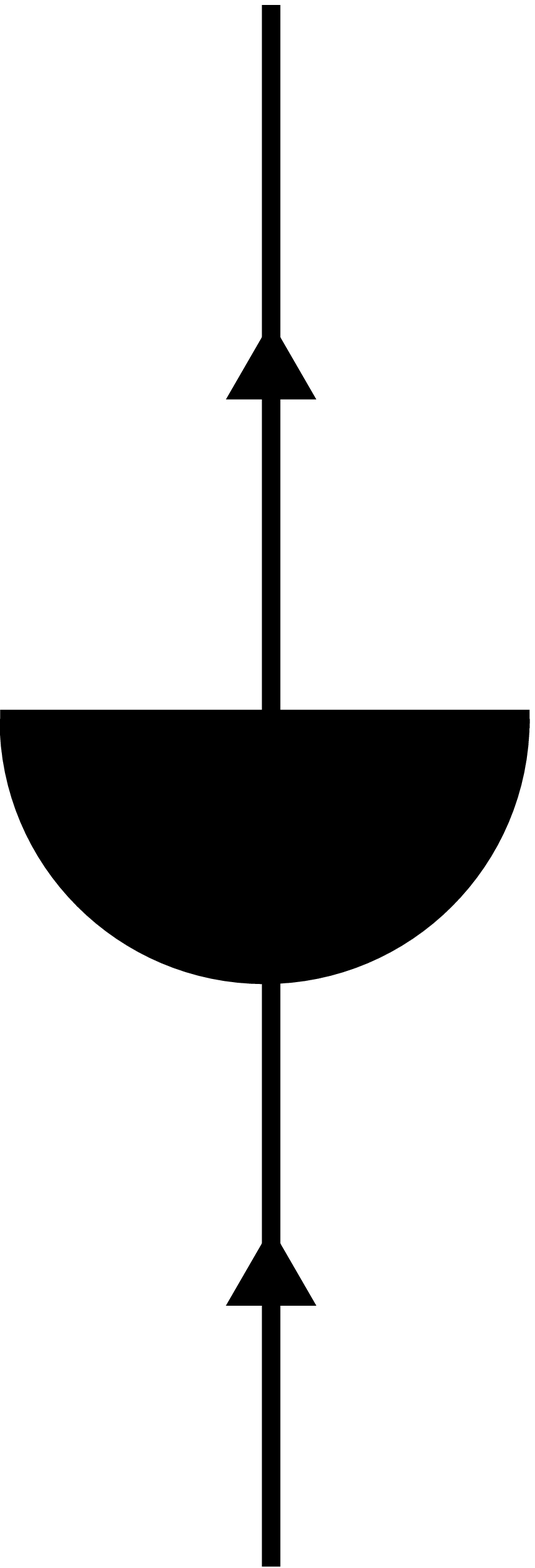}};
\node at (-6.7,0) {$e_O$};
\node at (-6.2,0) {$=$};
\node at (-5.1,-1) {\scriptsize $\Gcalop$};
\node at (-5.1,1) {\scriptsize $\widehat{\Gcalop}$};
\node at (-2.5,0) {\includegraphics[scale=0.1]{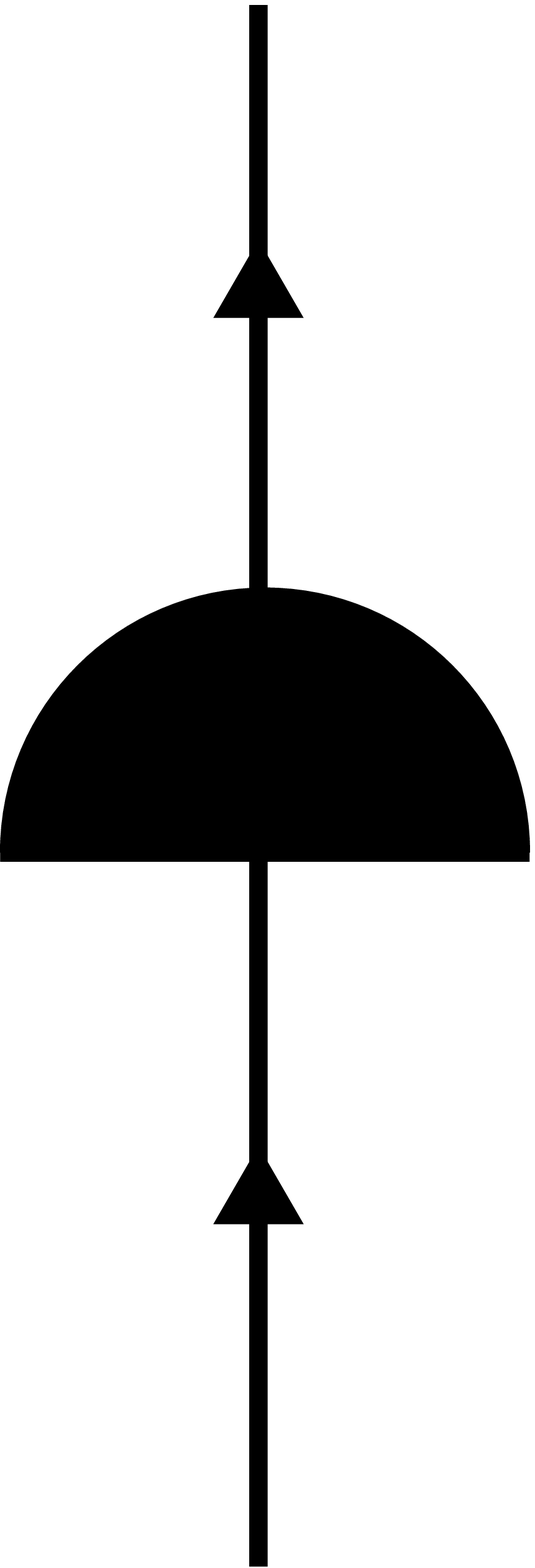}};
\node at (-2.1,-1) {\scriptsize $\widehat{\Gcalop}$};
\node at (-2.1,1) {\scriptsize $\Gcalop$};
\node at (-3.7,0) {$r_O$};
\node at (-3.2,0) {$=$};
\node at (2.5,0) {\includegraphics[scale=0.1]{figure128.eps}};
\node at (1.3,0) {$e_{cl}$};
\node at (1.8,0) {$=$};
\node at (2.9,-1) {\scriptsize $\Gcalcl$};
\node at (2.9,1) {\scriptsize $\widehat{\Gcalcl}$};
\node at (5,0) {\includegraphics[scale=0.1]{figure129.eps}};
\node at (3.8,0) {$r_{cl}$};
\node at (4.3,0) {$=$};
\node at (5.4,-1) {\scriptsize $\widehat{\Gcalcl}$};
\node at (5.4,1) {\scriptsize $\Gcalcl$};
\end{tikzpicture}
\end{center}

and note that similar to \cite[Lemma~4.4]{Kong_2014} we have

\begin{center}
\begin{tikzpicture}
\node at (-4,0) {\includegraphics[scale=0.15]{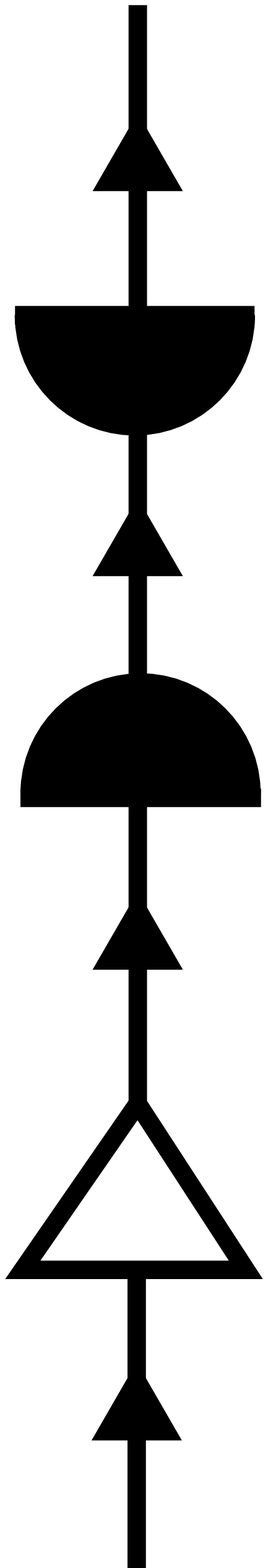}};
\node at (-2,0) {\includegraphics[scale=0.15,angle=180]{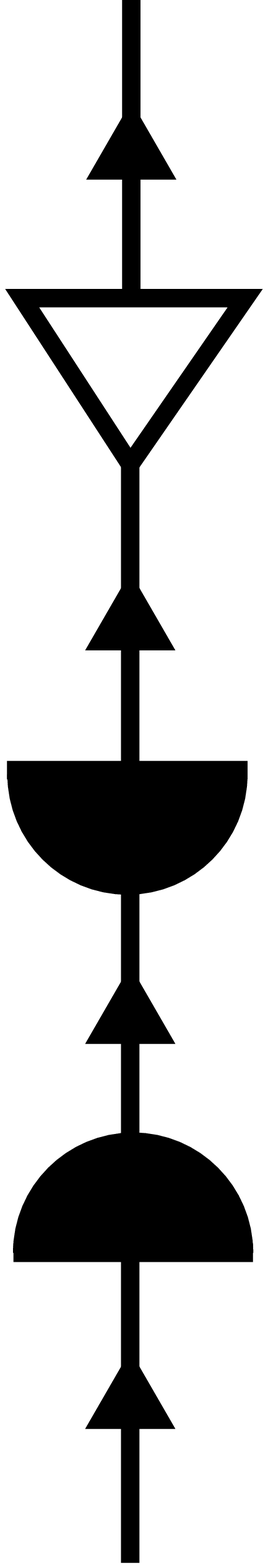}};
\node at (-5,0) {$\underset{\alpha}{\sum}$};
\node at (-4,-1.05) {\scriptsize $\alpha$};
\node at (-3.5,-1.7) {\scriptsize $i\otimes j$};
\node at (-3.6,0.5) {\scriptsize $\Gcalcl$};
\node at (-3.6,1.6) {\scriptsize $\widehat{\Gcalcl}$};
\node at (-1.6,-1.7) {\scriptsize $\widehat{\Gcalcl}$};
\node at (-1.6,-0.6) {\scriptsize $\Gcalcl$};
\node at (-2,1.05) {\scriptsize $\alpha$};
\node at (-1.6,1.6) {\scriptsize $i\otimes j$};
\node at (-0.5,0) {$=$};
\node at (2,0) {\includegraphics[scale=0.15]{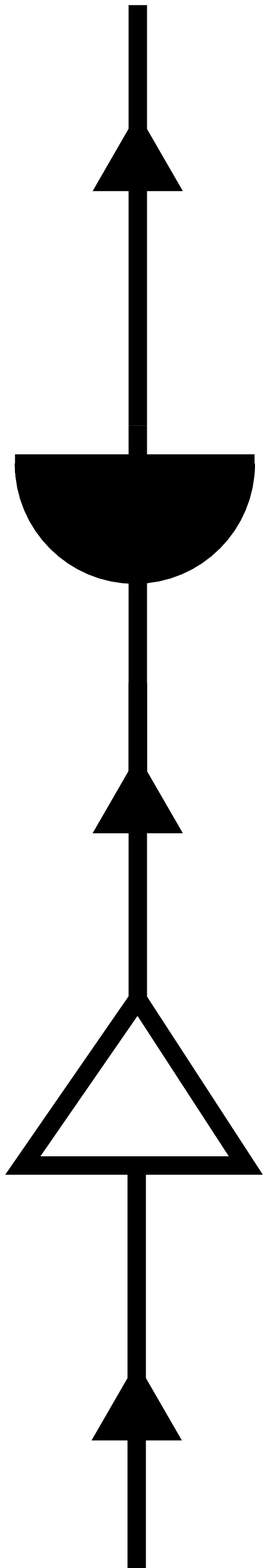}};
\node at (4,0) {\includegraphics[scale=0.15,angle=180]{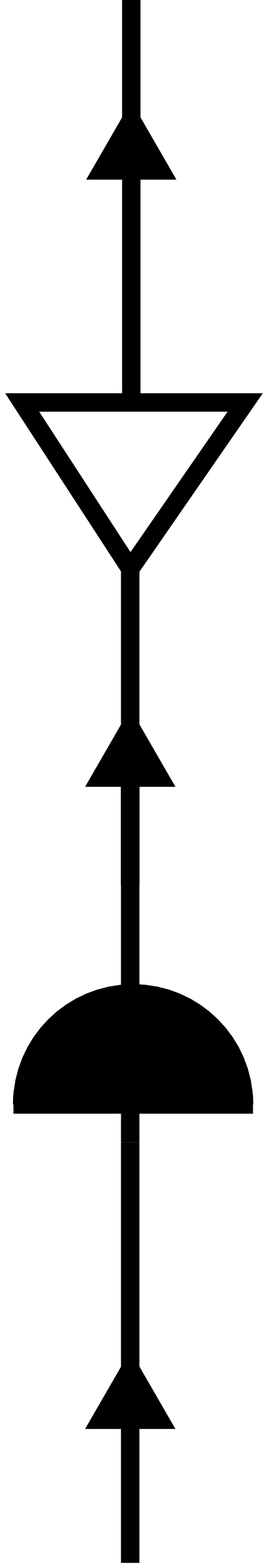}};
\node at (1,0) {$\underset{\beta}{\sum}$};
\node at (2.5,-1.7) {\scriptsize $i\otimes j$};
\node at (2.4,-0.1) {\scriptsize $\Gcalcl$};
\node at (2.4,1.6) {\scriptsize $\widehat{\Gcalcl}$};
\node at (2,-0.8) {\scriptsize $\beta$};
\node at (4.4,-1.7) {\scriptsize $\widehat{\Gcalcl}$};
\node at (4.4,-0.1) {\scriptsize $\Gcalcl$};
\node at (4.5,1.6) {\scriptsize $i\otimes j$};
\node at (4,0.75) {\scriptsize $\beta$}; 
\end{tikzpicture}
\end{center}

Using this we compute 

\begin{center}
\begin{tikzpicture}
\node at (0,0) {\includegraphics[scale=0.2]{figure52.eps}};
\node at (2.3,0) {$=$};
\node at (-2,-1.5) {\scriptsize $\Gcalcl$};
\node at (1,-1.5) {\scriptsize $i\otimes j$};
\node at (-2,0) {$\begin{aligned}\frac{d_id_j}{\Dsf^2}\end{aligned}$};
\node at (5,0) {\includegraphics[scale=0.2]{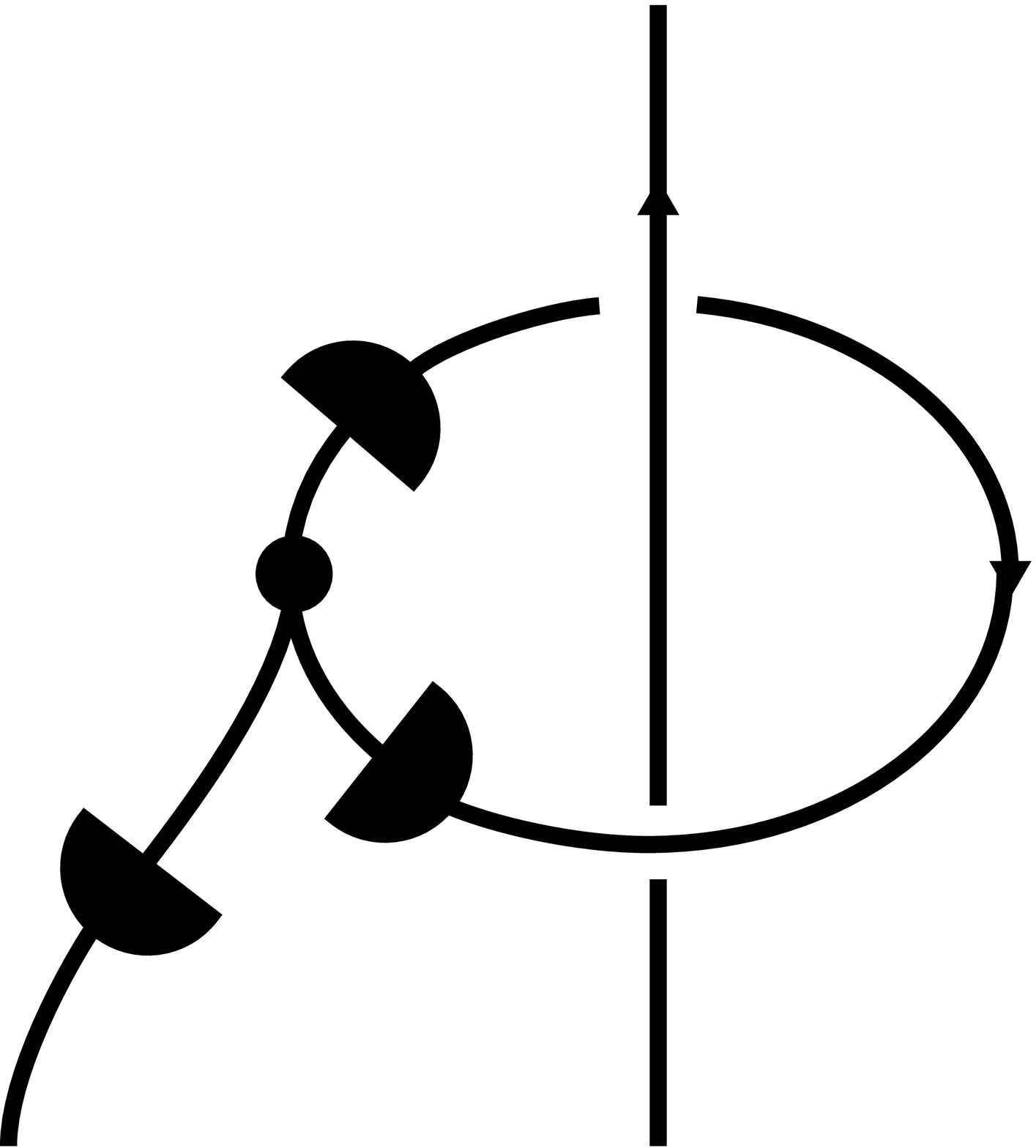}};
\node at (7.3,0) {$=$};
\node at (3,-1.5) {\scriptsize $\Gcalcl$};
\node at (6,-1.5) {\scriptsize $i\otimes j$};
\node at (3.7,-0.4) {\scriptsize $\widehat{\Gcalcl}$};
\node at (3,0) {$\begin{aligned}\frac{d_id_j}{\Dsf^2}\end{aligned}$};
\node at (10,0) {\includegraphics[scale=0.2]{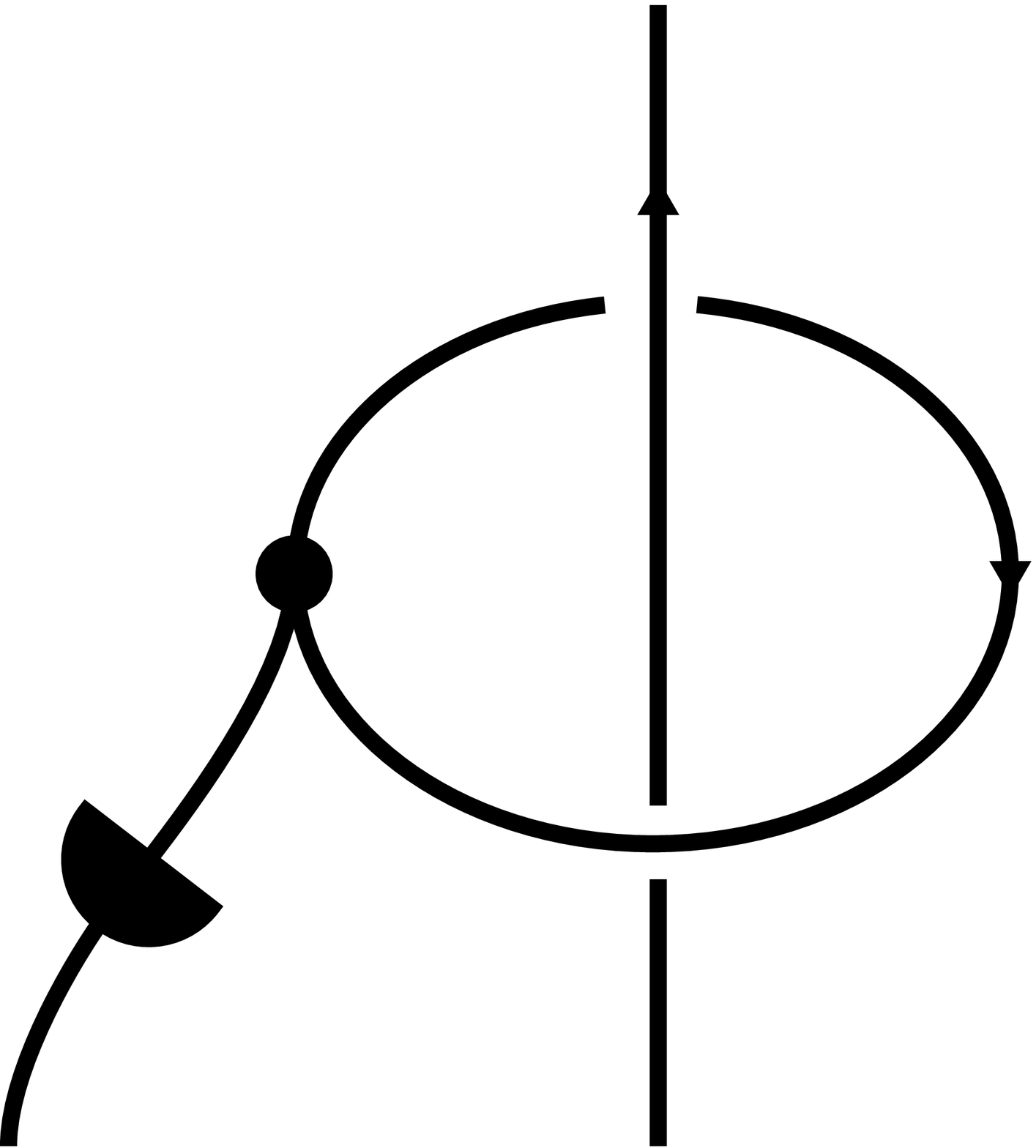}};
\node at (8,-1.5) {\scriptsize $\Gcalcl$};
\node at (11,-1.5) {\scriptsize $i\otimes j$};
\node at (8.7,-0.4) {\scriptsize $\widehat{\Gcalcl}$};
\node at (9.7,-0.4) {\scriptsize $\widehat{\Gcalcl}$};
\node at (8,0) {$\begin{aligned}\frac{d_id_j}{\Dsf^2}\end{aligned}$};
\node at (0,-5) {\includegraphics[scale=0.2]{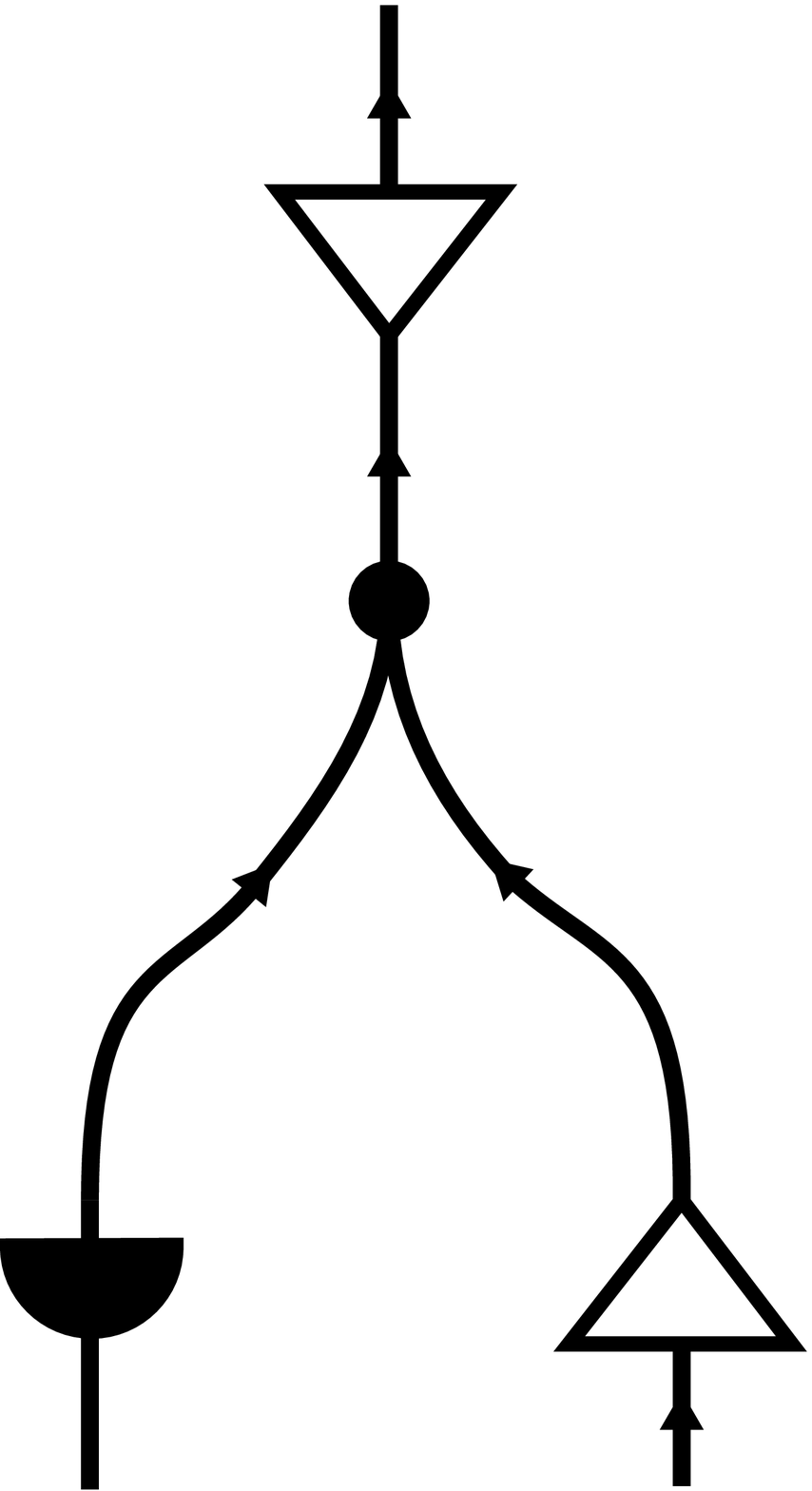}};
\node at (-2.5,-5) {$=$};
\node at (-1.5,-7.2) {\scriptsize $\Gcalcl$};
\node at (1.4,-7.2) {\scriptsize $i\otimes j $};
\node at (1,-6.85) {\scriptsize $\alpha$};
\node at (-1,-5.2) {\scriptsize $\widehat{\Gcalcl}$};
\node at (0.4,-2.7) {\scriptsize $i\otimes j$};
\node at (0,-3.28) {\scriptsize $\alpha$};
\node at (-2,-5) {$\begin{aligned}\sum_{\alpha} \end{aligned}$};
\node at (1.2,-6.1) {\scriptsize $\widehat{\Gcalcl}$};
\node at (5,-5) {\includegraphics[scale=0.135]{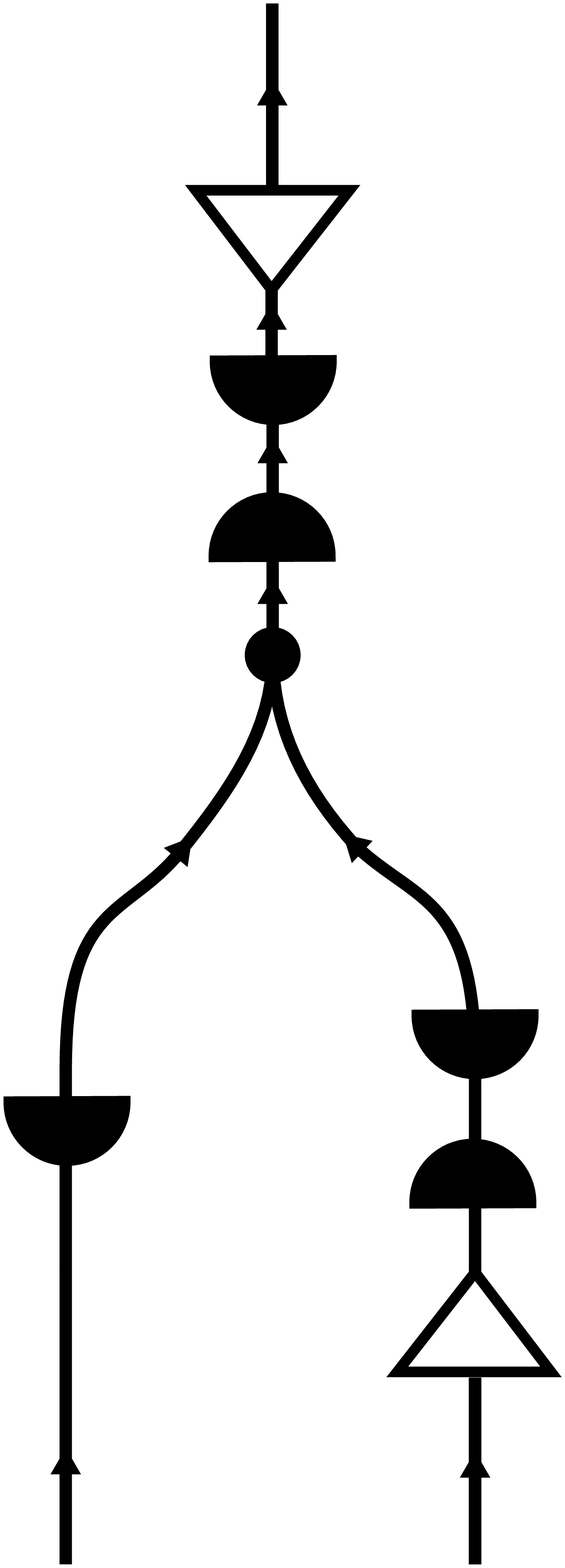}};
\node at (2.5,-5) {$=$};
\node at (3.9,-7.2) {\scriptsize $\Gcalcl$};
\node at (6.2,-7.3) {\scriptsize $i\otimes j $};
\node at (5.65,-6.8) {\scriptsize $\alpha$};
\node at (4,-5.1) {\scriptsize $\widehat{\Gcalcl}$};
\node at (5.4,-2.7) {\scriptsize $i\otimes j$};
\node at (5,-3.15) {\scriptsize $\alpha$};
\node at (3,-5) {$\begin{aligned}\sum_{\alpha} \end{aligned}$};
\node at (5.9,-5.3) {\scriptsize $\widehat{\Gcalcl}$};
\node at (10,-5) {\includegraphics[scale=0.135]{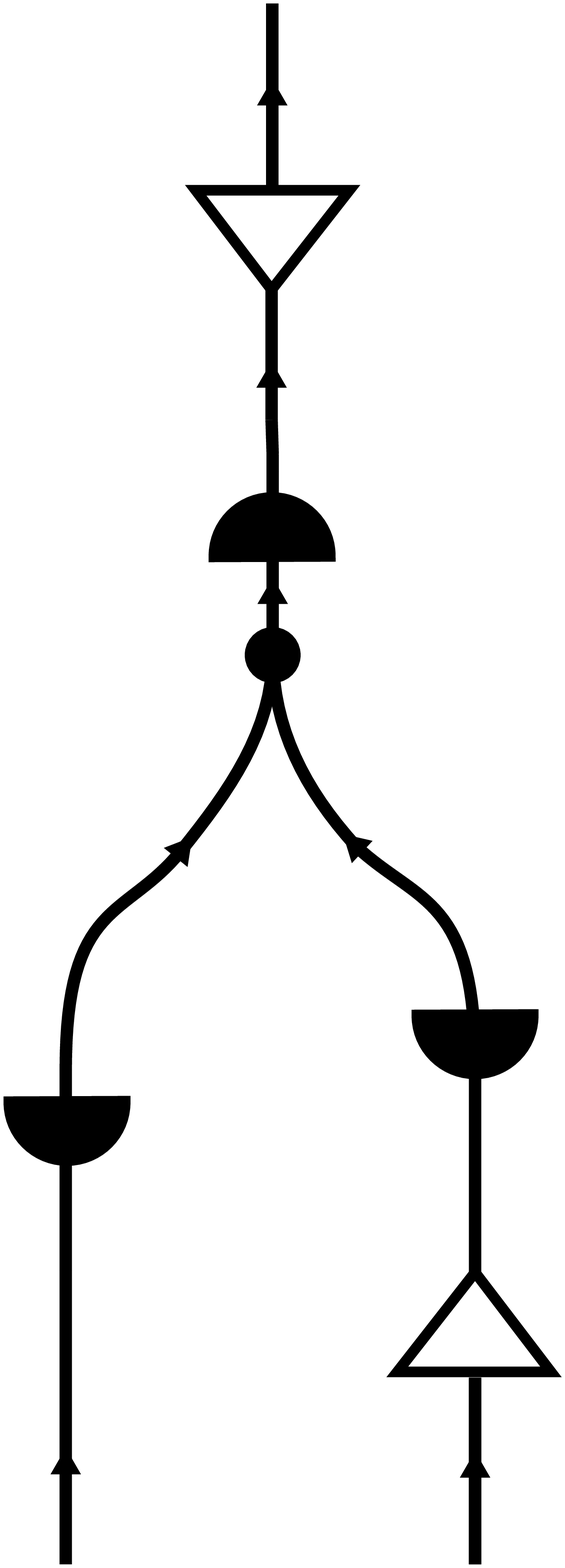}};
\node at (7.5,-5) {$=$};
\node at (9,-7.2) {\scriptsize $\Gcalcl$};
\node at (11.1,-7.3) {\scriptsize $i\otimes j $};
\node at (10.65,-6.8) {\tiny $\beta$};
\node at (9.3,-4.9) {\scriptsize $\widehat{\Gcalcl}$};
\node at (10.4,-2.7) {\scriptsize $i\otimes j$};
\node at (10,-3.15) {\tiny $\beta$};
\node at (8,-5) {$\begin{aligned}\sum_{\beta} \end{aligned}$};
\node at (10.6,-5.1) {\scriptsize $\widehat{\Gcalcl}$};
\node at (0,-11) {\includegraphics[scale=0.2]{figure53.eps}};
\node at (-2.5,-11) {$=$};
\node at (-1.4,-12.8) {\scriptsize $\Gcalcl$};
\node at (1.2,-13.5) {\scriptsize $i\otimes j $};
\node at (0.8,-12.85) {\scriptsize $\beta$};
\node at (0.2,-10.1) {\scriptsize $\Gcalcl$};
\node at (0.4,-8.8) {$i\otimes j$};
\node at (-0.15,-9.3) {\scriptsize $\beta$};
\node at (-2,-11) {$\begin{aligned}\sum_{\beta} \end{aligned}$};
\node at (1.1,-12.1) {\scriptsize $\Gcalcl$};
\end{tikzpicture}
\end{center}
which shows modularity.

The center and Cardy condition follow by similar computations. 
For the uniqueness part suppose we have chosen another set of retracts $(\Gcalcl^\prime,e_{cl}^\prime,r_{cl}^\prime)$ and $(\Gcalop^\prime,e^\prime_o,r^\prime_o)$, then it is easy to see that $f_o=r_o\circ e_o^\prime:\Gcalop^\prime\rightarrow \Gcalop$ and $f_{cl}=r_{cl}\circ e_{cl}^\prime:\Gcalcl^\prime\rightarrow \Gcalcl$ are isomorphisms of Frobenius algebras and in addition the diagram
\begin{center}
\begin{tikzcd}
\Gcalcl^\prime \ar[rr,"f_{cl}"] \ar[dd,"\iotaclop^\prime"']&  & \Gcalcl \ar[dd, "\iotaclop"]\\
& &  \\
L(\Gcalop^\prime)\ar[rr, "L(f_o)"'] & &  L(\Gcalop)
\end{tikzcd}
\end{center}
commutes.
\end{proof}

\section{Conclusion}
In this paper we have shown how string-nets on topological surfaces generate solutions to open-closed sewing relations. The major advantage of string-nets is the transportation of categorical graphical calculus onto surfaces, which allows to use the defining conditions for Cardy algebras directly when solving the sewing constraints. There are some open ends related to this work. First of all, as noted in \cite{Schweigert:2019zwt} one could further generalize the results including defects. This seems likely to be possible using the description of defect world sheets given in \cite{Fjelstad:2012mj}. Furthermore the qualifyer "rational" may be given up, leading to a more general notion of modular tensor categories, which are not fusion. As shown in \cite{Fuchs:2016wjr}\cite{Fuchs:2017unc} many of the categorical description can be transported to this situation by replacing sums over simple objects by coends. Since dragging curves along projector circles was the crucial point in manipulating string-nets for fusion categories there should be an appropriate procedure for string-nets with non-fusion colorings. 

Constructions of open-closed interactions using curves on surfaces have appeared in \cite{kaufmann2003arc}\cite{kaufmann2006closed} in the form of the $\mathsf{Arc}$-operad. Since the graphical representation of the construction very much resembles string-nets, there may be a connection between the two approaches. In general one may wonder about a (wheeled) PROP-description of string-nets, since null graphs give a pasting scheme for string-net diagrams. We plan to address some of these questions in future work.

\appendix

\section{Generating World Sheets}\label{appendix}
In this appendix we give all the generating world sheets in $\WSsf$. The following figures display the quotients of the orientation double for generating world sheets.
\begin{enumerate}[label=\Roman*)]
\item \textbf{Open World Sheets:}

\begin{center}
\begin{tikzpicture}
\node (opprop) at (0,0) {\includegraphics[scale=0.12]{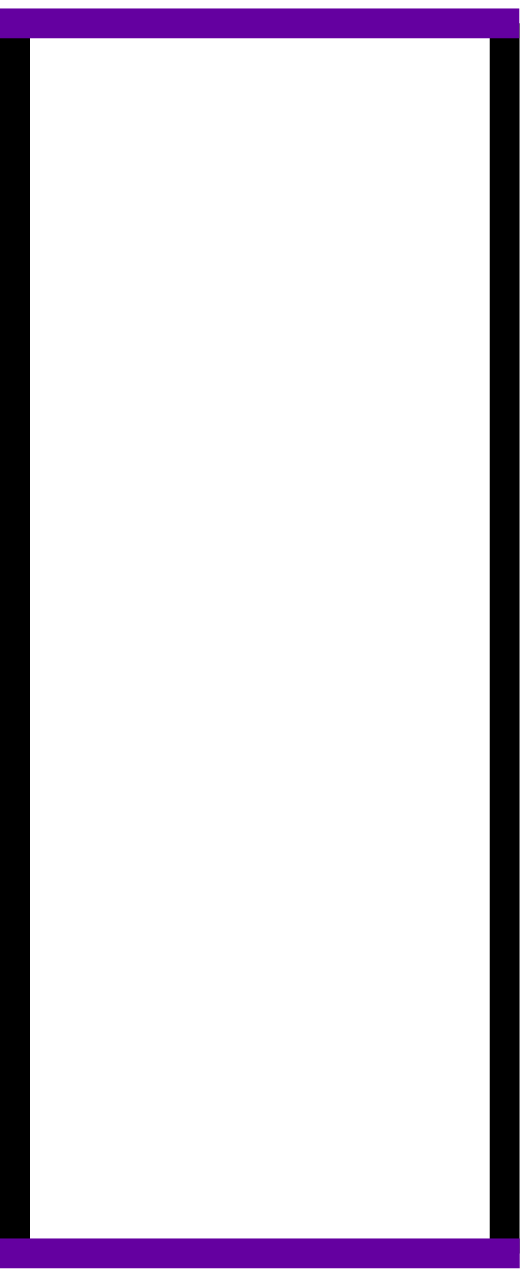}};
\node (mop) at (3,0) {\includegraphics[scale=0.12]{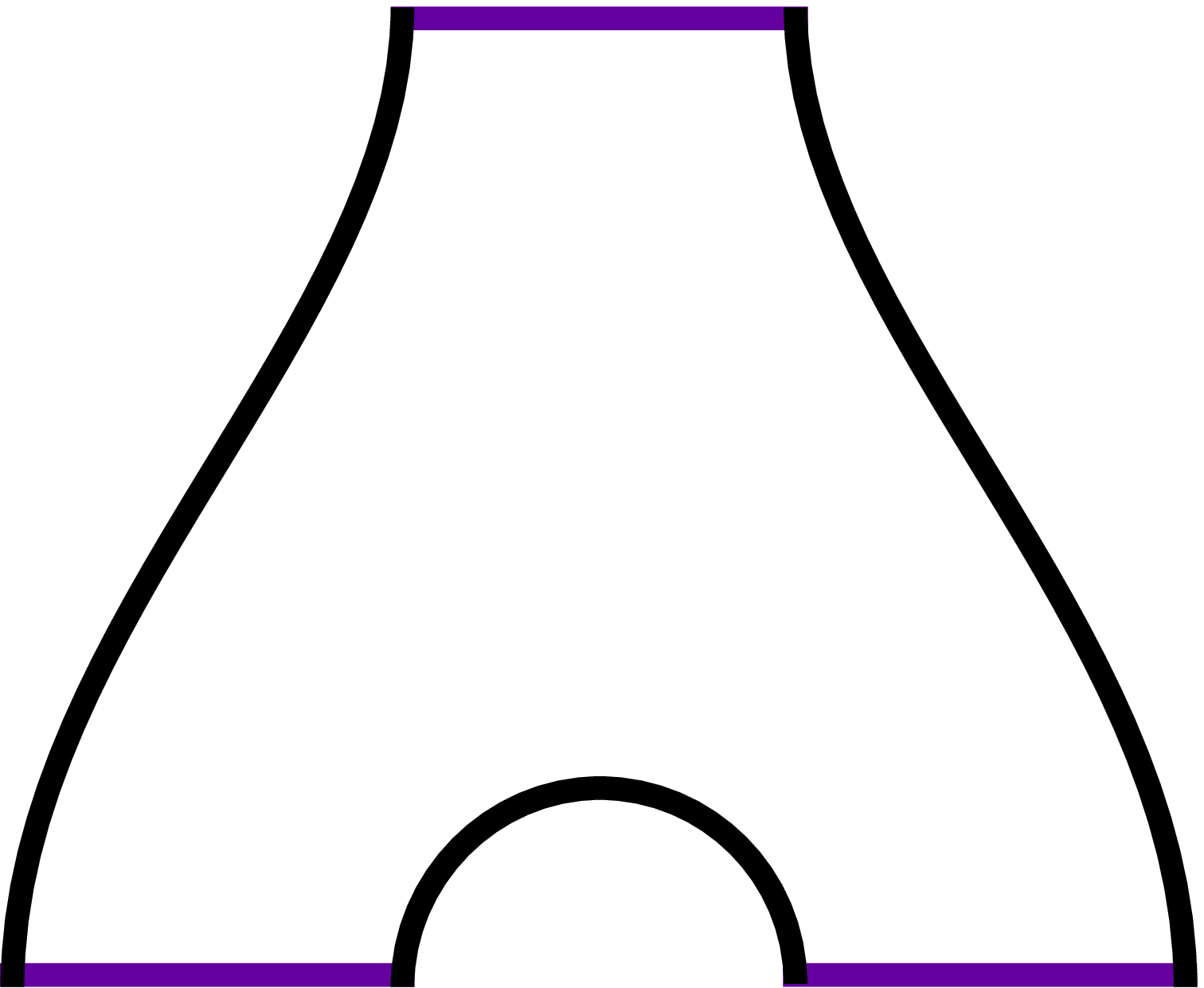}};
\node (Deltaop) at (6,0) {\includegraphics[scale=0.12, angle=180]{figure140.eps}};
\node (unitop) at (9,0) {\includegraphics[scale=0.12]{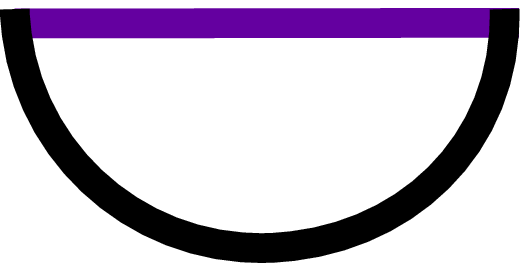}};
\node (counit) at (12,0) {\includegraphics[scale=0.12, angle=180]{figure141.eps}};
\node (Oprop) at (0,-1.2) {$O_{prop}$};
\node (Om) at (3,-1.2) {$O_m$};
\node (Odelta) at (6,-1.2) {$O_\Delta$};
\node (Ounit) at (9,-1.2) {$O_\eta$};
\node (Ocounit) at (12,-1.2) {$O_\epsilon$};
\end{tikzpicture}
\end{center}

Purple colored parts of the boundary correspond to open boundaries. Black boundaries are physical boundaries.
\item \textbf{Closed World Sheets:}

\begin{center}
\begin{tikzpicture}
\node (opprop) at (0,0) {\includegraphics[scale=0.12]{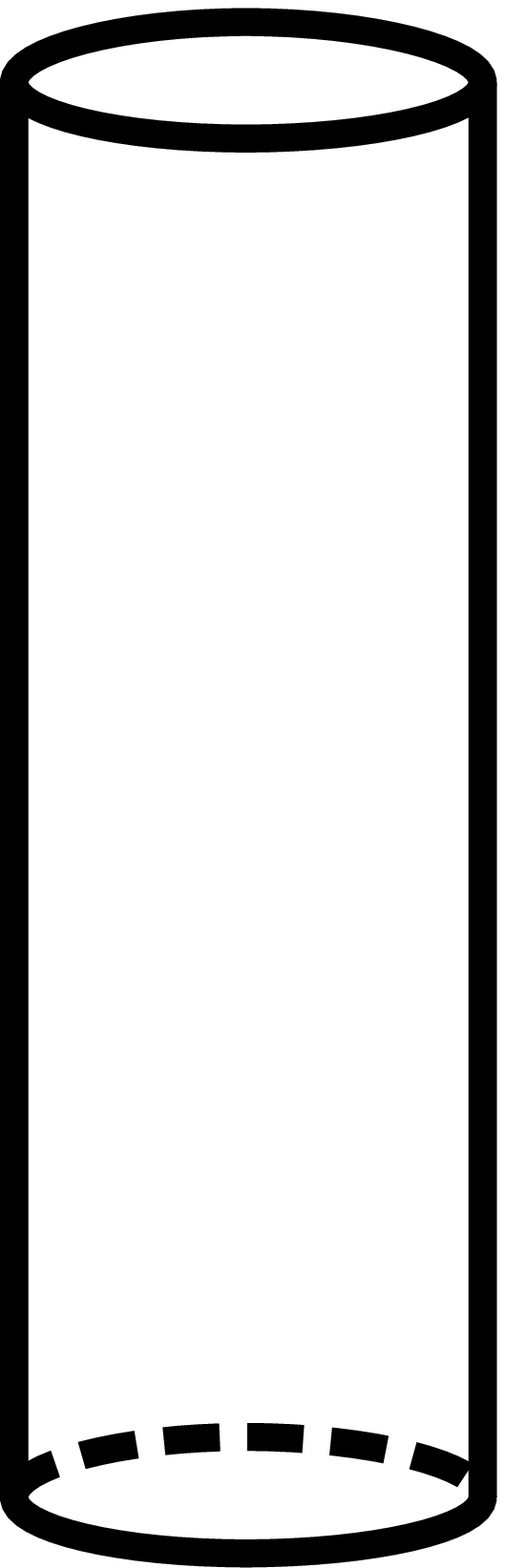}};
\node (mop) at (3,0) {\includegraphics[scale=0.12]{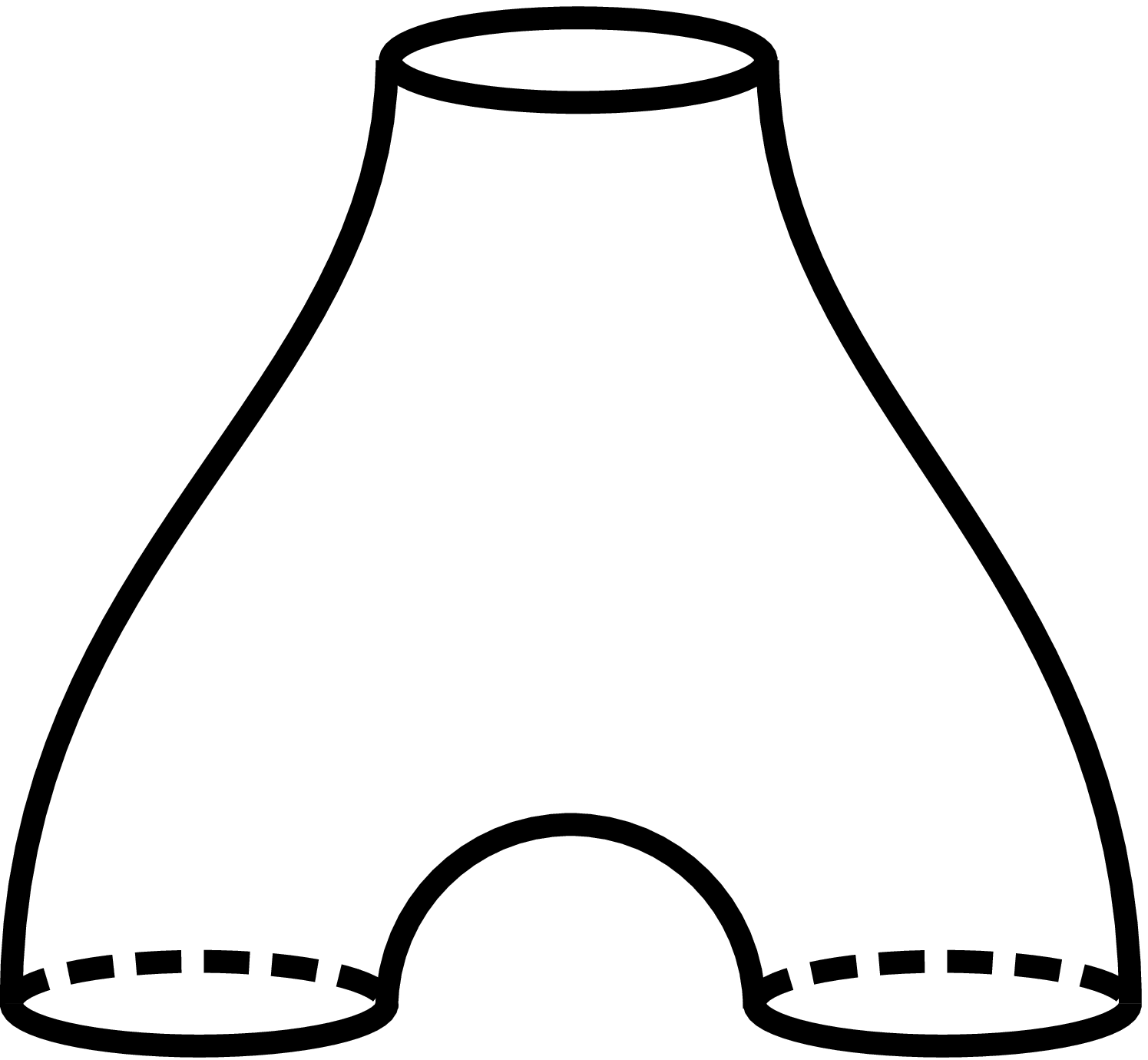}};
\node (Deltaop) at (6,0) {\includegraphics[scale=0.12, angle=180]{figure143.eps}};
\node (unitop) at (9,0) {\includegraphics[scale=0.12]{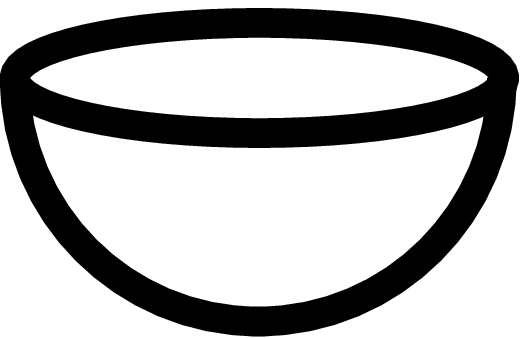}};
\node (counit) at (12,0) {\includegraphics[scale=0.12, angle=180]{figure144.eps}};
\node (Oprop) at (0,-1.3) {$C_{prop}$};
\node (Om) at (3,-1.3) {$C_m$};
\node (Odelta) at (6,-1.3) {$C_\Delta$};
\node (Ounit) at (9,-1.3) {$C_\eta$};
\node (Ocounit) at (12,-1.3) {$C_\epsilon$};
\end{tikzpicture}
\end{center}

\item \textbf{Open-Closed World Sheets:}
\begin{center}
\begin{tikzpicture}
\node (iota) at (0,0) {\includegraphics[scale=0.12]{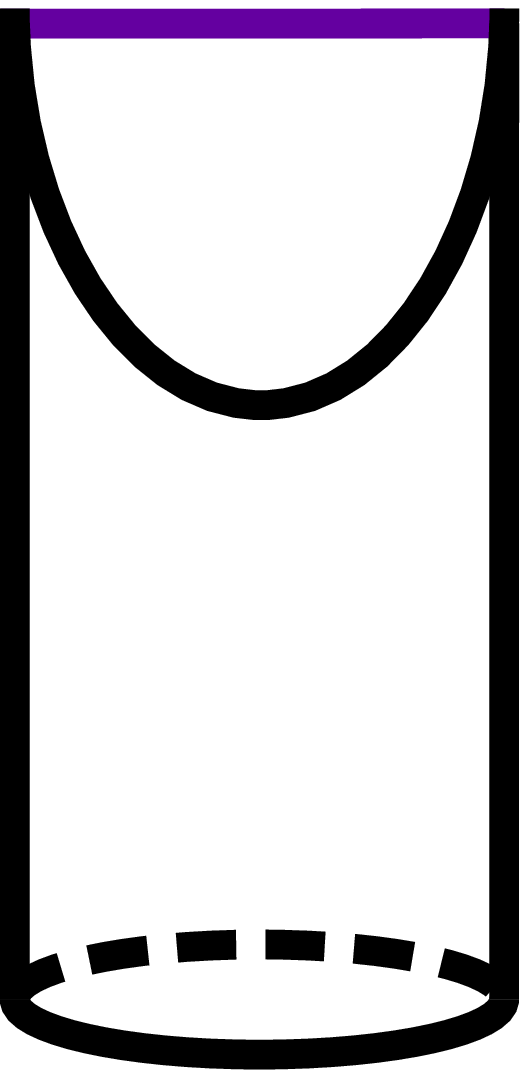}};
\node (iotadagger) at (3,0) {\includegraphics[scale=0.12, angle=180]{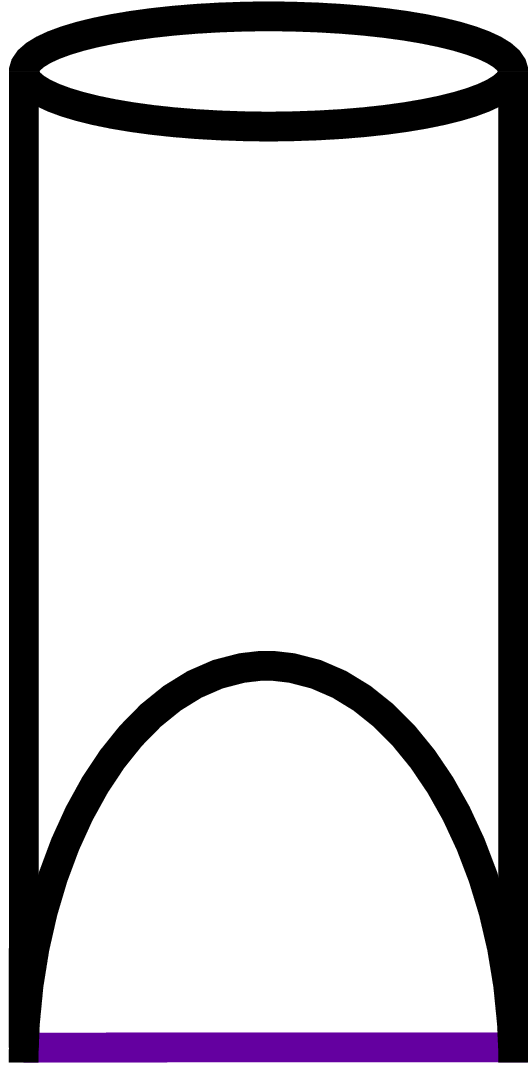}};
\node (I) at (0,-1.2) {$I$};
\node (Idagger) at (3,-1.2) {$I^\dagger$};
\end{tikzpicture}
\end{center}
\end{enumerate}

\vspace*{0.5cm}
\textbf{Acknowledgement:} The author thanks Ralph Blumenhagen, Ilka Brunner and Ingmar Saberi for valuable discussions and Ralph Blumenhagen for having a look on a first draft of this paper.

\bibliographystyle{alpha}
\bibliography{references}
\end{document}